\documentclass[a4paper,11pt]{article}
\usepackage{jheppub} 
\usepackage{lineno}
\usepackage[dvipsnames]{xcolor} 
\usepackage{dirtytalk}
\usepackage{ytableau}
\usepackage[dvipsnames]{xcolor}
\usepackage{dsfont}
\usepackage{mathdots}
\usepackage{colortbl}
\usepackage{braket}
\usepackage{amsthm}
\usepackage{amsmath}
\usepackage{float}
\usepackage{bbm, dsfont}
\usepackage[normalem]{ulem}
\usepackage{slashed}
\usepackage{mathtools}
\usepackage{multirow}
\usepackage{cancel}
\usepackage{cleveref}
\usepackage{orcidlink}
\usepackage{enumerate}
\usepackage{fontawesome}
\usepackage{adjustbox}
\usepackage{tcolorbox}
\usepackage{relsize}
\usepackage{csquotes}
\usepackage{enumitem}
\usepackage{tikz-cd}
\usepackage{tikz}
    \usetikzlibrary{positioning}
    \usetikzlibrary{arrows}
    \usetikzlibrary{shapes}
    \usepackage{pgfplots}
    \usetikzlibrary{calc}
    \usetikzlibrary{decorations.markings}
     \usetikzlibrary{decorations.pathmorphing}
    \tikzset{snake it/.style={decorate, decoration=snake}}

    \usetikzlibrary{matrix}

    \pgfplotsset{compat=1.18}

    \renewcommand{\ker}{\operatorname{Ker}}

\NewDocumentEnvironment{examplenew}{O{} m}
{
    \ifnum#2=1
    \setcounter{examplenew}{1}%
    \refstepcounter{subexampleone}%
    \textbf{Example \thesubexampleone}%
  \else\ifnum#2=2
    \setcounter{examplenew}{2}%
    \refstepcounter{subexampletwo}%
    \textbf{Example \thesubexampletwo}%
  \else\ifnum#2=3
    \setcounter{examplenew}{3}
    \refstepcounter{subexamplethree}%
    \textbf{Example \thesubexamplethree}%
     \else\ifnum#2=4
    \setcounter{examplenew}{4}
    \refstepcounter{subexamplefour}%
    \textbf{Example \thesubexamplefour}%
  \else
    \PackageError{examplenew}{Second argument must be 1..3}{}
  \fi\fi\fi\fi
  \ifx&#1&\else ~(#1)~\fi
}
{%
  \par \medskip \noindent
}
\newcounter{examplenew} 
\newcounter{subexampleone}[examplenew]
\newcounter{subexampletwo}[examplenew]
\newcounter{subexamplethree}[examplenew]
\newcounter{subexamplefour}[examplenew]
\renewcommand{\thesubexampleone}{\arabic{examplenew}.\arabic{subexampleone}}
\renewcommand{\thesubexampletwo}{\arabic{examplenew}.\arabic{subexampletwo}}
\renewcommand{\thesubexamplethree}{\arabic{examplenew}.\arabic{subexamplethree}}
\renewcommand{\thesubexamplefour}{\arabic{examplenew}.\arabic{subexamplefour}}

\definecolor{ourblue}{RGB}{0, 57, 120} %
\colorlet{ourlightblue}{ourblue!50!white}
\newcommand{\bs}[1]{\boldsymbol{#1}}
\newcommand{\eps}{\varepsilon}
\newcommand{\rd}{\mathrm{d}}

\DeclareMathOperator{\Ima}{Im}

\newcommand{\Gt}{\tilde{G}}
\newcommand{\Et}{\tilde{E}}

\newcommand{\bd}{\bs{d}}
\newcommand{\bk}{\bs{k}}
\newcommand{\bp}{\bs{p}}
\newcommand{\bq}{\bs{q}}

\newcommand{\bv}{\bs{v}}
\newcommand{\bx}{\bs{x}}
\newcommand{\bz}{\bs{z}}
\newcommand{\bA}{\bs{A}}
\newcommand{\bB}{\bs{B}}
\newcommand{\bC}{\bs{C}}
\newcommand{\bD}{\bs{D}}
\newcommand{\bE}{\bs{E}}

\newcommand{\bG}{\bs{G}}
\newcommand{\bH}{\bs{H}}
\newcommand{\bI}{\bs{I}}
\newcommand{\bJ}{\bs{J}}

\newcommand{\bL}{\bs{L}}
\newcommand{\bM}{\bs{M}}
\newcommand{\bN}{\bs{N}}
\newcommand{\bO}{\bs{O}}
\newcommand{\bP}{\bs{P}}
\newcommand{\bQ}{\bs{Q}}
\newcommand{\bR}{\bs{R}}
\newcommand{\bT}{\bs{T}}
\newcommand{\bS}{\bs{S}}
\newcommand{\bU}{\bs{U}}
\newcommand{\bV}{\bs{V}}
\newcommand{\bW}{\bs{W}}

\newcommand{\bZ}{\bs{Z}}
\newcommand{\bkappa}{\bs{\kappa}}
\newcommand{\bnu}{\bs{\nu}}
\newcommand{\cF}{\mathcal{F}}
\newcommand{\cG}{\mathcal{G}}
\newcommand{\cH}{\mathcal{H}}

\newcommand{\cU}{\mathcal{U}}

\newcommand{\bDelta}{\bs{\Delta}}
\newcommand{\bGamma}{\bs{\Gamma}}
\newcommand{\bSigma}{\bs{\Sigma}}
\newcommand{\bOmega}{\bs{\Omega}}
\newcommand{\bPi}{\bs{\Pi}}
\newcommand{\Gp}{\bG(\bp)}

\newcommand{\rint}{\mathrm{int}}
\newcommand{\rext}{\mathrm{ext}}
\newtheorem{proposition}{Proposition}

\newtheorem{theorem}{Theorem}
\newtheorem{definition}{Definition}
\newtheorem{lemma}{Lemma}
\newtheorem*{lemma*}{Lemma}

\theoremstyle{definition}
\newtheorem{example}{Example}
\newtheorem*{example*}{Example}

\def\beq{\begin{equation}}
\def\eeq{\end{equation}}
\def\bsp#1\esp{\begin{split}#1\end{split}}
\newcommand{\cV}{\mathcal{V}}
\DeclareMathOperator{\Aut}{Aut}
\DeclareMathOperator{\Sym}{Sym}
\DeclareMathOperator{\TSym}{TSym}
\DeclareMathOperator{\Tr}{Tr}
\DeclareMathOperator{\id}{id}
\DeclareMathOperator{\sign}{sign}
\DeclareMathOperator{\im}{Im}

\renewcommand{\setminus}{\backslash}

\author{Claude Duhr, Sara Maggio,  Cathrin Semper, Sven F. Stawinski}
\affiliation{
\vskip 0.5 em
Bethe Center for Theoretical Physics, Universit\"at Bonn, D-53115, Germany
\vskip 0.5 em}

\emailAdd{cduhr@uni-bonn.de, smaggio@uni-bonn.de, csemper@uni-bonn.de,
sstawins@uni-bonn.de}

\title{Discrete symmetries of Feynman integrals}
\abstract{We perform a comprehensive study of a certain class of discrete symmetries of families of Feynman integrals, defined as affine changes of variables that map different sectors of the family into each other. We show that these transformations are always encoded into permutations of the Feynman parameters that relate the Lee-Pomeransky polynomials of the two sectors, irrespective of the integral representation used to define the Feynman integrals. We then construct an affine map in loop-momentum space that encodes such a permutation. We also show that these symmetries can be naturally embedded into the framework of twisted cohomology theories, and the period and intersection parings are invariant under the symmetry transformations. If we focus on symmetries within a fixed sector, we obtain a group acting on the twisted cohomology group, and we study the decomposition of this action into irreducible representations. One of our main mathematical results is that the character of this representation is proportional to the Euler characteristic of the corresponding fixed-point set. We then study the implications for Feynman integrals, in particular for the intersection matrix in a canonical basis. We also present a formula for the number of master integrals in a given sector in the presence of a non-trivial symmetry group in terms of the Euler characteristics of fixed-point sets. As an application, we obtain the numbers of master integrals for banana integrals with up to four loops for arbitrary configurations of non-zero masses. In order to achieve our results, we had to combine tools from various different areas of mathematics, including graph theory, group theory and algebraic topology.}

\begin{document}

\begin{flushright}
    BONN-TH/2026-08
    \end{flushright}

\maketitle


\section{Introduction}
\label{sec:introduction}
Symmetries play a fundamental role in many areas of modern theoretical physics. For example, they are crucial in the construction of interacting gauge theories and Quantum Field Theories (QFTs), including the Standard Model of particle physics. They also put strong constraints on the structure of the $S$-matrix and correlation functions of a QFT. In some cases, the symmetries are so strong that they allow one to fully determine, or at least constrain, the exact results. The arguably most prominent examples of this are conformal and integrable QFTs, where symmetries can be exploited to fix certain observables. Even if the symmetries do not completely determine the final result, having a good understanding of all symmetries can substantially simplify the problem and it is often the first step towards devising new and efficient techniques for computations.

In this paper we focus on a certain class of discrete symmetries that arise in the computation of higher-order perturbative corrections to scattering amplitudes and correlation functions. The perturbative expansion is conveniently organized into Feynman diagrams, where the order of the expansion is directly linked to the number of loops of the graphs. A Feynman graph with $L$ independent loops gives rise to a Feynman integral where we integrate over $L$ momenta not fixed by momentum conservation. It is well known that Feynman integrals may exhibit symmetries that are not necessarily manifest from the action defining a QFT. For example, some Feynman integrals exhibit hidden or dual conformal symmetries~\cite{Drummond:2006rz,Loebbert:2020hxk}, which can sometimes even be enlarged to a Yangian symmetry~\cite{Chicherin:2017frs,Loebbert:2019vcj,Corcoran:2021gda,Kazakov:2023nyu,Loebbert:2024qbw,Loebbert:2025abz}. Feynman integrals also inherit symmetries from the associated Newton polytope~\cite{delaCruz:2024ssb}. Recently, also a new, mysterious, antipodal symmetry of certain classes of Feynman integrals and amplitudes was discovered~\cite{Arkani-Hamed:2017ahv,Dixon:2021tdw,Dixon:2022xqh,Dixon:2023kop,Dixon:2025zwj}.

The main focus of our paper is another class of symmetries for Feynman integrals, defined as certain affine changes of variables that permute the singular loci of the integrand, i.e., the propagators. The motivation to study such transformations comes from the fact that Feynman integrals are naturally organized into families of integrals that only differ by the exponents of the propagators. It is well known that using integration-by-parts (IBP) identities~\cite{Tkachov:1981wb,Chetyrkin:1981qh}, every member of a family can be written as linear combination of so-called master integrals, and the set of master integrals is known to be finite~\cite{smirnov2010numbermasterintegralsfinite}. There is a natural filtration by \emph{sectors}, coming from the set of propagators that are raised to positive powers. Understanding the symmetries of the family allows one to relate integrals from different sectors to each other~\cite{Pak:2011xt,Lee:2012cn,Wu:2024paw}. This may speed up the solution of the IBP system, and various public tools for IBP reduction try to exploit symmetries~\cite{Smirnov:2008iw,Smirnov:2014hma,Smirnov:2019qkx,Smirnov:2023yhb,Smirnov:2025prc,Lee:2012cn,Lee:2013mka,Maierhofer:2017gsa,Klappert:2020nbg,Lange:2025ofh,Wu:2025aeg}. These symmetries, however, are often determined in a heuristic manner, and their structure is still relatively poorly understood. For example, it was recently observed that there are symmetries between sectors that require kinematics-dependent transformations~\cite{Wu:2024paw}, and these symmetries had not been considered by IBP reduction codes so far. Moreover, symmetries may reduce the number of master integrals within a sector. While it is known that in the absence of symmetries the number of master integrals in a sector can often be computed as a certain Euler characteristic~\cite{Bitoun:2017nre,bitoun2018numbermasterintegralseuler,Mizera_2018} or by counting numbers of critical points~\cite{Lee:2013hzt}, these formul\ae\ are known to often overcount the number of master integrals in the presence of symmetries.

In our paper, we provide a detailed and comprehensive study of these classes of discrete symmetries. In order to achieve this, we combine a large variety of mathematical tools from graph theory, twisted cohomology theory, group theory and algebraic topology. In the first part of our paper, we address the question of how to characterize these symmetries in loop-momentum space, including the recently discovered cases where the transformations are kinematics-dependent~\cite{Wu:2024paw}. Using tools from graph theory, we show that these symmetries can be succinctly described as the permutations of the Feynman parameters that map the Lee-Pomeransky polynomials of the two sectors into each other. Via matroid theory, these transformations can in turn be described as operations on the underlying Feynman graphs. Our construction also clarifies the role of the kinematics-dependent transformations discussed in ref.~\cite{Wu:2024paw}.

In the second part we study these symmetries from a purely mathematical standpoint. It was realized a couple of years ago that an appropriate mathematical framework to study dimensionally-regulated Feynman integrals is twisted cohomology theory~\cite{Mastrolia:2018uzb}. We show that there is a natural way to incorporate symmetry transformations into twisted cohomology theories. 
If we focus on the symmetries within a fixed sector, then we obtain a group of symmetries acting linearly on the twisted cohomology group, leaving the scalar products on all relevant vector spaces invariant. An important tool when studying representations of groups is their decomposition into irreducible representations. In the context of Feynman integrals, these groups are always finite, and it is well known from group theory that for finite groups the decomposition into irreducible representations is controlled by the characters of the representation. Our main mathematical result is that we can describe the character of the representation on the twisted cohomology group, and thus its decomposition into irreducible representations, purely in terms of topological properties of the underlying geometric space, namely the Euler characteristics of the fixed-point sets. The proof of this result heavily relies on tools from algebraic topology, in particular Lefschetz numbers and the celebrated Lefschetz fixed-point theorem.

In the third part, we apply these purely mathematical results to Feynman integrals. After discussing the connection between symmetries and so-called canonical master integrals~\cite{Henn:2013pwa}, we derive our main result, namely a compact formula for the number of master integrals in a sector in the presence of a non-trivial symmetry group. Remarkably, the number of master integrals can still be obtained from an Euler characteristic computation even in the presence of non-trivial symmetries, generalizing the connection between Euler characteristics and numbers of master integrals known from the situation without symmetries. As an application, we compute the number of master integrals for banana integrals with up to four loops and with any configuration of massive propagators.

 Our work combines tools from various fields of mathematics, including graph and matroid theory, twisted cohomology, the representation theory of finite groups and algebraic topology, and we apply them to prove results relevant for Feynman integrals. In the design of the proofs we made use of AI, in particular ChatGPT 5.2. Our results show that such tools can be successfully used to perform frontier research in theoretical and mathematical physics. At the same time, we emphasize that in various instances the AI-generated mathematical arguments were (subtly) flawed, and we had to carefully check and (re-)derive all steps.

Our paper is structured as follows: Since our main results are of technical nature and the relevant mathematical proofs require knowledge from various different areas of mathematics, we present a summary of our main results from a physics perspective in section~\ref{sec:summary}. Section~\ref{sec:summary} is intended for readers primarily interested in applications to physics, and it can be read standalone. In section~\ref{sec:symfeyn} we explore symmetries of families of Feynman integrals from a graph-theoretical perspective, and we show that there is a bijection between symmetry transformations and permutations of Feynman parameters that relate the Lee-Pomeransky polynomials. In section~\ref{sec:symandtwist} we study symmetry transformations for twisted cohomology theories from a purely mathematical standpoint, and we show that the period and intersection pairings are invariant. In section~\ref{sec:twistedsymparam} we discuss how these symmetries are encoded into permutations of Feynman parameters,
irrespective of the twisted cohomology theory used to define the family. In section~\ref{sec:single_sector} we focus on symmetry groups of a single sector, and we discuss representation-theoretic aspects and the implications for dimensionally-regulated Feynman integrals. In section~\ref{sec:OrbitSpaceEuler} we prove that for finite groups the decomposition into irreducible representations is governed by purely topological quantities attached to the twisted cohomology groups, in particular that the character of the representation is given by the Euler characteristic of the fixed-point sets. In section~\ref{sec:number_MIs} we apply this mathematical result to Feynman integrals and present a formula for the number of master integrals in a sector in the presence of a non-trivial symmetry group. We draw our conclusions in section~\ref{sec:conclusion}. We also include several appendices where we collect some mathematical background material required throughout the main text.


\section{Summary of the main results}
\label{sec:summary}

The purpose of this paper is to provide a comprehensive study of certain classes of symmetries of Feynman integrals. To this effect, we combine tools from different areas, including twisted cohomology, graph theory, group theory and algebraic topology. We will discuss the relevant mathematical details in later sections, including rigorous proofs. In this section we present a summary of our main findings, to allow the reader primarily interested in the physics applications to grasp the main ideas without the detailed technical background required for the derivations. This section can be read standalone, without having to read the subsequent sections. 

\subsection{Feynman integrals}
\label{sec:summary_Feynman_integrals}
\paragraph{Definitions.} 
The main objects of interest in this paper are multiloop Feynman integrals, which can be defined as
\begin{equation}\label{eq:loop_momentum_def}
    I_{\bs{\nu}}(\bs{s},\eps)=e^{L\gamma_{\mathrm{E}}\eps}\int\left(\prod_{j=1}^L\frac{\rd^D k_j}{i\pi^{\frac{D}{2}}}\right)\frac{1}{D_1^{\nu_1}\dots D_P^{\nu_P}} \,,
\end{equation}
where $\gamma_{\mathrm{E}}=-\Gamma'(1)$ is the Euler-Mascheroni constant and $\bs{\nu}=(\nu_1,\dots ,\nu_P)$ is a vector of integers. The vector $\bs{s} = (\{p_i\cdot p_j\},\{m_i^2\})$ collects the independent kinematic variables (Mandelstam invariants and propagator masses) on which the integral depends. We will generally write the (inverse) propagators as $D_k=q_k^2-m_k^2$, where $m_k^2$ is the squared mass of the $k^{\textrm{th}}$ propagator and the \emph{edge momenta} $q_k$ are linear combinations of the $L$ loop momenta $k_j$ and the $E$ independent external momenta $p_l$, 
\beq\label{eq:q_to_kp}
q_k = \sum_{j=1}^LC_{jk}\, k_j + \sum_{l=1}^EE_{lk}\,p_l\,, \qquad k=1,\ldots,P\,.
\eeq
 We will mostly be interested in dimensionally-regulated integrals~\cite{tHooft:1972tcz,Cicuta:1972jf,Bollini:1972ui} in $D=d-2\eps$ dimensions, where $d$ is an (even) integer. 
 
 There are various other representations of Feynman integrals, like the Feynman parameter, Lee-Pomeransky~\cite{Lee:2013hzt} or Baikov representations~\cite{Baikov:1996iu,Frellesvig:2017aai,Frellesvig:2024ymq}. The Baikov and Lee-Pomeransky representations will play a role at some point in this paper. We therefore briefly review them here.
 
 The Lee-Pomeransky representation is given by
\begin{equation}\label{LeePomeransky}
    I_{\bnu}(\bs{s},\eps)=\frac{e^{L\gamma_{\mathrm{E}}\eps}\,(-1)^{\nu}\,\Gamma\left( \frac{D}{2}\right)}{\Gamma\left(\frac{(L+1)D}{2} -\nu\right)\prod_{j=1}^{P}\Gamma(\nu_j)}\left(\prod_{i=1}^P\int_0^{\infty}\rd x_i\, x_i^{\nu_i-1}\right)\cG(\bx,\bs{s})^{-\frac{D}{2}} \,,
\end{equation}
with $\nu=\sum_{i=1}^P \nu_i$.
Here $\cG$ is the \emph{Lee-Pomeransky polynomial}, given in terms of the first and second \emph{Symanzik polynomials} appearing in the Feynman parameter representation of the integral,
\begin{equation}\label{eq:GUF}
    \cG(\bx,\bs{s})=\cU(\bx)+\cF(\bx,\bs{s}) \,.
\end{equation}
The Symanzik polynomials can be explicitly written down from the topology of the associated Feynman graph (see review below).

In the Baikov representation a Feynman integral takes the form
\begin{align}
\label{eq_Baikov}
I_{\bnu}(\bs{s},\eps)&= \frac{e^{L\gamma_E\varepsilon}  \,\bG(p_1,\ldots,p_E)^\frac{-D+E+1}{2}}{\pi^{\frac{1}{2}(n-L)} \det \bC\, \prod_{j=1}^L \Gamma\!\left(\frac{D-E+1-j}{2}\right)}\,\hat{I}^D_{\bnu}(\bs{s},\eps)\,,
\end{align}
with
\begin{align}
\label{eq_Baikov2}
\hat{I}_{\bs{\nu}}^D \left(\{p_i\cdot p_j\},\{m_i^2\},\eps \right)&=\int_{\mathcal{C}} \rd^{n} z\, \mathcal{B}(\bs{z})^{\frac{D-L-E-1}{2}} \prod_{s=1}^{P} z_s^{-\nu_s}\, ,
\end{align}
where the integration cycle is given by (cf.,~e.g.,~ref.~\cite{Frellesvig:2024ymq})
\beq\label{eq:Baikov_contour}
\mathcal{C} = \left\{\bz\in\mathbb{R}^n: \frac{\mathcal{B}(\bz)}{G(\bs{p})}>0\right\}\,.
\eeq
The prefactor includes the Gram determinant
\begin{align}\label{eq:Gram_def}
   G(v_1,\dots, v_k)=\det \bG(v_1,\dots, v_k)\,,\qquad \bG(v_1,\dots, v_k)=\left(v_i\cdot v_j\right)_{1\le i,j\le k}\,,
\end{align}
and the integrand contains the \textit{Baikov polynomial}
\begin{align}\label{eq_Baikvpolynomial}
    \mathcal{B}(\bs{z}) =G(k_1,\ldots,k_L,p_1,\ldots,p_E) \,,
\end{align}
with $\bs{z}=(z_1,\ldots,z_{n})$, where the number $n$ of integration variables is 
\begin{align}\label{baikovnumervars}
    n= \frac{1}{2}L (L+1)+EL\, .  
\end{align}
We may choose the first $P$ integration variables to be the propagators themselves. If $P<n$, we introduce additional propagators and set their exponents $\nu_i$ to zero.  
Then, the Jacobian in eq.~\eqref{eq_Baikov} is given by $\det \bC$. 
Note that the Lee-Pomeransky and Baikov representations are closely related, and can be obtained from each other by a suitable variable change~\cite{Chen:2023eqx}.

\paragraph{Feynman graphs.} It is well known that one may attach a graph to a (scalar) Feynman integral. More specifically, to the integral in eq.~\eqref{eq:loop_momentum_def} we may attach the graph $G$ with $L$ loops whose $P$ \emph{internal edges} represent the propagators. The external momenta are represented by \emph{external edges}, distinguished by the fact that they are connected to an \emph{external vertex} of valency one.\footnote{The valency of a vertex is the number of edges attached to it.} Each internal edge $e$ is labelled by its edge momentum $q_e$, its mass $m_e$ and its exponent $\nu_e$. 

The Feynman graph $G$ contains enough information to reconstruct the Feynman integral. For example, if we focus on the momentum representation and
if we require that momentum is conserved at each vertex, then we can write every edge momentum as a linear combination of the $E$ independent external momenta and the $L$ loop momenta, cf.~eq.~\eqref{eq:q_to_kp}. This rewriting is not unique, and involves a choice of basis for the loop momenta $k_j$. One may perform this rewriting using only graph-theoretical information, and the coefficients $C_{jk}$ and $E_{lk}$ in eq.~\eqref{eq:q_to_kp} have a graph-theoretical interpretation. More precisely, let us package the independent loop, external and edge momenta into vectors,
\beq
\bq = \begin{pmatrix}q_1 \\ \vdots\\q_P\end{pmatrix}\,,\qquad
\bk = \begin{pmatrix}k_1\\\vdots\\k_L\end{pmatrix}\,,\qquad
\bp = \begin{pmatrix}p_1\\\vdots\\p_E\end{pmatrix}\,.
\label{eq:momentum_vectors}
\eeq
Equation~\eqref{eq:q_to_kp} can then be cast in the form 
\beq\label{eq:q_k_p}
\bq = \bC^T\bk + \bE^T\bp\,,
\eeq
where $\bC$ and $\bE$ are matrices that can be defined purely graph-theoretically. We will call the parametrization of the edge-momenta in eq.~\eqref{eq:q_k_p} an \emph{edge-momentum parametrization}. In particular, $\bC$ is an $L\times P$ matrix of rank $L$ called the \emph{cycle basis matrix}, and $\bE$ is an $E\times P$ matrix of rank $E$, which can be related to the \emph{edge-flow matrix} of the graph (see section~\ref{sec:edgeMomParams}).

Also the Symanzik polynomials can be defined purely graph-theoretically~\cite{Bogner:2010kv}. We have
\beq\bsp
\label{eq:UFDef}
\cU(\bx) &\,= \sum_{T\in F_1}\prod_{e\notin T}x_e\,,\\
 \cF(\bx,\bs{s})&\,=\left(\sum_{e}m_e^2\,x_e\right)\cU(\bx) +\sum_{(T_1,T_2)\in F_2}(-s_{T_1,T_2})\prod_{e\notin T_1\cup T_2}x_e\,,
\esp\eeq
where  $F_p$ denotes the set of spanning $p$-forests of $G$ and $s_{T_1,T_2}$ is the square of the momentum flowing from $T_1$ to $T_2$.

\paragraph{IBP reduction and master integrals.} It is often convenient to consider \emph{families} of Feynman integrals that only differ by the values of the exponents of the propagators, i.e., by the values of the vector $\bs{\nu}$. We may then consider the vector space generated by linear combinations of Feynman integrals for different $\bs{\nu}$, and the coefficients of the linear combinations are taken from the field $\cF=\mathbb{Q}(\eps,\bs{s})$ of rational functions in the dimensional regulator $\eps$ and the kinematic variables $\bs{s}$. We can assign a Feynman graph $G$ to a family of integrals (the Feynman graph for the member $\bs{\nu}=(1,\ldots,1)$ of the family), and we denote by $\cV_G$ the $\cF$-vector space generated by the members of the family.

It is well known that not all integrals of a given family are independent, but there are linear relations among them, the so-called \emph{integration-by-parts} (IBP) relations~\cite{Tkachov:1981wb,Chetyrkin:1981qh},
\begin{equation}\label{eq:IBP_tot_diff}
\int \rd^D k_j\,\frac{\partial}{\partial k_j^\mu}\left(\frac{v^{\mu}}{D_1^{\nu_1}\dots D_P^{\nu_P}} \right)=0\,,
\end{equation}
where $v^\mu$ is any internal or external momentum. This identity leads to $\cF$-linear relations among different members of the family. The IBP relations can be solved to express all integrals in the family in terms of a finite generating set of integrals. A basis for $\cV_G$ is commonly referred to as \emph{master integrals} in the literature. The number of master integrals is always finite~\cite{Smirnov:2010hn,Bitoun:2017nre}. In the following we denote the number of master integrals by $N$, and we collect the master integrals into a vector $\bI(\bs{s},\eps) = \big(I_1(\bs{s},\eps),\ldots,I_N(\bs{s},\eps)\big)^T$.

To every element $I_{\bs{\nu}}$ of the family, we assign a list $\vartheta^-(\bs{\nu})=(\theta(\nu_1),\dots ,\theta(\nu_P))\in\{0,1\}^P$, where
\begin{equation}\label{eq:heaviside}
    \theta(x)=\left\{ \begin{array}{cl} 
    1, & \text{if }x>0\,, \\
    0, & \text{else}\,,
    \end{array} \right.
\end{equation}
is the Heaviside step function. We then define a \emph{sector} as the subvector space of $\cV_G$ spanned by the integrals $I_{\bs{\nu}}$ that satisfy $\vartheta^-(\bs{\nu})\preceq\Theta$ for some fixed $\Theta=(r_1,\dots, r_P)\in\{0,1\}^{P}$.\footnote{Note that this definition slightly differs from the one usually given in the physics literature, where the sectors are defined by an equality $\vartheta^-(\bs{\nu}) = \Theta$. From a mathematical standpoint, it is more natural to include also integrals with fewer propagators, because IBP relations relate integrals with possibly fewer propagators.} Here $\preceq$ is the partial order on $\{0,1\}^{P}$ defined as follows: if $\Theta=(r_1,\dots, r_P)$ and $\Theta'=(r'_1,\dots, r'_P)$ are two elements of $\{0,1\}^{P}$, then we say that $\Theta'\preceq \Theta$ if $r_j'\le r_j$ for all $j=1,\ldots,P$. Integrals in the same sector share the same upper bound on the set of `active' propagators (by which we mean that they appear in the denominator of the integral), but may differ in the exponents of these propagators as well as their numerators. We will denote the set of indices of the active propagators in a sector $\Theta$ by $\bd_{\Theta}=\{d_1,\dots d_{P_\Theta}\}\subseteq\{1,\dots ,P\}$, i.e., if $\Theta=(r_1,\ldots,r_P)$, then
\beq
\bd_{\Theta} = \{i: r_i = 1\}\,.
\label{eq:dtheta_summary}
\eeq

We can associate a Feynman graph $G_{\Theta}$ to every sector $\Theta$, obtained by starting from the Feynman graph $G$ that represents the whole family and contracting all edges that do not correspond to active propagators in the sector $\Theta$. The \emph{top sector} is given by the collection of integrals $I_{\bnu}$ that satisfy $\vartheta^-(\bnu)=(1,\dots ,1)$, i.e., where all allowed propagators are active. Its associated graph is the graph $G$ representing the full family.

There is a natural partial order on the sectors, defined by the partial order $\preceq$ on $\{0,1\}^P$.  A sector that is less than the top sector is called a \emph{subsector}. We then say that a master integral is \emph{irreducible} and \emph{belongs to a sector $\Theta$} if it cannot be written as a linear combination of integrals from lower sectors $\Theta'\prec \Theta$. It is always possible to choose the master integrals that belong to a given sector in such a way that the numerator is trivial, i.e., the entries of the vector $\bs{\nu}$ are positive or zero. We assume from now on that the elements of our vector $\bI$ of master integrals have trivial numerators, and that they are irreducible and ordered in a way compatible with the partial order on the sectors.

\paragraph{Differential equations.}
It follows from the previous considerations that every member of our family can be computed if we know the corresponding vector of master integrals. 
One of the most commonly used methods to compute Feynman integrals is the method of differential equations~\cite{Kotikov:1990kg,Kotikov:1991hm,Kotikov:1991pm,Gehrmann:1999as,Henn:2013pwa}.  The vector $\bs{I}$  fulfils a differential equation of the form 
\begin{equation}\label{eq:DEQ_prototype}
\rd \bI(\bs{s},\eps) = \bOmega(\bs{s},\eps)\bI(\bs{s},\eps)\,,
\end{equation}
where $\rd = \sum_{i=1}^r\rd s_i\,\partial_{s_i}$ denotes the exterior derivative with respect to the external kinematic parameters $\bs{s}= (s_1,\ldots,s_r)$.  The entries of the matrix $\bOmega(\bs{s},\eps)$ are rational one-forms in $\bs{s}$ and rational functions in $\eps$, i.e., they can be expressed in the form 
\begin{equation}
\bOmega(\bs{s},\eps) = \sum_{i=1}^r\rd s_i\,\bOmega_i(\bs{s},\eps)\,,
\end{equation}
where the $\bOmega_i(\bs{s},\eps)$ are matrices of rational functions in $\bs{s}$ and $\eps$. Note that differentiation is compatible with the partial order on the sectors, and so the matrix $\bOmega(\bs{s},\eps)$ is block lower-triangular if the vector of master integrals is chosen appropriately. The blocks on the diagonal are associated with a given sector, and describe the differential equations satisfied by the maximal cuts of the master integrals that belong to that sector~\cite{Primo:2016ebd,Primo:2017ipr,Frellesvig:2017aai,Bosma:2017ens}.

Even if we choose a basis of master integrals that are irreducible and ordered according to the partial order on the sectors, there is some arbitrariness in the choice of basis. We may define a new vector of master integrals,
\begin{align}
\label{conventiontrafo}
 \bs{J}(\bs{s},\eps)    = \bs{U}(\bs{s},\eps)\bs{I}(\bs{s},\eps)\, ,
\end{align}
where $\bs{U}(\bs{s},\eps)$ is some matrix of full rank compatible with the partial order on the sectors. A judicious choice of basis may have an impact on our ability to solve the system of differential equations for the master integrals. 
A particularly convenient choice is a so-called \emph{canonical} basis, first introduced in ref.~\cite{Henn:2013pwa} in the context of Feynman integrals that evaluate to multiple polylogarithms~\cite{Goncharov:1998kja,Remiddi:1999ew,Gehrmann:1999as}. Various proposals have been made for how to obtain canonical differential equations for Feynman integrals that go beyond polylogarithms. While there is still no general consensus for what a good definition of a canonical basis beyond polylogarithms is, it is generally agreed that in a canonical basis the differential equations are $\eps$-factorized,\footnote{For an approach that advocates a form of the differential equations without $\eps$-factorization, see ref.~\cite{Chaubey:2025adn}.} i.e., the differential equation for $\bJ(\bs{s},\eps)$ can be cast in the form
\beq
\rd\bJ(\bs{s},\eps) = \eps\,\bs{\Omega}_c(\bs{s})\,\bJ(\bs{s},\eps)\,,
\eeq
where $\bs{\Omega}_c(\bs{s})$ is a matrix of one-forms that are independent of $\eps$. The requirement of $\eps$-factorization alone is not sufficient to uniquely characterize canonical bases (see, e.g., ref.~\cite{Frellesvig:2023iwr} for a discussion and comparison of different types of $\eps$-factorized bases). 
Here we work with the operational definitions of canonical bases from refs.~\cite{Pogel:2022vat,Pogel:2022ken,Pogel:2022yat,Gorges:2023zgv,Duhr:2025lbz,Maggio:2025jel,e-collaboration:2025frv}, which are expected to deliver equivalent results. In particular, in ref.~\cite{Duhr:2025lbz} it was proposed that a canonical basis is characterized not only by $\eps$-factorization, but in addition the differential forms in $\bs{\Omega}_c(\bs{s})$ should have at most simple poles and define independent cohomology classes. The advantage of such a canonical basis is that the differential equations can easily be solved as a power series in $\eps$ whose coefficients are constant linear combinations of iterated integrals~\cite{ChenSymbol}, and the individual iterated integrals are linearly independent~\cite{deneufchatel:hal-00558773,Duhr:2024xsy} (i.e., there are no hidden zeroes, in the sense that there is no non-trivial linear combination of iterated integrals that evaluates to zero).


\subsection{Symmetries of Feynman integrals}
\label{sec:symmOfFIsOverview}
\subsubsection{Definitions}
\label{sec:summary_symmetries_def}
The main topic of our paper are symmetries of Feynman integrals. We now precisely define what we mean by this.

\paragraph{Symmetry transformations between sectors.} Our definition of symmetry transformation follows closely the definition presented in ref.~\cite{Wu:2024paw}. Consider two sectors $\Theta_1$ and $\Theta_2$, and
let $(\bq_1, \bk_1, \bp_1)$ and $(\bq_2, \bk_2, \bp_2)$ denote the vectors of edge, loop and external momenta of the sectors $\Theta_1$ and $\Theta_2$, respectively (see eq.~\eqref{eq:momentum_vectors}).
A \emph{symmetry transformation} from $\Theta_1$ to $\Theta_2$ is an affine change of variables and a linear map on the external momenta
\begin{equation}
\label{eq:symmTrans}
    \bk_1=\bL\bk_2+\bM\bp_2,\qquad \bp_1=\bN\bp_2 \,,
\end{equation}
 such that the following properties hold~\cite{Wu:2024paw}: 
\begin{enumerate}
    \item The Jacobian is trivial: $\det\bL=\pm 1$.
    \item There is a bijection on the propagators, i.e., there is a bijective map $\alpha : \bd_{\Theta_2}\rightarrow \bd_{\Theta_1}$ such that $\bq^2_{1,\alpha(i)}(\bk_1,\bp_1)=\bq^2_{2,i}(\bk_2,\bp_2)$ and $m_i^2=m_{\alpha(i)}^2$ for all $i\in \bd_{\Theta_2}$.
    \item The scalar products between external momenta are invariant: $p_{1,i}\cdot p_{1,j}=p_{2,i}\cdot p_{2,j}$ for all $1\le i,j\le E$.\footnote{In case the Feynman integral factorizes into simpler integrals, we require invariance only for the scalar products that the individual factors depend on.} 
\end{enumerate}
We denote the set of all symmetry transformations from $\Theta_1$ to $\Theta_2$ by $\Sym(\Theta_1,\Theta_2,\bs{s})$. We stress that the set of symmetry transformations between sectors will typically depend on the kinematic point $\bs{s}$ (e.g., certain masses may need to be equal for a non-trivial symmetry transformation to exist). We will typically work with a fixed kinematic point $\bs{s}$,\footnote{We assume that the fixed kinematic point is generic, i.e., that there are no non-trivial relations between the off-diagonal entries of the Gram matrix $\bs{G}(\bs{p})$. In particular, we assume that we work in a generic number of dimensions, and we assume that the determinant of the Gram matrix $\bs{G}(\bs{p})$ is non-vanishing.} and we often suppress the dependence on $\bs{s}$, i.e., we will simply write $\Sym(\Theta_1,\Theta_2)$ instead of $\Sym(\Theta_1,\Theta_2,\bs{s})$.
Note that the second defining condition implies that a symmetry transformation between two sectors can only exist if these sectors have the same number of active propagators. The third defining condition can be cast in the form of an invariance property of the Gram matrix $\bG(\bp)$ (cf.~eq.~\eqref{eq:Gram_def}): 
\beq\label{graminv}
\bG(\bp_1) = \bN\bG(\bp_2)\bN^T = \bG(\bp_2)\,.
\eeq
Note that this implies 
\beq
\label{eq:detN}
\det\bN = \pm1\,.
\eeq

We may ask about the form of the entries of the matrices $\bL$, $\bM$ and $\bN$, in particular if they are constants. In ref.~\cite{Wu:2024paw} it was shown that there are examples of families of two-loop integrals where one can can relate two different sectors $\Theta_1$ and $\Theta_2$ by a symmetry transformation whose matrix $\bN$ is kinematics-dependent. We therefore assume from now on that the entries of these matrices are taken from the field $\cF$ introduced earlier.\footnote{We could even restrict the discussion to the subfield of $\cF$ of rational function that are independent of $\eps$. As we will see, this restriction is not relevant.} In ref.~\cite{Wu:2024paw} it was argued that the origin of the kinematic dependence can be traced back to the fact that subsectors may only depend on certain combinations of the independent external momenta, and an algorithm was presented to find such kinematics-dependent symmetry transformations. It is an interesting question to determine the complete set of allowed symmetry transformations between two sectors if arbitrary kinematics-dependent transformations are allowed. A first goal of our paper is to provide such a complete description. 

\paragraph{The symmetry group of a sector.} If we focus on symmetry transformations within a given sector $\Theta=\Theta_1=\Theta_2$, we obtain a group which we call the \emph{symmetry group of the sector $\Theta$}. We denote it by $\Aut(\Theta,\bs{s}) := \Sym(\Theta,\Theta,\bs{s})$. We will again often simply write $\Aut(\Theta)$ instead of $\Aut(\Theta,\bs{s})$. The elements of this group can be represented by $(L+E)\times(L+E)$ matrices
\beq
\label{eq:defTMatrix}
\bT = \begin{pmatrix} \bL & \bM \\ \bs{0} & \bN\end{pmatrix}\,,
\eeq
and the group law is simply given by ordinary matrix multiplication. Note that $\det \bT=\pm1$. These matrices act on the loop and external momenta via
\beq
\begin{pmatrix}\bk \\ \bp\end{pmatrix} = \bT \begin{pmatrix}\bk' \\ \bp'\end{pmatrix}\,.
\eeq
Following the previous discussion, we assume that the entries of the matrix $\bT$ are taken from the field $\cF$. 
An element of $\Aut(\Theta)$ sends integrals from the sector $\Theta$ to some other integrals from the same sector. In particular, the symmetry group acts on the master integrals. Since the number of master integrals in each sector is finite, we obtain a finite-dimensional representation of $\Aut(\Theta)$ on the $\cF$-vector space generated by the master integrals of the sector.

Some of the elements of $\Aut(\Theta)$ can be read off directly from the Feynman graph $G_{\Theta}$ associated to the sector. In particular, every automorphism of $G_{\Theta}$ defines an element of $\Aut(\Theta)$. However, as we will soon see, in some instances the group $\Aut(\Theta)$ may contain more elements than just the automorphisms. 

\subsubsection{A complete description of symmetry transformations}
\label{eq:symmTransfDescr}
We now present our first main result, namely a complete description of the set $\Sym(\Theta_1,\Theta_2)$ of symmetry transformations from a sector $\Theta_1$ to a sector $\Theta_2$. We have already mentioned that in ref.~\cite{Wu:2024paw} it was observed on examples that symmetry transformations between two sectors may require a kinematics-dependent matrix $\bs{N}$, and it was pointed out that this dependence on kinematics can be traced back to the appearance of \emph{momentum groups} in subsectors. We start by reviewing and explaining these concepts, and then separately discuss the cases with and without momentum groups.

\begin{figure}[t]
    \centering
    \includegraphics[width=0.25\textwidth]{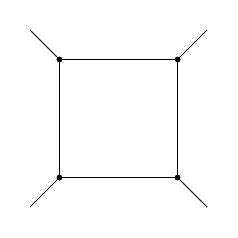}
    \includegraphics[width=0.25\textwidth]{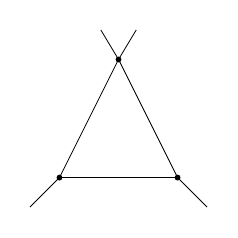}
    \caption{Example of a momentum-grouped graph (left) and a non momentum-grouped graph (right), corresponding to a subsector.}
    \label{fig:momGrouping}
\end{figure}
\paragraph{Momentum groups.} Consider a family of Feynman integrals with $E$ independent external momenta $\bp$ whose top sector can be represented by a Feynman graph $G$. We may assume without loss of generality that for every internal vertex in this graph, there is at most one external edge attached to it, and so there is at most one external momentum flowing in or out of any given internal vertex. 

Consider now a proper subsector $\Theta$ of this family.\footnote{By \emph{proper} we mean that this subsector is irreducible and not equal to the top sector.} Then the Feynman graph $G_{\Theta}$ that represents this sector is obtained from $G$ by contracting all edges that correspond to the propagators that are not active. The graph $G_{\Theta}$ may then have multiple external edges attached to the same internal vertex. Following ref.~\cite{Wu:2024paw}, we call the set of two or more external edges connected to the same vertex a \emph{momentum group}. We say that a graph or sector is \emph{momentum-grouped} if it does not contain any momentum groups, i.e., if at most one external edge is connected to every vertex (see figure \ref{fig:momGrouping} for an illustration). Note that we assume that the top sector is always momentum-grouped. 

If a sector $\Theta$ is not momentum-grouped, then the Feynman integrals from that sector depend on the external momenta $\bp$ only through the sum of the momenta in every momentum group. It is easy to see that we can always define a momentum-grouped graph $\widetilde{G}_{\Theta}$ with $\widetilde{E}<E$ independent external momenta $\bs{\tilde{p}} = (\tilde{p}_1,\ldots,\tilde{p}_{\widetilde{E}})^T$ that gives rise to the same Feynman integral as $G_{\Theta}$. The momenta $\bs{\tilde{p}}$ are linear combinations of the momenta $\bp$, and there is an $\widetilde{E}\times E$ matrix $\bS$ such that
\beq\label{eq:p_to_ptilde}
\bs{\tilde{p}} = \bS\bp\,.
\eeq
Note that $\bS$ must have full rank $\widetilde{E}<E$, because otherwise the grouped momenta cannot be independent.

\paragraph{Momentum-grouped sectors.} Let us discuss the concrete form of the set of symmetry transformations between two momentum-grouped sectors $\Theta_1$ and $\Theta_2$. We first introduce some sets of symmetry transformations between polynomials.

Consider two polynomials $f_1$, $f_2$ in $P$ variables $\bx = (x_1,\ldots,x_P)^T$. We define $\mathbb{S}(f_1,f_2)$ to be the set of permutations of these variables that maps one polynomial to the other, 
\beq \label{eq:symleepom}
\mathbb{S}(f_1,f_2) := \{\sigma\in S_P : f_1(\sigma(\bx)) = f_2(\bx)\}\,,
\eeq
where $(\sigma(\bx))_i=x_{\sigma^{-1}(i)}$. We also define $\mathbb{G}(f) := \mathbb{S}(f,f)$. In other words, $\mathbb{G}(f)$ is the group of permutations that leave the polynomial $f$ invariant.

We can see from the definition of a symmetry transformation in eq.~\eqref{eq:symmTrans} that it maps the set of active propagators from one sector bijectively to the active propagators from the other. This also induces a bijection $\sigma$ between the Feynman parameters $x_e$ that define the Symanzik and Lee-Pomeransky polynomials of the two sectors, and so we obtain an element $\sigma\in \mathbb{S}(\cG_1,\cG_2)$, where $\cG_1$ and $\cG_2$ denote the Lee-Pomeransky polynomials (for a fixed value of $\bs{s}$) obtained from the graphs $G_{\Theta_1}$ and $G_{\Theta_2}$, respectively.

In section~\ref{eq:symmMomGrouped} we will show that in the scenario in which the Mandelstam invariants $p_i\cdot p_j$ with $1\le i< j\le E$ take generic values, the converse of the previous statement is also true, namely that whenever the Mandelstam invariants take generic values, then every symmetry transformation is induced by a bijection $\sigma\in \mathbb{S}(\cG_1,\cG_2)$. More precisely, we will show that there is a bijection:
\beq\label{eq:S_T1_T2} 
\Sym(\Theta_1,\Theta_2,\bs{s}) \simeq \mathbb{Z}_2^c\times \mathbb{S}(\cG_1(\cdot,\bs{s}),\cG_2(\cdot,\bs{s}))  \,.
\eeq
Let us discuss this result in some detail. 
Let us start with the $\mathbb{Z}_2$-factors. They act on the loop and external momenta in such a way that a subset of edge momenta change sign (clearly changing the signs of edge momenta leaves the propagators invariant). The number $c$ of such $\mathbb{Z}_2$-factors depends on the precise form of the graph and the external momenta. Since in any case these transformations act trivially on the propagators and the Feynman integrals, we will not discuss them any further. 

The non-trivial symmetry transformations come from non-trivial bijections $\sigma\in \mathbb{S}(\cG_1,\cG_2)$. One of our main results is that it is possible to use concepts from graph theory to explicitly construct the matrices $\bL_{\sigma}$, $\bM_{\!\sigma}$ and $\bN_{\!\sigma}$ that correspond to the bijection $\sigma\in \mathbb{S}(\cG_1,\cG_2)$. To see how this works, assume that we pick independent edge, loop and external momenta $\bq_l$, $\bk_l$ and $\bp_l$ in each sector $\Theta_l$ respectively. 
Since every $\sigma\in \mathbb{S}(\cG_1,\cG_2)$ is a bijection of the Feynman parameters, it induces a bijection between the active propagators in the two sectors. We can then write  
\beq\label{eq:permutation}
\bq_1(\bk_1,\bp_1) = \bP_{\!\!\sigma}\,\bq_2(\bk_2,\bp_2)\,,
\eeq
where $\bP_{\!\!\sigma}$ is a (signed) permutation matrix.\footnote{A signed permutation matrix is a matrix that has exactly one non-zero entry in each row and column, and the value of that entry is $\pm1$. The signs correspond to orientation changes of the edges.}
From eq.~\eqref{eq:q_k_p} we know that we have
\beq
\label{eq:q12ParamGeneral}
\bq_l = \bC_l^T\bk_l + \bE_l^T\bp_l\,,\qquad l=1,2\,,
\eeq
where $\bC_l$ is the cycle basis matrix obtained from the graph $G_{\Theta_l}$ associated with the sector $\Theta_l$, and $\bE_l$ is related to a corresponding edge-flow matrix. In appendix~\ref{sec:proofSymLemma} we will show that the matrices $\bL_{\sigma}$, $\bM_{\!\sigma}$ and $\bN_{\!\sigma}$ can be constructed using only the permutation matrix $\bP_{\!\!\sigma}$ and graph-theoretical information contained in $\bC_l$ and $\bE_l$:
\begin{align}
\label{eq:symmMatricesFromAut_sum}
    \bL_{\sigma}&=(\bC_1\bC_1^T)^{-1}\bC_1\bP_{\!\!\sigma}\bC_2^T \,, \nonumber\\
    \bN_{\!\sigma}&=(\bE_1\bPi_{C_1}^{\perp}\bE_1^T)^{-1}\bE_1\bPi_{C_1}^{\perp}\bP_{\!\!\sigma}\bPi_{C_2}^{\perp}\bE_2^T \,, \\
    \bM_{\!\sigma}&=(\bC_1\bC_1^T)^{-1}\bC_1\bP_{\!\!\sigma}\bE_2^T-(\bC_1\bC_1^T)^{-1}\bC_1\bE_1^T\bN_{\!\sigma} \,, \nonumber
\end{align}
where we defined the projectors,
\beq
\label{eq:projectorDef}
\bPi_{C_l}=\bC_l^T(\bC_l\bC_l^T)^{-1}\bC_l\textrm{~~~and~~~}\bPi_{C_l}^{\perp}=\mathds{1}-\bPi_{C_l}\,.
\eeq
Equation~\eqref{eq:symmMatricesFromAut_sum} explicitly establishes the bijection between the set of symmetry transformations between two sectors and the set of permutations that relate the corresponding Feynman parameters (modulo the $\mathbb{Z}_2$ factors). 

Let us make some comments about this result. First, if there are non-trivial linear relations among the Mandelstam invariants, one may find permutations  
$\sigma\in\mathbb{S}(\cG_1,\cG_2)$ that cannot be lifted to a symmetry transformation using the formulas in eq.~\eqref{eq:symmMatricesFromAut_sum}.  Second, $\mathbb{S}(\cG_1,\cG_2)$ is always by definition a finite set, and so the set of symmetry transformations $\Sym(\Theta_1,\Theta_2,\bs{s})$ for generic $\bs{s}$ is always finite.
Third, since the right-hand side of eq.~\eqref{eq:symmMatricesFromAut_sum} only involves constant matrices (either the signed permutation matrix $\bP_{\!\sigma}$ or the purely graph-theoretical matrices $\bC_l$ and $\bE_l$), we see that symmetry transformations between momentum-grouped sectors are always independent of the kinematics. Finally, it is natural to ask if the symmetry transformations themselves have interpretations in terms of graph theory. We will discuss this in more detail in section~\ref{eq:symmMomGrouped}. Here it suffices to say that every symmetry transformation indeed corresponds to a map between the graphs $G_{\Theta_1}$ and $G_{\Theta_2}$: it is either an isomorphism between $G_{\Theta_1}$ and $G_{\Theta_2}$ or a Whitney twist that maps $G_{\Theta_1}$ into $G_{\Theta_2}$, where one of the two components of the twist does not depend on any external momenta (see section~\ref{eq:symmMomGrouped} for details).

\paragraph{Non momentum-grouped sectors.} Assume that the family of integrals depends on $E$ independent external momenta $\bp$, and consider two sectors $\Theta_1$ and $\Theta_2$ of this family. They may not be momentum-grouped. Instead the corresponding integrals may depend on $\widetilde{E}<E$ independent external momenta $\bs{\tilde{p}}_l$,\footnote{If there is a symmetry between these two sectors, the integrals will necessarily depend on the same number of independent external momenta, so we can assume that the number $\widetilde{E}$ of independent external momenta is the same for both.} related to $\bp$ by (cf.~eq.~\eqref{eq:p_to_ptilde}),
\beq\label{eq:ptk_to_pk}
\bs{\tilde{p}}_l = \bS_l\bp\,,\qquad l=1,2\,,
\eeq
where $\bS_1$ and $\bS_2$ are two $\widetilde{E}\times E$ matrices of full rank $\widetilde{E}$.

To understand the structure of the possible symmetry transformations from $\Theta_1$ to $\Theta_2$, we may first look at the momentum-grouped graph $\widetilde{G}_{\Theta_l}$ obtained by replacing each momentum group in $G_{\Theta_l}$ by a sum of external momenta, or equivalently by a single external momentum from $\bs{\tilde{p}}_l$. Since the graphs $\widetilde{G}_{\Theta_l}$ are momentum-grouped, we may apply the result from the previous subsection. It follows that every symmetry transformation is induced by a bijection $\sigma\in \mathbb{S}(\cG_1,\cG_2)$. We may then use eq.~\eqref{eq:symmMatricesFromAut_sum} to construct matrices $\bs{\widetilde{L}}_{\sigma}$, $\bs{\widetilde{M}}_{\!\sigma}$ and $\bs{\widetilde{N}}_{\!\sigma}$ such that
\beq
\label{eq:symmTransfMomGrouped}
\bk_1 = \bs{\widetilde{L}}_{\sigma}\bk_2+\bs{\widetilde{M}}_{\!\sigma}\bs{\tilde{p}}_2\,,\qquad\bs{\tilde{p}}_1=\bs{\widetilde{N}}_{\!\sigma}\bs{\tilde{p}}_2 \,,
\eeq
defines a symmetry transformation between the associated momentum-grouped graphs $\widetilde{G}_{\Theta_1}$ and $\widetilde{G}_{\Theta_2}$. Note that the matrices $\bs{\widetilde{L}}_{\sigma}$, $\bs{\widetilde{M}}_{\!\sigma}$ and $\bs{\widetilde{N}}_{\!\sigma}$ are again constructed from purely graph-theoretical input, and in particular they are constant matrices.

We may now lift any symmetry transformation from $\widetilde{G}_{\Theta_1}$ to $\widetilde{G}_{\Theta_2}$ to a symmetry transformation between the non momentum-grouped sectors $\Theta_1$ and $\Theta_2$. The detailed derivation is provided in section~\ref{sec:nonMomGroupedSymms}. Here we only show the result. 
 The matrices $\bs{\widetilde{L}}_{\sigma}$,  $\bs{\widetilde{M}}_{\!\sigma}$ and $\bs{\widetilde{N}}_{\!\sigma}$ can be lifted to a symmetry transformation in $\Sym(\Theta_1,\Theta_2)$ via the relation
\beq
 \label{eq:symmTransfNonMomGrouped}
\begin{pmatrix}\bs{{L}}_{\sigma} & \bs{{M}}_{\!\sigma}\\ \bs{0} & \bs{{N}}_{\!\sigma}\end{pmatrix} = \begin{pmatrix}\mathds{1} & \bs{0}\\\bs{0} & \bS_1\\\bs{0} & \bS'_1(\bs{s})\end{pmatrix}^{-1}\begin{pmatrix}\bs{\widetilde{L}}_{\!\sigma} & \bs{\widetilde{M}}_{\!\sigma} &\bs{0}\\ \bs{0}
 & \bs{\widetilde{N}}_{\!\sigma}&\bs{0}\\\bs{0}&\bs{0} & \bs{\overline{O}}\end{pmatrix}\begin{pmatrix}\mathds{1} & \bs{0}\\\bs{0} & \bS_2\\\bs{0} & \bS'_2(\bs{s})\end{pmatrix}\,,
 \eeq
where $\bs{\overline{O}}$ is a matrix whose detailed form is not important here, and whose role will be explained in section~\ref{sec:nonMomGroupedSymms}. The matrix $\bS_l$ is defined in eq.~\eqref{eq:ptk_to_pk}, and $\bs{S'}_{\!\!l}(\bs{s})$ is the $(E-\widetilde{E})\times E$ matrix defined by the equation, 
\beq
\label{eq:SpDef}
\bS_l\bG(\bp)\bs{S'}_{\!\!l}(\bs{s})^T = 0\,,
\eeq
where $\bG(\bp)$ is the Gram matrix associated with the family. This equation can always be solved, and the rows of the matrix $\bs{S'}_{\!\!l}(\bs{s})$ form a basis for the kernel of $\bS_l\bG(\bp)$. Note that, since $\bG(\bp)$ obviously depends on the kinematics, the solution $\bs{S'}_{\!\!l}(\bs{s})$ will also be kinematics-dependent in general. This kinematic dependence then feeds into the symmetry transformation through the matrix $\bs{S}'_l(\bs{s})$. We note that this is the only source of kinematic dependence. 

\subsubsection{Constructing all symmetry transformations between two sectors}

We can use the results from the previous subsections to  obtain a concrete procedure to find and construct \emph{all} symmetry transformations between two sectors $\Theta_1$ and $\Theta_2$ (provided that the kinematics is generic):

\begin{enumerate}
\item We know that  symmetry transformations exist only between sectors with the same number of propagators, so it is enough to focus on such pair of sectors $(\Theta_1,\Theta_2)$ (including the cases $\Theta_1=\Theta_2$). This implies that the Lee-Pomeransky polynomials $\cG_1$ and $\cG_2$ associated with these two sectors depend on the same number of Feynman parameters $\bx$.
\item Given such a pair, we look for bijections $\sigma$ of the Feynman parameters such that $\cG_2(\bx,\bs{s}) = \cG_1(\sigma(\bx),\bs{s})$. An efficient algorithm to find these bijections was for example described in ref.~\cite{Pak:2011xt}. 
\item Given such a bijection $\sigma$, we may lift it to a symmetry transformation in loop-momentum space via eq.~\eqref{eq:symmMatricesFromAut_sum}. This lifting can be done by exclusively using information from the bijection $\sigma$ and from the underlying Feynman graphs.
\item If the sectors contain momentum groups encoded in the matrices $\bS_1$ and $\bS_2$, we can lift the symmetry transformation to the original sectors $\Theta_1$ and $\Theta_2$. This requires one to construct the matrices $\bs{S'}_{\!\!1}(\bs{s})$ and $\bs{S'}_{\!\!2}(\bs{s})$, which may in turn introduce a kinematic dependence into the symmetry transformation.
\end{enumerate}
We obtain in this way an effective algorithm to find symmetry transformations between sectors. This is not the first time that such an algorithm is proposed, cf.,~e.g.,~refs.~\cite{Pak:2011xt,Lee:2012cn,Wu:2024paw}. However, we feel that our proposal has several advantages: First, it allows one to construct the change of variables only using information from the permutation symmetries of the Lee-Pomeransky polynomials and the underlying Feynman graph, including cases where the symmetry transformation is kinematics-dependent. As far as we know, only the method of ref.~\cite{Wu:2024paw} manages to systematically find such transformations. Second, many approaches to finding symmetries use heuristics to find linear changes of variables in loop-momentum space that define symmetry transformations, and there is no guarantee that all transformations are found. Since eq.~\eqref{eq:S_T1_T2} is a bijection, we are guaranteed to find \emph{all} symmetry transformations between sectors.


\subsection{Twisted cohomology and symmetry transformations}\label{sec_cohomandsym}

In the previous subsection we have discussed the structure of the set of symmetry transformations between two sectors. It was realized a couple of years ago that an appropriate mathematical framework to study Feynman integrals in dimensional regularization is twisted cohomology~\cite{Mastrolia:2018uzb,aomoto_theory_2011}. We now discuss how we can incorporate symmetry transformations into this framework.

\subsubsection{Review of twisted cohomology}\label{introtwisted}
Loosely speaking, to define a twisted cohomology group, we split the integrand into a multivalued part, the so-called \textit{twist} $\Psi$, and a single-valued differential $n$-form $\varphi$ integrated over some cycle $\gamma$
\begin{align}\label{eq:twisted_int}
    \int_\gamma \Psi \varphi\,.
\end{align}
The twist takes the form
\beq\label{eq:twist}
\Psi = P_1(\bs{z})^{\mu_1}\ldots  P_r(\bs{z})^{\mu_r}\,,
\eeq
where the $P_k(\bs{z})$ are polynomials in the integration variables $\bs{z}\in\mathbb{C}^n$ and the $\mu_k$ are non-integer.
For Feynman integrals the twist is for example related to the Baikov or Lee-Pomeransky polynomials, with the multivaluedness arising from the presence of the dimensional regulator $\eps$ in the exponents in eqs.~\eqref{LeePomeransky} and~\eqref{eq_Baikov2}. 
The zeroes of the twist $\Psi$ are called \textit{regulated boundaries}, and we define the \textit{twisted variety}:
\begin{align}\label{twistedvariety}
\Sigma:=\{\bz\in \mathbb{C}^n: \Psi(\bz)=0\} \,.
\end{align}
The single-valued $n$-form $\varphi$ is typically a rational differential form with poles along the twisted variety. If $\varphi$ has poles at loci not contained in the twisted variety, we refer to them as \emph{unregulated boundaries}. They are contained in the \emph{on-shell variety}, defined as the locus
\begin{align}\label{onshellvariety}
    D_-=\{\bz\in \mathbb{C}^n: \bz \text{ pole of }\, \Psi \varphi\}\,.
\end{align}
The twisted (co)-homology groups are then defined on a space
\begin{align} 
X=\mathbb{C}^n\backslash(\Sigma\cup D_-)\,,
\end{align}
in terms of the covariant derivative
\begin{align}
\nabla_\omega=\rd+\omega\,\wedge\,,\quad \omega=\rd\!\log\Psi\,.
\end{align}
Explicitly, we define the \textit{twisted cohomology group} $H_\text{dR}^n(X,\nabla_\omega)$ as the group of equivalence classes of differential forms that are closed with respect to $\nabla_\omega$ modulo forms that are exact.\footnote{Note that $\Psi \nabla_\omega\varphi=\rd (\Psi \varphi)$. In particular, the condition that $\Psi\varphi$ is closed or exact can be expressed solely in terms of the differential form $\varphi$ via $\nabla_\omega$.}  Typically, for Feynman integrals in the Baikov representation, the case $D_-=\emptyset$ corresponds to a maximal cut of a Feynman integral (because we have taken residues at all propagator poles, and so there are no unregulated boundaries anymore). In general, for example in the Lee-Pomeransky representation, there will be an additional set of boundaries, defining a \textit{relative twisted cohomology group}. To keep the discussion simple, we will only consider the non-relative case here, and we refer to section~\ref{reltwistedrew} for more details on the more general relative case.

Since the twist is multivalued, the integration cycle in eq.~\eqref{eq:twisted_int} is a so-called \textit{loaded cycle} on $X$, i.e., a cycle on $X$ together with a choice of branch of $\Psi$.
This multivaluedness is captured by a \textit{local system} $\check{\mathcal{L}}_\omega$.\footnote{The cohomology group $H^n(X,\mathcal{L}_\omega)$ is isomorphic to $H^n_\text{dR}(X,\nabla_\omega)$ introduced above.}
For details on the construction of the specific local system $\check{\mathcal{L}}_\omega$ see e.g., ref.~\cite{aomoto_theory_2011}.
The (relative) twisted homology group $H_n(X,\check{\mathcal{L}}_\omega)$ is then the group generated by closed $n$-cycles modulo boundaries. A basis of this group is given by \textit{regularized contours} $\gamma$, which in the one-dimensional case can for example be Pochhammer contours encircling two branch points or a closed loop around a pole. The details are not essential for this paper, and can for example be found in refs.~\cite{yoshida_hypergeometric_1997,kita_intersection_1994-2,aomoto_theory_2011, Bhardwaj:2023vvm, Duhr:2023bku}.

In the case of Feynman integrals defined in the Baikov representation, the on-shell variety is typically a union of hyperplanes defined by the vanishing of the (inverse) propagators. 
It is also possible to define a twist $\Psi$ and twisted on-shell varieties $\Sigma$ and $D_-$  for a Feynman integral directly from the loop-momentum representation~\cite{Caron-Huot:2021xqj}.
Note that $\Psi$, $\Sigma$ and $D_-$ are different for each of these representations (for example, not even the number $n$ of integration variables needs to be the same).
Hence, we see that there is no unique twisted cohomology theory that one can associate with a given Feynman integral.

For twisted (co)homology groups, there naturally exist {dual groups} of twisted forms (and twisted cycles). There are associated pairings between the differential forms and their duals, called \textit{intersection pairings}.
In the simplest case, for example for the Baikov representation on the maximal cut (where $D_-=\emptyset$), the group and its dual can both be defined on the same space $X=\mathbb{C}^n\backslash\Sigma$, where the dual group is associated with the inverse twist $\Psi^{-1}$. Explicitly it is defined as $H_{\text{dR},c}^n(X,\check{\nabla}_\omega)$ with $\check{\nabla}_\omega=\rd-\omega\,\wedge$\,, and the dual homology group is $H_n^\text{lf}(X,\mathcal{L}_\omega)$, where $\mathcal{L}_\omega$ is the local system dual to $\check{\mathcal{L}}_{\omega}$.  Here the subscript `$c$' refers to \emph{cohomology with compact support}, while the  superscript `lf' indicates that the chains are \textit{locally finite}.
If $D_-\neq\emptyset$, the dual group is a relative twisted cohomology group, which we will discuss in more detail in section~\ref{reltwistedrew}.

\paragraph{Pairings.} We now have four different groups, namely the twisted cohomology and homology groups, and their duals. There exist non-degenerate pairings between these groups, and all these groups have the same dimension $N$. Let us fix bases 
\beq\bsp
\bs{\varphi} &\,= (\varphi_1,\ldots,\varphi_N)\,,\\
\bs{\gamma} &\,= (\gamma_1,\ldots,\gamma_N)\,,\\
\bs{\check{\varphi}} &\,= (\check{\varphi}_1,\ldots,\check{\varphi}_N)\,,\\
\bs{\check{\gamma}} &\,= (\check{\gamma}_1,\ldots,\check{\gamma}_N)\,.
\esp\eeq
We now define the pairings via these bases. 

We may pair the loaded cycles $\gamma_j$ with the twisted cocycles $\varphi_i$ via the \textit{period pairing}
\begin{align} \label{eq:integration_pairing}  P_{ij}=\langle\varphi_i|\gamma_j]=\int_{\gamma_j} \Psi \varphi_i\,.
\end{align}
The matrix $\bP$ is called the \emph{period matrix}. Similarly, we may define the \emph{dual period pairing} and the \emph{dual period matrix},
\begin{align}
    {\check{P}}_{ij} = [\check{\gamma}_j|\check{\varphi}_i \rangle = \int_{\check{\gamma}_j}\Psi^{-1}\check{\varphi}_i\, .
\end{align}
Dual periods are not as well understood in the context of Feynman integrals, cf.,~e.g.,~refs. \cite{Caron-Huot:2021xqj,Caron-Huot:2021iev,Giroux_2023}. If we focus on maximal cuts in the Baikov representation of Feynman integrals in dimensional regularization, where the on-shell variety is empty, there is a choice of dual bases such that the periods coincide with the dual periods, 
up to a change in the sign of the dimensional regulator~\cite{Duhr:2024xsy}
\begin{align}\label{eq:self-duality}
\bs{\check{P}}(\eps)=\bP(-\eps)\,.
\end{align}
We refer to such a basis choice as a \textit{self-dual} basis.

We may also pair the twisted (co-)homology groups with their duals.
The bilinear pairings between the basis elements $\varphi_i$ and the elements $\check{\varphi}_j$ of a dual basis is given by the cohomology intersection matrix 
\begin{align}\label{eq:cohompairing}
    {C_{ij}} =  \frac{1}{(2\pi i)^n} \langle \varphi_i |\check{\varphi}_j\rangle = \frac{1}{(2\pi i)^n}\int_X \varphi_i\wedge \check\varphi_{j}\, .
\end{align}
Methods for its computation are discussed in detail in the literature, see refs.~\cite{yoshida_hypergeometric_1997,Mizera:2017rqa,Mizera:2019gea, aomoto_theory_2011}.
Similarly the homology intersection matrix $\bH$ is defined by the topological intersection numbers between basis cycles $\gamma_i$ of the homology group and their duals $\check{\gamma}_j$, weighted by the value of the twist at the intersection points,
\beq
H_{ij} = [\check{\gamma}_{j}|\gamma_i]\,.
\eeq
Details can be found in refs.~\cite{yoshida_hypergeometric_1997,Duhr:2023bku,kita_intersection_1994-2}. 

We note that the matrices $\bP$, $\bs{\check{P}}$, $\bC$ and $\bH$ are not independent. This can be seen as a consequence of the completeness relation for the bases of (dual) twisted cocycles (a similar relation exists for the bases of (dual) loaded cycles):
\beq\sum_{i,j}\frac{1}{(2\pi i)^n}|\check{\varphi}_i\rangle(\bC^{-1})_{ij}\langle\varphi_j| = \mathds{1}\,.
\eeq
The completeness relation gives rise to the \textit{twisted Riemann bilinear relations} (TRBRs)~\cite{Cho_Matsumoto_1995}
\begin{align}
\label{eq:TRBRs}
\bH = \frac{1}{(2\pi i)^n}\,\bP^T\big(\bC^{-1}\big)^T\bs{\check{P}}\,.
\end{align}

\paragraph{Differential equations.} We now assume that the integral in eq.~\eqref{eq:twisted_int} (and thus the twist $\Psi$ and the differential forms $\bs\varphi$) depends on some external parameters $\bs{s}$. In the context of Feynman integrals, $\bs{s}$ is the vector collecting the kinematic invariants on which the integral depends.
The functional dependence of the matrices $\bP$, $\bs{\check{P}}$, $\bC$ and $\bH$ is governed by some first-order linear differential equations. For example, the period matrix $\bP$ satisfies the equation (cf.~eq.~\eqref{eq:DEQ_prototype})
\begin{align}
\label{deqperiod}
    \rd \bs{P} = \bs{\Omega} \bs{P}\, ,
\end{align}
where $\bs{\Omega}$ is a matrix of rational functions in $\eps$ and of rational one-forms in $\bs{s}$, and $\rd = \sum_i\rd s_i\,\partial_{s_i}$ is the total differential with respect to the external parameters $\bs{s}$. Similarly, the dual period matrix satisfies
\begin{align}
\label{deqperiod_dual}
    \rd \bs{\check{P}} = \bs{\check{\Omega}} \bs{\check{P}}\, .
\end{align}
In the context of Feynman integrals, if we focus on maximal cuts and we work with a self-dual basis, then eq.~\eqref{eq:self-duality} implies that $\bs{\check{\Omega}}(\eps)=\bOmega(-\eps)$.
The homology intersection matrix is independent of $\bs{s}$, $\rd\bH=0$. As a consequence, the TRBRs and the differential equations~\eqref{deqperiod} and~\eqref{deqperiod_dual} for the period matrix and its dual combine to a differential equation for the cohomology intersection matrix
\begin{align}\label{DEQC}
\rd \bC=\bs{\Omega} \bC+\bC\bs{\check{\Omega}}^T\,.
\end{align} 
Up to normalization, the intersection matrix $\bC$ can be characterized as the unique rational solution to this equation~\cite{cohomIntMatrixUniqueness}.

\subsubsection{Twisted symmetry transformations}
\label{sec:summary_twisted_sym}
We now argue that the concept of symmetry transformations for (families of) Feynman integrals has a very natural interpretation in the context of twisted cohomology.

\paragraph{Definition.} Consider a twisted cohomology theory as defined in the previous section (which may or may not be associated to a family of Feynman integrals). There is a natural way to define a concept of `sectors' on such a twisted cohomology theory. The on-shell variety can be decomposed into different components along which the single-valued $n$-forms develop poles not regulated by the twist:
\beq\label{eq:on-shell_decomposition}
D_- = \bigcup_{i=1}^P D_-^{(i)}\,.
\eeq
To every cocycle $\varphi$ we may then associate a list $\vartheta^-(\varphi) = (r_1^-,\ldots,r_P^-)$  with
\beq\label{eq_cocycle}
r_i^- = \left\{\begin{array}{ll}
1\,, & \textrm{ if $\varphi$ has a pole along $D_-^{(i)}$},\\
0\,,& \textrm{ otherwise.}
\end{array}\right.
\eeq
A \emph{sector} is then defined as the subspace of the twisted cohomology group generated by those $\varphi$ with a $\vartheta^-(\varphi) \preceq \Theta$, for some fixed $\Theta = (r_1^-,\ldots,r_P^- )\in \{0,1\}^P$. Note that this definition naturally mirrors the definition of a sector encountered in the context of Feynman integrals (cf.~section~\ref{sec:summary_Feynman_integrals}). Indeed, for Feynman integrals the on-shell variety $D_-$ is the locus where a subset of propagators is singular. It may be decomposed as in eq.~\eqref{eq:on-shell_decomposition}, with $D_-^{(i)}$ the locus in loop-momentum space where the $i^{\textrm{th}}$ inverse propagator vanishes, $D_i=0$. Similarly to the case of Feynman integrals, we define $\bd^-_{\Theta}\subseteq \{1,\ldots,P\}$ to be the set of indices such that there are differential forms from this sector that have a pole along $D_-^{(i)}$, $i\in\bd^-_{\Theta}$. We also define
\beq \label{eqDmT}
D_{-,\Theta} = \bigcup_{i\in \bd^-_{\Theta}}D_-^{(i)}\,,\qquad X_{\Theta} = \mathbb{C}^n\setminus (\Sigma\cup D_{-,\Theta})\,.
\eeq

We can now define the analogue of the symmetry transformations from section~\ref{sec:summary_symmetries_def}.
A \emph{twisted symmetry transformation} from $\Theta_1$ to $\Theta_2$ of a twisted cohomology theory is an affine and bijective map \beq\label{eq_twistedsym}
f:X_{\Theta_2}\to X_{\Theta_1}, \qquad \bx_1 = \bA\bx_2+\bs{b}\,,
\eeq
such that
\begin{enumerate}
    \item The Jacobian is trivial, $\det(f) := \det\bA=\pm1$.
    \item The twist is invariant, $\Psi(f(\bx),\bs{s}) = \Psi(\bx,\bs{s})$. This implies in particular that $f$ maps the twisted variety $\Sigma$ to itself.
    \item $f$ bijectively maps $D_{-,\Theta_2}$ to $D_{-,\Theta_1}$, i.e., there is a bijection $\alpha_-: \bd^-_{\Theta_2}\to \bd^-_{\Theta_1}$ such that $f\Big(D_-^{(i)}\Big) = D_-^{(\alpha_-(i))}$, for all $i\in \bd^-_{\Theta_2}$.
\end{enumerate}
We denote the set of symmetry transformations between $\Theta_1$ and $\Theta_2$ by $\TSym(\Theta_1,\Theta_2,\bs{s})$.
For $\Theta_1=\Theta_2$, the twisted symmetry transformations from a sector to itself form a group, which we denote by $\operatorname{TAut}(\Theta,\bs{s})$. If it is clear from the context what the value of $\bs{s}$ is, we will typically simply write $\TSym(\Theta_1,\Theta_2)$ and $\operatorname{TAut}(\Theta)$.

\paragraph{Invariance of the pairings.}

Consider a twisted symmetry transformation $f:X_{\Theta_2}\to X_{\Theta_1}$ between two sectors with $f(\bx_2) = \bA\bx_2+\bs{b}$. It induces linear maps between the corresponding twisted cohomology and homology groups. To see how that works in detail, consider a differential form from the sector $\Theta_1$,
\beq
\varphi = \rd^nx_1\,R(\bs{x}_1)\,,
\eeq
where $R$ is a rational function with possible poles along $\Sigma\cup D_{-,\Theta_1}$, and holomorphic everywhere else on $X_{\Theta_1}$. 
Then we associate to it the following differential form: 
\beq \label{eqpullback}
f^*\varphi = \rd^nx_2 \,R(f(\bx_2))\,\det(f)=\rd^nx_2 \,R(\bA\bx_2+\bs{b})\,\det(f)\,.
\eeq
It  is easy to see that, since $f$ is a twisted symmetry transformation, $f^*\varphi$ has poles at most along $\Sigma\cup D_{-,\Theta_2}$ and it is holomorphic everywhere else on $X_{\Theta_2}$. Hence, $f^*\varphi$ defines a differential form in the sector $\Theta_2$. Similarly, given a loaded cycle from sector $\Theta_2$, we can define the loaded cycle
$f_*\gamma = f(\gamma)$ in sector $\Theta_1$, and we obtain in this way a linear map between the twisted homology groups. The linear maps $f^*$ and $f_*$ are called \emph{pullback} and \emph{pushforward}, respectively. Note that they go in opposite directions: $f^*$ takes a differential form from sector $\Theta_1$ and associates to it a differential form from sector $\Theta_2$. Instead, $f_*$ starts from a cycle from sector $\Theta_2$ and returns a cycle from sector $\Theta_1$.

Let us now consider the period pairing between twisted cycles and differental forms. Changing variables according to $\bx_1 = f(\bx_2)$, we find
\beq\bsp \label{eq:invarianceP}
\langle \varphi|f_*\gamma] &\,= \int_{f(\gamma)}\rd^nx_1\,\Psi(\bx_1,\bs{s})\,R(\bx_1)\\
&\,= \int_{\gamma}\rd^nx_2\,\Psi(f(\bx_2),\bs{s})\,R(f(\bx_2))\,\det(f)\\
&\,= \langle f^*\varphi|\gamma]\,,
\esp\eeq
where in the last step we used the fact that the twist is invariant. In other words, we see that the period pairing is invariant under twisted symmetries.
In section~\ref{subsecpairing} we will show that also the dual period pairing and the interaction pairings are invariant:
\beq\bsp\label{prop1}
[f_*\check{\gamma}|\check{\varphi}\rangle &\,= [\check{\gamma}|f^*\check{\varphi}\rangle\,,\\
\langle f^*\varphi_1|f^*\check{\varphi}_2\rangle &\,= \langle\varphi_1|\check{\varphi}_2\rangle\,,\\
[f_*\check{\gamma}_1|f_*\gamma_2]&\,= [\check{\gamma}_1|\gamma_2]\,.
\esp\eeq

\paragraph{The symmetry group of a sector.} Let us now focus on the case $\Theta_1=\Theta_2=\Theta$, i.e., we consider the group $\operatorname{TAut}(\Theta)$ of twisted symmetry transformations of $\Theta$. We assume $\operatorname{TAut}(\Theta)$ to be a finite group (which will be the case relevant for Feynman integrals). We then obtain a representation of $\operatorname{TAut}(\Theta)$ on the twisted cohomology group $\cH=H_\text{dR}^n(X,\nabla_{\omega})$. Similarly, we obtain linear representations on the twisted homology group, and on the dual groups. We can then apply tools from the representation theory of finite groups to study the twisted symmetry transformations of the sector $\Theta$ (see appendix~\ref{recapgrouptheory} for a concise review of group theory).

An important question in the representation theory of groups is if a given representation is reducible. In section~\ref{sec:single_sector} we study the decomposition of the action of $\operatorname{TAut}(\Theta)$ on $\cH$ into irreducible representations. 
We may pick bases of the twisted (co-)homology groups such that the elements $\varphi_j$ and $\gamma_i$ transform in irreducible representations of $\operatorname{TAut}(\Theta)$. The invariance of the period pairing, together with standard arguments from group theory, then imply that $\langle \varphi_j|\gamma_i]=0$ unless $\varphi_j$ and $\gamma_i$ transform in some equivalent irreducible representations. In particular, we can order the basis elements such that the period matrix $\bP$ is block-diagonal, and the blocks correspond to irreducible representations of $\operatorname{TAut}(\Theta)$. A similar argument applies to the other pairings, and so we conclude that, if we pick bases of all (co-)homology groups and their duals aligned with the decomposition into irreducible representations, then the period matrices $\bP$ and $\bs{\check{P}}$ and the intersection matrices $\bC$ and $\bH$ are block-diagonal, and the blocks reflect the decomposition into irreducible representations.
 This implies that the matrices $\bOmega$ and $\bs{\check{\Omega}}$ describing the differential equations~\eqref{deqperiod} and~\eqref{deqperiod_dual} will also be block diagonal.

It is well known from the theory of finite groups that the decomposition into irreducible representations is controlled by the characters of the representation, defined as the trace of the representation matrices $\bD$,
\beq
\chi_{\bD}(f) = \Tr(\bD_f)\,,\qquad f\in \operatorname{TAut}(
\Theta)\,.
\eeq
Details about characters and how they can be used to decompose a representation into irreducible representations are reviewed in appendix~\ref{recapgrouptheory}. Here it suffices to recall that characters are constant on conjugacy classes, and the multiplicity $m_k$ with which an irreducible representation $\bD^{(k)}$ appears in the decomposition of $\bD$ is given by projecting the character $\chi_{\bD}$ of $\bD$ onto the character $\chi_k := \chi_{\bD^{(k)}}$. We then see that, in order to understand the decomposition of $\bD$ into irreducible representation, we need to study its characters $\chi_{\bD}$. In section~\ref{sec:OrbitSpaceEuler} we prove one of the main mathematical results of our paper: Under certain technical assumptions on the twisted cohomology group, the character $\chi_{\bD}(f)$ can be related to a purely topological invariant of $X$, namely the absolute value of  the Euler characteristic of the manifold of fixed points of $f$ in $X$:
\beq
\label{eq:summary_character_to_Euler}
\chi_{\bD}(f) = (-1)^n\chi(X_f)\,,
\eeq
where $X_f = \big\{\bs{x}\in X: f(\bs{x})=\bs{x}\big\}$ is the fixed-point set of $f$. The detailed proof of this result will be presented in section~\ref{sec:OrbitSpaceEuler} and heavily relies on tools from algebraic topology, in particular Lefschetz numbers. We find it remarkable that such a simple relation between  group-theoretical aspects of $\operatorname{TAut}(\Theta)$ and topological properties of the space $X$ exists. Equation~\eqref{eq:summary_character_to_Euler} will be the key to understand how we can compute the number of master integrals in the presence of symmetries. 

\paragraph{Symmetries of families of integrals.} So far, we have defined twisted symmetry transformations as affine maps on the space $X$ that preserve the data that define the twisted cohomology theory. As a consequence, we obtain  natural actions of twisted symmetry transformations on the (co-)homology groups and their duals, and these actions preserve the pairings between these groups. In other words, twisted symmetry transformations act on \emph{integrands} and \emph{integration contours}. The symmetry transformations defined in section~\ref{sec:summary_symmetries_def}, however, act on families of \emph{integrals}. We now briefly comment on the relationship between these two notions of symmetry transformations. A more detailed discussion will be provided in section~\ref{sec:symmetries_of_families}. 

Let us fix an integration cycle $\gamma$ from the relevant twisted homology group. We assume that $\gamma\subseteq \mathbb{R}^n$ is a real cycle.\footnote{This assumption is always satisfied for Feynman integrals. See, e.g., the integration cycles in the Baikov representation in eq.~\eqref{eq:Baikov_contour} and the Lee-Pomeransky representation in eq.~\eqref{LeePomeransky}. Note that any complex contour deformation due to singularities can be reabsorbed into the integrand by an $i\epsilon$ prescription with the contour remaining real.} Let $\cV_{\gamma}$ denote the $\cF$-vector space generated by all twisted periods of the form $\langle\varphi|\gamma]$, where $\varphi$ is a twisted differential form. Then $\cV_{\gamma}$ defines a family of integrals, and we may define sectors on $\cV_{\gamma}$ by using the sectors on the twisted cohomology group. We now pick a differential form $\varphi$ from a sector $\Theta_1$, and a twisted symmetry transformation $f\in\TSym(\Theta_1,\Theta_2)$. Then via the change of variables $\bs{x}_1=f(\bs{x}_2)$, we obtain the differential form $f^*\varphi$ from the sector $\Theta_2$. However, this change of variables will not automatically leave the integration cycle $\gamma$ invariant. We only obtain an integral $\langle f^*\varphi|\gamma] = \langle\varphi|f_*\gamma]$ from the family $\cV_\gamma$ if $f_*\gamma=f(\gamma)$ leaves $\gamma$ invariant setwise, possibly up to changing its orientation. The change of orientation is captured by the determinant of $f$, and so we need to demand that $f_*\gamma=\det(f)\,\gamma$. Based on these considerations, we define the set of \emph{symmetry transformations between $\Theta_1$ and $\Theta_2$ of the family $\cV_{\gamma}$} by
\beq
\Sym_{\gamma}(\Theta_1,\Theta_2,\bs{s}) = \big\{f\in\TSym(\Theta_1,\Theta_2,\bs{s}): f_*\gamma=\det(f)\,\gamma\big\}\,.
\eeq
We also use the notation $\Aut_{\gamma}(\Theta,\bs{s}) := \Sym_{\gamma}(\Theta,\Theta,\bs{s})$, and we will often omit the dependence on the kinematic point $\bs{s}$. Then, the set of symmetry transformations of the family $\cV_\gamma$ is a subset of all twisted symmetries. Let us make some comments about this definition. First, in the context of Feynman integrals defined in loop-momentum space discussed in section~\ref{sec:summary_symmetries_def}, we did not need to distinguish between the symmetry transformations of the integrand and those of the full integral. This is due to the fact that the integration contour in loop-momentum space is $\mathbb{R}^{DL}$, and this space is invariant under all invertible affine transformations. In more general situations, however, one can easily construct affine transformations that leave the twist invariant, but not the  integration cycle $\gamma$ defining the family. Second, in the case $\Theta_1=\Theta_2=\Theta$, it is easy to see that the determinant defines a one-dimensional (and thus irreducible) representation of $\Aut_{\gamma}(\Theta)\subseteq \operatorname{TAut}(\Theta)$. In other words, the integration cycle $\gamma$ transforms in this irreducible representation. Since the integration pairing is zero unless $\gamma$ and $\varphi$ transform in equivalent irreducible representations of $\Aut_{\gamma}(\Theta)$, we conclude that the differential forms $\varphi\in\cH$ that provide a non-zero element $\langle\varphi|\gamma]\in\cV_{\gamma}$ are those that transform in the determinant representation of $\Aut_{\gamma}(\Theta)$. We denote the irreducible subspace of $\cH$ transforming in the determinant representation by
\beq\label{eq:summary_HG}
\cH_{G} = \big\{\varphi\in\cH: f^*\varphi = \det(f)\,\varphi\big\}\,.
\eeq
Let us briefly compare this result to the result obtained in ref.~\cite{Gasparotto:2023roh}, where it was argued that in the presence of symmetries the relevant differential forms transform in the trivial representation of the symmetry group (rather than the determinant representation considered here). There is no contradiction. Indeed, ref.~\cite{Gasparotto:2023roh} only considered symmetry transformations with determinant $+1$, and in that case the determinant and trivial representations agree. In more general situations, where we allow orientation changes, we need to consider the determinant representation. This distinction is particularly crucial when studying the decomposition into irreducible representations, which will be the key to understand the number of master integrals in the presence of a non-trivial symmetry group.

\subsection{Symmetries and Feynman integrals}\label{sec:summary_twistedFeynman}

In section~\ref{sec:summary_twisted_sym} we have introduced symmetry transformations for twisted cohomology theories, and we have discussed some of their properties, like the invariance of the pairings and the connection between the decomposition into irreducible representations. We now connect this back to Feynman integrals.

\subsubsection{Twisted symmetry transformations and Feynman integrals}
We have already seen that there is no unique twisted cohomology theory that we can associate with a family of Feynman integrals. Hence, we would naively expect that also the set of symmetry transformations of a family may depend on the choice of twisted cohomology theory. This, however, sounds unnatural, as we expect the properties of a family of Feynman integrals to be independent of the integral representation used to define the family. 

In section~\ref{sec:twistedsymparam} we show that there is indeed a universal set of symmetry transformations that one can attach to a family of Feynman integrals, irrespective of whether we work with the twisted cohomology theory obtained from the loop-momentum, Baikov, Lee-Pomeransky or Feynman parameter representations.
More precisely, we will see in section~\ref{sec:twistedsymparam} that for each of these integral representations, the set of symmetry transformations between the sectors $\Theta_1$ and $\Theta_2$ is in bijection with $\mathbb{S}(\cG_1,\cG_2)\times \mathbb{X}(\Theta_1,\Theta_2)$, where $\cG_i$ is the Lee-Pomeransky polynomial of the sector $\Theta_i$. Here $\mathbb{X}(\Theta_1,\Theta_2)$ is a set of symmetry transformations that depend on the specific twisted cohomology theory used to define the family, but the elements of $\mathbb{X}(\Theta_1,\Theta_2)$ leave the integrals trivially invariant. The first factor $\mathbb{S}(\cG_1,\cG_2)$ is universal, and acts non-trivially on the integrals. Remarkably, this universal factor can be described explicitly as the set of permutations of the Feynman parameters that map the Lee-Pomeransky polynomial from one sector to the other. 

In eq.~\eqref{eq:S_T1_T2}, we have seen that the set of symmetry transformations in loop-momentum space can be described by $\mathbb{S}(\cG_1,\cG_2)$ and $\mathbb{Z}_2$ factors that flip the signs of a subset of edge momenta and trivially leave the integrals invariant. We have also seen that the permutations from $\mathbb{S}(\cG_1,\cG_2)$ can be lifted to symmetry transformations in loop-momentum space via purely graph-theoretical methods. In section~\ref{sec:twistedsymparam} we will show that this is a more general phenomenon: also in the Baikov, Lee-Pomeransky and Feynman parameter representations can all twisted symmetry transformations of the family that act non-trivially be constructed from a permutation from $\mathbb{S}(\cG_1,\cG_2)$.  Since $\mathbb{S}(\cG_1,\cG_2)$ is universal and does not depend on the integral representation used to define the family, we will refer to it as \emph{the} set of symmetry transformations of the family associated to the Feynman graph $G$, and we denote it by $\Sym_G(\Theta_1,\Theta_2,\bs{s})$ (or $\Aut_G(\Theta,\bs{s})$ in the case $\Theta_1=\Theta_2=\Theta$). We will again often suppress the dependence on the kinematic point $\bs{s}$. 
{We stress that these symmetries consist by definition of affine change of variables. Very recently, also non-linear changes of variables have been considered that leave the Feynman parameter representation invariant~\cite{Buccioni:2023okz,Becchetti:2025rrz,Coro:2025vgn}. These are not considered by our consideration, and it would be interesting to investigate if or how these non-linear symmetries are related to the symmetries considered here.

\subsubsection{Symmetries and the canonical intersection matrix}\label{sec_canonicalinv}
Let us now focus on a fixed sector $\Theta$.
In section~\ref{sec:summary_twisted_sym} we have seen that we may pick a basis of the twisted homology and cohomology groups and their duals such that the period and intersection matrices are block-diagonal with respect to irreducible representations of the symmetry group $G_{\bs{s}}:=\Aut_G(\Theta,\bs{s})$ of that sector. We now argue that if we work with a canonical basis of master integrals, the result is even stronger. 

Consider the vector $\bI(\bs{s},\eps)$ of maximal cuts of the master integrals of the sector $\Theta$, and assume that we rotate to a canonical basis as in eq.~\eqref{conventiontrafo}. In ref.~\cite{Duhr:2024xsy} it was shown that after rotation to the canonical basis, the cohomology intersection matrix is constant in $\bs{s}$,\footnote{We assume that we work in a self-dual basis of maximal cuts, where the vector of dual maximal cut integrals is $\bs{\check{I}}(\bs{s},\eps) = \bI(\bs{s},-\eps)$~\cite{Duhr:2024rxe}.}
\beq\label{eq_deltacanonical}
\bU(\bs{s},\eps)\bC(\bs{s},\eps)\bU(\bs{s},-\eps)^T = f(\eps)\,\bDelta\,,
\eeq
where $f$ is a rational function, and $\bDelta$ is a constant matrix with rational numbers as entries.

Consider now the symmetry group $G_{\bs{s}}$ of this sector. We know that the symmetry group depends on the kinematic point $\bs{s}$. Assume that there is a point $\bs{s}_0$ where the symmetry is enlarged. For example, assume that for generic values of $\bs{s}$ the group $G_{\bs{s}}$ is trivial, but $G_{\bs{s}_0}$ is not. 
In section~\ref{eq:canonical_basis} we will show that we may then pick a canonical basis $\bJ(\bs{s},\eps)$ such that the entries of $\bJ(\bs{s}_0,\eps)$ transform in irreducible representations of $G_{\bs{s}_0}$. Note that $\bJ(\bs{s}_0,\eps)$ is the original vector of master integrals evaluated at $\bs{s}=\bs{s}_0$, which may be larger than the vector obtained by solving IBPs directly for $\bs{s}=\bs{s}_0$ (including symmetries). From the previous argument, we know that the intersection matrix $\bDelta$ in such a basis is block-diagonal and constant, and thus independent of $\bs{s}$. Hence, $\bDelta$ is the same at the point $\bs{s}_0$ with enlarged symmetry and at a generic kinematic point $\bs{s}$. In other words, we find that, if the canonical master integrals for general $\bs{s}$ are chosen according to the irreducible representations of $\Aut(\Theta,\bs{s}_0)$, then remarkably the canonical intersection matrix is block-diagonal even for general $\bs{s}$, with the blocks on the diagonal corresponding to the irreducible representations of the enlarged symmetry group.

\subsection{The number of master integrals and Euler characteristics}\label{sec:summary_number_MIs}

One of the simplest numbers one may attach to a given sector $\Theta$ of a family of Feynman integrals is the number $N_{\Theta}$ of master integrals in that sector. It is an interesting question if it is possible to determine the number of master integrals of the sector $\Theta$ without having to solve the IBP relations explicitly. This issue was first addressed in ref.~\cite{Lee:2013hzt}, where it was shown that, under certain assumptions, the number of master integrals for a given sector can be obtained by counting the number of critical points associated with the Lee-Pomeransky polynomial, i.e., the number of solutions to the equations
\beq
\partial_{x_e}\log\cG(\bx,\bs{s}) = 0\,,\qquad 1\le e\le P\,.
\eeq
In ref.~\cite{Bitoun:2017nre} it was shown that (using genericity assumptions for the external kinematics and the propagator exponents), the number of master integrals can be given as the Euler characteristic of a certain complex, and in ref.~\cite{aomoto_theory_2011} it was shown that this Euler characteristic is naturally associated with the twisted cohomology group describing the sector.

However, it was soon realized that in the presence of symmetries, the Euler characteristic overcounts the number of master integrals. In section~\ref{sec:number_MIs}, we will show that the number $N_{\Theta}$ of master integrals can also be obtained from an Euler characteristic computation in the presence of symmetries:
\beq\label{eq:Euler_orbit}
N_{\Theta} = \frac{1}{|G_{\bs{s}}|}\sum_{\sigma\in G_{\bs{s}}}\big|\chi(X_{\sigma})| = \frac{1}{|G_{\bs{s}}|}\sum_{[\sigma]\in \mathcal{C}_{G_{\bs{s}}}}n_{[\sigma]}\big|\chi(X_{\sigma})|\,,
\eeq
where $\mathcal{C}_{G_{\bs{s}}}$ is the set of conjugacy classes of $G_{\bs{s}}$ and $n_{[\sigma]}$ is the number of distinct elements conjugate to $\sigma$. Here $\chi(X_{\sigma})$ is the Euler characteristic of the submanifold $X_{\sigma}$ of $\mathbb{C}^P\setminus\Sigma$ that is left invariant by  $\sigma$ (where $\sigma$ acts on $\mathbb{C}^P$ in the natural way by permutation of the coordinates), and $\Sigma$ is the twisted variety.  We will rigorously prove eq.~\eqref{eq:Euler_orbit} in section~\ref{sec:number_MIs}. In a nutshell, the proof proceeds as follows: we know that the differential forms relevant to define our family of integrals transform in the determinant representation of $G_{\bs{s}}$ (cf.~eq.~\eqref{eq:summary_HG}). Under the right assumptions on the the twisted cohomology theory (which are very similar to those that underlie the original relation between the number of master integrals and the Euler characteristic in refs.~\cite{Lee:2013hzt,Bitoun:2017nre,Mastrolia:2018uzb}), the number of master integrals can be related to the dimension of $\cH_G$. This dimension can in turn be computed using standard techniques from group theory using the character of the representation, and from eq.~\eqref{eq:summary_character_to_Euler} we know that the characters are related to the Euler characteristics of the fixed-point sets. Note that, in order to prove eq.~\eqref{eq:Euler_orbit}, we need to combine various results introduced in previous subsections.

Let us conclude by commenting on the practical applicability of eq.~\eqref{eq:Euler_orbit}. One may wonder how easy or complicated it is to compute the different ingredients entering the right-hand side of eq.~\eqref{eq:Euler_orbit} in practice. It turns out that in the context of the Lee-Pomeransky polynomial, all ingredients can be efficiently computed. First, the group $G_{\bs{s}}$ is easy to determine, e.g., using the algorithm from ref.~\cite{Pak:2011xt}. Second, the invariant submanifolds $X_{\sigma}$ are also easy to describe: It is well known that every permutation $\sigma\in S_P$ admits a decomposition into a product of \emph{cycles}, $\sigma = c_1\cdots c_r$, and this decomposition is unique up to the ordering of the cycles. Here the cycle $c_k$ is just a cyclic permutation,
\beq
c_k : i_{k,1}\to i_{k,2}\to\ldots \to i_{k,|c_k|} \to i_{k,1}\,.
\eeq
The invariant submanifold $X_{\sigma}$ is then simply the set of points $\bx\in\mathbb{C}^P$ such that
\beq
x_{i_{k,1}} = \ldots= x_{i_{k,|c_k|}}\,,\qquad \textrm{for all } 1\le k\le r\,.
\eeq
Finally, the Euler characteristic $\chi(X_{\sigma})$ can then be obtained by computing the number of critical points of the Lee-Pomeransky polynomial restricted to $X_{\sigma}$, i.e., $\chi(X_{\sigma})$ is obtained as the number of solutions to the equations
\beq
\partial_{x_e}\log\cG(\bs{x},\bs{s})|_{X_{\sigma}}=0\,.
\eeq
The number of critical points can easily be computed using public computer codes, e.g.,~ref.~\cite{ChestnovCrisanti2025}, and a public package implemeting eq.~\eqref{eq:Euler_orbit} is currently being developed~\cite{packagemasters}.

\subsection{Examples}
In this section we will illustrate our formula to compute the number of master integrals on banana integrals with up to four loops. 
The Symanzik polynomials of the $L$-loop banana graph are
\beq\bsp\cU_L(x_1,\ldots,x_{L+1}) &\,= x_1\cdots x_{L+1}\sum_{k=1}^{L+1}\frac{1}{x_k}\,,\\
\cF_{L}(x_1,\ldots,x_{L+1},m_1^2,\ldots,m_{L+1}^2) &\,= (-s)x_1\cdots x_{L+1} + \cU_L(x_1,\ldots,x_{L+1})\sum_{k=1}^{L+1}m_k^2x_k\,, 
\esp\eeq
with $s=p^2$.
The Lee-Pomeransky polynomial is
\beq
\cG_L(\bs{x},\bs{s}) = \cF_L(\bs{x},\bs{s}) + \cU_L(\bs{x})\,,
\eeq
with $\bs{x}=(x_1,\ldots,x_{L+1})$ and $\bs{s}=(s,m_1^2,\ldots,m_{L+1}^2)$.
In the following, we assume that all masses are non zero.

In the bulk, where all the masses are distinct, the symmetry group of the top sector\footnote{By top sector we mean here the sector where at most the $(L+1)$ propagators defining the banana integral are active.} $\Theta_{\text{top}}$ is trivial. Then, as already mentioned above, under similar assumptions as in refs.~\cite{Lee:2013hzt,Bitoun:2017nre,Mastrolia:2018uzb} the number of master integrals in the top sector (modulo lower sectors) can be computed from the Euler characteristic 
\beq\label{eq:bananq_Euler_bulk}
N_{\Theta_{\text{top}}} = \big|\chi(\mathbb{C}^{L+1}\setminus \Sigma)\big|\,,
\eeq
where $\Sigma$ is the locus in $\mathbb{C}^{L+1}$ where the Lee-Pomeransky polynomial vanishes. The precise value of $N_{\Theta_{\text{top}}}$ depends on $L$.

In the equal-mass case, $m^2:=m_1^2=\ldots=m_{L+1}^2$, the top-sector of the family has a $S_{L+1}$ symmetry, corresponding to a permutation of the $L+1$ propagators, or equivalently, of the $L+1$ Feynman parameters. There are always precisely $L$ master integrals in the top sector for equal masses (cf.,~e.g.,~refs.~\cite{Klemm:2019dbm,Bonisch:2020qmm}). Except for $L=1$, this number is considerably lower than the absolute value of the Euler characteristic in eq.~\eqref{eq:bananq_Euler_bulk}.

In the remainder of this section, we present the number of master integrals for all possible configurations where some of the masses are equal (and non-zero) up to four loops. All of these configurations have symmetry groups that are subgroups of $S_{L+1}$. We have checked that we obtain the correct number of master integrals for each case by comparing to KIRA~\cite{Maierhofer:2017gsa,Klappert:2020nbg,Lange:2025ofh}.

\begin{example}[The two-loop sunrise integral]
\label{exsunrise}

Let us start by discussing the case $L=2$, which corresponds to the two-loop sunrise integral. In the bulk, where all three masses are distinct, there are four master integrals in the top sector. 

The group $S_3$ is generated by the two permutations $(12)$ and $(123)$. They act on $\mathbb{C}^3$ via
\beq\bsp
\bP_{\!(12)}\bs{x}&=(x_2,x_1,x_3)^T\,,\quad \text{ with }\quad \bP_{\!(12)}=\left(
\begin{array}{ccc}
 0 & 1 & 0 \\
 1 & 0 & 0 \\
 0 & 0 & 1 \\
\end{array}
\right)\,,\\
\bP_{\!(123)}\bs{x}&=(x_3,x_1,x_2)^T\,,\quad \text{ with }\quad \bP_{\!(123)}=\left(
\begin{array}{ccc}
 0 & 0 & 1 \\
 1 & 0 & 0 \\
 0 & 1 & 0 \\
\end{array}
\right)\,.
\esp\eeq
There are three distinct conjugacy classes, which we may choose as (see appendix~\ref{recapgrouptheory} for a review of the representation theory of the symmetric group):
\beq
[\mathds{1}]\,,\qquad [(12)]\,,\qquad [(123)]\,,
\eeq
and
\beq
n_{[\mathds{1}]} = 1\,,\qquad n_{[(12)]} = 3\,,\qquad n_{[(123)]} = 2\,.
\eeq

We proceed by parametrizing the invariant subspace of each group element. For example, the fixed-point locus for the generator $(12)$ is given by 
\begin{align}
(x_1,x_2,x_3)^T=(e_1+e_2)z_1+e_3 z_2=(z_1,z_1,z_2)^T\,,
\end{align}
where the $e_i$ are standard unit vectors in $\mathbb{C}^3$.
The Lee-Pomeransky polynomial localized on the fixed-point set is
\begin{align}
\mathcal{G}_{(12)}=z_1(z_1+2 z_2)+z_1(z_1+2 z_2)(z_2+2 z_1)m^2-sz_1^2 z_2\,.
\end{align}
Let us define the set of critical points for $\cG_{(12)}$ by
\beq
\operatorname{Crit}(\cG_{(12)}) = \big\{(z_1,z_2)\in\mathbb{C}^{2}: \partial_{z_1}\log\mathcal{G}_{(12)}=\partial_{z_2}\log\mathcal{G}_{(12)}=0\big\}\,.
\eeq
We find 
\begin{align}
|\chi(X_{[(12)]})| = \# \operatorname{Crit}(\cG_{(12)})=2\,.
\end{align}
Repeating this reasoning for all conjugacy classes, we find: 
\begin{align}
\label{eq:sunrise_X}
|\chi(X_{[\mathds{1}]})|=4\,, \quad |\chi(X_{[(12)]})|=2\,,\quad |\chi(X_{[(123)]})|=1\,,
\end{align}
such that
\begin{align}
N_{\Theta_{\text{top}}}(m_1^2=m_2^2=m_3^2)=\frac{1}{6}(4+3\times 2+2\times 1)=2\,,
\end{align}
in agreement with the known result (cf., e.g., refs.~\cite{Caffo:1998du,Laporta:2004rb}).

We can repeat exactly the same reasoning for the situation where only two of the three masses are equal, for example $m_1^2=m_2^2\neq m_3^2$. In that case the symmetry group of the top sector is $S_2 = \{\mathds{1},(12)\big\} \subset S_3$. The Euler characteristics of the fixed-point sets $X_{[\mathds{1}]}$ and $X_{[(12)]}$ are still given by eq.~\eqref{eq:sunrise_X}. However, now every conjugacy class only contains a single element. We then find
\begin{align}
N_{\Theta_{\text{top}}}(m_1^2=m_2^2\neq m_3^2)=\frac{1}{2}(4+2)=3\,,
\end{align}
in agreement with the known result (cf.,~e.g., refs.~\cite{Remiddi:2013joa,Campert:2020yur}).
\end{example}

\begin{example}[The three-loop banana integral]
\label{ex3loopbanana}
Let us now discuss the case $L=3$. Three-loop banana integrals depending on four distinct masses have recently been considered in refs.~\cite{Duhr:2025kkq,Pogel:2025bca}, and configurations where at least two masses are distinct have been studied in refs.~\cite{duhr2025modularformsthreeloopbanana,Duhr:2025ppd,Maggio:2025jel,Duhr:2025ouy}. The equal-mass case has been analyzed in refs.~\cite{Bloch:2016izu,Bloch:2014qca,Broedel:2019kmn,Broedel:2021zij,Pogel:2022yat}.

If all four masses are distinct, there are 11 master integrals in the top sector.
The symmetry group $S_4$ of the equal-mass case is generated by the permutations $(14)$ and $(1234)$. The elements of $S_4$ can be grouped into 5 conjugacy classes,
\beq
[\mathds{1}]\,,\quad [(12)]\,,\quad [(123)]\,,\quad [(12)(34)]\,,\quad [(1234)]\,.
\eeq
We can compute the Euler characteristics of the fixed-point sets, for example by counting numbers of critical points (here the values of the Euler characteristics are the same in the equal- and unequal-mass cases). We find:
\begin{align}\label{s4conjeuler}
|\chi(X_{[\mathds{1}]})|&=11\,, \quad |\chi(X_{[(12))]})|=5\,,\quad\,|\chi(X_{[(123))]})|=2 \,,\\
|\chi(X_{[(12)(34))]})|&=3\,,\quad |\chi(X_{[(1234))]})|=1\notag\,.\quad 
\end{align}

Let us now compute the number of master integrals for the different mass configurations. Each mass configuration is characterized by a symmetry group $G\subseteq S_4$, and the number of master integrals is the sum of the Euler characteristics of the relevant conjugacy classes, weighted by the number of elements in that conjugacy class. The results are summarized in table~\ref{tab:3-loop-banana}.

\begin{table}[ht]
\centering
\begin{tabular}{llll}
\hline
\textbf{Mass configuration} & \textbf{\# of  masses} & \textbf{Symmetry group} & \textbf{\# of master integrals} \\
\hline
$m_1 = m_2 = m_3 =m_4$             & 1 & $S_4$                 & 3  \\
$m_1 = m_2 = m_3$             & 2 & $S_3$                 & 5  \\
$m_1 = m_2,\;\; m_3 = m_4$    & 2 & $S_2 \times S_2$     & 6  \\
$m_1 = m_2$                   & 3 & $S_2$                 & 8 \\
all distinct                   & 4 & $S_1$                 & 11 \\
\hline
\end{tabular}
\caption{Number of master integrals for different mass configurations of the three-loop banana integral family.}
\label{tab:3-loop-banana}
\end{table}
\end{example}

\begin{example}[The four-loop banana integral]
\label{ex4loopbanana}
Finally, let us discuss the case $L=4$. The equal-mass case was studied in ref.~\cite{Pogel:2022ken}, and some non-equal masses cases were considered in ref.~\cite{Maggio:2025jel}.

We proceed in  the same way as in the previous cases. The Euler characteristics of the fixed point sets for the individual conjugacy classes are given by
\begin{align}
|\chi(X_{[\mathds{1}]})|&=26\,, \quad |\chi(X_{[(12)]})|=12\,,\quad |\chi(X_{[(123)]})|=5\,,\\
|\chi(X_{[(12)(34)]})|&=6\,,\quad |\chi(X_{[(1234)]})|=2\,,\quad|\chi(X_{[(123)(45)]})|=3\,,\quad |\chi(X_{[(12345)]})|=1\,. \notag
\end{align}
We use these results to compute the number of master integrals for the different mass configurations. The results are shown in table~\ref{tab:banana_masters}.

\begin{table}[ht]
\centering
\begin{tabular}{llll}
\hline
\textbf{Mass configuration} & \textbf{\# of  masses} & \textbf{Symmetry group} & \textbf{\# of master integrals} \\
\hline
$m_1 = m_2 = m_3 =m_4= m_5$             & 1 & $S_5$                 & 4  \\
$m_1 = m_2 = m_3 = m_5$             & 2 & $S_4$                 & 7  \\
$m_1 = m_2,\;\; m_3 = m_4 = m_5$    & 2 & $S_2 \times S_3$     & 9  \\
$m_1 = m_2 = m_3$                   & 3 & $S_3$                 & 12 \\
$m_1 = m_2,\;\; m_3 = m_4$          & 3 & $S_2 \times S_2$     & 14 \\
$m_1 = m_2$                         & 4 & $S_2$                 & 19 \\
distinct masses & 5 & $S_1$ & 26\\
\hline
\end{tabular}
\caption{Number of master integrals for different mass configurations of the four-loop banana integral family.}
\label{tab:banana_masters}
\end{table}
\end{example}

\section{Symmetries of Feynman integrals from graph theory}\label{sec:symfeyn}

In this section we will give a detailed account of symmetry transformations for Feynman integrals. First, we will introduce a graph-theoretical interpretation of the edge-momentum parametrizations in loop-momentum space. Then we will explicitly show that the symmetry transformations between two sectors are captured by the permutations of Feynman parameters that map the corresponding Lee-Pomeransky polynomials into each other. The proof is constructive, in the sense that it allows us to explicitly construct the transformation in loop-momentum space corresponding to a particular permutation of the Feynman parameters. 
Unless otherwise stated, the material covered in this section is rather standard in graph theory, but we believe it has not been applied as needed in this section for Feynman integrals.

\subsection{A brief review of graph theory}
\label{sec:GraphTheory}

Let us start by reviewing some basic facts from graph theory. The material is entirely standard and can be found in many textbooks, see, e.g., ref.~\cite{oxleyBook}.

An \emph{(oriented) graph} $G$ is a triple $(E_G,V_G,\phi)$, where $E_G$ is the set of edges, $V_G$ is the set of vertices and 
\begin{equation}
    \phi:E_G\rightarrow V_G\times V_G \,,
\end{equation}
is the \emph{incidence relation}, which maps every edge of $G$ to its start and end point. In the following it will be convenient to lift the sets of edges and vertices to the free $\mathbb{Z}$-modules generated by these sets. In an abuse of notation, we will again denote these modules by $E_G$ and $V_G$. Sometimes we will also need the lift them to $\mathbb{Q}$-vector spaces, which we indicate by a subscript, e.g., $E_{G,\mathbb{Q}} = E_G\otimes \mathbb{Q}$. We can then reformulate the incidence relation via the \emph{boundary operator}
\begin{equation}
    \partial:E_G\rightarrow V_G, \qquad \partial e=v_2-v_1, \text{ for } \phi(e)=(v_1,v_2) \,.
\end{equation}

For Feynman graphs it is useful to split the edges and vertices into \emph{internal} and \emph{external},
\begin{equation}
    E_G=E_G^{\rint}\oplus E_G^{\rext}, \qquad V_G=V_G^{\rint}\oplus V_G^{\rext} \,,
\end{equation}
where a vertex is external if it is of valency one. The unique edge connected to this vertex is then called an external edge. 
For each graph $G$ we define the number of (internal) vertices $V=\dim V_G^{\rint}$, the number of internal edges (or propagators) $P=\dim E_G^{\rint}$, the number of external vertices $V_{\rext}=\dim V_G^{\rext}=\dim E_G^{\rext}$ and the number of loops $L=\dim H_1(G,\mathbb{Z})$, where
\begin{equation}
    H_1(G,\mathbb{Z})=\Big\langle x\in E_G^{\rint}\,| \, \partial x=0 \Big\rangle_{\mathbb{Z}} \,,
\end{equation}
is the \emph{first graph homology group} of $G$. We denote the dimension (or {rank}) of a free $\mathbb{Z}$-module $M$, i.e., the number of its generators, by $\dim M$.

A \emph{graph isomorphism} between two graphs $G,\Gt$ is a tuple of bijections $(\phi_E:E_G\rightarrow E_{\Gt},\,\phi_V: V_G\rightarrow V_{\Gt})$ compatible with the graph structure,
\begin{equation}
\label{eq:graph_structure}
    \tilde{\partial}\circ\phi_E=\phi_V\circ \partial \,,
\end{equation}
where $\partial,\tilde{\partial}$ are the boundary operators of $G$ and $\Gt$, respectively. If $\Gt=G$, we talk about a graph automorphism. A graph isomorphism preserves the incidence relations. This implies that the valency of a vertex is preserved, and hence a graph isomorphism separately identifies the external and internal edges and vertices.

The advantage of thinking about the sets of edges and vertices as free modules comes from the fact that we can perform linear algebra operations on them. We will encounter various matrices whose rows and columns are labeled by the edges and/or vertices of the graph. These matrices can be understood as linear operators on the free modules $E_G$ and $V_G$.
The first instance of such a matrix is the \emph{incidence matrix} $\bB\in\mathbb{Z}^{V\times P}$, defined by
\begin{equation}
\label{eq:incidenceMatrix}
    \bB_{ve}=\left\{ \begin{array}{rl} 
    +1, & \text{if } v \text{ end-point of } e\,, \\
    -1, & \text{if } v \text{ initial point of } e\,, \\
    0, & \text{else} \,,
    \end{array}\right.
\end{equation}
for every $v\in V_G^{\rint},\, e\in E_G^{\rint}$. Note that this is precisely the matrix corresponding to the $\mathbb{Z}$-module homomorphism $\partial:E^{\rint}_G\rightarrow V^{\rint}_G$.

\subsection{Edge-momentum parametrizations}
\label{sec:edgeMomParams}
Let us first discuss the edge-momentum parametrization in eq.~\eqref{eq:q_k_p} in more detail. We will introduce two concrete choices of parametrization that will be convenient for us in the following. One of these will have a clear physical interpretation, while the second one will be of technical importance in section \ref{sec:proofSymLemma}. 

Consider a Feynman graph $G$ that we associate to a sector $\Theta$ of some Feynman integral family. We can interpret the matrix $\bC$ appearing in an edge-momentum parametrization (cf. eq.~\eqref{eq:q_k_p}) as a \emph{cycle basis matrix} of the graph $G$. This means that there is a basis of $H_1(G,\mathbb{Z})$ such that 
\begin{equation}
\label{eq:CDef}
    \bC_{\ell e}=\left\{ \begin{array}{rl} 
    +1, & \text{if } e \text{ is in } \ell \text{ in the same direction}\,, \\
    -1, & \text{if } e \text{ is in } \ell \text{ in the opposite direction} \,, \\
    0, & \text{else} \,,
    \end{array}\right.
\end{equation}
for $e\in E_G^{\rint}$ and $\ell$ a basis element of $H_1(G,\mathbb{Z})$ (i.e., $\ell$ is a single loop). A change of basis corresponds to a rotation by an element of $\mathrm{GL}(L,\mathbb{Z})$ from the left. Note that $\bC^T$ can be viewed as the matrix corresponding to the inclusion map $H_1(G,\mathbb{Z})\hookrightarrow E_G^{\mathrm{int}}$ (where we interpret $E_G^{\mathrm{int}}$ as the free $\mathbb{Z}$-module generated by the internal edges of $G$). In particular, we have $\mathrm{rank} \, \bC=L$. 

Besides the cycle basis matrix $\bC$, the parametrization of the edge momenta in eq.~\eqref{eq:q_k_p} also requires the knowledge of a matrix $\bE$ that tracks the flow of the external momenta through the graph. The choice of this matrix is not unique (not even for a fixed choice of $\bC$). We now present two different ways to construct such a matrix $\bE$ from graph-theoretical input. 
The first parametrization  is the \emph{pure loop-momentum parametrization}, defined by the fact that exactly $L$ edge momenta have no contribution from external momenta. To define this parametrization we will need to choose a spanning tree\footnote{We will view $T$ as a subgraph of $G$ with the same vertex set, but with only a subset of edges, i.e., we view it as a triple $(E_T,V_G,\partial|_{E_T})$.} $T$ of $G$ as well as an external vertex $v_0$. Then, the parametrization takes the form 
\begin{equation}
\label{eq:qParamhat}
    \bq(\bk,\bp)=\bC^T\bk+\hat{\bD}{}^T\bp \,,
\end{equation}
with $\hat{\bD}\in\mathbb{Z}^{(V_{\mathrm{ext}}-1)\times P}$ the \emph{edge-flow matrix}, defined by
\begin{equation}\label{edgeflow}
    \hat{\bD}_{v e}=\left\{ \begin{array}{rl} 
    +1, & \text{if } e \text{ lies on } \gamma_v \text{ with the same orientation}\,, \\
    -1, & \text{if } e \text{ lies on } \gamma_v \text{ with the opposite orientation} \,, \\
    0, & \text{else} \,,
    \end{array}\right.
\end{equation}
where $v\in V_G^{\rext}\setminus\{v_0\}, \, e\in E_G^{\rint}$ and $\gamma_v$ is the unique path in $T$ from $v$ to $v_0$. We note that for momentum-grouped graphs we have $\mathrm{rank}\,\hat{\bD}=V_{\rext}-1$. Furthermore, observe that we can view $\hat{\bD}{}^T$ as the matrix corresponding to the map $V_G^{\rext}\setminus\{v_0\}\rightarrow E_T^{\mathrm{int}}$ that maps $v$ to (the internal edges along) $\gamma_v$. 
Since the columns of $\hat{\bD}$ corresponding to the complement of the spanning tree $T$ are zero by definition, these edges will be pure loop momenta by construction. The vertex $v_0$ corresponds to the external momentum that we eliminate using overall momentum conservation. 

There is another choice of parametrization which will be  convenient in the following. We will refer to it as an \emph{orthogonal parametrization}, and we write it as
\begin{equation}
\label{eq:qParam}
    \bq(\bk,\bp)=\bC^T\bk+\bD^T\bp \,,
\end{equation}
for a matrix $\bD\in\mathbb{Q}^{(V_{\mathrm{ext}}-1)\times P}$, which is distinguished by the fact that it is orthogonal to $\bC$ in the sense that $\bD\bC^T=0$. We can reach an orthogonal parametrization starting from an arbitrary parametrization, defined by matrices $\bC$ and $\bE$ as in eq.~\eqref{eq:q_k_p}, where $\bE$ does not necessarily admit an interpretation as an edge-flow matrix. To obtain the orthogonal parametrization we then shift the loop momenta as 
\begin{equation}
    \bk\rightarrow \bk-(\bC\bC^T)^{-1}\bC\bE^T\bp \,,
\end{equation}
and hence obtain
\begin{equation}
\label{eq:D_ort}
    \bD^T=(\mathds{1}-\bPi_C)\bE^T \,,
\end{equation}
where we project out the `loop directions' using the projection matrix $\bPi_C$ from eq.~\eqref{eq:projectorDef}.\footnote{Note that $\bPi_C$ really projects onto the $\mathbb{Q}$-vector space $H_1(G,\mathbb{Q})=\mathbb{Q}\otimes H_1(G,\mathbb{Z})$, since the matrix $\bC\bC^T$ does not necessarily have an inverse over the integers.} Observe that in this parametrization, in general every edge momentum can have contributions from the external momenta.

\begin{example}[The crossed box]
\begin{figure}[ht]
    \centering
    \includegraphics[width=0.5\textwidth]{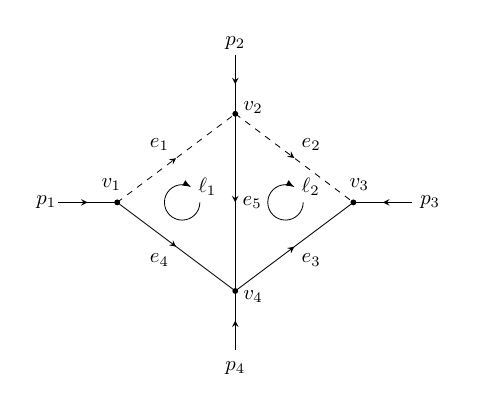}
    \caption{The oriented graph corresponding to the crossed box integral. The internal edges are labeled by $e_1,\dots ,e_5$ and the external momenta are given by $p_1,p_2,p_3,p_4$ with $p_4=-p_1-p_2-p_3$. The edges that are denoted by dashed lines are the ones that are deleted to obtain the spanning tree $T_{\mathrm{cb}}$. The loops denoted by $\ell_1,\ell_2$ form a basis of the homology group of this graph.}
    \label{fig:crossedBox}
\end{figure}
Let us illustrate the above parametrizations on the example of the crossed box graph in figure \ref{fig:crossedBox}. We first construct the pure loop-momentum parametrization described above. The cycle basis matrix associated with the basis $\ell_1,\ell_2$ of the first graph homology group shown in figure \ref{fig:crossedBox} reads:
\begin{equation}
    \bC_{\mathrm{cb}}=\begin{pmatrix}
        1 & 0 & 0 & -1 & 1 \\
        0 & 1 & -1 & 0 & -1 
    \end{pmatrix} \,.
\end{equation}
Next, let us construct the edge-flow matrix $\hat{\bD}_{\mathrm{cb}}$. We choose a spanning tree $T_{\mathrm{cb}}$ obtained from the graph by deleting the edges $e_1,e_2$ (denoted by dashed lines in figure \ref{fig:crossedBox}). We fix $v_0$ to be the vertex where the momentum $p_4$ is inflowing. We obtain the matrix
\begin{equation}
    \hat{\bD}_{\mathrm{cb}}=\begin{pmatrix}
        0 & 0 & 0 & 1 & 0 \\
        0 & 0 & 0 & 0 & 1 \\
        0 & 0 & -1 & 0 & 0
    \end{pmatrix} \,.
\end{equation}
The resulting parametrization then takes the form
\begin{equation}
    \bq_{\mathrm{cb,plm}}(\bk,\bp)=\bC_{\mathrm{cb}}^T\begin{pmatrix}
        k_1 \\ k_2
    \end{pmatrix} +  \hat{\bD}{}_{\mathrm{cb}}^T \begin{pmatrix}
        p_1 \\ p_2 \\ p_3
    \end{pmatrix}=\begin{pmatrix}
        k_1 \\ k_2 \\-k_2-p_3 \\ -k_1+p_1 \\k_1-k_2+p_2
    \end{pmatrix} \,.
\end{equation}
Notice that the first two edge momenta are pure loop momenta, because we chose to delete the edges $e_1,e_2$ to obtain the spanning tree. 

Let us now switch to an orthogonal parametrization by shifting the loop momenta. As discussed above, this yields the matrix
\begin{equation}
    \bD_{\mathrm{cb}}=\hat{\bD}_{\mathrm{cb}}(\mathds{1}-\bPi_{C_{\mathrm{cb}}})=\frac{1}{8}\begin{pmatrix}
        3 & 1 & -1 & 5 & 2 \\
        -2 & 2 & -2 & 2 & 4 \\
        -1 & -3 & -5 & 1 & 2
    \end{pmatrix} \,,
\end{equation}
which leads to the edge-momentum functions
\begin{equation}
    \bq_{\mathrm{cb,o}}(\bk,\bp)=\bC_{\mathrm{cb}}^T\begin{pmatrix}
        k_1 \\ k_2
    \end{pmatrix} +  \bD_{\mathrm{cb}}^T \begin{pmatrix}
        p_1 \\ p_2 \\ p_3
    \end{pmatrix}=\begin{pmatrix}
        k_1+\frac{1}{8}(3p_1-2p_2-p_3) \\
        k_2+\frac{1}{8}(p_1+2p_2-3p_3) \\
        -k_2-\frac{1}{8}(p_1+2p_2+5p_3) \\
        -k_1+\frac{1}{8}(5p_1+2p_2+p_3) \\
        k_1-k_2+\frac{1}{8}(2p_1+4p_2+2p_3)
    \end{pmatrix} \,.
\end{equation}
Although this does not seem like a very natural edge-momentum parametrization, it is distinguished by the fact that the matrices $\bC_{\mathrm{cb}}$ and $\bD_{\mathrm{cb}}$ determining this parametrization satisfy $\bD_{\mathrm{cb}}\bC_{\mathrm{cb}}^T=0$.
\end{example}

\subsection{Symmetry transformations between momentum-grouped sectors}
\label{eq:symmMomGrouped}

Let us now prove that the symmetry transformations between two momentum-grouped sectors as defined in section~\ref{sec:summary_symmetries_def} are in bijection with the permutations mapping the corresponding Lee-Pomeransky polynomials into each other (up to trivial redundancies), i.e., let us prove eq.~\eqref{eq:S_T1_T2}. To this end, we consider two momentum-grouped sectors $\Theta_1$ and $\Theta_2$ of some Feynman integral family, with associated Feynman graphs $G_1,G_2$. We will denote the corresponding Lee-Pomeransky polynomials in the two sectors by $\cG_i=\cU_i+\cF_i$, for $i=1,2$. Let us furthermore denote the number of propagators in these two sectors by $P$ (we have already seen in section~\ref{sec:summary_symmetries_def} that if there is a non-trivial symmetry transformation between sectors, then they necessarily have the same number of propagators). First, it will be useful to give an equivalent characterization of $\mathbb{S}(\cG_1,\cG_2)$. 
\begin{lemma}
\label{symmTransEquivDef}
    Let $\sigma\in \mathbb{S}(\cG_1,\cG_2)$. Then,
    \begin{enumerate}[label=(\alph*)]
        \item the bijection maps the $\cU$ and $\cF$ polynomials of the two sectors into each other: $\sigma\in \mathbb{S}(\cU_1,\cU_2)$ and $\sigma\in \mathbb{S}(\cF_1,\cF_2)$.
        \item the mass assignment of the Feynman graphs is compatible with the bijection, i.e., $m_e^2=m_{\sigma(e)}^2$ for all $e\in E_{G_2}^{\rint}$. 
    \end{enumerate}
    The converse direction also holds, i.e., any bijection satisfying (a) and (b) has to lie in $\mathbb{S}(\cG_1,\cG_2)$.
\end{lemma}
\begin{proof}
    The converse direction is trivial, so let us focus on showing that (a) and (b) hold for any $\sigma\in\mathbb{S}(\cG_1,\cG_2)$.
    First, recall that the Symanzik polynomials $\cU_i$ and $\cF_i$ are homogeneous polynomials in the Feynman parameters $x_e$ of degrees $L$ and $L+1$, respectively \cite{Bogner:2010kv}. The degree $L$ part of the Lee-Pomeransky polynomial $\cG_i$ is precisely given by $\cU_i$. Since $\sigma$ simply permutes the variables, it cannot mix different homogeneity degrees, and hence it follows that $\sigma\in \mathbb{S}(\cU_1,\cU_2)$. Since $\sigma\in \mathbb{S}(\cG_1,\cG_2)$, (a) follows.

    To show that the two terms in the $\cF_i$ (see eq.~\eqref{eq:UFDef}) are already individually mapped to each other, consider terms in $\cF_i$ which are of degree $2$ in a single variable $x_e$. From the definitions it is clear that such terms can only arise from the mass term. As above, $\sigma$ cannot mix such terms with terms of degree at most one in all variables, and it hence follows that $\sigma$ maps the mass terms of the two $\cF$ polynomials into each other, which proves part (b).

\end{proof}
\begin{figure}[ht]
    \centering
    \includegraphics[width=0.35\textwidth]{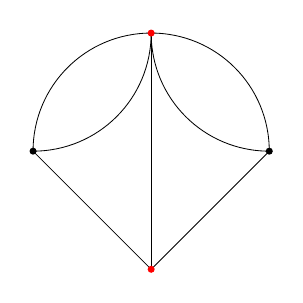}
    \includegraphics[width=0.35\textwidth]{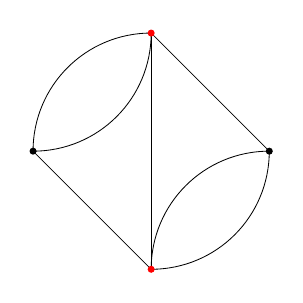}
    \caption{Example of two graphs related by a Whitney twist, which in this case is implemented by disconnecting the graphs at the top and bottom vertices, colored in red.}
\label{fig:whitneyTwist}
\end{figure}

\paragraph{Symmetry transformations as matroid isomorphisms.}
The fact that the bijection maps $\cU_2$ to $\cU_1$ implies that it maps spanning trees of $G_2$ to spanning trees of $G_1$, when viewed as a bijection on the edges of the (undirected) Feynman graphs. This implies that the map $\sigma$ is actually an isomorphism of the \emph{cycle matroids} associated to the graphs $G_1$ and $G_2$ (see appendix \ref{app:matroidIntro} for a brief introduction to matroid theory, and also ref.~\cite{Bogner:2010kv} for an introduction in the context of Feynman integrals). Such isomorphisms have been fully classified by Whitney and (for connected graphs) can only take the form of graph isomorphisms or \emph{Whitney twists} \cite{whitneyTheorem}. A Whitney twist is a procedure on the graph where we disconnect the graph at two vertices into two disconnected pieces and then reidentify the vertices in the opposite way to create another connected graph, not necessarily isomorphic to the original one; see figure \ref{fig:whitneyTwist} for an illustration. Hence, we can view $\mathbb{S}(\cG_1,\cG_2)$ as a subset of the set of matroid isomorphisms between the cycle matroids of $G_1$ and $G_2$. It will generally be a proper subset, because the additional requirement of $\sigma\in \mathbb{S}(\cF_1,\cF_2)$ (and compatibility with the mass assignment) imposes further restrictions. 

As an example for the constraints coming from the requirement $\sigma\in \mathbb{S}(\cF_1,\cF_2)$, let us show that for a Whitney twist to be a symmetry, the kinematics has to be strongly constrained.
\begin{lemma}
\label{lemmaWhitneyKin}
    Consider two Feynman graphs with Lee-Pomeransky polynomials $\cG_1,\cG_2$, associated with momentum-grouped sectors of a Feynman integral family, and assume that there is a bijection $\sigma\in\mathbb{S}(\cG_1,\cG_2)$, implemented on the graph by a Whitney twist.
    Furthermore, assume that the kinematics of the family is generic, in the sense that the only linear relation involving Mandelstam invariants of the form $p_i\cdot p_j$ for $i\neq j$ stems from momentum conservation. Then, all of the external momenta have to lie in one of the two components separated by the Whitney twist.
\end{lemma}
\begin{figure}[ht]
    \centering
    \includegraphics[width=0.7\textwidth]{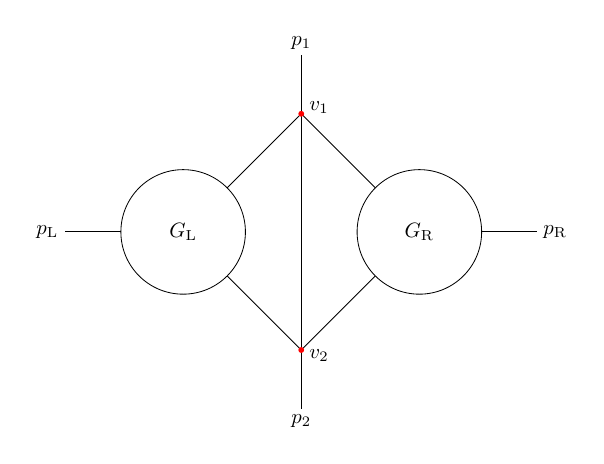}
    \caption{General setup for studying Whitney twists, which is implemented by disconnecting the graph at the vertices $v_1$ and $v_2$ (colored in red) into two components containing the subgraphs $G_{\mathrm{L}}$ and $G_{\mathrm{R}}$, and then reidentifying them in the opposite way. Note that this is a schematic figure representing multiple possible momentum-grouped graphs. In particular, the internal edges represent any number of edges in the underlying graph. Similarly, the momenta $p_{\mathrm{L}}$ and $p_{\mathrm{R}}$ denote the total momentum flowing into the left or right components and might be the sum of multiple external momenta. The momenta $p_1$ and $p_2$, however, each represent a single external momentum, because we assume the underlying graph to be momentum grouped.}
    \label{fig:whitneyTwistSetup}
\end{figure}
\begin{figure}[ht]
    \centering
    \includegraphics[width=0.48\textwidth]{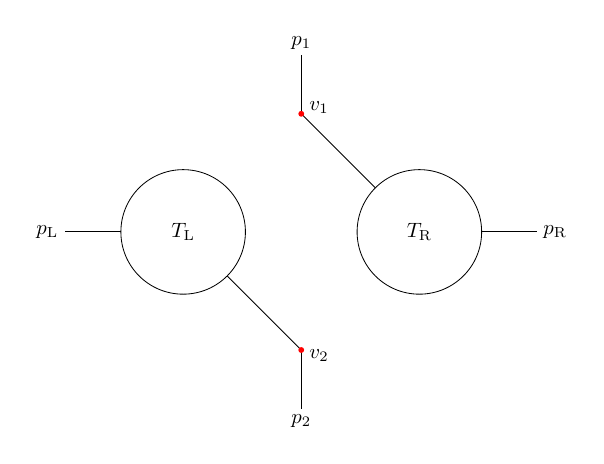}
    \includegraphics[width=0.48\textwidth]{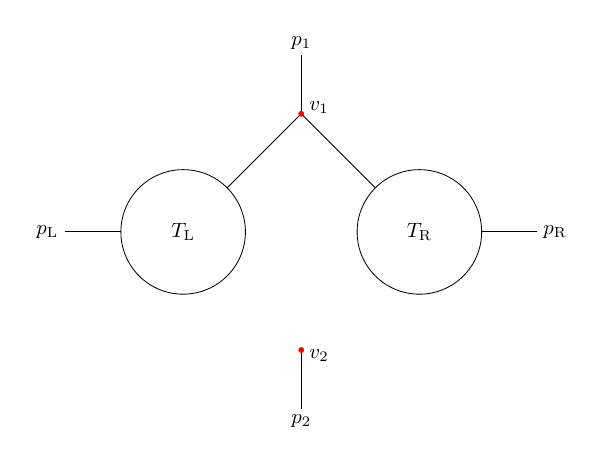}
    \caption{Two $2$-forests of a graph schematically represented by figure \ref{fig:whitneyTwistSetup}, which are mapped into each other by the Whitney twist under consideration. Here all drawn edges are single edges in the graph and $T_{\mathrm{L}}$ and $T_{\mathrm{R}}$ are spanning trees of the subgraphs $G_{\mathrm{L}}$ and $G_{\mathrm{R}}$, respectively. }
    \label{fig:whitneyTwist2Forests}
\end{figure}
\begin{proof}
    The two graphs have the same schematic form shown in figure \ref{fig:whitneyTwistSetup}. They differ by a Whitney twist acting on the component containing $G_{\mathrm{L}}$, disconnecting it from the right at the vertices $v_1,v_2$ and reconnecting it in the opposite way. Our choice of the left component being affected by the Whitney twist is arbitrary and only for definiteness, since the other choice simply differs by a graph isomorphism. The requirement that $\sigma\in\mathbb{S}(\cG_1,\cG_2)$ implies that the external invariants associated to two $2$-forests connected by $\sigma$, have to agree (cf.~eq.~\eqref{eq:UFDef}). Choosing the $2$-forests as in figure \ref{fig:whitneyTwist2Forests} (and an analogous setup reflected along the horizontal) implies the constraints
    \begin{equation}
    \label{eq:2forestRelWhitney}
        p_{\mathrm{L}}^2+2p_{\mathrm{L}}\cdot p_1=0 \,, \qquad p_{\mathrm{L}}^2+2p_{\mathrm{L}}\cdot p_2=0 \,.
    \end{equation}
    Let us separate the proof into different cases.
    \begin{itemize}
        \item Assume that $p_1,p_2\neq 0$. Then eq.~\eqref{eq:2forestRelWhitney} and our assumption of generic kinematics imply that there are no momenta entering into $G_{\mathrm{L}}$.
        \item Next, assume that $p_1\neq 0$, but $p_2=0$. The constraint in eq.~\eqref{eq:2forestRelWhitney} now implies that $p_{\mathrm{L}}^2=0$ and $p_{\mathrm{L}}\cdot p_1=0$. The only way that the condition $p_{\mathrm{L}}^2=0$ does not violate our genericity assumption is if there is only one momentum entering into $G_{\mathrm{L}}$, i.e., if  $p_{\mathrm{L}}$ is an individual (massless) momentum. The condition $p_{\mathrm{L}}\cdot p_1=0$ then further needs $p_1$ to be the only other external momentum to ensure generic kinematics. However, then there are no Lorentz invariants left over, and our graphs are equivalent to two vacuum graphs by Lorentz invariance. The same argument applies to the situation $p_1= 0,\,p_2\neq0$.       
        \item  Finally, consider the case $p_1=p_2=0$. Then, the constraint in eq.~\eqref{eq:2forestRelWhitney} implies that $p_{\mathrm{L}}^2=0$, and by momentum conservation also $p_{\mathrm{R}}^2=0$. The only way that this does not violate the genericity of the external momenta is if either one side has no external legs attached, or if there is a single momentum entering on both sides. In the latter case, however, there are again no Lorentz invariants and we can put $p_{\mathrm{L}}=p_{\mathrm{R}}=0$, and we end up with two vacuum graphs again. 
    \end{itemize}
   Hence, we see that in all cases we end up either with the two graphs having no external momenta at all or having all external momenta enter into one of the two components separated by the Whitney twist, as claimed.
\end{proof}

\paragraph{From symmetry transformations to Lee-Pomeransky symmetries.}
Let us now show that the set $\mathbb{S}(\cG_1,\cG_2)$ precisely captures the symmetry transformations in loop-momentum space (up to trivial redundancies), i.e., eq.~\eqref{eq:S_T1_T2}. To this end, we will show both inclusions in turn.

Let us first consider a symmetry transformation, i.e., an element of $\mathrm{Sym}(\Theta_1,\Theta_2)$, in the notation of section~\ref{sec:summary_symmetries_def}. This element can be represented by matrices $\bL_{\sigma},\bM_{\sigma},\bN_{\sigma}$, for some signed bijection $\sigma$. By forgetting about signs, this gives rise to an ordinary bijection $\alpha\in \mathbb{S}(\cG_1,\cG_2)$, which corresponds to the mapping of propagators under the symmetry transformation. In particular, by assumption, the mass assignment is compatible with $\alpha$. We are therefore left to show that this bijection maps the $\cU$ and $\cF$ polynomials of the two sectors to each other. This is the content of the following lemma.
\begin{lemma}
    The bijection $\alpha$ maps the Symanzik polynomials of the sectors $\Theta_1$ and $\Theta_2$ into each other: $\alpha\in\mathbb{S}(\cU_1,\cU_2)$ and $\alpha\in\mathbb{S}(\cF_1,\cF_2)$.
\end{lemma}
\begin{proof}
We will use the definitions of the $\cU$ and $\cF$ polynomials from the loop-momentum representation, i.e., we write \cite{Bogner:2010kv}
\begin{align}
\label{eq:UFPolyDef}
    \cU_i(\bx)&=\det\bV_i(\bx) \,, \nonumber \\
    \cF_i(\bx,\bs{s})&=\det\bV_i(\bx)\left(\bp^T\bW_i^T(\bx)\bV_i^{-1}(\bx)\bW_i(\bx)\bp+J_i(\bx) \right) \,,
\end{align}
where the quantities $\bV_i,\bW_i,J_i$ for $i=1,2$ can be computed via
\begin{equation}
\label{eq:VRJDef}
    \sum_{k=1}^{P_1} x_kD_{i,k}(\bk,\bp)=\bk^T\bV_i(\bx)\bk-2\bk^T\bW_i(\bx)\bp-J_i(\bx) \,,
\end{equation}
where $D_{i,k}$, with $k=1,\dots ,P_1$, are the propagators in the sector $\Theta_i$. Since $\alpha$ maps the propagators of the two sectors into each other, we have
\begin{equation}
\label{eq:paramRelFeynman}
    \sum_{k=1}^{P_1} x_k D_{1,k}(\bk,\bp)=\sum_{k=1}^{P_1} x_k D_{2,\alpha^{-1}(k)}(\bk',\bp')=\sum_{k=1}^{P_1} x_{\alpha(k)} D_{2,k}(\bk',\bp') \,.
\end{equation}
Using the symmetry transformation associated with $\sigma$, cf.~eqs.~\eqref{eq:symmTrans}, and~\eqref{eq:VRJDef}, we furthermore find
\begin{equation}
 \sum_{k=1}^{P_1} x_k D_{1,k}(\bk,\bp)=\bk'^T\bV'_1(\bx)\bk'
    -2\bk'^T\bW'_1(\bx)\bp'-J'_1(\bx) \,,
\end{equation}
with
\begin{align}
    \bV'_1(\bx)&=\bL_{\sigma}^T\bV_1(\bx)\bL_{\sigma} \,, \nonumber \\
    \bW'_1(\bx)&=\bL_{\sigma}^T\left(\bW_1(\bx)\bN_{\sigma}-\bV_1(\bx)\bM_{\sigma}\right) \,, \\
    J'_1(\bx)&=J_1(\bx)-\bp'^T\bM_{\sigma}^T\left(\bV_1(\bx)\bM_{\sigma}-2\bW_1(\bx)\bN_{\sigma}\right)\bp' \nonumber \,.
\end{align}
Now observe that, using eq.~\eqref{eq:UFPolyDef}, the quantities $\bV'_1(\bx),\bW'_1(\bx),J'_1(\bx)$ yield precisely the same Symanzik polynomials as the original $\bV_1(\bx),\bW_1(\bx),J_1(\bx)$. 
Consider the expansion 
\begin{equation}
    \sum_{k=1}^{P_1} x_{\alpha(k)} D_{2,k}(\bk',\bp')=\bk'^T\bV_2(\alpha^{-1}(\bx))\bk'
    -2\bk'^T\bW_2(\alpha^{-1}(\bx))\bp'-J_2(\alpha^{-1}(\bx)) \,.
\end{equation}
The equality in eq.~\eqref{eq:paramRelFeynman} implies that we should get the same Symanzik polynomials using the quantities $\bV_2(\alpha^{-1}(\bx)),\bW_2(\alpha^{-1}(\bx)),J_2(\alpha^{-1}(\bx))$ or $\bV'_1(\bx),\bW'_1(\bx),J'_1(\bx)$. Together with the previous results we can conclude that 
\begin{equation}
    \cU_1(\alpha(\bx))=\cU_2(\bx) \,,\qquad \cF_1(\alpha(\bx),\bs{s})=\cF_2(\bx,\bs{s}) \,,
\end{equation}
which proves the claim.
\end{proof}
This lemma, together with Lemma \ref{symmTransEquivDef}, establishes the fact that there is a map from $\mathrm{Sym}(\Theta_1,\Theta_2)$ to $\mathbb{S}(\cG_1,\cG_2)$. This map is not quite injective, because there are non-trivial symmetry transformations in loop-momentum space which only change the signs of edge momenta, and hence correspond to the trivial bijection in $\mathbb{S}(\cG_1,\cG_2)$. Such transformations that only change signs of edge momenta are strongly restricted by momentum conservation. In particular, if there is no subset of edge momenta at a vertex that already satisfy momentum conservation on their own, then the only possibility is to flip all internal and external momenta. The only way that such a subset of edge momenta already satisfies momentum conservation is if there is a vacuum subgraph that connects to the rest of the graph through a single joint vertex. In conclusion, while the map $\mathrm{Sym}(\Theta_1,\Theta_2)\rightarrow \mathbb{S}(\cG_1,\cG_2)$ is not injective, there is an injective map
\begin{equation}\label{eq:symtheta1theta2}
    \mathrm{Sym}(\Theta_1,\Theta_2)\rightarrow \mathbb{Z}_2^{c}\times \mathbb{S}(\cG_1,\cG_2) \,,
\end{equation}
where the $\mathbb{Z}^c_2$ factor precisely takes into account the ambiguities discussed above. The value of $c$ is the number of vacuum subgraphs that can be detached by deleting a single vertex, plus 1 (where the additional $\mathbb{Z}_2$ comes from changing the signs of all internal and external momenta simultaneously). 

\paragraph{From Lee-Pomeransky symmetries to symmetry transformations.}
Let us now consider the converse direction, i.e., let us consider a bijection $\alpha\in\mathbb{S}(\cG_1,\cG_2)$ and construct an associated symmetry transformation in loop-momentum space. To this end we need to lift the bijection $\alpha$ to a \emph{signed} bijection
\begin{equation}
\label{eq:signedPermFromPerm}
    \sigma=(\bkappa,\alpha)\in \mathbb{Z}_2^P\times \mathbb{S}(\cG_1,\cG_2) \,.
\end{equation}
Here $\alpha$ is the `genuine' bijection and the signs $\bkappa$ take into account orientation reversal of edges. This lift can be done in an almost unique way, where the non-uniqueness comes from precisely the same sign ambiguities as described before. In particular, to every element $\alpha\in\mathbb{S}(\cG_1,\cG_2)$ we can associate a $\sigma=(\bkappa,\alpha)\in\mathbb{Z}_2^P\times \mathbb{S}(\cG_1,\cG_2)$ where $\bkappa$ is only determined up to sign changes from $\mathbb{Z}_2^c$. Hence there is a unique signed bijection that we can associate to an element in $\mathbb{Z}_2^{c}\times \mathbb{S}(\cG_1,\cG_2)$. The following lemma then establishes a map to the set of symmetry transformations in loop-momentum space.
\begin{lemma}
\label{lemmaLPToSymm}
    For every signed bijection $\sigma\in \mathbb{Z}_2^{c}\times \mathbb{S}(\cG_1,\cG_2)$ there is a symmetry transformation from the sector $\Theta_1$ to $\Theta_2$ given as in eq.~\eqref{eq:symmMatricesFromAut_sum}.\nonumber
\end{lemma}
We will defer the somewhat technical proof of Lemma \ref{lemmaLPToSymm} to appendix~\ref{sec:proofSymLemma}. For now let us just note that this implies the existence of an inverse to the map constructed in the previous paragraph, and hence proves the bijection between the set $\mathrm{Sym}(\Theta_1,\Theta_2)$ of symmetry transformations between the sectors $\Theta_1,\Theta_2$ and the set $\mathbb{\mathbb{S}}(\cG_1,\cG_2)$ of bijections that map the corresponding Lee-Pomeransky polynomials into each other, up to the described redundancies captured by the additional $\mathbb{Z}_2^c$ factor. We have hence established eq.~\eqref{eq:S_T1_T2}. Let us note that we can specialize this statement to symmetry transformations of a single momentum-grouped sector in which the set of symmetry transformations has a group structure. The symmetry group of a momentum-grouped sector with Lee-Pomeransky polynomial $\cG$ is hence given by
\begin{equation}\label{eqleepomloopiso}
    \mathrm{Aut}(\Theta)=\mathbb{Z}_2^c\rtimes \mathbb{G}(\cG) \,.
\end{equation}
\subsection{Symmetry transformations between non momentum-grouped sectors}
\label{sec:nonMomGroupedSymms}
We now show how we can lift the results from momentum-grouped sectors to non momentum-grouped sectors, and in particular we will derive eq.~\eqref{eq:symmTransfNonMomGrouped}. To this end, consider the setup introduced in section~\ref{eq:symmTransfDescr}. In particular, we assume that we have chosen (auxiliary) momentum-grouped sectors with momenta (cf.~eq.~\eqref{eq:p_to_ptilde})
\begin{equation}
    \widetilde{\bp}_l=\bS_l\bp_l \,, \quad l=1,2  \,.
\end{equation}
Furthermore we assume that we have a symmetry transformation as in eq.~\eqref{eq:symmTransfMomGrouped} between the momentum-grouped sectors defined by the matrices $\widetilde{\bL}_{\sigma},\widetilde{\bM}_{\sigma},\widetilde{\bN}_{\sigma}$.\footnote{Note that that momentum-grouped sectors do not necessarily live in the same `auxiliary' family (with less external momenta), since $\widetilde{\bp}_1$ and $\widetilde{\bp}_2$ might live in distinct subspaces of the full space of external momenta spanned by $\bp_1$ or $\bp_2$. It is easy to see that the proof from the previous section carries over to this setting nonetheless.} We will now construct matrices $\bL_{\sigma},\bM_{\sigma},\bN_{\sigma}$ defining a symmetry transformation between the non momentum-grouped sectors, as in eq.~\eqref{eq:symmTrans}. By simply comparing the two, we can immediately read off the relations:
\begin{equation}\label{eq:nonmongrouptrafo}
    \bL_{\sigma}=\widetilde{\bL}_{\sigma}, \qquad \bM_{\sigma}=\widetilde{\bM}_{\sigma}\bS_2 \,,
\end{equation}
where $\bS_2$ was defined in eq.~\eqref{eq:ptk_to_pk}. We furthermore find the consistency condition
\begin{equation}
\label{eq:NLiftCond}
    \bS_1\bN_{\sigma}=\widetilde{\bN}_{\sigma}\bS_2 \,.
\end{equation}
To solve this condition, it is convenient to define the rotation matrices
\begin{equation}
    \bR_l(\bs{s})=\begin{pmatrix}
        \bS_l \\ \bS'_l(\bs{s}) 
    \end{pmatrix} \,,
\end{equation}
for $l=1,2$, with $\bS'_l(\bs{s})$ defined in eq.~\eqref{eq:SpDef}. Note that this means that the $E-\Et$ rows of $\bS'_l(\bs{s})$ form a basis of the kernel of the rank $\Et$ matrix $\bS_l\bG(\bp_l)$. Furthermore, note that the matrices $\bR_l(\bs{s})$ satisfy
\begin{equation}
\label{eq:RInvRel}
    \bS_l\bR_l(\bs{s})^{-1}=\begin{pmatrix}
        \mathds{1}_{\Et} & 0_{\Et\times (E-\Et)}
    \end{pmatrix} \,.
\end{equation}
The choice of rotation $\bR_l(\bs{s})$ leads to a block structure in the Gram matrices
\begin{equation}
    \bR_l(\bs{s})\bG(\bp_l)\bR_l(\bs{s})^T=\begin{pmatrix}
        \bG(\tilde{\bp}_l) & 0 \\ 0 & \overline{\bG}_l^{\perp}
    \end{pmatrix} \,,
\end{equation}
for $l=1,2$ with $\overline{\bG}_l^{\perp}\in \cF^{(E-\Et)\times (E-\Et)}$. After rotating the conditions in eq.~\eqref{eq:NLiftCond} and the consistency condition $\bN_{\sigma}\bG(\bp_2)\bN_{\sigma}^T=\bG(\bp_1)$ (cf. eq.~\eqref{graminv}) for the matrix $\bN_{\sigma}$ into the new bases defined by $\bR_l$ and using eq.~\eqref{eq:RInvRel}, we can solve them, and we find that
\begin{equation}\label{NNtilde}
    \bN_{\sigma}=\bR_1(\bs{s})^{-1}\begin{pmatrix}
        \widetilde{\bN}_{\sigma} & 0 \\ 0 & \overline{\bs{O}}(\bs{s})
    \end{pmatrix}\bR_2(\bs{s}) \,,
\end{equation}
where the matrix $\overline{\bs{O}}(\bs{s})$ satisfies
\begin{equation}
\label{eq:OCond2Sectors}
    \overline{\bO}\overline{\bG}^{\perp}_2\overline{\bO}^T= \overline{\bG}^{\perp}_1 \,.
\end{equation}
Hence, we see that the symmetry transformations between non momentum-grouped sectors are given by precisely the symmetries of the associated momentum-grouped sectors, together with a choice of matrix $\overline{\bs{O}}$. In particular, for non momentum-grouped sectors the redundancies in the loop-momentum representation are further enhanced, such that we find 
\begin{equation}
\label{eq:symmSetGeneral}
    \mathrm{Sym}(\Theta_1,\Theta_2)= \mathbb{Z}_2^c\times \mathbb{O}(\overline{\bG}^{\perp}_1,\overline{\bG}^{\perp}_2) \times \mathbb{S}(\cG_1,\cG_2) \,,
\end{equation}
where $\mathbb{O}(\overline{\bG}^{\perp}_1,\overline{\bG}^{\perp}_2)$ is the set of matrices $\overline{\bs{O}}$ satisfying eq.~\eqref{eq:OCond2Sectors}. 
If we specialize to symmetry transformations in a single non momentum-grouped sector, we have $\overline{\bG}^{\perp}_1=\overline{\bG}^{\perp}_2$ and hence the matrix $\overline{\bs{O}}$ has to live in the orthogonal group 
\begin{equation}
    \overline{\bs{O}}\in\mathrm{O}\big(\overline{\bG}^{\perp}\big)=\{\bM\in \cF^{(E-\Et)\times (E-\Et)} \, | \, \bM\overline{\bG}^{\perp}\bM^T=\overline{\bG}^{\perp} \} \,.
\end{equation}
In contrast to the momentum-grouped case, the transformation of the external momenta will in general depend on the external kinematical variables. The existence of such kinematics-dependent symmetry transformations has already been noted in ref.~\cite{Wu:2024paw}. Here, our derivation makes the origin of this dependence completely transparent: it comes from the additional rotations $\bR_l$, which contain the matrices $\bS'_l$ that form bases of the kernel of $\bS\bG(\bp_l)$.

\begin{example}[A four-point triangle integral]
\label{fourpointtriangle}
    \begin{figure}[ht]
    \centering
    \includegraphics[width=0.4\textwidth]{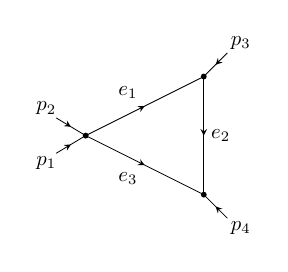}
    \caption{The oriented graph corresponding to the four-point triangle integral. The internal edges are labeled by $e_1,e_2,e_3$ and the external momenta are given by $p_1,\dots ,p_4$, with $p_4=-p_1-p_2-p_3$.}
    \label{fig:fourPtTriangle}
\end{figure}
As an illustration of the above discussion, let us consider the four-point triangle integral shown in figure \ref{fig:fourPtTriangle} with the external momenta on-shell, $p_i^2=0$. 
This example is inspired by ref.~\cite{Wu:2024paw}, where it was observed that this sector\footnote{The authors of ref.~\cite{Wu:2024paw} considered a symmetry between two sectors which take the same form as our triangle with one edge replaced by a bubble. For the transformation of the external momenta, which is our main interest here, this makes no difference, which is why we consider this simpler setup.} admits a kinematics-dependent symmetry. We will see that we can simply recover this symmetry transformation with the dependence on the kinematics naturally arising from the rotations $\bR_l$.

Following the discussion above we first associate a momentum-grouped sector to this graph, which corresponds to the ordinary triangle integral with external momenta\footnote{Note that we have $\bS_{\mathrm{tri},1}=\bS_{\mathrm{tri},2}\equiv \bS_{\mathrm{tri}}$ here since we are focusing on a fixed sector, and similarly for all other relevant matrices.}
\begin{equation}
    \tilde{\bp}_{\mathrm{tri}}=\bS_{\mathrm{tri}}\,\bp_{\mathrm{tri}} , \qquad \bS_{\mathrm{tri}}=\begin{pmatrix}
            1 & 1 & 0 \\ 0 & 0 & 1 
        \end{pmatrix} \,,
\end{equation}
with $\bp_{\mathrm{tri}}=(p_1,p_2,p_3)^T$ and $\tilde{\bp}_{\mathrm{tri}}=(\tilde{p}_1,\tilde{p}_2)^T$. We choose the edge momenta in the auxiliary sector to be parametrized as 
\begin{equation}
    \tilde{\bq}_{\mathrm{tri}}=\begin{pmatrix}
        k+p_1+p_2 \\ k+p_1+p_2+p_3 \\ -k
    \end{pmatrix} \,.
\end{equation}
The auxiliary sector then admits a Feynman graph automorphism corresponding to the signed bijection $\sigma_{\mathrm{tri}}=\big( \bkappa_{\mathrm{tri}},\alpha_{\mathrm{tri}}\big)$ with
\begin{equation}
    \bkappa_{\mathrm{tri}}=(1,-1,1)\,, \qquad \alpha_{\mathrm{tri}}:(x_1,x_2,x_3)\mapsto (x_3,x_2,x_1) \,,
\end{equation}
swapping edges $e_1$ and $e_3$ and reversing the orientation of $e_2$. The resulting symmetry transformation is given by the matrices
\begin{equation}
    \tilde{\bL}_{\sigma_{\mathrm{tri}}}=-1, \qquad 
    \widetilde{\bM}_{\sigma_{\mathrm{tri}}}=\begin{pmatrix}
        -1 & 0
    \end{pmatrix}, \qquad
    \widetilde{\bN}_{\sigma_{\mathrm{tri}}}=\begin{pmatrix}
        1 & 0 \\ -1 & -1 
    \end{pmatrix} \,,
\end{equation}
and takes the form
\begin{equation}
    k=-k'-\tilde{p}_1',\qquad \tilde{p}_1=\tilde{p}_1', \quad \tilde{p}_2=-\tilde{p}'_1-\tilde{p}'_2 \,.
\end{equation}
Note that this transformation preserves the Gram matrix, $\bG(\tilde{\bp})=\bG(\tilde{\bp}')$. We can now lift this symmetry transformation to a symmetry transformation of the non momentum-grouped sector, as discussed above. We immediately have
\begin{equation}
    \bL_{\sigma_{\mathrm{tri}}}=-1, \qquad \bM_{\sigma_{\mathrm{tri}}}=\widetilde{\bM}_{\sigma_{\mathrm{tri}}}\bS_{\mathrm{tri}}=\begin{pmatrix}
        -1 & -1 & 0
    \end{pmatrix} \,.
\end{equation}
To lift the transformation on the external momenta, we need to find the rotation $\bR_{\mathrm{tri}}$. To this end, we compute the kernel of the matrix $\bS_{\mathrm{tri}}\bG(\bp_{\mathrm{tri}})$, which is spanned by the vector $(t,s+t,s+2t)$. We can hence choose
\begin{equation}
    \bR_{\mathrm{tri}}=\begin{pmatrix}
        1 & 1 & 0 \\ 0 & 0 & 1 \\ t & s+t & s+2t
    \end{pmatrix} \,.
\end{equation}
Here we defined the Mandelstam invariants $s=(p_1+p_2)^2$ and $t=(p_2+p_3)^2$. The rotation $\bR_{\mathrm{tri}}$ by construction brings the Gram matrix into the block-diagonal form
\begin{equation}
    \overline{\bG}_{\mathrm{tri}}=\bR_{\mathrm{tri}}\bG(\bp_{\mathrm{tri}})\bR_{\mathrm{tri}}^T=\begin{pmatrix}
        s & -\frac{s}{2} & 0 \\ -\frac{s}{2} & 0 & 0 \\ 0 & 0 & \overline{\bG}^{\perp}_{\mathrm{tri}}
    \end{pmatrix} , \qquad \overline{\bG}^{\perp}_{\mathrm{tri}}=st(s+t) \,.
\end{equation}
We can now construct a symmetry transformation of the non momentum-grouped sector by making a choice of element of $\textrm{O}\big(\overline{\bG}^{\perp}_{\mathrm{tri}}\big)$, which here simply corresponds to a choice of sign
\begin{equation}
    o\in \textrm{O}\big(\overline{\bG}^{\perp}_{\mathrm{tri}}\big)=\{\pm1\} \,.
\end{equation}
For general $o$ we find the transformation
\begin{align}
    p_1&=-\frac{(1+o)t}{s}p_1'-\frac{os+(1+o)t}{s}p_2'-\frac{(1+o)(s+2t)}{s}p_3' \,, \nonumber \\
    p_2&=\frac{s+(1+o)t}{s}p_1'+\frac{(1+o)(s+t)}{s}p_2'+\frac{(1+o)(s+2t)}{s}p_3' \,, \\
    p_3&=-p_1'-p_2'-p_3' \,. \nonumber
\end{align}
We can see that for the choice $o=1$, we precisely recover the form of the symmetry transformation presented in ref.~\cite{Wu:2024paw}. Note, however, that in this example we could also make the choice $o=-1$, which leads to a completely equivalent, but kinematics-independent, symmetry transformation. This implies that in this example it is also possible to implement the symmetry transformation from ref.~\cite{Wu:2024paw} in a kinematics-independent way. In general, however, it is not always possible to find a transformation that is independent of the kinematics.
\end{example}

\subsection{The symmetry groupoid of a family of Feynman integrals}
\label{sec:loopMomRep}

In the previous subsections we have discussed how we can describe, and explicitly construct, the set of symmetry transformations $\Sym(\Theta_1,\Theta_2)$ between two sectors $\Theta_1$ and $\Theta_2$ of a given family of Feynman integrals. It is well known that symmetries of an object form a group, and we have already mentioned that the set of symmetries $\Aut(\Theta) = \Sym(\Theta,\Theta)$ is indeed a group. The set $\Sym(\Theta_1,\Theta_2)$, however, does in general not form a group, because we cannot compose elements unless $\Theta_1=\Theta_2$. It is then natural to ask if there is a more general mathematical structure that captures symmetry transformations between different sectors. 

In this section we argue that the relevant mathematical structure is given by \emph{groupoids}. A groupoid can be defined as follows:\footnote{On a more formal level, a groupoid is a category in which every morphism is invertible.} Consider a set of objects $\{O_1,O_2,O_3,\ldots,\}$. To each pair of objects $(O_i,O_j)$ we assign a set $\textrm{Hom}(O_i,O_j)$ of maps from $O_i$ to $O_j$. We denote a generic element of $\textrm{Hom}(O_i,O_j)$ by $g_{ij}\in \textrm{Hom}(O_i,O_j)$, and we require that the following properties hold:
\begin{enumerate}
    \item $g_{ij}g_{jk}\in \textrm{Hom}(O_i,O_k)$,
    \item $(g_{ij}g_{jk})g_{kl} = g_{ij}(g_{jk}g_{kl})$,
    \item $\textrm{Hom}(O_i,O_i)$ contains the identity map $\textrm{id}$,
    \item for every $g_{ij}\in \textrm{Hom}(O_i,O_j)$ there is $g_{ji}\in \textrm{Hom}(O_j,O_i)$ such that $g_{ij}g_{ji}=\textrm{id}$.
\end{enumerate}
Note that these properties imply that $\textrm{Hom}(O_i,O_i)$ is always a group.

It is now relatively easy to see that this structure is precisely the situation described in the previous sections, where the objects $O_i$ can be identified with the sectors $\Theta_i$ of a given family of Feynman integrals, and the maps between the objects are precisely the symmetry transformations between these sectors, i.e., $\textrm{Hom}(\Theta_i,\Theta_j) = \Sym(\Theta_i,\Theta_j)$. In other words, we see that we obtain in a very natural way a \emph{symmetry groupoid} attached to a family of Feynman integrals.

Groupoids are a mathematical generalization of groups, and they share some of their properties. In particular, one may define matrix representations of groupoids, in complete analogy with matrix representations of groups. In particular, the matrices constructed in the previous sections can be understood as matrix representations of the symmetry groupoid attached to a family of Feynman integrals. We now discuss this in more detail for the case of the symmetry group $\Aut(\Theta)$. The extension to the full groupoid is immediate. For simplicity, we focus on momentum-grouped sectors, but the extension to general sectors is again immediate.

From appendix \ref{sec:proofSymLemma} and section~\ref{sec:nonMomGroupedSymms} we know that a symmetry transformation of a sector $\Theta$ can be associated with a signed permutation $\sigma\in \Aut(\Theta)$.
From $\sigma$ we can construct the matrices $\bs{\widetilde{L}}_{\sigma}$, $\bs{\widetilde{M}}_{\sigma}$ and $\bs{\widetilde{N}}_{\sigma}$ in eq.~\eqref{eq:LMNfromIsomorph}. These matrices have the following property:
\begin{lemma}\label{lemma:group_structure_LMN}
    Let $\sigma,\sigma'$ be signed permutations corresponding to symmetries $\mathrm{Aut}(\Theta)$ of some momentum-grouped sector $\Theta$. Then the following relations hold.
    \begin{enumerate}[label=(\alph*)]
        \item $\bs{\widetilde{L}}_{\sigma}\bs{\widetilde{L}}_{\sigma'}=\bs{\widetilde{L}}_{\sigma\sigma'}$.
        \item $\bs{\widetilde{N}}_{\sigma}\bs{\widetilde{N}}_{\sigma'}=\bs{\widetilde{N}}_{\sigma\sigma'}$.
        \item $\bs{\widetilde{L}}_{\sigma}\bs{\widetilde{M}}_{\sigma'}+\bs{\widetilde{M}}_{\sigma}\bs{\widetilde{N}}_{\sigma'}=\bs{\widetilde{M}}_{\sigma\sigma'}$.
    \end{enumerate}
\end{lemma}
\begin{proof}
All of these statements follow from Lemma \ref{lemmaIsomorphOnQ}. Let us consider them in turn.
    \begin{enumerate}[label=(\alph*)]
        \item The product is given by
        \begin{equation}
            \bs{\widetilde{L}}_{\sigma}\bs{\widetilde{L}}_{\sigma'}=(\bC\bC^T)^{-1}\bC\bP_{\!\sigma}\bPi_C\bP_{\!\sigma'}\bC^T \,.
        \end{equation}
    
        The image of $\bC^T$ lies in $H_1(G,\mathbb{Q})$, which is preserved by $\bP_{\!\sigma'}$. The projector $\bPi_C$ hence acts as the identity, and we have
        \begin{equation}
            \bs{\widetilde{L}}_{\sigma}\bs{\widetilde{L}}_{\sigma'}=
            (\bC\bC^T)^{-1}\bC\bP_{\!\sigma}\bP_{\!\sigma'}\bC^T
            =(\bC\bC^T)^{-1}\bC\bP_{\!\sigma\sigma'}\bC^T
            =\bs{\widetilde{L}}_{\sigma\sigma'} \,.
        \end{equation}
        \item The product is given by
        \begin{equation}
        \label{eq:NNProduct}
            \bs{\widetilde{N}}_{\sigma}\bs{\widetilde{N}}_{\sigma'}=(\bD\bD^T)^{-1}\bD\bP_{\!\sigma}\bPi_D\bP_{\!\sigma'}\bD^T \,.
        \end{equation}
        The image of $\bD^T$ is mapped via $\bP_{\!\sigma'}$ to a part in $\mathrm{Im}\, \bD^T$ and to a part in $\mathrm{Im}\, \bC^T$. The projector $\bPi_D$ acts as the identity on the first part, and it annihilates the second. We hence need to argue that the second part does not contribute to the full product. To this end, we imagine dropping the projector. An element of $\mathrm{Im}\, \bC^T$ is then preserved by $\bP_{\!\sigma}$, but annihilated by $\bD$, due to the orthogonality property. Hence, dropping the projector $\bPi_D$ does not change the map defined by the matrix in eq.~\eqref{eq:NNProduct}, and hence it does not change the matrix itself. We thus find
        \begin{equation}
            \bs{\widetilde{N}}_{\sigma}\bs{\widetilde{N}}_{\sigma'}=\bs{\widetilde{N}}_{\sigma\sigma'} \,,
        \end{equation}
        as above.
        \item The product on the left-hand side is given by
        \begin{equation}
            \bs{\widetilde{L}}_{\sigma}\bs{\widetilde{M}}_{\sigma'}+\bs{\widetilde{M}}_{\sigma}\bs{\widetilde{N}}_{\sigma'}=(\bC\bC^T)^{-1}\bC \bP_{\sigma}(\bPi_C+\bPi_D)\bP_{\sigma'}\bD^T \,.
        \end{equation}
        Arguing as before, the image of $\bD^T$ is mapped to the sum of $\mathrm{Im}\,\bD^T$ and $\mathrm{Im}\,\bC^T$, on which $\bPi_C+\bPi_D$ acts as the identity. Hence we find
        \begin{equation}
            \bs{\widetilde{L}}_{\sigma}\bs{\widetilde{M}}_{\sigma'}+\bs{\widetilde{M}}_{\sigma}\bs{\widetilde{N}}_{\sigma'}=\bs{\widetilde{M}}_{\sigma\sigma'} \,,
        \end{equation}
        as claimed.
    \end{enumerate}
\end{proof}

 We can now obtain the same relations for the matrices $\bL_{\sigma}, \,\bM_{\sigma}, \, \bN_{\sigma}$ defined in eq.~\eqref{eq:symmMatricesFromAut_sum} (restricted to a single sector) using the relations in eq.~\eqref{eq:fromOrthToGenParam} with $\bGamma=\bGamma_1=\bGamma_2$. For the first two relations this is trivial. For the last one, we compute
\begin{align}
    &\bL_{\sigma}\bM_{\sigma'}+\bM_{\sigma}\bN_{\sigma'} \nonumber\\
    &=\widetilde{\bL}_{\sigma}(\widetilde{\bM}_{\sigma'}+\widetilde{\bL}_{\sigma'}\bGamma-\bGamma\widetilde{\bN}_{\sigma'})
    +(\widetilde{\bM}_{\sigma}+\widetilde{\bL}_{\sigma}\bGamma-\bGamma\widetilde{\bN}_{\sigma})\widetilde{\bN}_{\sigma'} \nonumber\\
    &=\widetilde{\bL}_{\sigma}\widetilde{\bM}_{\sigma'}+\widetilde{\bM}_{\sigma}\widetilde{\bN}_{\sigma'}+\widetilde{\bL}_{\sigma}\widetilde{\bL}_{\sigma'}\bGamma-\bGamma\widetilde{\bN}_{\sigma}\widetilde{\bN}_{\sigma'} \\
    &=\widetilde{\bM}_{\sigma\sigma'}+\widetilde{\bL}_{\sigma\sigma'}\bGamma-\bGamma\widetilde{\bN}_{\sigma\sigma'}\nonumber \\
    &=\bM_{\sigma\sigma'} \,. \nonumber
\end{align}
Let us now define
\beq
\bs{T}_{\sigma} = \begin{pmatrix}\bL_{\sigma}&\bM_{\sigma}\\
 \bs{0} & \bN_{\sigma}
\end{pmatrix}\,.
\eeq
It is then easy to see that Lemma~\ref{lemma:group_structure_LMN} implies
\begin{equation}
    \bT_{\textrm{id}} = \mathds{1}\textrm{~~~and~~~}\bT_{\sigma}\bT_{\sigma'}=\bT_{\sigma\sigma'} \,.
\end{equation}
In other words, $\bT_{\sigma}$ defines a representation of $\mathrm{Aut}(\Theta)$. This can similarly be shown for non-momentum-grouped sectors, i.e., also the matrices in eq.~\eqref{eq:symmTransfNonMomGrouped} form a representation of the symmetry group.

\subsection{Example: the one-loop on-shell box integral}
\begin{figure}[ht]
\centering
\scalebox{1.5}{
\begin{tikzpicture}[scale=1.1, >=latex]

\tikzstyle{topology}=[inner sep=0pt]
\tikzset{
  midarrow/.style={
    postaction={
      decorate,
      decoration={
        markings,
        mark=at position 0.5 with {\arrow{>}}
      }
    }
  },
  midarrowrev/.style={
    postaction={
      decorate,
      decoration={
        markings,
        mark=at position 0.5 with {\arrow{<}}
      }
    }
  }
}

\node[topology] (box) at (6,-0.5) {
\begin{tikzpicture}[scale=0.8]

\draw[midarrow]    (0,0) -- (0,2);   
\draw[midarrow]    (0,2) -- (2,2);   
\draw[midarrow]    (2,2) -- (2,0);   
\draw[midarrow]    (2,0) -- (0,0);   

\node at (-0.25,1) {\scriptsize$e_1$};   
\node at (1,2.25)  {\scriptsize$e_2$};
\node at (2.25,1)  {\scriptsize$e_3$};
\node at (1,-0.25) {\scriptsize$e_4$};

\draw[midarrow] (-1,3) -- (0,2)
  node[pos=0.62, sloped, above, yshift=1pt]{\scriptsize$p_1$};

\draw[midarrow] (3,3) -- (2,2)
  node[pos=0.62, sloped, above, yshift=1pt]{\scriptsize$p_2$};

\draw[midarrow] (2.8,-1) -- (2,0)
  node[pos=0.62, sloped, above, yshift=1pt]{\scriptsize$p_3$};

\draw[midarrow] (-1,-1) -- (0,0)
  node[pos=0.62, sloped, above, yshift=1pt]{\scriptsize$p_4$};

\draw[<-, thin] (1.35,1) arc[start angle=0, end angle=330, radius=0.35];
\node at (1,1) {\scriptsize$k$};
\end{tikzpicture}
};

\end{tikzpicture}}
\caption{The Feynman graph representing the family of one-loop box integrals  with on-shell external momenta.}
\label{fig:box}
\end{figure}
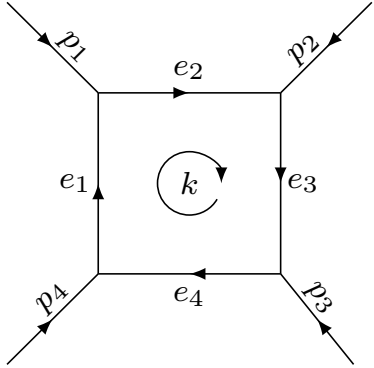
Let us illustrate the concepts and results from the previous subsections on a simple example. We consider the family of one-loop box integrals defined by the Feynman graph depicted in figure~\ref{fig:box}, with all internal masses equal, $m_i^2=m^2$, and all external momenta on shell, $p_i^2=0$, such that 
\begin{align}
    \bG_\text{Box}(\bp)=\left(
\begin{array}{ccc}
 0 & \frac{s}{2} & \frac{t}{2} \\
 \frac{s}{2} & 0 & -\frac{s+t}{2} \\
 \frac{t}{2} & -\frac{s+t}{2} & 0 \\
\end{array}
\right)\,,
\end{align}
where we defined the Mandelstam invariants $s=2p_1\cdot p_2$ and $t=2p_1\cdot p_3$. 
The family defined by the box Feynman graph with this restriction on the external kinematics has $15$ non-zero sectors. We will see that, by determining the symmetry transformations between these sectors, we can reduce the number of sectors to $7$. We will also describe the symmetry transformations within these sectors. We note that this one-loop box integral is well understood, and the results presented here are not new. However, this particular family of integrals is simple enough to illustrate all the concepts introduced in the preceding subsections, in particular how we can lift permutation symmetries from Lee-Pomeranky polynomials to symmetry transformations in the loop-momentum representation. 

\paragraph{Symmetry transformations in the Lee-Pomeransky representation.}
From eq. \eqref{eq:S_T1_T2}, we know that the symmetry group $\mathrm{Aut}(\Theta)$ of a sector is essentially given by $\mathbb{G}(\mathcal{G})$, the group of all permutations leaving the Lee-Pomeransky polynomial $\mathcal{G}$ of the sector $\Theta$ invariant. The Lee-Pomeransky polynomial $\cG_{1111} = \cG_{\Theta_{\text{top}}}$ of the top sector $\Theta_{\text{top}}=(1,1,1,1)$ of the on-shell box is given by
\begin{align}\label{Gpolyboxtop}
\mathcal{G}_{1111}(\bx)=(x_1+x_2+x_3+x_4)+m^2 (x_1+x_2+x_3+x_4)^2+x_2 x_4 (s+t)-s x_1 x_3\,.
\end{align}
The symmetry group of the top sector is 
\begin{align}
\label{eq:symmGroupBox}
\mathbb{G}(\mathcal{G}_{1111})=\{e,(13),(24),(13)(24)\}\simeq\mathbb{Z}_2\times \mathbb{Z}_2\,,
\end{align}
where $e$ denotes the identity permutation and $(ij)$ is the transposition that exchanges the Feynman parameters $x_i$ and $x_j$ (and leaving all other Feynman parameters fixed).
The Lee-Pomeransky polynomials for the subsectors can be obtained from eq.~\eqref{Gpolyboxtop} by setting the $x_i$ corresponding to propagators in this sector to zero, and the symmetry transformations of the subtopologies follow from their respective Lee-Pomeransky polynomials, e.g., 
\begin{align}
\mathbb{G}(\mathcal{G}_{{1110}})&=\{e,(13)\}\,,\quad    \mathbb{G}(\mathcal{G}_{{1101}})=\{e,(24)\}\,,\\
\quad  \mathbb{G}(\mathcal{G}_{{1010}})&=\{e,(13)\}\,,\quad  \mathbb{G}(\mathcal{G}_{{1001}})=\{e,(14)\}\,.\notag
\end{align}
We see that the triangle and the bubble sectors each have a $\mathbb{Z}_2$ symmetry.

Additionally, there are  symmetry transformations between sectors. For example, the symmetry transformations from $(1,0,1,1)$ to $(1,1,1,0)$ or from $(0,1,1,1)$ to $(1,1,0,1)$ are explicitly given as 
\begin{align}\label{morphbox}
\mathbb{S}( \cG_{1011},\cG_{1110})&=\{(x_1,x_2,x_3)\mapsto (x_1,x_4,x_3), (x_1,x_2,x_3)\mapsto (x_3,x_4,x_1)\}\,,\quad\\
\mathbb{S}(\cG_{0111},\cG_{1101})&=\{(x_1,x_2,x_4)\mapsto (x_3,x_2,x_4),(x_1,x_2,x_4)\mapsto (x_3,x_4,x_2)\}\,,\notag
\end{align}
As discussed in subsection~\ref{sec:loopMomRep}, we obtain a groupoid where the individual sectors are the objects with permutations as morphisms between the sectors (or automorphisms within one sector). A graphical representation of the groupoid is given in figure~\ref{fig:grupoid}. 
Due to the symmetries between different sectors, there are only two independent triangle master integrals, which we may take for example as $I_{0111}$ and $I_{1011}$. In the bubble sectors, the two sectors $(1,0,1,0)$ and $(0,1,0,1)$ and one of the remaining four (which can all be connected by the maps in eq.~\eqref{morphbox}) are independent.\footnote{The bubble integral with a massless external leg can obviously be reduced to a tadpole integral. However, this is not detectable from symmetry transformations alone, but follows, e.g., from explicitly solving the IBP relations in these sectors. Since our goal is to illustrate the structure of the symmetry groupoid of the family, we keep these reducible sectors in our discussion.} Similarly, only one tadpole integral can be independent.

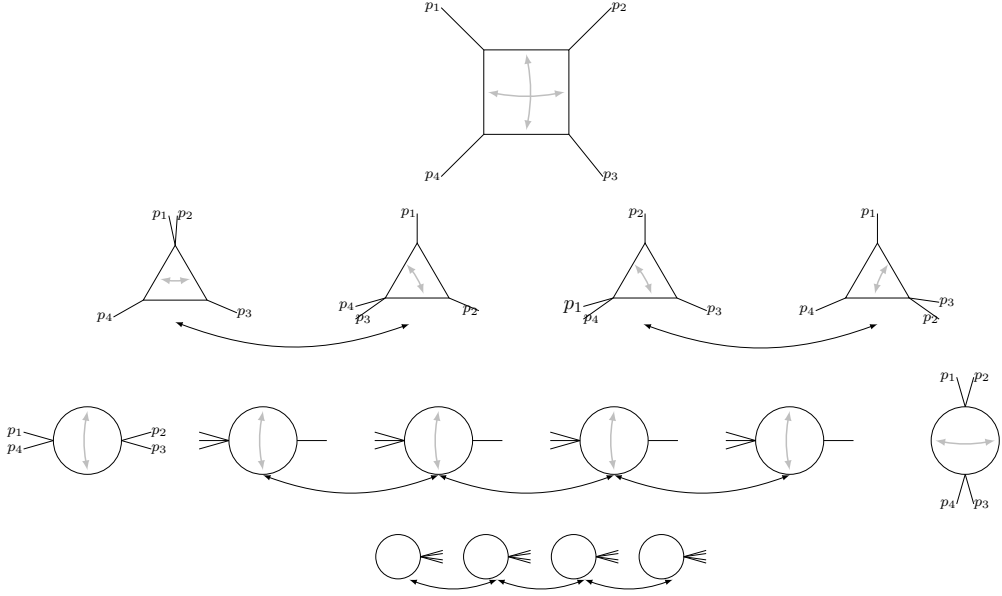
\begin{figure}[ht]
\centering
\scalebox{0.7}{
\begin{tikzpicture}[scale=1.1, >=latex]

\tikzstyle{topology}=[inner sep=0pt]
\tikzstyle{arrow}=[->, thick]

\node[topology] (box) at (6,-0.5) {
\begin{tikzpicture}[scale=0.8]
\draw (0,0) rectangle +(2,2);
\draw (0,2) -- (-1,3) node[left]{\scriptsize$p_1$} {};
\draw (2,0) -- (2.8,-1)node[right]{\scriptsize$p_3$};
\draw (2,2) -- (3,3)node[right]{\scriptsize$p_2$};
\draw (0,0) -- (-1,-1)node[left]{\scriptsize$p_4$};
\draw[thick,<->,gray!50](0.1,1) to[bend right=10] (1.9,1);
\draw[thick,<->,gray!50](1,0.1) to[bend right=10] (1,1.9);
\end{tikzpicture}
};

\node[topology] (tri1) at (0,-3.5) {
\begin{tikzpicture}[scale=0.8]
\draw (0,0) -- (1.5,0) -- (0.75,1.3) --cycle;
\draw (0,0) -- (-0.7,-0.4)node[left]{\scriptsize$p_4$} ;
\draw (1.5,0) -- (2.2,-0.3)node[right]{\scriptsize$p_3$} ;
\draw (0.75,1.3) -- (0.6,2.0)node[left]{\scriptsize$p_1$} ;
\draw (0.75,1.3) -- (0.8,2.0)node[right]{\scriptsize$p_2$} ;
\draw[thick,<->,gray!50](0.4,0.5) to[bend right=10] (1.1,0.5);
\end{tikzpicture}
};

\node[topology] (tri2) at (12,-3.5) {
\begin{tikzpicture}[scale=0.8]
\draw (0,0) -- (1.5,0) -- (0.75,1.3) --cycle;
\draw (0,0) -- (-0.7,-0.3)node[left]{\scriptsize$p_4$} ;

\draw (1.5,0) -- (2.2,-0.1)node[right]{\scriptsize$p_3$} ;
\draw (1.5,0) -- (2.2,-0.5)node[left]{\scriptsize$p_2$} ;

\draw (0.75,1.3) -- (0.75,2.0)node[left]{\scriptsize$p_1$} ;

\draw[thick,<->,gray!50] (0.69,0.1) to[bend left=10] (1.,0.8);
\end{tikzpicture}
};

\node[topology] (tri3) at (8,-3.5) {
\begin{tikzpicture}[scale=0.8]
\draw (0,0) -- (1.5,0) -- (0.75,1.3) --cycle;
\draw (0,0) -- (-0.7,-0.5)node[right]{\scriptsize$p_4$} ;
\draw (0,0) -- (-0.7,-0.2)node[left]{$p_1$} ;

\draw (1.5,0) -- (2.2,-0.3)node[right]{\scriptsize$p_3$} ;
\draw (0.75,1.3) -- (0.75,2.0)node[left]{\scriptsize$p_2$} ;
\draw[thick,<->,gray!50] (0.9,0.1) to[bend right=10] (0.5,0.8);
\end{tikzpicture}
};

\node[topology] (tri4) at (4,-3.5) {
\begin{tikzpicture}[scale=0.8]
\draw (0,0) -- (1.5,0) -- (0.75,1.3) --cycle;
\draw (0,0) -- (-0.7,-0.5)node[right]{\scriptsize$p_3$} ;
\draw (0,0) -- (-0.7,-0.2)node[left]{\scriptsize$p_4$} ;

\draw (1.5,0) -- (2.2,-0.3)node[left]{\scriptsize$p_2$} ;
\draw (0.75,1.3) -- (0.75,2.0)node[left]{\scriptsize$p_1$} ;
\draw[thick,<->,gray!50] (0.9,0.1) to[bend right=10] (0.5,0.8);
\end{tikzpicture}
};

\node[topology] (bub1) at (-1.5,-6.5) {
\begin{tikzpicture}[scale=0.8]
\draw (0,0) circle (0.8);
\draw (-0.8,0) -- (-1.5,0.2)node[left]{\scriptsize$p_1$};
\draw (-0.8,0) -- (-1.5,-0.2)node[left]{\scriptsize$p_4$}; 

\draw (0.8,0) -- (1.5,0.2)node[right]{\scriptsize$p_2$};
\draw (0.8,0) -- (1.5,-0.2)node[right]{\scriptsize$p_3$};
\draw[thick,<->,gray!50] (0,0.7)  to[bend right=10](0,-0.7);
\end{tikzpicture}
};

\node[topology] (bub2) at (13.5,-6.5) {
\begin{tikzpicture}[scale=0.8]
\draw (0,0) circle (0.8);
\draw (0,0.8) -- (0.2,1.5)node[right]{\scriptsize$p_2$};
\draw (0,0.8) -- (-0.2,1.5)node[left]{\scriptsize$p_1$};
\draw (0,-0.8) -- (0.2,-1.5)node[right]{\scriptsize$p_3$};
\draw (0,-0.8) -- (-0.2,-1.5)node[left]{\scriptsize$p_4$};
\draw[thick,<->,gray!50] (-0.7,0) to[bend right=10](0.7,0);

\end{tikzpicture}
};

\node[topology] (bub3) at (1.5,-6.5) {
\begin{tikzpicture}[scale=0.8]
\draw (0,0) circle (0.8);
\draw (-0.8,0) -- (-1.5,0.2);
\draw (-0.8,0) -- (-1.5,0);
\draw (-0.8,0) -- (-1.5,-0.2); 
\draw (0.8,0) -- (1.5,0);
\draw[thick,<->,gray!50] (0,0.7) to[bend right=10](0,-0.7);

\end{tikzpicture}
};
\node[topology] (bub4) at (4.5,-6.5) {
\begin{tikzpicture}[scale=0.8]
\draw (0,0) circle (0.8);
\draw (-0.8,0) -- (-1.5,0.2);
\draw (-0.8,0) -- (-1.5,0);
\draw (-0.8,0) -- (-1.5,-0.2); 
\draw (0.8,0) -- (1.5,0);
\draw[thick,<->,gray!50] (0,0.7) to[bend right=10](0,-0.7);
\end{tikzpicture}
};
\node[topology] (bub5) at (7.5,-6.5) {
\begin{tikzpicture}[scale=0.8]
\draw (0,0) circle (0.8);
\draw (-0.8,0) -- (-1.5,0.2);
\draw (-0.8,0) -- (-1.5,0);
\draw (-0.8,0) -- (-1.5,-0.2); 
\draw (0.8,0) -- (1.5,0);
\draw[thick,<->,gray!50] (0,0.7) to[bend right=10](0,-0.7);
\end{tikzpicture}
};
\node[topology] (bub6) at (10.5,-6.5) {
\begin{tikzpicture}[scale=0.8]
\draw (0,0) circle (0.8);
\draw (-0.8,0) -- (-1.5,0.2);
\draw (-0.8,0) -- (-1.5,0);
\draw (-0.8,0) -- (-1.5,-0.2); 
\draw (0.8,0) -- (1.5,0);
\draw[thick,<->,gray!50] (0,0.7) to[bend right=10](0,-0.7);
\end{tikzpicture}
};

\node[topology] (tad1) at (4,-8.5) {
\begin{tikzpicture}[scale=0.6]
\draw (0,0) circle (0.7);
\draw (0.7,0) -- (1.4,-0.2);
\draw (0.7,0) -- (1.4,-0.1);
\draw (0.7,0) -- (1.4,0.1);
\draw (0.7,0) -- (1.4,0.2);
\end{tikzpicture}
};

\node[topology] (tad2) at (5.5,-8.5) {
\begin{tikzpicture}[scale=0.6]
\draw (0,0) circle (0.7);
\draw (0.7,0) -- (1.4,-0.2);
\draw (0.7,0) -- (1.4,-0.1);
\draw (0.7,0) -- (1.4,0.1);
\draw (0.7,0) -- (1.4,0.2);
\end{tikzpicture}
};

\node[topology] (tad3) at (7,-8.5) {
\begin{tikzpicture}[scale=0.6]
\draw (0,0) circle (0.7);
\draw (0.7,0) -- (1.4,-0.2);
\draw (0.7,0) -- (1.4,-0.1);
\draw (0.7,0) -- (1.4,0.1);
\draw (0.7,0) -- (1.4,0.2);
\end{tikzpicture}
};

\node[topology] (tad4) at (8.5,-8.5) {
\begin{tikzpicture}[scale=0.6]
\draw (0,0) circle (0.7);
\draw (0.7,0) -- (1.4,-0.2);
\draw (0.7,0) -- (1.4,-0.1);
\draw (0.7,0) -- (1.4,0.1);
\draw (0.7,0) -- (1.4,0.2);
\end{tikzpicture}
};
\draw[<->, bend right=20] (tad1.south) to (tad2.south);
\draw[<->, bend right=20] (tad2.south) to (tad3.south);
\draw[<->, bend right=20] (tad3.south) to (tad4.south);

\draw[<->, bend right=20] (bub5.south) to (bub6.south);
\draw[<->, bend right=20] (bub4.south) to (bub5.south);
\draw[<->, bend right=20] (bub3.south) to (bub4.south);

\draw[<->, bend right=20] (tri1.south) to (tri4.south);
\draw[<->, bend left=20] (tri2.south) to (tri3.south);
\end{tikzpicture}}
\caption{Groupoid structure of the $15$ non-zero sectors of the on-shell box. Internal arrows indicate internal symmetries within the sector. External arrows are the morphisms between the sectors.}
\label{fig:grupoid}
\end{figure}

\paragraph{Symmetry transformations within sectors.} Let us derive the symmetry transformations within a given sector.
 We can parametrize the edge momenta of the box integral in terms of the matrices (cf.~eq.~\eqref{eq:q_k_p})
\begin{align}
\bC_\text{Box}=\begin{pmatrix}
        1 & 1 & 1 &1
    \end{pmatrix} , \qquad 
    \bE_\text{Box}= \left(
\begin{array}{cccc}
 0 & 1 & 1 & 1 \\
 0 & 0 & 1 & 1 \\
 0 & 0 & 0 & 1 \\
\end{array}
\right)\,.
\end{align}
This leads to
\begin{align}
\label{eq:q_box}
    \bq(\bk,\bp)=\bC^T_\text{Box}\,\bk+\bE^T_\text{Box}\,\bp=\left(\begin{array}{c}k\\
    k+p_1\\k+p_1+p_2\\k+p_1+p_2+p_3
\end{array}\right)\,.
\end{align}
The group of symmetry transformations 
of the top-sector in eq.~\eqref{eq:symmGroupBox} is generated by two (signed) transpositons, which can be represented by the signed permutation matrices:
\begin{align}
\bP_{(13)}=\left(
\begin{array}{cccc}
 0 & 0 & -1 & 0 \\
 0 & -1 & 0 & 0 \\
 -1 & 0 & 0 & 0 \\
 0 & 0 & 0 & -1 \\
\end{array}
\right)\,, \quad \bP_{(24)}=\left(
\begin{array}{cccc}
 -1 & 0 & 0 & 0 \\
 0 & 0 & 0 & -1 \\
 0 & 0 & -1 & 0 \\
 0 & -1 & 0 & 0 \\
\end{array}
\right)\,.
\end{align}
Note that the signs arise from the fact that the map on the edges reverses the orientation (see fig.~\ref{fig:box}).
The corresponding symmetry transformations in loop-momentum space can be obtained from eq.~\eqref{eq:symmMatricesFromAut_sum}, and they are given by the matrices
\begin{equation}
\begin{aligned}
    \bL_{(13)}&=(-1)\,, \qquad & 
    \bM_{(13)}&=\left(
\begin{array}{ccc}
 -1 &
 -1&
 0 
\end{array}
\right)\,, \qquad &  \bN_{(13)}&=\left(
\begin{array}{ccc}
 0 & 1 & 0 \\
 1 & 0 & 0 \\
 -1 & -1 & -1 \\
\end{array}
\right)\,, \\
    \bL_{(24)}&=(-1)\,, \qquad & \bM_{(24)}&=\left(
\begin{array}{ccc}
 0 &
 0 &
 0 
\end{array}
\right)\,, \qquad& \bN_{(24)}&=\left(
\begin{array}{ccc}
 -1 & -1 & -1 \\
 0 & 0 & 1 \\
 0 & 1 & 0 \\
\end{array}
\right)\,,
\end{aligned}
\end{equation}
where we can easily check that $\bN_\sigma\bG_\text{Box}(\bp)\bN^T_\sigma=\bG_\text{Box}(\bp)$ with $\sigma\in\mathbb{G}(\cG_{1111})$, is satisfied for the chosen kinematics.
Then the symmetry transformation for the $(24)$ permutation is  
\beq
k= -k'\,,\quad  p_{1}= -p_{1}'-p_{2}'-p_{3}'\,,\quad p_{2}=p_{3}'\,, \quad p_{3}=p_{2}'\,,
\eeq
such that, for a choice of external momenta $\bp'$ satisfying $\bG(\bp')=\bG_\text{Box}$ there are two parametrizations $\bq(\bk,\bp)$ and $\bq(\bk',\bp')$ of the box, related via the permutation
\begin{align}
\bq(\bk,\bp)=\bP_{(24)}\bq(\bk',\bp')\,.
\end{align}
For the $(13)$ permutation we find 
\beq
k= -k'-p_{1}'-p_{2}'\,,\quad p_{1}=p_{2}'\,,\quad p_{2}=p_{1}'\,,
\quad p_{3}= -p_{1}'-p_{2}'-p_{3}'\,,
\eeq
such that 
\begin{align}
\bq(\bk,\bp)=\bP_{(13)}\bq(\bk',\bp')\,.
\end{align}
Let us now consider the triangle subsectors. As before, they are not momentum-grouped, and so we first determine the symmetry transformations for the associated momentum-grouped sectors, and then lift them to the non momentum-grouped versions. 
The discussion is analogous to example~\ref{fourpointtriangle}.  For the non momentum-grouped sectors,  we choose the following parametrization  for the $(24)$ permutation of the $(0,1,1,1)$ sector 
\begin{align}
\label{eq:q_tri}
    \bq_\text{Tri,1}(\bk,\bp)=\left(\begin{array}{c}k\\
    k+p_2\\k+p_2+p_3
\end{array}\right)\,.
\end{align}
Then we find 
\beq
k= -k'-p_{2}'-p_{3}'\,,\quad p_{1}= p_{1}'-\frac{(p_{2}'-p_{3}')(s-t)}{s+t}\,, \quad p_{2}=p_{3}'\,, \quad p_{3}=p_{2}'\,,
\eeq 
such that $\bq_\text{Tri,1}(\bk,\bp)=\bP_{\text{Tri},1}\,\bq_\text{Tri,1}(\bk',\bp')$
with 
\begin{align}
\bP_{\text{Tri,1}}=    \left(
\begin{array}{ccc}
 0 & 0 & -1 \\
 0 & -1 & 0 \\
 -1 & 0 & 0 \\
\end{array}
\right)\,.
\end{align}
For the $(1,0,1,1)$ sector with 
\begin{align}
\label{eq:q_tri2}
    \bq_\text{Tri,2}(\bk,\bp)=\left(\begin{array}{c}k\\
    k+p_1+p_2\\k+p_1+p_2+p_3
\end{array}\right)\,,
\end{align}
we find the transformation \beq
k=-k'-p_{1}'-p_{2}'\,,\quad p_{1}= p_{2}'\,,\quad p_{2}= p_{1}'\,,\quad p_{3} =-p_{1}'-p_{2}'-p_{3}'\,. 
\eeq
Then the two parametrizations $\bq_\text{Tri,2}(\bk,\bp)=\bP_{\text{Tri,2}}\,\bq_\text{Tri,2}(\bk',\bp')$ are related via the permutation
\begin{align}
 \bP_{\text{Tri,2}}   =\left(
\begin{array}{ccc}
 0 & -1 & 0 \\
 -1 & 0 & 0 \\
 0 & 0 & -1 \\
\end{array}
\right)\,.
\end{align}

There are two non momentum-grouped bubble sectors. For the $(0,1,0,1)$ sector, with parametrization $\bq_\text{Bubb,1}(\bk,\bp)=(k,k+p_2+p_3)^T$, the $(24)$ permutation is
\begin{equation}
    k= -k'-p_{2}'-p_{3}'\,, \qquad \bp=\bp'\,,
\end{equation}
and for the $(1,0,1,0)$ sector with parametrization $\bq_\text{Bubb,2}(\bk,\bp)=(k,k+p_1+p_2)^T$, the $(24)$ permutation is achieved by 
\begin{equation}
    k= -k'-p_{1}'-p_{2}'\,, \qquad \bp=\bp' \,,
\end{equation}
such that for both sectors we have
\begin{equation}
    \bq_\text{Bubb,i}(\bk,\bp)=\bP_\text{Bubb}\,  \bq_\text{Bubb,i}(\bk',\bp') \,, \quad \text{with} \quad \bP_\text{Bubb}=\left(
\begin{array}{cc}
 0 & -1 \\
 -1 & 0 \\
\end{array}
\right)\,.
\end{equation}

 \paragraph{Symmetry transformations in loop-momentum space between sectors.}
Let us now proceed to derive the symmetry transformations between different sectors. The only non-trivial transformations are those between the non momentum-grouped triangle sectors. We first parametrize the auxiliary momentum-grouped sectors and derive the symmetry transformations there. We then follow subsection~\ref{sec:nonMomGroupedSymms} to lift them to the non momentum-grouped sectors using eqs.~\eqref{eq:nonmongrouptrafo} and~\eqref{NNtilde}. 

As an example, let us derive a symmetry transformation between the sectors $\Theta_1=(1,0,1,1)$ and $\Theta_2=(1,1,1,0)$, with graphs $G_{1}$ and $G_{2}$ and associated Lee-Pomeransky polynomials $\cG_{1}$ and $\cG_{2}$.
The momentum-grouped external momenta  are $\bs{\tilde{p}}_{1}=(p_1+p_2,p_3)=\bS_{1}\bp$ and  $\bs{\tilde{p}}_{2}=(p_1,p_2)=\bS_{2}\bp$,
with
\begin{equation}
\begin{aligned}
\bS _{1}&=  \left(
\begin{array}{ccc}
 1 & 1 & 0 \\
 0 & 0 & 1 \\
\end{array}
\right) \,,\qquad\bS_{2}= \left(
\begin{array}{ccc}
 1 & 0 & 0 \\
 0 & 1 & 0 \\
\end{array}
\right)\,.
\end{aligned}
\end{equation}
The Lee-Pomeransky polynomials of the two subsectors are 
\beq\bsp
\mathcal{G}_{1}&=(x_1+x_3+x_4)+m^2 (x_1+x_3+x_4)^2-s x_1 x_3\,,\\
\mathcal{G}_{2}&=(x_1+x_2+x_3)+m^2 (x_1+x_2+x_3)^2-s x_1 x_3\,.
\esp\eeq
We want to construct the loop momentum transformation induced by the permutation $\alpha \in  \mathbb{S}( \cG_1,\cG_2)$ with  $\alpha(x_1,x_2,x_3)=(x_1,x_4,x_3)$.
Therefore, we fix a parametrization for the two auxiliary sectors according to
\beq\bsp
    \widetilde{\bq}_{i}(\bk_{i},\widetilde{\bp}_i)&=\widetilde{\bC}_{\text{Tri}}^T \,\bk_i+\widetilde{\bE}^T_{\text{Tri}}\,\widetilde{\bp}_i\,, \quad i\in \{1,2\}\,,
\esp\eeq
with
\begin{align}
\widetilde{\bC}_{\text{Tri}}&=\begin{pmatrix}
        1 & 1 & 1 
    \end{pmatrix} , \qquad 
    \widetilde{\bE}_{\text{Tri}}= \left(
\begin{array}{cccc}
 0 & 1 & 1 \\
 0 & 0 & 1 
\end{array}
\right)\,.
\end{align}
We label the edge momenta such that $\widetilde{q}_{i,j}$ is the momentum associated to the edge with Feynman parameter $x_j$ in the graph $G_i$.
Explicitly, $\widetilde{\bq}_{1}=(\widetilde{q}_{1,1},\widetilde{q}_{1,3},\widetilde{q}_{1,4})^T$ and $\widetilde{\bq}_{2}=(\widetilde{q}_{2,1},\widetilde{q}_{2,2},\widetilde{q}_{2,3})^T$. We emphasize that this is a choice, and different choices of parametrizations will give rise to different symmetry transformations.

 We aim to find a symmetry transformation of the form
\begin{align}\bk_1=\bL_{1,2}\,\bk_2+\bM_{1,2}\bp_{2},\qquad \bp_{1}=\bN_{1,2}\,\bp_{2} \,,\end{align} 
 such that (cf.~eq.~\eqref{eq:symmTrans}), 
 \begin{align}
  \left(
\begin{array}{c}\label{eqparabox}
  \widetilde{q}_{1,1}(\bk_{1},\bp_{1}) \\
  \widetilde{q}_{1,3}(\bk_{1},\bp_{1}) \\
 \widetilde{q}_{1,4}(\bk_{1},\bp_{1})\\
\end{array}
\right)=  \left(
\begin{array}{c}
  \widetilde{q}_{1,\alpha(1)}(\bk_{1},\bp_{1}) \\
  \widetilde{q}_{1,\alpha(3)}(\bk_{1},\bp_{1}) \\
 \widetilde{q}_{1,\alpha(2)}(\bk_{1},\bp_{1})\\
\end{array}
\right)
 = -\left(
\begin{array}{c}
  \widetilde{q}_{2,1}(\bk_{2},\bp_{2}) \\
  \widetilde{q}_{2,3}(\bk_{2},\bp_{2}) \\
 \widetilde{q}_{2,2}(\bk_{2},\bp_{2})\\
\end{array}
\right)=\bP_{1,2}\left(
\begin{array}{c}
  \widetilde{q}_{2,1}(\bk_{2},\bp_{2}) \\
  \widetilde{q}_{2,2}(\bk_{2},\bp_{2}) \\
 \widetilde{q}_{2,3}(\bk_{2},\bp_{2})\\
\end{array}
\right)\,,
 \end{align}
with the signed permutation matrix, 
\begin{align}
\bP_{1,2}   = \left(
\begin{array}{ccc}
 -1 & 0 & 0 \\
 0 & 0 & -1 \\
 0 & -1 & 0 \\
\end{array}
\right)\,.
\end{align}
Then, according to eq.~\eqref{eq:symmMatricesFromAut_sum}, the symmetry transformations of the momentum-grouped sectors are given by 
\begin{align}
    \widetilde{\bL}_{1,2}=(-1)\,,\quad \widetilde{\bN}_{1,2}=\left(
\begin{array}{cc}
 -1 & -1 \\
 0 & 1 \\
\end{array}
\right)\,,\quad \widetilde{\bM}_{1,2}=\left(
\begin{array}{cc}
 0  &
 0  
\end{array}
\right)\,.
\end{align}
Let us now lift these symmetry transformations to the non momentum-grouped sectors. To this end, we compute
\begin{equation}
\begin{aligned}
\bR _{1}=\left(
\begin{array}{ccc}
 1 & 1 & 0 \\
 0 & 0 & 1 \\
 \frac{s+t}{s+2 t} & \frac{t}{s+2 t} & 1 \\
\end{array}
\right)\,, \qquad \bR _{2}=    \left(
\begin{array}{ccc}
 1 & 0 & 0 \\
 0 & 1 & 0 \\
    \frac{s+t}{s} & -\frac{t}{s} & 1 \\
\end{array}
\right)
\,.
\end{aligned}
\end{equation}
We rotate the Gram matrices into a block-diagonal form as in eq.~\eqref{eq:RInvRel}:
\begin{equation}
\begin{aligned}
\bR _{1}(\bs{s})\bG(\bp _{1})\bR _{1}(\bs{s})^T&=\left(
\begin{array}{ccc}
 s & -\frac{s}{2} & 0 \\
 -\frac{s}{2} & 0 & 0 \\
 0 & 0 & \overline{\bG}_{1}^{\perp}  \\
\end{array}
\right)\,, \qquad & \overline{\bG}_{1}^{\perp}&=\frac{s t (s+t)}{(s+2 t)^2}\,,\\
\bR_{2}(\bs{s})\bG(\bp_{2})\bR _{2}(\bs{s})^T&=\left(
\begin{array}{ccc}
 0 & \frac{s}{2} & 0 \\
 \frac{s}{2} & 0 & 0 \\
 0 & 0 & \overline{\bG}_{2}^{\perp} \\
\end{array}
\right)
\,, \qquad &  \overline{\bG}_{2}^{\perp}&=\frac{t (s+t)}{s}  \,.
\end{aligned}
\end{equation}
Then $\overline{\bO}$ is determined by $\overline{\bO}\overline{\bG}_{2}^{\perp}\overline{\bO}^T=\overline{\bG}_{1}^{\perp}$. Explicitly, we can choose
\begin{equation}
    \overline{\bO}=\frac{s}{s+2t} \,.
\end{equation}
Finally, the symmetry transformation for the non momentum-grouped sector is given by 
\begin{align}
\bL_{1,2}=-1\,,\quad \bN_{1,2}=\left(
\begin{array}{ccc}
 \frac{2 t}{s}+1 & -\frac{2 t}{s}-1 & 1 \\
 -\frac{2 (s+t)}{s} & \frac{2 t}{s} & -1 \\
 0 & 1 & 0 \\
\end{array}
\right)\,,\quad \bM_{1,2}=\left(
\begin{array}{ccc}
 0&
 0 &
 0 
\end{array}
\right)\,.
\end{align}
Then the parametrizations
 \begin{align}
    \bq_\text{1}(\bk,\bp)=\left(\begin{array}{c}k\\
    k+p_1+p_2\\k+p_1+p_2+p_3
\end{array}\right)\, \quad \text{and} \quad   \bq_\text{2}(\bk,\bp)=\left(\begin{array}{c}k\\
    k+p_1\\k+p_1+p_2
\end{array}\right)\,,
\end{align}
of the two sectors are mapped into each other as $ \bq_\text{1}(\bk_1,\bp_1)= \bP_{1,2}\,\bq_\text{2}(\bk_2,\bp_2)$.

We proceed analogously to derive a symmetry transformation between the sectors $\Theta_3=(0,1,1,1)$ and $\Theta_4=(1,1,0,1)$. The momentum-grouped external momenta are $\bs{\widetilde{p}}_{3}=(p_2,p_3)$ and $\bs{\widetilde{p}}_{4}=(p_1,p_2+p_3)$. With $\overline{\bO}=1$ we find $\bk_3=\bL_{3,4}\bk_4+\bM_{3,4}\bp_4$ and $\bp_3=\bN_{3,4}\bp_4$ with 
\begin{align}
\bL_{3,4}=-1\,,\quad \bN_{3,4}=\left(
\begin{array}{ccc}
 \frac{2 (s-t)}{s+t} & \frac{s-t}{s+t} & \frac{2 s}{s+t} \\
 1 & 0 & 0 \\
 -1 & -1 & -1 \\
\end{array}
\right)\,, \quad \bM_{3,4}=\left(
\begin{array}{ccc}
 -1 &
 0 &
 0 
\end{array}
\right)\,,
\end{align}
such that for  
 \begin{align}
    \bq_\text{3}(\bk,\bp)=\left(\begin{array}{c}k\\
    k+p_2\\k+p_2+p_3
\end{array}\right)\, \quad \text{and} \quad   \bq_\text{4}(\bk,\bp)=\left(\begin{array}{c}k\\
    k+p_1\\k+p_1+p_2+p_3
\end{array}\right)\,,
\end{align}
we have $ \bq_\text{3}(\bk_3,\bp_3)= \bP_{3,4}\,\bq_\text{4}(\bk_4,\bp_4)$, where 
\begin{align}
   \bP_{3,4}=\left(
\begin{array}{ccc}
 0 & -1 & 0 \\
 -1 & 0 & 0 \\
 0 & 0 & -1 \\
\end{array}
\right)\,. 
\end{align}

The four bubble sectors $(0,0,1,1)$, $(1,1,0,0)$, $(0,1,1,0)$ and $(1,0,0,1)$ can all be trivially related by $p_1\leftrightarrow p_2\leftrightarrow p_3\leftrightarrow p_4$.

\section{Symmetries and twisted cohomology}\label{sec:symandtwist}
(Relative) twisted cohomology provides a framework to describe Feynman integrals in dimensional regularization~\cite{Mastrolia:2018uzb}. 
It is hence natural to formulate a notion of symmetry transformations directly in the language of twisted cohomology, which is the goal of this section. In particular, we will define a notion of \emph{twisted symmetry transformations} acting on the twisted (co)homology groups and their duals, which leave the various pairings between these groups invariant. While these symmetry transformations naturally act on integrands and integration contours, we will also comment on how to apply this notion to families of integrals. We will keep the discussion in this section very general and make the connection back to Feynman integrals only in the next section.

\subsection{Review of relative twisted cohomology groups}\label{reltwistedrew}
Before we discuss how we can incorporate a notion of symmetry transformations that extends the ideas introduced in the previous section to general twisted cohomology theories, we give a brief review of the relevant mathematical background. We have already reviewed the non-relative case in section~\ref{introtwisted}, and so we primarily focus on relative twisted cohomology groups~\cite{matsumoto_relative_2019-1}. 

In subsection~\ref{introtwisted}, we introduced twisted cohomology groups associated to a twist $\Psi$, and we defined the twisted variety $\Sigma$ in eq.~\eqref{twistedvariety} and the on-shell variety $D_-$ in eq.~\eqref{onshellvariety}. Generally, the twist may depend on some variables $\bs{s}$, and we denote by $\mathcal{F}=\mathbb{Q}(\bs{s},\mu_i)$  the field of rational functions in $\bs{s}$ and $\mu_i$ with rational coefficients (where the $\mu_i$ are the non-integer exponents of the polynomials that define the twist $\Psi$, cf. eq.~\eqref{eq:twist}). 

Loosely speaking, relative (co-)homology theories generalize the non-relative case to situations where the cycles are allowed to have boundaries in some prescribed locus $D_+$.
We define
\begin{align}
X = \mathbb{C}^n\backslash\Sigma\,,\qquad X_-=\mathbb{C}^n\backslash(\Sigma\cup D_-)\,,\qquad X_+=\mathbb{C}^n\backslash(\Sigma\cup D_+)\,.
\end{align}
The relative twisted homology group $H_n(X_-,D_+,\check{\mathcal{L}}_{\omega})$ is defined as the group generated by all twisted cycles with boundaries in $D_+$ modulo closed cycles (with boundaries in $D_+$).

Similarly, one can consider the \textit{(relative) twisted cohomology group} $H^n_{\textrm{dR}}(X_-,D_+,\nabla_\omega)$, as the group generated by equivalence classes of (relative) differential forms that are closed with respect to $\nabla_\omega$, modulo (relative) forms that are exact. 
A construction of relative differential forms can be found in refs.~\cite{matsumoto_relative_2019-1,Caron-Huot:2021xqj, Caron-Huot:2021iev}. The details are not essential for this paper. Here it suffices to take a pragmatic approach, and we present an explicit set of generators sufficient for most applications. Using the Leray exact sequence, we may represent classes in $H^n_{\textrm{dR}}(X_-,D_+,\nabla_\omega)$ as differential forms defined on $X_-$ or supported on $D_+$. More precisely, let us write $D_+$ as a union of irreducible components,
\beq\label{eq:D+_decomposition}
D_+ = \bigcup_{i=1}^Q D_+^{(i)}\,.
\eeq
Then for each boundary component $D_+^{(i)}$, we introduce a differential form with support only on a very small (infinitesimal) neighborhood around $D_+^{(i)}$,
\beq
\delta_i := \rd\theta_i\,,\quad\textrm{with} \quad
\theta_i(x) = \theta\Big(\epsilon-d\big(x,D_+^{(i)}\big)\Big)\,,
\eeq
where $\epsilon$ is infinitesimal and $d(x,D) = \inf_{y\in D}d(x,y)$ is the distance from $x\in X$ to $D$, and $\theta(x)$ is the Heaviside step  function from eq.~\eqref{eq:heaviside}.
We refer to such a differential form as a \emph{$\delta$-form}.
Details on the construction and meaning of $\rd \theta$ can be found in refs.~\cite{Caron-Huot:2021iev,Caron-Huot:2021xqj,Giroux:2022wav}. Here it suffices to recall that the derivative of the step function is the Dirac $\delta$-function, so that $\rd\theta_i$ has support only on an infinitesimal neighborhood around $D_+^{(i)}$. Hence, integrals over $\delta_i$ localize on the boundary component $D_+^{(i)}$. A generic relative differential form can then be cast in the form
\begin{align}\label{formrel}
    \varphi=R(\bx)\,
    \left(\bigwedge_{i\in I}\delta_{i}\right)\wedge\left(\bigwedge_{j\in \bar{I}}\rd x_j\right)\,,
\end{align}
where $R$ is a function with poles at most along $\Sigma\cup D_-$, $I\subseteq \{1,\ldots,Q\}$ with $|I|\le n$ and $\bar{I}$ is its complement in $\{1,\ldots,n\}$.
For more details, see refs.~\cite{matsumoto_relative_2019-1,Caron-Huot:2021iev,Caron-Huot:2021xqj,Giroux:2022wav}. 

There are isomorphisms relating different (twisted) cohomology groups. First of all, the (singular) twisted cohomology group is the dual of the (singular) twisted homology group, $H^n(X_-,D_+,\mathcal{L})\simeq H_n(X_-,D_+,\check{\mathcal{L}})^\vee$. Since the period pairing is non-degenerate, we have an isomorphism between the de Rham and singular cohomology groups, after enlarging the coefficients to complex numbers,
\beq
H^n_{\text{dR}}(X_-,D_+,\nabla_{\omega})\otimes\mathbb{C} \simeq 
H^n(X_-,D_+,\mathcal{L})\otimes\mathbb{C}\,.
\eeq
For this reason, we will often drop the subscript `dR' when talking about the de Rham cohomology group.
There is a second isomorphism, relating the dual of relative cohomology groups to the group with the roles of $D_-$ and $D_+$ interchanged~\cite{matsumoto_relative_2019-1,Caron-Huot:2021xqj,Caron-Huot:2021iev}
\beq
H^n(X_-, D_+,\nabla_{\omega})^\vee\simeq H^n(X_+, D_-,\check{\nabla}_{\omega})\,.
\eeq

\subsection{Twisted symmetry transformations}\label{subsecpairing}
After this short review of relative twisted cohomology groups, we want to define a framework that allows us to extend the concept of symmetry transformations between sectors from families of Feynman integrals to twisted cohomology theories. 

We start by generalizing the notion of sectors introduced in subsection~\ref{sec:summary_twisted_sym} to the relative case. 
In eq.~\eqref{eq_cocycle} we associated a list $\vartheta^-(\varphi)$ to every cocycle. For 
 a cocycle $\varphi\in H^n_{\mathrm{dR}}(X_-,D_+,\nabla_\omega)$, let us define an additional list $\vartheta^+(\varphi) = (r^+_1,\ldots,r^+_{Q})$ with 
\beq\label{eq_relcocycle}
r^+_i = \left\{\begin{array}{ll}
0\,, & \textrm{ if $\varphi$ contains $\delta_{i}$}\,,\\
1\,,& \textrm{ otherwise}\,.
\end{array}\right.
\eeq
Then, additionally to $\Theta^- = (r_1
^-,\ldots,r_P^- )\in \{0,1\}^P$, we define $\Theta^+ = (r^+_1,\ldots,r^+_{Q} )\in \{0,1\}^{Q}$, and there are  two filtrations 
\beq\bsp
   \text{Sec}_{\Theta^{\pm}} &:=\{\varphi\in H^n_{\text{dR}}(X_-, D_+,\nabla_\omega):\vartheta^{\pm}(\varphi)\preceq\Theta^{\pm}\}\,.
\esp\eeq
A \emph{sector of a relative twisted cohomology group} is determined by the combined filtration 
\begin{align}
 \text{Sec}_{\Theta^-,\Theta^{+}}:=  \text{Sec}_{\Theta^-}\cap \text{Sec}_{\Theta^{+}}  \,.
\end{align}
We will often use the notation ${\Theta} = (\Theta^-,\Theta^+)$.
In the case where $D_+=\emptyset$ we recover the notion of a sector introduced in section~\ref{sec:summary_twisted_sym}. We also define $\bd_{\Theta}^{\pm}$ as the sets of indices $i$ with $r_i^{\pm}=1$.
A cocycle from the sector $\Theta$ can then be cast in the form (cf.~eq.~\eqref{formrel})
\begin{align}\varphi=R(\bx)\,
\delta_{{\Theta'}}\wedge\left(\bigwedge_{j\in {\bd}^+_{\Theta'}}\rd x_j\right)\,,
\label{eq:general_phi}
\end{align}
where $\Theta'\preceq\Theta$. We introduce the shorthands 
\beq
\delta_{\Theta'}=\bigwedge_{i\in\bd^+_{{\Theta'}}}\delta_i\,,\qquad D_{+}^{{\Theta'}} =\bigcap_{i\in \bd_{\Theta'}^+}D_+^{(i)}\,,\qquad D_{+,\Theta'} = \bigcup_{i\in \bd^+_{\Theta'}}D_+^{(i)}\,.
\eeq
We also define 
$X_{\pm,\Theta}=\mathbb{C}^n\setminus (\Sigma\cup D_{\pm,\Theta})$.

\begin{definition}\label{def:twisted_symmetry}
    A twisted symmetry transformation from the sector ${\Theta}_1$ to the sector ${\Theta}_2$ is an affine and bijective map
    \begin{align}
     \label{eq_twistedsymgen}
f:X_{-,\Theta_2}\to X_{-,\Theta_1}\,, \quad \bx_2\mapsto\bx_1 = \bA\bx_2+\bs{b}\,,
 \end{align}
 such that
 \begin{enumerate}
    \item the Jacobian is trivial, $\det(f):=\det\bA=\pm1$,
    \item the twist restricted to the support of $\delta_{{\Theta}_1}$ is mapped to the twist on the support of $\delta_{{\Theta}_2}$, i.e., $(\Psi|_{D_{+}^{\Theta_1}})\circ f = \Psi|_{D_{+}^{\Theta_2}}$. 
    \item $f$ bijectively maps $D_{\pm,{\Theta}_2}$ to $D_{\pm,{\Theta}_1}$.
\end{enumerate}
\end{definition}
Since $f$ is affine and bijective, it is continuous (even smooth), and so it maps irreducible components of $D_{\pm,{\Theta}_2}$ to irreducible components of $D_{\pm,{\Theta}_1}$. From this we deduce the existence of two permutations
\beq
\alpha_{\pm}:  \bd^{\pm}_{\Theta_2}\to \bd^{\pm}_{{\Theta}_{1}}\,,
\eeq
such that
\beq\bsp
f\Big(D_{\pm}^{(i)}\Big) = D_{\pm}^{\alpha_{\pm}(i)}\,,\qquad &\textrm{for all }i\in \bd^{\pm}_{\Theta_2}
\,.
\esp\eeq
Note that for $D_+=\emptyset$,  Definition~\ref{def:twisted_symmetry} reduces to the one given in section~\ref{sec:summary} for non-relative cohomology groups.
We will denote the set of all twisted symmetry transformations from $\Theta_1$ to $\Theta_2$ by $\TSym(\Theta_1,\Theta_2,\bs{s})$. We also use the notation $\operatorname{TAut}(\Theta,\bs{s}) := \TSym(\Theta,\Theta,\bs{s})$, and we often drop the dependence on $\bs{s}$. By a very similar argument as in section~\ref{sec:loopMomRep}, we see that the twisted symmetry transformations form a groupoid whose objects are the sectors.

In section~\ref{sec:summary}, we have already discussed that a twisted symmetry transformation induces an action on the twisted homology and cohomology groups by pushforwards and pullbacks, respectively. This statement is essentially unchanged in the case of relative (co-)homology groups. The additional conditions in Definition~\ref{def:twisted_symmetry} ensure that the pushforwards and pullbacks behave correctly in the presence of the $\delta$-forms. In particular, Condition 2 ensures that the twist localizes correctly onto the supports of the $\delta$-forms. Morevoer, Condition 3 implies that for every twisted symmetry transformation $f\in \TSym(\Theta_1,\Theta_2)$, $f^*\delta_{\Theta}$ defines a valid products of $\delta$-forms attached to the sector $\Theta_2$, for all $\Theta\preceq\Theta_1$.

Let us conclude this discussion by briefly commenting on how twisted symmetries act on compactly-supported cohomology groups and locally-finite homology groups (see also ref.~\cite{PhamSingularitiesOfIntegrals} for a discussion in a mathematical context).
Let $\varphi$ be a cocycle from the sector $\Theta_1$.
The support of $\varphi$ is defined as 
\begin{align}\label{eq:supp}
\mathrm{supp}(\varphi)&=\overline{\{\bx_1\in X_{-,\Theta_1}: \varphi(\bx_1)\neq 0\}}\,\notag.
\end{align} 
By definition, the differential forms used to define the compactly-supported cohomology groups all have compact support. 
The support of the pullback $f^*\varphi$
is the preimage of the support of  $\varphi$:
\begin{align}
\mathrm{supp}(f^*\varphi)
=f^{-1}\big(\mathrm{supp}(\varphi)\big)\,.
\end{align}
As twisted symmetry transformations
are required to be affine and bijective transformations, compact sets are mapped to compact sets, and the preimage of a compact set is compact.\footnote{In particular $f$ is continuous and so is its inverse (it is an instance of a homeomorphism). Continuous functions from a manifold to itself preserve compactness. It follows that $f$ is a so-called \emph{proper map}.} It then follows from eq.~\eqref{eq:supp} that the property of being compactly supported is preserved under the action of $f$.
Similarly, (homology classes of) locally finite chains are (possibly infinite) linear combinations $\gamma=\sum_{i}n_i\,\gamma_i$ of chain elements\footnote{Recall that the chain elements are differentiable maps $\gamma_i:\Delta^n\to X$, where $\Delta^n$ is an $n$-dimensional simplex.} $\gamma_i$ with coefficients $n_i$ from the local system, such that only a finite number of terms intersect any given compact set in $X_{-,\Theta_2}$. It is then easy to see that, if $f$ sends compact sets to compact sets, $f_*\gamma$ is locally finite whenever $\gamma$ is.

\subsection{Invariance of the pairings}
An important feature of (co-)homology theories is the existence of perfect pairings between the different cohomology and homology groups. We now show that these pairings are invariant under twisted symmetry transformations. While the mathematical statements established in the following are conceptually straightforward, we include them for completeness, as we have not found them formulated in this precise form in the  context of Feynman integrals.
The invariance of the period pairing was already shown in eq.~\eqref{eq:invarianceP}, and the invariance of the dual period pairing follows analogously. We therefore focus on the invariance of the intersection pairings. We present the argument in detail in the non-relative case, though we expect the same reasoning to apply to relative cases.

Let us consider the cohomology intersection pairing defined in eq.~\eqref{eq:cohompairing}.
Consider a differential form $\varphi\in H^n_{\text{dR}}(X,\nabla_{\omega})$ and a compactly-supported differential form $\check{\varphi}\in  H_{\mathrm{dR},c}^n(X,\check{\nabla}_\omega)$. Note that $\check{\varphi}$ is generally not holomorphic, because the compact support requires the introduction of anti-holomorphic contributions.\footnote{Explicit constructions of compactly-supported forms introduce for example a Heaviside step function vanishing in a small neighborhood around the boundaries, cf.~e.g., ref.~\cite{Matsumoto:1998ojf}.}
Let us write the most general compactly-supported form as $\check{\varphi}=\rd^n \bar{\bx}_1\,\check{R}_{0}(\bx_1)+\sum_{i=1}^n\check{R}_{i}(\bx_1)\wedge \rd x_{1,i}$, where $\check{R}_i(\bx_1)$ is a compactly-supported $(n-1)$-form and $\check R_0(\bx_1)$ is a compactly supported zero-form.
The intersection pairing is then given by 
\begin{align}
\langle\varphi|\check{\varphi}\rangle&=
\int_{\mathrm{supp} (\check{\varphi})}\rd^n\bx_1\wedge\rd^n\bar{\bx}_1\,R(\bx_1)\check{R}_0(\bx_1)\,.
\end{align}
Since $f$ is a bijection, we have
\beq
\mathrm{supp} (\check{\varphi})=f(f^{-1}(\mathrm{supp} (\check{\varphi}))) = f(\mathrm{supp} (f^*\check{\varphi}))\,.
\eeq
Hence:
\beq\bsp
\langle\varphi|\check{\varphi}\rangle&=
\int_{f(\mathrm{supp} (f^*\check{\varphi}))}\rd^n\bx_1\wedge\rd^n\bar{\bx}_1\,R(\bx_1)\check{R}_0(\bx_1)\\
&=\int_{\mathrm{supp} (f^*\check{\varphi})}\rd^n \bx_2 \wedge \,\rd^n\bar{\bx}_2\,R(f(\bx_2))\check{R}_0(f(\bx_2))\\
&=\int_{\mathrm{supp} (f^*\check{\varphi})}(f^*\varphi)\wedge (f^*\check{\varphi})\\
&=\langle f^*\varphi|f^*\check{\varphi}\rangle\,.
\esp\eeq
We thus see that the intersection pairing in cohomology is invariant.
The invariance of the homology intersection pairing follows by inserting the completeness relation
\begin{align}
\frac{1}{(2\pi i)^n}\sum_{i,j}|f^*\check{\varphi}_i\rangle C^{-1}_{ij}\langle f^*\varphi_j|=\frac{1}{(2\pi i)^n}\sum_{i,j}|\check{\varphi}_i\rangle C^{-1}_{ij}\langle\varphi_j|=\mathds{1}\,,
    \end{align}
into the twisted Riemann bilinear relations in eq.~\eqref{eq:TRBRs} and using the invariance of the (dual) period pairing 
\begin{align}
[\check{\gamma}|\gamma]&
=\frac{1}{(2\pi i)^n}\sum_{i,j}[\check{\gamma}|f^*\check{\varphi}_i\rangle C^{-1}_{ij}\langle f^*\varphi_j|\gamma]\notag \\
&=\frac{1}{(2\pi i)^n}\sum_{i,j}[f_*\check{\gamma}|\check{\varphi}_i\rangle C^{-1}_{ij}\langle \varphi_j|f_*\gamma]=[f_*\check{\gamma}|f_*\gamma]\,.
\end{align}
Note that for simplicity we restricted the above discussion to the case of twisted symmetry transformations. The arguments, however, easily generalize to any homeomorphism with trivial Jacobian that leaves the twist and the on-shell variety invariant. 

\subsection{Symmetry transformations for families of integrals}
\label{sec:symmetries_of_families}
Twisted symmetry transformations are defined as a set of affine transformations that act on the corresponding (co-)homology groups and that leave the natural pairings between these groups invariant. In particular, they act on \emph{integrands} (and also on integration cycles). The symmetry transformations from section~\ref{sec:symfeyn}, instead, act on a family of \emph{integrals}, obtained by integrating a set of differential forms over a {fixed} contour. It will turn out to be important to distinguish these two notions in a clean way.

Consider a fixed cycle $\gamma\in H_n(X\setminus D_-,D_+,\check{\mathcal{L}_\omega})$. We define the \emph{family of integrals} associated with the cycle $\gamma$ as the vector space generated by all twisted periods obtained by integrating a twisted cocycle over $\gamma$,
\beq\label{eq:V_gamma_Def}
\cV_{\gamma} := \Big\langle \langle\varphi|\gamma]: \varphi\in H^n(X\setminus D_-,D_+,\nabla_{\omega})\Big\rangle_{\cF}\,.
\eeq
As usual, the integrals may depend on some external variables $\bs{s}$ (the external kinematic data in the context of Feynman integrals), and we assume that the scalars are taken from the field of rational functions $\cF$ introduced earlier. In section~\ref{reltwistedrew} we have defined the sector filtration on the twisted cohomology group. We can use it to define a sector filtration for $\cV_{\gamma}$:
\beq\label{sectorfiltration}
S_{{\Theta}}\cV_{\gamma} = \Big\langle \langle \varphi|\gamma] : \varphi\in \operatorname{Sec}_{\Theta}\Big\rangle_{\cF}\,.
\eeq

In the (modern) literature on Feynman integrals, the space $\cV_\gamma$ and the cohomology group $H^n(X\setminus D_-,D_+,\nabla_{\omega})$ are often identified. While this identification is justified for many applications, we prefer to keep the two spaces separate.
In particular, the vector space $\cV_{\gamma}$ may not even have the same dimension as the twisted cohomology group. Instead, the dimension is bounded by the dimension of $H^n(X\setminus D_-,D_+,\nabla_{\omega})$.\footnote{See also refs.~\cite{Gasparotto:2023roh,e-collaboration:2025frv} for a discussion.}
In analogy to families of Feynman integrals, we refer to a choice of basis for $\cV_{\gamma}$ as a set of \emph{master integrals} for $\cV_{\gamma}$.

We now discuss symmetry transformations that map integrals from a sector ${\Theta}_1=(\Theta_{1,-},\Theta_{1,+})$ to a sector ${\Theta}_2=(\Theta_{2,-},\Theta_{2,+})$. Before we present the precise definition, let us give some motivation. Let $\varphi = \rd^nx_1\,R(\bs{x}_1)\in \operatorname{Sec}_{{\Theta}_1}$ be a cocycle from the sector ${\Theta}_1$, and let $f\in \TSym({\Theta}_1,{\Theta}_2)$ be a twisted symmetry from ${\Theta}_1$ to ${\Theta}_2$. If we write $\bs{x}_1 = f(\bx_2)=\bA\bs{x}_2+\bs{b}$ for some $\bA \in \operatorname{GL}(n,\cF)$ with $\det(f)=\det\bA=\pm1$ and $\bs{b}\in\cF^n$, then $f^*\varphi = \det\bA\,\rd^nx_2\,R(\bA\bs{x}_2+\bs{b})$ is a cocycle from sector ${\Theta}_2$. Performing this change of variables in the integral representation for $\langle \varphi|\gamma]$, we find
\beq
\langle \varphi|\gamma] = \int_\gamma\rd^nx_1\,R(\bs{x}_1)=\int_{f^{-1}(\gamma)}\rd^nx_2\,R(\bA\bs{x}_2+\bs{b})\,\det\bA = \langle f^*\varphi|f_*^{-1}\gamma]\,.
\eeq
Since $f^*\varphi\in \operatorname{Sec}_{{\Theta_2}}$, we see that $\langle f^*\varphi|f_*^{-1}\gamma]\in S_{{\Theta}_2}\cV_{\gamma}$, provided that $f^{-1}(\gamma) = f_*^{-1}\gamma$ is proportional to $\gamma$ itself. This motivates the following definition:

\begin{definition}
    The set of symmetry transformations from  a sector ${\Theta}_1$ to a sector ${\Theta}_2$ of the family of integrals $\cV_{\gamma}$ is
    \beq
    \Sym_{\gamma}({\Theta}_1,{\Theta}_2,\bs{s}) =\{f\in\TSym({\Theta}_1,{\Theta}_2,\bs{s}): f_*\gamma=\lambda_f\,\gamma, \mathrm{~for~some~} \lambda_f\in \cF^\times\}\,.
\eeq
\end{definition}
We also define $\Aut_{\gamma}(\Theta,\bs{s}) = \Sym_{\gamma}(\Theta,\Theta,\bs{s})$, and we often drop the dependence on $\bs{s}$.
Note that, just like in the case of Feynman integrals, these symmetry transformations form a groupoid whose objects are the sectors (see section~\ref{sec:loopMomRep}). We see that there is a close relationship between the symmetry groupoid $\Sym_{\gamma}$ of a family of integrals and the groupoid $\TSym$ of twisted symmetries. However, it is useful to keep the two notions distinct, because they admit different interpretations: $\TSym$ acts on the twisted (co-)homology groups, and transforms a twisted cocycle into another cocyle. Instead, the elements of $\Sym_{\gamma}$ are changes of variables that map integrals from a family to integrals from the same family. We say that two sectors ${\Theta}_1$ and ${\Theta}_2$ of $\cV_{\gamma}$ are \emph{equivalent}, ${\Theta}_1\sim {\Theta}_2$, if there is a symmetry transformation between them, i.e., if $\Sym_{\gamma}({\Theta}_1,{\Theta}_2)\neq \emptyset$. It is easy to check that this defines a genuine equivalence relation on the set of all sectors, and we have
$S_{{\Theta}_1}\cV_{\gamma} = S_{{\Theta}_2}\cV_{\gamma}$ precisely if ${\Theta}_1\sim {\Theta}_2$. We write
\beq
\cV_{\gamma}\simeq \bigoplus_{{\Theta}\textrm{ ineq.}}\operatorname{Gr}^S_{{\Theta}}\cV_{\gamma}\,,
\eeq
where the direct sum runs over all inequivalent sectors and the different summands are the associated graded spaces
\beq\label{eq:grV}
\operatorname{Gr}^S_{{\Theta}}\cV_{\gamma} = \frac{S_{{\Theta}}\cV_{\gamma}}{\sum\limits_{{\Theta}'\prec {\Theta}}S_{{\Theta}'}\cV_{\gamma}}\,.
\eeq
We define the \emph{number of master integrals in the sector $\Theta$} as
\beq\label{eq:N_Theta}
N_{{\Theta}} := \dim_{\cF} \operatorname{Gr}^S_{{\Theta}}\cV_{\gamma}\,.
\eeq

In the rest of the paper, we focus on a special case, which is the one relevant to understanding symmetry transformations for families of Feynman integrals. We assume from now on that the cycle $\gamma$ is real, i.e., $\gamma\subseteq \mathbb{R}^n$. Furthermore, $\cF$ consists of rational functions in $\bs{s}$ and $\eps$ with rational coefficients, and so $f$ is a \emph{real} affine map whenever $\bs{s}$ and $\eps$ take real values. Note that the condition $f_*\gamma=\lambda_f\gamma$ requires $f$ to fix $\gamma$ pointwise, and $f$ can at most change the orientation of $\gamma$. The possible change in orientation induced by a real affine map on a real cycle is captured by the sign of the determinant, and so in this setup we have
\beq
\lambda_{f} = \det(f)\,.
\eeq
We assume from here on that $\lambda_f$ takes this particular value.

To conclude, the groupoid $\Sym_{\gamma}$ of twisted symmetry transformations of  a family of  integrals defined by the real cycle $\gamma\subseteq\mathbb{R}^n$ is the subgroupoid of $\TSym$ such that $\gamma$ is fixed pointwise, possibly with a change of orientation.


\section{Feynman integrals and twisted symmetry transformations}\label{sec:twistedsymparam}

The construction of the symmetry transformations in section~\ref{sec:symfeyn} was based on the loop-momentum representation of Feynman integrals.
In section~\ref{sec:symfeyn} we have also identified the set of symmetry transformations in loop-momentum space with the set of permutations that relate the Lee-Pomeransky polynomials of the two sectors (up to factors that act trivially). 
We would like to connect the symmetry groupoid defined in section~\ref{sec:symfeyn} to the notion of twisted symmetry transformations introduced in section~\ref{sec:symandtwist}. While it is possible to define a twisted cohomology theory for Feynman integrals directly from loop-momentum space~\cite{Caron-Huot:2021xqj}, here we prefer to do that via the Lee-Pomeransky and Baikov representations.
However, we already mentioned that there is no unique twisted cohomology theory that one can attach to a family of Feynman integrals, and so there is also no unique way to assign a symmetry groupoid to a family of Feynman integrals. On the other hand, we expect the resulting family of integrals and their properties to be independent of the chosen integral representation. For example, we expect the number of master integrals in a sector to be an invariant of the family, irrespective of the integral representation. In the remainder of this section we explicitly show that the symmetry groupoid of the family is essentially the same for the loop-momentum, Baikov, Lee-Pomeransky and Feynman parameter representations.

\subsection{Symmetry transformations from the Lee-Pomeransky representation}
\label{sec:LP_sym}

Let us discuss the symmetry groupoid for a family of Feynman integrals defined by the Lee-Pomeransky representation in eq.~\eqref{LeePomeransky}. We start by briefly reviewing the relevant twisted cohomology theory (see, e.g., ref.~\cite{Lu:2024dsb}).
\begin{itemize}
    \item The twisted variety $\Sigma$ is defined by the vanishing of the Lee-Pomeransky polynomial $\cG$ of the top
    sector,
    \beq
    \Sigma = \big\{\bx \in \mathbb{C}^P : \cG(\bx,\bs{s}) = 0\big\}\,,
    \eeq
    where we recall that $P$ is the number of propagators in the top sector.
    \item The on-shell variety is empty, $D_-=\emptyset$.
    \item The irreducible components of  the relative boundary $D_+$ are given by the coordinate hyperplanes, 
    \beq
    D_+^{(i)} = \big\{\bx\in\mathbb{C}^P : x_i=0 \big\}\,,\qquad 1\le i\le P\,.
   \eeq 
\end{itemize}
Hence, we need to consider the relative twisted cohomology group $H^P(\mathbb{C}^P\setminus\Sigma,D_+,\nabla_{\omega})$.
Integrals as in eq.~\eqref{LeePomeransky} are only well defined if $\nu_i> 0$ for all $i$. Such an integral lies in the top sector of the family. Integrals from subsectors correspond to twisted cocycles with $\delta$-forms inserted~\cite{Lu:2024dsb}. Since $D_-=\emptyset$, the filtration $\operatorname{Sec}_{\Theta^-}$ is trivial, and so the sector filtration agrees with the filtration $\operatorname{Sec}_{\Theta^+}$ on the relative twisted cohomology group.

The integration cycle in eq.~\eqref{LeePomeransky} is the positive orthant
\beq
\gamma_F = \mathbb{R}_+^P\,.
\eeq
This is a relative cycle whose boundary lies in $D_+$,
and it may change in a non-trival fashion under an arbitrary  twisted symmetry transformation.

\begin{proposition}\label{prop:LP_symmetries}
    The set $\Sym_{\gamma_F}(\Theta_1,\Theta_2)$ of symmetry transformations in the Lee-Pomeransky representation from a sector $\Theta_1$ to a sector $\Theta_2$ is given by
    \beq\label{eq:Sym_LP}
    \Sym_{\gamma_F}(\Theta_1,\Theta_2) \simeq \mathbb{S}(\cG_1,\cG_2) \times \big(D(P-P_2,\mathbb{R}_+)\rtimes S_{P-P_2}\big)\,,
    \eeq
    where $\cG_k$ are the Lee-Pomeransky polynomials of the two sectors, $P_2$ is the number of active propagators in the sector $\Theta_2$ and $D(p,\mathbb{R}_+)$ is the group of $p\times p$ diagonal matrices with positive real entries and determinant 1.
\end{proposition}
Before we present the proof, let us connect this proposition to the results of section~\ref{sec:symfeyn}. In eq.~\eqref{eq:symmSetGeneral} we have shown that the twisted symmetry transformations between two sectors in loop-momentum space can be identified (up to factors that act trivially on the integrals) with the set of permutations of the Feynman parameters relating the Lee-Pomeransky polynomials. Proposition~\ref{prop:LP_symmetries} states that this set of permutations coincides with the set of symmetry transformations of the family obtained from the Lee-Pomeransky representation. This proves the claim that the set of symmetry transformations of the family obtained from the Lee-Pomeransky and loop-momentum representations are identical (as usual, up to factors that act trivially). 

\begin{proof}
Let $\cG$, $\cG_1$ and $\cG_2$ denote the Lee-Pomeransky polynomials of the top sectors and the sectors $\Theta_1$ and $\Theta_2$ respectively. 
These three polynomials are related by 
\begin{align}
\mathcal{G}|_{D_{+}^{\Theta_1}}=\mathcal{G}_1 \qquad \text{ and } \qquad \mathcal{G}|_{D_{+}^{\Theta_2}}=\mathcal{G}_2\,.
\end{align}
Since $D_-=\emptyset$, every permutation  $\sigma\in \mathbb{S}(\mathcal{G}_{1},\mathcal{G}_{2})$ defines a twisted symmetry transformation, and the remaining factors leave the integrals trivially invariant.
Our goal is now to show that also every element from $\Sym_{\gamma_F}(\Theta_1,\Theta_2)$ defines an element on the right-hand side of eq.~\eqref{eq:Sym_LP}.

Let us start by describing the set of twisted symmetry transformations that act on $\gamma_F$ as multiplication by the determinant. Any affine map that sends $\gamma_F$ to a multiple of itself must be linear and map $D_+$ to itself, i.e., it must be a linear map that permutes the coordinate hyperplanes. 
Such a map is a composition of a permutation and a rescaling of the Feynman parameters.

From eq.~\eqref{eq:UFDef} we know how the Lee-Pomeransky polynomial is related to the two Symanzik polynomials. Let $G_k$ , $\mathcal{U}_k$ and $\mathcal{F}_k$ denote the graphs and the Symanzik polynomials of the two sectors, respectively. We focus on the situation where $G_2$ is 2-vertex connected (otherwise the associated loop integrals factorize, and we analyze each factor separately). The two Symanzik polynomials are of different degrees in the Feynman parameters, hence they need to be separately invariant. Let us consider a permutation that sends an edge $e$ of $G_2$ to an edge $e'$ of $G_1$ and that rescales the corresponding Feynman parameter by $\lambda$. In total, such a transformation is $x_e\mapsto \lambda_ex_{e'}$. The different monomials in the first Symanzik polynomial correspond to the (complements of) spanning trees of the graph, and they all have coefficient 1. Since we map $\mathcal{U}_2$ to $\mathcal{U}_1$, this implies $\prod_{e\notin T}\lambda_e=1$ for all the spanning trees $T$ of $G_2$. Since $G_2$ is 2-vertex connected, an arbitrary edge $e_0\in E_{G_2}^{\rint}$ is an element of some cycle $C_{e_0}\in H_1(G_2,\mathbb{Z})$. Consider any spanning tree $T$ with $e_0\notin T$, but containing all other edges of $C_{e_0}$ (such a spanning tree always exists). By adding $e_0$ to $T$ we close the loop $C_{e_0}$ and obtain some other spanning tree $T^\prime$ by removing some edge $e_0^\prime \in C_{e_0}$. Then 
\begin{align}
1=\frac{\prod_{e\notin T}\lambda_e}{\prod_{e\notin T^\prime }\lambda_e}=\frac{\lambda_{e_0}}{\lambda_{e_0^\prime}}\,.
\end{align}
It follows that, if there exists some $C\in H_1(G_2,\mathbb{Z})$ that connects the two edges $e$ and $e^\prime$, then $\lambda_e=\lambda_{e^\prime}$. In particular, for a 2-vertex connected graph such a cycle always exists for any pair of edges. This implies $\lambda_e=\lambda$ for all $e\in E_{G_2}^\text{int}$.
Under a rescaling of all Feynman parameters by $\lambda$, the Symanzik polynomials transform as $\mathcal{U}_2\rightarrow \lambda^L \mathcal{U}_1$ and $\mathcal{F}_2\rightarrow \lambda^{L+1} \mathcal{F}_1$. Hence all rescalings are trivial, $\lambda=1$, and we only need to consider permutations of the Feynman parameters that send $\cG_2$ to $\cG_1$. These considerations explain the factor $\mathbb{S}(\cG_1,\cG_2)$ in eq.~\eqref{eq:Sym_LP}.

It remains to explain the remaining factors in eq.~\eqref{eq:Sym_LP}. The $P-P_2$ Feynman parameters not associated to $\Theta_1$ and $\Theta_2$ only enter the twisted cocycle through the $\delta$-forms. Hence, they can be permuted and rescaled among themselves. Here the rescaling factors are not bound to be unity, but their product still needs to be $1$ (because they are positive and the resulting symmetry transformation should have determinant $\pm1$). This explains the remaining factors in eq.~\eqref{eq:Sym_LP}. Note that this factor acts trivially on the integrals. 
\end{proof}

\subsection{Symmetry transformations from Baikov representations}\label{subsec:symBaikov}

Let us now discuss the symmetry transformations of a family defined from the democratic Baikov representation. We will comment on some subtleties that appear in the loop-by-loop representation~\cite{Frellesvig:2017aai,Frellesvig:2024ymq} at the end of this subsection. 
The relevant twisted cohomology group is $H^n(\mathbb{C}^n\setminus(\Sigma\cup D_-),\nabla_{\omega})$, with the following data (cf.~ref.~\cite{Mastrolia:2018uzb}):
\begin{itemize}
    \item The twisted variety $\Sigma$ is defined by the vanishing of the Baikov polynomial, 
    \beq
    \Sigma = \big\{\bz\in\mathbb{C}^n: \mathcal{B}(\bz)=0\big\}\,,
    \eeq
    where $\mathcal{B}(\bz)$ denotes the Baikov polynomial associated with the top-sector of the family of Feynman integrals.
    \item The on-shell variety $D_-$ is a union of hyperplanes $D_{-}^{(i)}$ corresponding to the vanishing of the (inverse) propagators. Indeed, since the Baikov variables are quadratic in the momenta, the inverse propagators define linear forms $h_i$ in the Baikov variables, $D_i = h_i(\bz)$, and so each $D_-^{(i)}$ is a hyperplane:
    \beq\label{hyperplanesbaikov}
    D_{-}^{(i)} = \big\{\bz\in\mathbb{C}^n: h_i(\bz)=0\big\}\,,\quad i=1,\ldots,P\,.
    \eeq
    We can choose these hyperplanes to be coordinate hyperplanes, which is the choice made in eq.~\eqref{eq_Baikov2}.
    \item The relative boundary is empty,
    \beq\label{eq:D+_Baikov}
    D_+ = \emptyset.
    \eeq
\end{itemize}

Let us now study symmetry transformations obtained from the Baikov representation.
From eq.~\eqref{eq:D+_Baikov} it follows that the filtration $\operatorname{Sec}_{\Theta^+}$ is trivial, and only the filtration $\operatorname{Sec}_{\Theta^-}$ is relevant. This filtration can easily be identified with the sector filtration on the family of Feynman integrals.
The integration cycle $\mathcal{C}$ (common to all members of the family) is given in eq.~\eqref{eq:Baikov_contour}. 

From the previous discussion, it follows that an element of $\Sym_{\mathcal{C}}(\Theta_1,\Theta_2)$ can be described as an affine map $f$ with $\det(f)=\pm1$ such that 
\begin{enumerate}
    \item the Baikov polynomial is invariant, $\mathcal{B}(f(\bz)) = \mathcal{B}(\bz)$,
\item $f$ exchanges the linear forms corresponding to the active propagators of the two sectors, i.e., there is a bijection $\sigma:\bd_{\Theta_1}^-\to \bd_{\Theta_2}^-$ such that
\beq
h_i(f(\bz)) = h_{\sigma(i)}(\bz)\,,\qquad \textrm{for all } i\in \bd_{\Theta_1}^-\,.
\eeq
\end{enumerate}
At this point, however, there is no reason to assume that the set of twisted symmetry transformations $\Sym_{\mathcal{C}}(\Theta_1,\Theta_2)$ obtained from the Baikov representation agrees with the set $\Sym(\Theta_1,\Theta_2)$ obtained from the loop-momentum representation or the set $\Sym_{\gamma_F}(\Theta_1,\Theta_2)$ from the Lee-Pomeransky representation.  Nevertheless, all these sets are related, as we will now argue. We will only explicitly work out the relationship between the symmetry groupoids in the Baikov and loop-momentum representations, because we already know that the latter is related to the one in the Lee-Pomeransky representation.

\begin{proposition} \label{lemmaequivalence}
Assuming $\det\bG(\bp)\neq 0$, there is a map from the symmetry transformations in the loop-momentum representation to those in the Baikov representation.
 \end{proposition}

\begin{proof}
It is convenient to write the Gram matrix in eq.~\eqref{eq_Baikvpolynomial} 
as 
\begin{align}\label{BaikovDef}
	\bG(\bk_p)=	\left(
	\begin{array}{cc}
		\bG(\bk) & \bs{Q} \\
		\bs{Q}^T & \bG(\bp)
	\end{array}
	\right)\,, \text{~~~where~~~} \bs{Q}=(k_i\cdot p_j)_{1\leq i\leq L, 1 \leq j \leq E}\,,
\end{align}
and the Gram matrix and the vectors  $\bk$ and $\bp$ have been defined in eqs.~\eqref{eq:Gram_def} and~\eqref{eq:momentum_vectors}, respectively. Furthermore we defined $\bk_p=(k_1,\dots ,k_L,p_1,\dots ,p_E)^T$.

Let us denote the set of symmetric $(L+E)\times (L+E)$ matrices with one fixed block $\Gp$ as
\begin{align}
\text{sym}^{\Gp}_{L+E}:=\left\{\left(\begin{array}{cc}
	\bB_{11}&\bB_{12}\\
    \bB_{12}^T&\Gp
\end{array}\right)\in \text{sym}_{L+E}\right\}\,,
\end{align} 
where $\text{sym}_k:=\{\bB\in \mathcal{F}(\bz)^{k\times k}\,|\, \bB^T=\bB\}$, and $\mathcal{F}(\bz)$ is the field of rational functions in the Baikov variables $\bz$ with coefficients in $\mathcal{F}$.
There is a bijection between  $\text{sym}^{\Gp}_{L+E}$ and the space of $n$-dimensional vectors with $n$ as in eq.~\eqref{baikovnumervars}
\begin{align}\label{bijB}
\operatorname{v}: \text{sym}^{\Gp}_{L+E}&\rightarrow \mathcal{F}(\bz)^{n},\\\left(\begin{array}{cc}
	\bB_{11}&\bB_{12}\\
    \bB_{12}^T&\Gp
\end{array}\right)&\mapsto (\text{vech}(\bB_{11}),\text{vec}(\bB_{12}))^T\,,\notag
\end{align}
with the half-vectorization $\mathrm{vech}$ and vectorization $\mathrm{vec}$ as defined as follows.
The \textit{vectorization} is an isomorphism between the vector space of matrices $\bQ\in \mathcal{F}^{L\times E}$ and $\mathcal{F}^{LE}$ obtained by joining the column vectors
\begin{align}\label{vecdef}
\text{vec}:\bQ\mapsto (\bQ^T_1,\dots ,\bQ^T_k)\,.
\end{align}
The \textit{half-vectorization} $\text{vech}(\bG(\bk))$ of a symmetric matrix is the vectorization of the lower-triangular part of the matrix.
The Baikov variables can then be chosen as
\begin{align}\label{zdef}
\bs{z}=\operatorname{v}(\bG(\bk_p))=\big(\mathrm{vech}(\bG(\bk)),\mathrm{vec}(\bs{Q})\big)^T\,.
\end{align}
The Baikov polynomial is $\mathcal{B}(\bz) = \det \bG(\bz)$, where by abuse of notation we write $\bG(\bz)$ for the Gram matrix $\bG(\bk_p)$ expressed in our choice of Baikov variables.

Our goal is to describe twisted symmetry transformations from $\Theta_1$ to $\Theta_2$ for our choice of Baikov variables. They take the form
\begin{align}\label{trafobaikovvars}
\bz_1=\bs{\mathcal{M}}\bz_2+\bs{b}\,,\qquad\bs{\mathcal{M}}\in \mathrm{GL}\left(n,\mathbb{\mathcal{F}}\right)\,,\quad \bs{b}\in \mathcal{F}^n,
 \end{align}
such that 
\begin{itemize}
\item $\det \bs{\mathcal{M}}=\pm 1$,
\item the Baikov polynomial is invariant: $\mathcal{B}(\bz_1)=\mathcal{B}(\bz_2)$, 
\item the transformation acts as a bijective map between the active propagators, i.e., $h_{\alpha_-(i)}(\bz_1)=h_{i}(\bz_2)$, with $i \in \bd^-_{\Theta_2}$ and $\alpha_{-}$ a bijection from $\bd^-_{\Theta_2}$ to $\bd^-_{\Theta_1}$ (with $\alpha_-=\sigma^{-1}$ compared to before). 
\end{itemize}
The hyperplanes $h_i$ are defined as in eq.~\eqref{hyperplanesbaikov}.
Note that our special choice of Baikov variables in eq.~\eqref{zdef} is not a restriction. Indeed, every other choice of Baikov variables $\bz'$ is related to our choice by an invertible affine transformation, $\bz ' = \bs{\mathcal{J}}\bz + \bz_0$, with $\bs{\mathcal{J}}\in \mathrm{GL}\left(n,\mathbb{\mathcal{F}}\right)$ and $ \bz_0\in \mathcal{F}^n$. It is then easy to work out how a twisted symmetry transformation acts in the Baikov variables $\bz'$, and all conclusions remain unchanged.

Our goal is to show that every symmetry transformation in loop-momentum space gives rise to a twisted symmetry transformation in the Baikov representation. The former have been completely described in section~\ref{sec:symfeyn}, where we have shown that every symmetry transformation in loop-momentum space can be cast in the form 
\beq
\bk_p \mapsto \bT\bk_p\,,\qquad \bT = \left(	\begin{array}{cc}
		\bs{L} & \bs{M}  \\
	0& \bs{\bN}
	\end{array}\right) \,,
    \eeq
    and the matrices $\bs{L}$, $\bs{M}$ and $\bs{N}$ can be constructed using graph-theoretical input.
It is easy to see that the matrix $\bT$ acts on the Gram matrix $\bG(\bk_p)$ via the map, 
\begin{align}\label{BaikovTraf}
\bG(\bk_p)\mapsto T_G(\bG(\bk_p)) := \bs{T} \bG(\bk_p) \bs{T}^T\,,
\end{align}
The map in eq.~\eqref{BaikovTraf} is explicitly given by
\begin{align}\label{mapexpl}
    \bG(\bk)&\mapsto \bL\bG(\bk)\bL^T+\bL\bQ\bM^T+\bM\bQ^T\bL^T+\bM\bG(\bp)\bM^T\,,\notag\\
    \bQ&\mapsto \bL\bQ \bN^T+\bM\bG(\bp)\bN^T\,,\\
    \bG(\bp)&\mapsto \bN\bG(\bp)\bN^T\,.\notag
\end{align}
Note that eq.~\eqref{BaikovTraf} is insensitive to the sign of $\bT$, i.e., $\pm\bT$ induce the same symmetry transformation for the Baikov parametrization. We know from section~\ref{sec:symfeyn} that $\det \bT=\pm 1$. Hence, the determinant of the Gram matrix, and thus the Baikov polynomial, are invariant. Thus, every symmetry transformation $\bT$ in loop-momentum space induces a transformation $T_G$ on the Gram matrix which leaves the Baikov polynomial invariant. 

This by itself is not yet sufficient to conclude that the transformation $T_G$ induced by $\bT$ defines a twisted symmetry transformation in the Baikov variables. First, note that we obtain an affine transformation on the Baikov variables. More precisely, the action of $T_G$ on the Baikov variables is $\operatorname{v} \circ T_G\circ \operatorname{v}^{-1}$, and it is easy to check that this transformation is affine. We write
\begin{align}
(\operatorname{v} \circ T_G\circ \operatorname{v}^{-1})(\bz)= \bs{\mathcal{M}}\bz+\bs{b}\,,
\end{align}
for some $\bs{\mathcal{M}}\in \mathcal{F}^{n\times n}$ and a shift $\bs{b}\in \mathcal{F}^n$.
 Second, the condition that we obtain a bijection between $\bd^-_{\Theta_1}$ and $\bd^-_{\Theta_2}$ is obvious, because this condition was already satisfied in loop-momentum space.  
 
It then finally remains to show  that $\det \bs{\mathcal{M}}=\pm1$.
From eq.~\eqref{mapexpl} it becomes clear that the $\mathrm{vec}(\bs{Q})$ part of the vector in eq.~\eqref{zdef} does not receive contributions from $\mathrm{vech}(\bG(\bk))$, such that $\bs{\mathcal{M}}$ is upper block-triangular,
\begin{align}
    \label{eq:baikovM}
\bs{\mathcal{M}}=	\left(\begin{array}{cc}
	\bs{\mathcal{M}}^{(1)} & * \\
0 & 	\bs{\mathcal{M}}^{(2)}
\end{array}\right)\,,
\end{align}
with  $\bs{\mathcal{M}}^{(1)}\in\text{GL}\big(\tfrac{L(L+1)}{2},\mathcal{F}\big)$ and $\bs{\mathcal{M}}^{(2)}\in \text{GL}(EL,\mathcal{F})$.
We will not be concerned with the upper right entry, because $\det \bs{\mathcal{M}}=\det\bs{\mathcal{M}}^{(1)} \det \bs{\mathcal{M}}^{(2)}$. Equation~\eqref{mapexpl} then shows that 
\beq\bsp
    \bs{\mathcal{M}}^{(1)}\,\mathrm{vech}(\bG(\bk))&=\mathrm{vech}(\bL\bG(\bk)\bL^T)\,\\
     \bs{\mathcal{M}}^{(2)}\,\mathrm{vec}(\bQ)&=\mathrm{vec}(\bL\bQ\bN^T)\,.
    \esp\eeq
We now show how to compute the determinants of $\bs{\mathcal{M}}^{(1)}$ and $\bs{\mathcal{M}}^{(2)}$.

 As a starting point, consider an $\mathcal{F}$-vector space $V$ of (finite) dimension $d$ and $T_V:V\rightarrow V$ a linear operator acting on it. The determinant of a linear operator is the determinant of the operator represented in some basis of the vector space. The determinant is invariant under similarity transformations, and hence basis-independent.  Furthermore, for two isomorphic spaces $V\simeq W$, any linear map on $V$ is similar to an induced linear map on $W$.
Let us consider the linear operator $T:\mathcal{F}^{E\times L}\rightarrow \mathcal{F}^{E\times L}$ given by 
\begin{align}\label{matrixtrafo}
T(\bQ) = \bN \bQ \bL^T\,, \quad \text{ with } \quad \bN\in \text{GL}(E,\mathcal{F}),\,\bL\in \text{GL}(L,\mathcal{F})\,.
\end{align}
On the corresponding $EL$-dimensional vector it acts like $\text{vec}\circ T\circ \text{vec}^{-1}$, where the image of $\textrm{vec}(\bQ)$ is explicitly given by
\begin{align}
	\text{vec}(T(\bQ))=\text{vec}(\bL \bQ \bN^T)=	(\bN\otimes \bL) \text{vec}(\bQ)=  \bs{\mathcal{M}}^{(2)}\text{vec}(\bQ)\,.
\end{align}
Clearly,
$\det(\bs{\mathcal{M}}^{(2)})=\det(\bN\otimes \bL)=\det(\bN)^L \det(\bL)^E$, and we have seen in section~\ref{sec:symfeyn} that the determinants of $\bL$ and $\bN$ are $\pm 1$. Hence $\det(\bs{\mathcal{M}}^{(2)})=\pm1$.

Similarly, we consider a transformation on the half-vectorization $\text{vech}(\bG(\bk))$ for $\bG(\bk)\in \mathcal{F}^{L\times L}$ as
\begin{align}
\text{vech}(\bG(\bk))\mapsto (\bL\otimes \bL)^+\text{vech}(\bG(\bk))=  \bs{\mathcal{M}}^{(1)}\text{vech}(\bG(\bk))\,,
\end{align}
where $(\bL\otimes \bL)^+$ is the action of the tensor product on the symmetric subspace.

To compute $\det(\bL\otimes \bL)^+$, we first triangularize $\bL$. Over a field $\mathcal{F}$ with $\mathbb{C}\subseteq \mathcal{F}$, every matrix is triangularizable, so we write
\begin{align}
\bL = \bU \bs{D} \bU^{-1}, \qquad 
\bs{D} =
\begin{pmatrix}
\lambda_1 & *        & \cdots & * \\
0         & \lambda_2 & \cdots & * \\
\vdots    & \ddots   & \ddots & \vdots \\
0         & \cdots   & 0      & \lambda_L
\end{pmatrix}\,, \qquad \det\bU=1.
\end{align}
The matrix $\bs{D}$ acts on a vector space $V_L$ and we denote its standard basis by $\{\bs{e}_1,\dots \bs{e}_L\}$.
The operator $D$ then acts on a basis vector $\bs{e}_j$ as
\begin{align}
\bD\,(\bs{e}_j) = \sum_{a\le j} d_{aj}\,\bs{e}_a\,\quad d_{jj}=\lambda_j, \quad \forall j\,.
\end{align}
It is a standard fact that $\bD\otimes \bD$ is also upper triangular. Then from acting on the tensor product basis $\{\bs{e}_i\otimes \bs{e}_j\}$ (in its standard ordering) of $V_L\otimes V_L$, which yields
\begin{align}\label{tensorprod}
(\bD\otimes \bD)(\bs{e}_i\otimes \bs{e}_j)
= \sum_{a\le i,\; b\le j} d_{ai}d_{bj}\,\bs{e}_a\otimes \bs{e}_b \,,
\end{align}
we see that the diagonal entries are given by $d_{ii}d_{jj}=\lambda_{i}\lambda_{j}$

Next, let us decompose the basis into a symmetric and antisymmetric part. We perform a change of basis such that 
\begin{align}\label{decompsym}
		V_L\otimes V_L=\text{Span}\Big(& \left\{\frac{1}{2}(\bs{e}_i\otimes \bs{e}_j+\bs{e}_j\otimes \bs{e}_i); i
        \leq j \right\} \cup\left\{\frac{1}{2}(\bs{e}_i\otimes \bs{e}_j-\bs{e}_j\otimes \bs{e}_i); i <j \right\}\Big)\,.
		\end{align}
We order the symmetric basis vectors $f_{ij}=\frac{1}{2}(\bs{e}_i\otimes \bs{e}_j+\bs{e}_j\otimes \bs{e}_i)$, $i\le j$,  lexicographically by the pairs $(i,j)$, i.e. $(i,j)<(i^\prime,j^\prime)$ if $i<i^\prime$ or $i=i^\prime \land j<j^\prime$.
As we know the action of $\bD$ on the pairs $\bs e_i\otimes \bs e_j$ we find
\begin{align}
( \bD\otimes  \bD)(\bs e_i\otimes \bs e_j\pm \bs e_j\otimes \bs e_i)
=
\sum_{a\le i,\ b\le j} d_{ai}d_{bj}\,\bs (\bs e_a\otimes \bs e_b
\pm
\bs e_b\otimes \bs e_a)\,.
\end{align}
First notice from this expression that the corresponding matrix in this basis is block diagonal  
\begin{align}
 \bs{D} \otimes \bs{D}=	\left(\begin{array}{cc}
 		(\bD\otimes \bD)^+ & 0\\
 		0 & (\bD\otimes \bD)^-
 	\end{array}\right)\,.
 	\end{align}
Then notice that 
 \begin{align}
( \bD\otimes  \bD)f_{ij}
=
\sum_{a\le i,\ b\le j} d_{ai}d_{bj}\,\bs (\bs e_a\otimes \bs e_b
\pm
\bs e_b\otimes \bs e_a)=\sum_{a\le i,\ b\le j} d_{ai}d_{bj}\,f_{\min(a,b),\max(a,b)}\,.
\end{align}
 One can show that
\begin{align}
a\leq i\land b\leq j\land i\leq j\implies (\min(a,b),\max(a,b))\leq (i,j)\,.
\end{align}
Hence, on the RHS, we only sum over basis elements that appear before (or are equal to) $f_{ij}$ in the ordering and
the matrix $(\bs{D} \otimes \bs{D})^+$  is upper triangular with diagonal entries $\lambda_i\lambda_j$.
 Now we can easily compute the determinant
    \begin{align}\det (\bs{D} \otimes \bs{D})^+=\prod_{i=1}^L\lambda^2_i \prod_{j=1}^L\lambda_j  \prod_{l<j}\lambda_{l}=\prod_{i=1}^L \lambda^2_i \prod_{j=1}^L\lambda_j^{L-1}=\prod_{i=1}^L\lambda_i^{L+1}=\det(\bs{D})^{L+1}\,.
    \end{align}
Hence,
  \begin{align}
 	\det   (\bs{\mathcal{M}}^{(1)})
&=\det((\bD\otimes \bD)^+)=\det(\bs{D})^{L+1}=\det(\bs{L})^{L+1}=\pm1\notag\,,
 		\end{align}  
        and so we see that every  symmetry transformation from $\Theta_1$ to $\Theta_2$ in loop-momentum space induces a twisted symmetry transformation in Baikov variables.

\end{proof}

Let us now discuss what happens in the converse direction. 
\begin{proposition} \label{lemmabaikovconverse}
The set $\Sym_{\mathcal{C}}(\Theta_1,\Theta_2)$ of twisted symmetry transformations in the Baikov representation  takes the form
\beq\label{proddec}
\Sym_{\mathcal{C}}(\Theta_1,\Theta_2) \simeq \mathbb{S}(\mathcal{G}_1,\mathcal{G}_2)\times \mathbb{A}_{\mathcal{C}}(\Theta_2)\,,
\eeq
where $\mathbb{A}_{\mathcal{C}}(\Theta_2)$ is a set of transformations that acts trivially on the integrals of the sector. Note that this is only a bijection of sets and does not respect the underlying groupoid structure.
 \end{proposition}

\begin{proof}

We have the following maps between the sets of symmetry transformations in the loop-momentum, the Baikov and the Lee-Pomeransky representation:
\[
\begin{array}{ccc}
\Sym(\Theta_1,\Theta_2)
&\xrightarrow{\rho}&
\Sym_{\mathcal C}(\Theta_1,\Theta_2)
\\[2mm]
\pi\downarrow && \downarrow \pi_{\mathcal{C}}
\\[2mm]
\mathbb S(\mathcal G_1,\mathcal G_2)
&=&
\mathbb S(\mathcal G_1,\mathcal G_2)
\end{array}\,.
\]
The maps $\pi$ and $\pi_{\mathcal{C}}$ are projections that isolate the underlying bijection $\sigma$ on the propagators induced by an element of the symmetry in loop momentum space or the Baikov representation. By construction, a map $M_\sigma \in \Sym(\Theta_1,\Theta_2)$ induces a bijection $\sigma \in \mathbb S(\mathcal G_1,\mathcal G_2)$. Then according to proposition~\ref{lemmaequivalence}, $M_\sigma$ in turn determines an element  $f_\sigma\in \Sym_{\mathcal C}(\Theta_1,\Theta_2)$, i.e. there is a map $\rho:\Sym(\Theta_1,\Theta_2)\to \Sym_{\mathcal C}(\Theta_1,\Theta_2)$ with $\rho(M_\sigma)=f_\sigma$, which induces the same bijection $\sigma$ on the propagators in the Baikov representation. In other words, the diagram above commutes $\pi=\pi_{\mathcal{C}}\circ \rho$. As shown in section~\ref{sec:symfeyn}, the projection $\pi$ is surjective. By the commutativity of the diagram, surjectivity of $\pi_{\mathcal{C}}$ follows. 
Hence there exists a set $\mathbb{A}_{\mathcal{C}}(\Theta_2)$ with
\beq
\Sym_{\mathcal{C}}(\Theta_1,\Theta_2) \hookrightarrow\mathbb{S}(\mathcal{G}_1,\mathcal{G}_2)\times \mathbb{A}_{\mathcal{C}}(\Theta_2)\,,
\eeq
where $\times$ is a cartesian product of sets.
Indeed the above is not only an inclusion but a bijection. Note that, because of the groupoid structure, for any two $S_{1\sigma},S_{2\sigma}\in\Sym_{\mathcal{C}}(\Theta_1,\Theta_2) $  with $\pi_\mathcal{C}(S_{1\sigma})=\pi_\mathcal{C}(S_{2\sigma})=\sigma$, we have $\pi_\mathcal{C}(S_{1\sigma}\circ S_{2\sigma}^{-1})=\pi_\mathcal{C}(S_{1\sigma})\circ \pi_\mathcal{C}(S_{2\sigma})^{-1}=\mathrm{id}$. Define 
\begin{align}
\mathbb{A}_\mathcal{C}(\Theta_2):=\{a\in \Sym_{\mathcal{C}}(\Theta_2,\Theta_2)\,, \text{ with } \pi_\mathcal{C}(a)=\mathrm{id}\}\,.
\end{align}
Now fix a reference element in each fiber $S_\sigma \in \pi_\mathcal{C}^{-1}(\sigma)$. Then, by definition any other element $S^\prime_\sigma \in \pi_\mathcal{C}^{-1}(\sigma)$ can be written as  $S^\prime_\sigma=S_\sigma\circ a$ with $a= S_\sigma^{-1}\circ S^\prime_\sigma \in \mathbb{A}_\mathcal{C}$ and hence is uniquely identified with the element $(\sigma,a)\in \mathbb{S}(\mathcal{G}_1,\mathcal{G}_2)\times \mathbb{A}_{\mathcal{C}}(\Theta_2)$. Conversely, for any $(\sigma,a)\in \mathbb{S}(\mathcal{G}_1,\mathcal{G}_2)\times \mathbb{A}_{\mathcal{C}}(\Theta_2)$ there exists an $S_\sigma\circ a\in 
\Sym_{\mathcal{C}}(\Theta_1,\Theta_2)$. Hence
\beq
\Sym_{\mathcal{C}}(\Theta_1,\Theta_2) \simeq \mathbb{S}(\mathcal{G}_1,\mathcal{G}_2)\times \mathbb{A}_{\mathcal{C}}(\Theta_2)\,.
\eeq

It remains to show that  $\mathbb{A}_{\mathcal{C}}(\Theta_2)$ acts trivially on the family of integrals. To understand this, we start by noting that we can always pick a basis of master integrals without numerators. The existence of such a basis follows directly from the the Lee-Pomeransky representation and the associated relative twisted cohomology group. We know that the relative twisted cohomology group is generated by $\delta$-forms and non-relative cocycles. The $\delta$-forms arise from propagators raised to the power $\nu_i=0$ in eq.~\eqref{LeePomeransky}, because
\beq
x_i^{\nu_i-1} = \frac{1}{\nu_i}\delta(x_i) + \mathcal{O}(\nu_i^0)\,,
\eeq
and the pole cancels against the divergent $\Gamma$-function in the prefactor in eq.~\eqref{LeePomeransky}. Numerator factors correspond to propagators raised to powers $\nu_i<0$ and give rise to derivatives of $\delta$-functions (cf.~e.g.,~ref.~\cite{Chen:2023eqx}). Since we can find a basis of the relative twisted cohomology group that only involves $\delta$-functions, but not their derivatives, we conclude that there is a basis of master integrals where no propagators are raised to powers $\nu_i<0$, i.e., a basis of master integrals without numerator factors.

If we now pick such a basis of master integrals and write all basis integrals in the Baikov representation, we can easily see that the integrand of a master integral from the sector $\Theta$ only involves the Baikov polynomial and linear forms $h_i$ with $i\in\bd_{\Theta}^-$ raised to negative powers, but no additional polynomials in the numerator. The Baikov polynomial and the linear forms are (by definition) invariant under $\mathbb{A}_{\mathcal{C}}(\Theta)$, and hence all master integrands from the sector $\Theta$ are invariant under $\mathbb{A}_{\mathcal{C}}(\Theta)$. We can write any other integral of the family as a linear combination of the particular master integrals with no numerators. Since $\mathbb{A}_{\mathcal{C}}(\Theta)$ acts in a trivial way on these we see that $\mathbb{A}_{\mathcal{C}}(\Theta)$ acts in a trivial way on any integral from this sector. Note that it might not act trivially on any \textit{integrand}, however all non-trivial changes will correspond to a total derivative.
\end{proof}
\paragraph{The loop-by-loop Baikov representation.}
So far we have only discussed the democratic Baikov representation. Let us conclude by making some comments about what happens in the loop-by-loop Baikov representation~\cite{Frellesvig:2017aai,Frellesvig:2024ymq}.

The loop-by-loop Baikov representation can be obtained from the democratic Baikov representation by integrating out some of the {irreducible scalar products (ISPs)} \cite{Jiang:2023qnl}. As a consequence, some (but not all) of the symmetries of the integrals may already become manifest at the level of the twist. Accordingly, it can happen that the dimension of the twisted cohomology group from a loop-by-loop Baikov representation is less than the dimension of the twisted cohomology groups obtained from the democratic one. 

At the same time, it can always happen that the twist includes a factor given by an ISP, which regulates possible poles in this variable. In such a case the twisted cohomology group includes differentials where this ISP appears in the denominator, i.e., acts as a genuine propagator. Such differentials can be associated with super-sectors, see also the discussion in ref.~\cite{Pogel:2025bca}.   In particular, even if the democratic Baikov representation is free of super sectors, they might appear in the loop-by-loop representation. This leads to an increase of the dimension of the twisted cohomology group from a loop-by-loop representation compared to the democratic one. 

Each element of the twisted cohomology group in the democratic Baikov representation induces an element of the twisted cohomology group in the loop-by-loop representation (see e.g., refs.~\cite{Frellesvig:2019kgj,Frellesvig:2024ymq} for a discussion of how this works for differentials involving the ISPs that are integrated out). In that sense, the twisted symmetry transformations acting on the cohomology group from the democratic representation induce an action on the one from the loop-by-loop representation. It would be interesting to see if this induced action is always associated to a twisted symmetry transformation in the loop-by-loop approach. We checked some loop-by-loop representations of the three-loop equal-mass banana family, where this was indeed the case.

\subsection{Summary}
\label{sec:sec5_summary}

Let us briefly summarize the results of this section and what we have learned about the different representations of Feynman integrals. We have considered the set of symmetry transformations between two sectors $\Theta_1$ and $\Theta_2$ in the loop-momentum (LM), Baikov (B) and Lee-Pomeransky (LP) representations. They are:\footnote{For the loop-momentum representation, we could also include the Lorentz transformations.} 
\beq\begin{split}
    \Sym^{\textrm{LM}}(\Theta_1,\Theta_2) &\,\simeq \mathbb{S}(\cG_1,\cG_2) \times \mathbb{Z}_2^c\times \mathbb{O}\big(\overline{\bG_1}^{\perp},\overline{\bG_2}^{\perp}\big)\,,\\
        \Sym_{\mathcal{C}}^{\textrm{B}}(\Theta_1,\Theta_2) &\,\simeq \mathbb{S}(\cG_1,\cG_2)\times\mathbb{A}_{\mathcal{C}}(\Theta)\,,\\ 
        \Sym_{\gamma_F}^{\textrm{LP}}(\Theta_1,\Theta_2) &\,\simeq \mathbb{S}(\cG_1,\cG_2) \times\big(S_{P-P_2}\rtimes D(P-P_2,\mathbb{R}_+)\big)\,,
        \end{split}\eeq
        where $\Sym^{\textrm{LM}}(\Theta_1,\Theta_2)=\Sym(\Theta_1,\Theta_2)$ is the set of symmetry transformations in loop-momentum space introduced in section~\ref{sec:symfeyn}.
        We see that in all cases the groupoid of symmetry transformations takes the form $\mathbb{S}(\cG_1,\cG_2) \times \mathbb{X}$, where $\mathbb{X}$ is a set of transformations specific to the representation leaving the integrands trivially invariant, at least when working with an appropriately chosen set of basis integrals. Since in the subsequent sections, we will be concerned with how the symmetry transformations act on a basis of integrals or cocycles, we may assume that they have been chosen in such a way that the action of $\mathbb{X}$ is trivial. We may encode this into a short exact sequence
        \beq
        0\to \mathbb{X} \to \Sym^\textrm{X}(\Theta_1,\Theta_2) \to \mathbb{S}(\cG_1,\cG_2)\to 0\,,
        \eeq 
        where $\Sym^\textrm{X}(\Theta_1,\Theta_2)$ is the symmetry groupoid in the representation X. The map from $\mathbb{X}$ to $\Sym^\textrm{X}(\Theta_1,\Theta_2)$  is the inclusion, and the kernel of the subsequent map is $\mathbb{X}$.
        The part that acts non-trivially instead is universal across all representations, and it is given by $\mathbb{S}(\cG_1,\cG_2)\simeq \Sym^X(\Theta_1,\Theta_2)/\mathbb{X}$, i.e., by the permutations that relate the Lee-Pomeransky polynomials of the two sectors. It is this set of permutations that canonically encodes the symmetry groupoid of a family of Feynman integrals, irrespective of the representation used to define the integrals. We will therefore from now on refer to $\mathbb{S}(\cG_1,\cG_2)$ as the \emph{set of symmetry transformations from $\Theta_1$ to $\Theta_2$}. If $G$ is the Feynman graph that defines the top sector of the family, we denote this universal part of the groupoid of symmetry transformations by 
        \beq
        \Sym_G(\Theta_1,\Theta_2) \simeq \mathbb{S}(\cG_1,\cG_2)\,.
        \eeq
        In the case of symmetry transformations from a sector $\Theta$ to itself, we use the notation
        \beq
        \Aut_G(\Theta) = \Sym_G(\Theta,\Theta)\,.
        \eeq

        Let us conclude by making a comment. We can easily extend our analysis to include also the Feynman parameter representation. Indeed, it follows from the proof in section~\ref{sec:LP_sym} that any element of $\Sym_G(\Theta_1,\Theta_2)$ will map the Symanzik polynomials of the sectors into each other. Conversely, every permutation that maps the Symanzik polynomials into each other will lie in $\Sym_G(\Theta_1,\Theta_2)$. Hence, we see that $\Sym_G(\Theta_1,\Theta_2)$ is also the canonical set of symmetry transformations between the sectors in the Feynman parameter representation, in agreement with Lemma~\ref{symmTransEquivDef}.

\subsection{Example: 3-loop banana integrals}\label{sec3loopbanana}
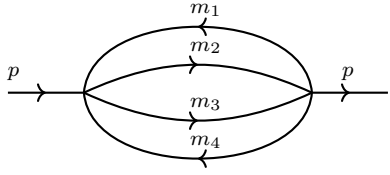
\begin{figure}[!th]
\centering
\begin{tikzpicture}
\coordinate (llinks) at (-2.5,0);
\coordinate (rrechts) at (2.5,0);
\coordinate (links) at (-1.5,0);
\coordinate (rechts) at (1.5,0);
\begin{scope}[very thick,decoration={
    markings,
    mark=at position 0.5 with {\arrow{>}}}
    ] 
\draw [-, thick,postaction={decorate}] (links) to [bend right=25]  (rechts);
\draw [-, thick,postaction={decorate}] (links) to [bend right=-25]  (rechts);

\draw [-, thick,postaction={decorate}] (rechts) to [bend right=85]  (links);
\draw [-, thick,postaction={decorate}] (llinks) to [bend right=0]  (links);
\draw [-, thick,postaction={decorate}] (rechts) to [bend right=0]  (rrechts);
\end{scope}
\begin{scope}[very thick,decoration={
    markings,
    mark=at position 0.5 with {\arrow{>}}}
    ]
\draw [-,  thick,postaction={decorate}] (rechts) to  [bend left=85] (links);
\end{scope}
\node (d1) at (0,1.1) [font=\scriptsize, text width=.2 cm]{$m_1$};
\node (d2) at (0,0.6) [font=\scriptsize, text width=.2 cm]{$m_2$};
\node (d3) at (1.4,-0.15) [font=\scriptsize, text width=3 cm]{$m_3$};
\node (d4) at (1.4,-.65) [font=\scriptsize, text width=3 cm]{$m_4$};
\node (p1) at (-2.0,.25) [font=\scriptsize, text width=1 cm]{$p$};
\node (p2) at (2.4,.25) [font=\scriptsize, text width=1 cm]{$p$};
\end{tikzpicture}
\caption{The three-loop banana graph $G_{\text{ban}}$.}
\label{fig:banana}
\end{figure}
As an illustration of the discussion above, we consider (the maximal cut of) the three-loop banana integral (see figure~\ref{fig:banana}) in $D=2-2\eps$ dimensions, defined by

\begin{equation}
    \label{eq:ban_int_def}
    I_{\nu_1,\ldots,\nu_9}=e^{3\gamma_E\eps}\int\left(\prod_{a=1}^{3}\frac{d^Dk_a}{i\pi^{\frac{D}{2}}}\right) \frac{1}{D_1^{\nu_1}\,D_2^{\nu_2}\,D_3^{\nu_3}\,D_4^{\nu_4}\,D_5^{\nu_5}\,D_6^{\nu_6}\,D_7^{\nu_7}\,D_8^{\nu_8}\,D_9^{\nu_9}}\,,
\end{equation}
with the propagators
\begin{align}\label{3looppropchoice}
    &D_1=k_1^2-m_1^2\,,\qquad D_2=k_2^2-m_2^2\,,\qquad D_3=(k_1-k_3)^2-m_3^2\,,\nonumber\\
    & D_4=(k_2-k_3-p)^2-m_4^2\,,\qquad D_5=k_3^2\,,\qquad D_6=k_3\cdot p\,,\\
    & D_7=k_1\cdot p\,,\qquad D_8=k_2\cdot p\,,\qquad D_9=k_1\cdot k_2\,.\nonumber
\end{align}
We require $\nu_i\le 0$ for $i\ge 5$, and we consider the sector  $\Theta_{\text{ban}}=(1,1,1,1,0,0,0,0,0)$, corresponding to the graph in figure~\ref{fig:banana}.
This integral was recently studied in detail in ref.~\cite{Duhr:2025kkq,Pogel:2025bca}, and we follow closely the conventions and notations of ref.~\cite{Duhr:2025kkq}. 
After solving the IBP identities using \texttt{LiteRed}~\cite{Lee:2012cn}, we identify 15 master integrals, 11 of which are in the sector $\Theta_{\text{ban}}$ but not in lower sectors. For simplicity, we will in the following focus on the maximal cut. As a basis of master integrals ${\bs I} = (I_1,\ldots,I_{11})^T$ we choose
\begin{equation}
\begin{aligned}
\label{bananbasis}
    I_1 & =     I_{1,1,1,1,0,0,0,0,0}\,, \qquad &
    I_2 &=     I_{2,1,1,1,0,0,0,0,0}\,, \qquad &
    I_3 & =     I_{1,2,1,1,0,0,0,0,0}\,,\\
    I_4 &=     I_{1,1,2,1,0,0,0,0,0}\,,\qquad& 
    I_5  &=     I_{1,1,1,2,0,0,0,0,0}\,, \qquad &
    I_{6} & =     I_{1,1,1,1,-1,0,0,0,0}\,, \\
    I_{7} & =     I_{1,1,1,1,0,-1,0,0,0}\,,\qquad &
    I_{8}  &=     I_{1,1,1,1,0,0,-1,0,0}\,,\qquad &
    I_{9} & =     I_{1,1,1,1,0,0,0,-1,0}\,, \\
    I_{10}  &=     I_{1,1,1,1,0,0,0,0,-1}\,,\qquad&
    I_{11} & =     I_{3,1,1,1,0,0,0,0,0} \,. \qquad & &
\end{aligned} 
\end{equation}

For generic values of the masses $m_i$, the symmetry group of the sector $\Theta_{\text{ban}}$ is trivial.
Therefore,  we focus from here on on the equal-mass case, where $m_i^2=m^2$, $1\le i\le 4$. 
It is easy to see that the group of symmetry transformations of the top sector is
\beq
\Aut_{G_{\text{ban}}}(\Theta_{\text{ban}})\simeq S_4\,,\quad \textrm{ if } m_i^2=m^2, 1\le i\le 4\,.
\eeq
This symmetry is easy to identify at the level of the Feynman graph in figure~\ref{fig:banana}. Previously we have argued that $\Aut_{G_{\text{ban}}}$ is the part of the symmetry group of the sector $\Theta_{\text{ban}}$ that is universal across all representations. Let us explicitly illustrate that we can identify the $S_4$ permutation symmetry from every representation.
As a starting point, we note that $S_4$ is generated by the transposition $\tau=(14)$ and the cyclic permutation $\sigma=(1234)$, subject to the relations,
\beq\label{eq:S4_presentation}
\tau^2 = \sigma^4 = (\tau\sigma)^3 = \id\,.
\eeq
This is a presentation of the group $S_4$, i.e., every group generated by two elements modulo the relations in eq.~\eqref{eq:S4_presentation} is necessarily isomorphic to $S_4$.
Hence, it is enough to identify these two special symmetry transformations and their relations for every representation.

Let us start by discussing the Lee-Pomeransky and Feynman parameter representations. The Symanzik polynomials for the equal-mass banana family are 
\beq\bsp
\mathcal{U}_{\text{ban,em}} &\,=x_1x_2x_3+x_1x_2x_4+x_1x_3x_4+x_2x_3x_4\,,\\
\mathcal{F}_{\text{ban,em}}&\,=(-p^2)\,x_1x_2x_3x_4+m^2\,(x_1+x_2+x_3+x_4)\,\mathcal{U}_{\text{ban,em}}\,.
\esp\eeq
It is easy to see that both Symanzik polynomials, and thus also the Lee-Pomeransky polynomial $\cG_{\text{ban,em}} = \mathcal{F}_{\text{ban,em}}+\mathcal{U}_{\text{ban,em}}$, are invariant under any permutation of the four Feynman parameters. Hence
\beq
\mathbb{G}(\cG_{\text{ban,em}})\simeq S_4\,.
\eeq

Let us now turn to the loop-momentum representation. Using the strategy introduced in section~\ref{sec:symfeyn}, we can lift the two permutations $\tau$ and $\sigma$ to symmetry transformations in loop-momentum space. If we use the notation $\bk_{\bp} = (k_1,k_2,k_3,p)^T$, we find, 
\beq\bsp
\bk_{\bp}\mapsto \bT_{\tau}\bk_{\bp}\textrm{~~~and~~~}
\bk_{\bp} \mapsto \bT_{\sigma}\bk_{\bp}\,,
\esp\eeq
with
\beq
\bT_{\tau} = \left(
\begin{array}{cccc}
 0 & 1 & -1 & -1 \\
 0 & 1 & 0 & 0 \\
 -1 & 1 & 0 & -1 \\
 0 & 0 & 0 & 1 \\
\end{array}
\right) \textrm{~~~and~~~} 
\bT_{\sigma} = \left(
\begin{array}{cccc}
 0 & 1 & -1 & -1 \\
 -1 & 0 & 0 & 0 \\
 0 & 0 & -1 & -1 \\
 0 & 0 & 0 & 1 \\
\end{array}
\right)\,.
\eeq
It is easy to check that these matrices satisfy
\beq
\bT_{\tau}^2=\bT_{\sigma}^4 = (\bT_{\tau}\bT_{\sigma})^3 = \mathds{1}\,.
\eeq
Hence, the group of matrices generated by $\bT_{\tau}$ and $\bT_{\sigma}$ is isomorphic to $S_4$.

Finally, let us discuss the Baikov representation. If we choose the Baikov variables to be the nine propagators defined in eq.~\eqref{3looppropchoice}, then the Baikov polynomial for the maximal cut of the equal-mass banana integral is (we put $p^2=1$ without loss of generality)
\begin{align}\label{3LoopBaikov}
\mathcal{B}^{\text{m.c.}}_{\text{ban},\mathrm{em}}&=4 m^4 \left(z_5-z_6^2\right)-m^2 \left[2 z_5^2+z_5 \left(4 \left(-z_6 z_7+z_6+z_7^2\right)+2\right)+(2 z_6+1)^2\right]\notag\\
&-4 m^2 z_8^2 (z_5+2 z_6+1)+4 m^2 (z_6+1) z_8 (z_5+2 z_6+1)+4 z_9^2 \left(z_6^2-z_5\right)\\
&+2 z_9 [4 z_7 z_8 (z_5+z_6)+(z_5+2 z_6+1) (z_5-2 z_6 z_7)\notag\\
&-2 z_5 (z_6+1) z_8]+[z_7 (z_5+2 z_6-2 z_8+1)-z_5 z_8]^2\,.\notag
\end{align}
Since we are working on the maximal cut, the Baikov polynomial only depends on the 5 variables $\bs{z}=(z_5,\ldots,z_9)$. It is straightforward to check that the Baikov polynomial is invariant under the following transformations
\beq\bsp
\mathcal{B}^{\text{m.c.}}_{\text{ban},\mathrm{em}}(\bs{\mathcal{M}}_{\tau}\bz+\bs{b}_{\tau})&\,=\mathcal{B}^{\text{m.c.}}_{\text{ban},\mathrm{em}}(\bz)\,,\\
\mathcal{B}^{\text{m.c.}}_{\text{ban},\mathrm{em}}(\bs{\mathcal{M}}_{\sigma}\bz+\bs{b}_{\sigma})&\,=\mathcal{B}^{\text{m.c.}}_{\text{ban},\mathrm{em}}(\bz)\,,
\esp\eeq
with 
\begin{equation}
\begin{aligned}
  \bs{\mathcal{M}}_{\tau}=&\left(
\begin{array}{ccccc}
 0 & 0 & 2 & -2 & -2 \\
 0 & 0 & -1 & 1 & 0 \\
 0 & -1 & 0 & 1 & 0 \\
 0 & 0 & 0 & 1 & 0 \\
 -\frac{1}{2} & -1 & 0 & 0 & 0 \\
\end{array}
\right)\,,\qquad &  \bs{b}_\tau&=\left(
\begin{array}{ccccccccc}2 m^2+1\\-1\\-1\\0\\\frac{1}{2} (2 m^2-1)\end{array}\right)\,,\\
\bs{\mathcal{M}}_{\sigma}&=\left(
\begin{array}{ccccc}
 1 & 2 & 0 & 0 & 0 \\
 0 & -1 & 0 & 0 & 0 \\
 0 & 0 & 0 & -1 & 0 \\
 0 & -1 & 1 & 0 & 0 \\
 \frac{1}{2} & 1 & 0 & -1 & -1 \\
\end{array}
\right)\,, \qquad & \bs{b}_\sigma&=\left(
\begin{array}{ccccccccc}1\\-1\\0\\0\\\frac{1}{2} \end{array}\right)\,.
\end{aligned}
\end{equation}
If we define 
\beq
f_{\tau}(\bs{z}) = \bs{\mathcal{M}}_{\tau}\bs{z}+\bs{b}_{\tau}\,,\qquad 
f_{\sigma}(\bs{z}) = \bs{\mathcal{M}}_{\sigma}\bs{z}+\bs{b}_{\sigma}\,,
\eeq
then we see that
\beq
f_{\tau}^2 = f_{\sigma}^4 = (f_{\tau}f_{\sigma})^3=\id\,,
\eeq
and so $f_{\tau}$ and $f_{\sigma}$ again generate a group isomorphic to $S_4$.

To conclude, we see that, irrespective of the integral representation used to identify the symmetry transformations of the top sector of the equal-mass banana family, we find that
\beq\label{eq:ban_em_Aut}
\Aut_{G_{\text{ban}}}(\Theta_{\text{ban}})\simeq S_4\,.
\eeq

\section{Symmetry transformations within a given sector}
\label{sec:single_sector}

In the previous sections we have seen that there is a natural notion of groupoid of (twisted) symmetry transformations of a twisted cohomology group or a family of Feynman integrals. If we focus on the special case of symmetry transformations from a sector to itself, then these symmetry transformations form a group, and in our case the relevant groups are all finite. In this section we study this group using tools from the theory of finite groups.

\subsection{The symmetry group of a sector of a twisted cohomology group}
\label{subsec:symmgrouptwisted}

Let us start by discussing the twisted symmetry transformations from a sector $\Theta$ of a twisted cohomology group to itself. 
In the case relevant to Feynman integrals, they form a finite group $G=\Aut_G(\Theta,\bs{s})\simeq \mathbb{G}(\cG(\cdot,\bs{s}))$, with $\cG$ the Lee-Pomeransky polynomial of the sector $\Theta$ and $\bs{s}$ the vector of external kinematic data. 
Without loss of generality, we may assume $\Theta$ to be the top sector of the family (otherwise, simply restrict to the family defined only by the integrals from this sector).

By definition, $G$ is a subgroup of the group of twisted symmetry transformations of the associated twisted cohomology group. We then obtain an action of $G$ on $\cH := H^P(\mathbb{C}^P\setminus\Sigma,D_+,\nabla_{\omega})$. If we fix a basis $\{\varphi_k: k=1,\ldots,N\}$ of $\cH$ (with $N$ the dimension of the twisted cohomology group), then we obtain a matrix representation of $G$. More precisely, if $f\in G$, then its action on the basis is
\beq\label{eq:Dsf}
\varphi_k \mapsto f^*\varphi_k =  \big(\bD_{f}\big)_{kl}\,\varphi_l\,,
\eeq
where $\bD_{f}\in\cF^{N\times N}$ is an $N\times N$ matrix with entries in $\cF$ (note that the representation matrices may depend on the kinematic point $\bs{s}$). It satisfies
\beq\label{eq:D_rep}
\bD_{\textrm{id}} = \mathds{1} \textrm{~~~and~~~}\bD_{f_1f_2} = \bD_{f_1}\bD_{f_2}\,,
\eeq
for $f_1,f_2\in G$.
Thus, we may associate to every sector  a group of matrices that encode the twisted symmetry transformations of that sector.

Since $G$ is a finite group, we know on general grounds that we may decompose the representation $\bD$ into irreducible representations of $G$,
\beq\label{eq:irrep_dec_2}
\bD=\bigoplus_R m_R\bD^{(R)}\,,
\eeq
where the direct sum runs over all irreducible representations $R$ of $G$ and $m_{R}$ is the multiplicity of the irreducible representation $R$ in $\bD$.
Equivalently, we may write $\cH$ as a direct sum,
\beq\label{eq:irrep_dec_1}
\cH= \bigoplus_{R} \cH_R\,,\qquad \dim \cH_R = m_R\,d_R\,,
\eeq
where elements of $\cH_R$ transform in the irreducible representation $R$, and $d_R$ is the dimension of $R$.
 We can use standard group theory arguments to construct projectors onto the irreducible subspaces $\cH_R$:
\beq\label{eq:Reynolds}
\bP_{R} = \frac{d_R}{|G|} \sum_{f\in G}\chi_R(f)^*\,\bD_{f}\,,
\eeq
with $\chi_R(f)=\operatorname{Tr}(\bD^{(R)}_f)$ the character of $\bD^{(R)}$.

The group $G$ also acts on the twisted homology group, and on the dual (co-)homology groups. If we fix bases of all relevant (co-)homology groups, 
then $G$ acts on the basis of the twisted homology group via  representation matrices which we denote by $\bA_{f}\in \mathbb{C}(e^{i\pi \mu_k})^{N\times N}$ . Similarly, we obtain representation matrices $\bs{\check{D}}_f\in \cF^{N\times N}$ and $\bs{\check{A}}_f\in \mathbb{C}(e^{i\pi  \mu_k})^{N\times N}$ for the dual groups (we use the notation for the bases introduced in section~ \ref{introtwisted}):
\beq\bsp\label{eq:AAD}
\check{\varphi}_k &\,\mapsto f^*\check{\varphi}_k =  \big(\bs{\check{D}}_{f}\big)_{kl}\,\check{\varphi}_l\,,\\
\gamma_k &\,\mapsto f_*\gamma_k =  \gamma_l\,\big(\bA_{f}\big)_{lk}\,,\\
\check{\gamma}_k &\,\mapsto f_*\check{\gamma}_k =  \check{\gamma}_l\,\big(\bs{\check{A}}_{f}\big)_{lk}\,,
\esp\eeq
and the action on the twisted cohomology group is given in eq.~\eqref{eq:Dsf}.
Each of these matrices defines a representation of $G$ (cf.~eq.~\eqref{eq:D_rep}):
\beq\bsp
&\bs{\check{D}}_{\textrm{id}} = \mathds{1}\,,\qquad \bs{\check{D}}_{f_1f_2} = \bs{\check{D}}_{f_1}\bs{\check{D}}_{f_2}\,,\\
&\bA_{\textrm{id}} = \mathds{1}\,,\qquad \bA_{f_1f_2} = \bA_{f_1}\bA_{f_2}\,,\\
&\bs{\check{A}}_{\textrm{id}} = \mathds{1}\,,\qquad \bs{\check{A}}_{f_1f_2} = \bs{\check{A}}_{f_1}\bs{\check{A}}_{f_2}\,.
\esp\eeq
Without loss of generality we may pick bases aligned with the decompositions into irreducible representations, i.e., we assume that the basis vectors transform in some irreducible representation.

While these matrices define a genuine matrix representation of $G$, there is a difference in the way that $G$ acts on the homology and cohomology groups. We can see from eqs.~\eqref{eq:Dsf} and ~\eqref{eq:AAD} that $G$ acts on the cohomology group and its dual from the right, and it acts on the homology group and its dual from the left. This is a direct reflection of how pushforwards and pullbacks behave under composition:
\beq
(f_1f_2)_* = f_{1*}f_{2*}\,,\qquad (f_1f_2)^* = f_2^*f_1^*\,.
\eeq
We can pair the basis elements to form the period matrix $\bP$ and its dual $\bs{\check{P}}$, as well as the cohomology and homology intersection matrices $\bC$ and $\bH$ (see section~\ref{sec:summary}). The invariance of the pairings under twisted symmetry transformations in eqs.~\eqref{eq:invarianceP} and~\eqref{prop1} implies the relations: 
\begin{align}\label{matrixinvcond}
   \bs{{D}}_{f}\bP=\bP\bA_f\,,\quad     &\bs{\check{D}}_{f}\check{\bP}=\check{\bP}\bs{\check{A}}_f\,,\quad 
   \bC=\bs{{D}}_{f}\bC\bs{\check{D}}^T_{f}\,,\quad    \bH=\bA_f^T\bH \bs{\check{A}}_f\,.
\end{align}

We may ask how each of these representations decomposes into irreducible representations of $G$. Since the period and intersection matrices have full rank, it is easy to see that the representations on all the relevant groups are related. 
For example, the representation $\bA_f$ on the twisted homology group is equivalent to the one on the cohomology group,
\beq
\bA_f = \bP^{-1}\bD_f\bP\,.
\eeq
Instead, the representation on the dual groups are equivalent to the contragredient representation of $\bD$:
\beq\bsp
\label{similarityrep}
\bs{\check{D}}_f&\,=\bC^T\big(\bs{{D}}_f^{-1}\big)^T\big(\bC^{-1}\big)^T\,,\\
\bs{\check{A}}_f&\,=\big((\bP^{-1})^T\bH\big)^{-1}(\bD_f^{-1})^T(\bP^{-1})^T\bH\,.
\esp\eeq
 For finite groups, all irreducible representations are unitary, which implies that their contragredient representations are the same as the complex conjugate, e.g., $(\bD^{(R)-1})^T = \bD^{(R)*}$. We then see that $\bA$ admits the same decomposition into irreducible representations as $\bD$ in eq.~\eqref{eq:irrep_dec_2}, while the representations $\bs{\check{A}}$ and $\bs{\check{D}}$ on the dual groups admit the same decomposition as in eq.~\eqref{eq:irrep_dec_2}, but with $\bD^{(R)}$ replaced by $\bD^{(R)*}$. 

A standard argument from group theory implies that if we have an invariant perfect bilinear pairing between two spaces on which $G$ acts, then the pairing between any two vectors is zero, unless one of the following two conditions holds:
\begin{itemize}
    \item they transform in equivalent irreducible representations, if the group acts on the two spaces from different sides (i.e., on one space it acts from the right, and on the other from the left),
    \item they transform in irreducible representations contragredient to one another, if the  the group acts on the two spaces from the same side.
\end{itemize}
We know from section~\ref{sec:symandtwist} that the period and intersection pairings are invariant under twisted symmetry transformations, and so we immediately conclude that this result applies to our case. For example, we then have $
C_{ij} = \langle\varphi_i|\check{\varphi}_j\rangle =0$, unless $\varphi_i$ and $\check{\varphi}_j$ transform in the contragredient irreducible representation. These arguments apply in precisely the same way to the homology intersection matrix $\bH$ and the period matrices $\bP$ and $\bs{\check{P}}$. Hence, we see that in appropriately chosen bases the intersection and period matrices are block-diagonal, with the blocks corresponding to the decomposition into the irreducible representations. 

In applications, we typically deal with differential forms that depend on some external parameters $\bs{s}$ (the external kinematic data in the case of Feynman integrals), and we assume that the symmetry group $G$ is the same for a small connected subset in the parameter space (we will discuss the situation when the group changes when we change the values of the parameters in section~\ref{sec:kinematic_dependence}). In general, also the representation matrices will be functions of $\bs{s}$. However, after rotation to a basis aligned with the decompositions into irreducible representations, the representation matrices are block diagonal, and on each block we have a copy of a representation matrix $\bD^{(R)}$ corresponding to an irreducible representation $R$. The latter are defined purely group-theoretically, and are independent of $\bs{s}$. We thus conclude that we can find a basis aligned with the decomposition into irreducible representations where all representation matrices are independent of $\bs{s}$. It is easy to see from eq.~\eqref{matrixinvcond} that in such a basis the matrix $\bOmega$ appearing in the differential equation satisfied by the period matrix $\bP$ (cf.~eq.~\eqref{deqperiod}) commutes with the representation,
\beq
\bOmega\bD_f = \bD_f\bOmega\,,\qquad \textrm{for all } f\in G\,.
\eeq
Schur's lemma then immediately implies that $\bOmega$ must be block-diagonal as well (but the differential equation may still mix different subsectors).

To conclude, we see that, if the bases are chosen appropriately, all period and intersection matrices are block-diagonal, in line with the decomposition into irreducible representations. It follows that we can always study the contributions from a given irreducible representation independently from the others. For example, all linear (IBP) relations,  differential equations or twisted Riemann bilinear relations will be separately satisfied within each irreducible subspace $\cH_R$.

\subsection{The symmetry group of a sector of a family of integrals}\label{sec:decomposition}

So far we have focused on the group theory aspects related to the twisted symmetry transformations acting on the (co-)homology groups and their duals. We now focus on the group $G$ of symmetry transformations in a fixed sector of a family of integrals. 

We focus on the case of a family of Feynman integrals.
The elements of $G$ are permutations of the Feynman parameters, and $G$ is a subgroup of the group of twisted symmetry transformations of the sector. From the discussion in section~\ref{sec:twistedsymparam}, we know that it does not matter which twisted cohomology theory we study, and we focus on the Lee-Pomeransky representation for concreteness. 

Note that $G$ acts as a group of matrices. Every matrix group has the \emph{determinant representation},
\beq
D_{\det}: G\to \mathbb{C}^{\times}\,,\quad \bM \mapsto \det \bM\,.
\eeq
This representation is obviously one-dimensional, and therefore irreducible. In our case, $G$ is a group of permutations, and the determinant coincides with the signature of the permutation. We stress, however, that depending on the group, the determinant representation may coincide with the trivial representation. Indeed, if all matrices have determinant $+1$ (or, in our case, if $G$ is a subgroup of the group of permutations of positive signature), then the determinant representation is trivial.

The determinant representation plays a special role. 
By definition, there is a surjective map
\beq
\mathcal{I}_{\gamma_F}: \cH \to \cV_{\gamma_F},\qquad \varphi \mapsto \langle\varphi|\gamma_F]\,.
\eeq
It is still an open question how to precisely describe the kernel of this map, but it is known that symmetry transformations can provide elements in the kernel. To see this, recall that the integration pairing is zero unless both the cycle and cocycle transform in the same irreducible representation. By definition of $G=\Aut_{\gamma_F}(\Theta)$, $\gamma_F$ transforms in the determinant representation, and so the integration pairing is zero unless $\varphi$ transforms in the determinant representation as well, i.e., $\varphi\in\cH_G$, with $\cH_G$ the subspace of $\cH$ transforming in the determinant representation,
\beq
\cH_G:= \big\{\varphi\in \cH: f^*\varphi = \det(f)\,\varphi\,, \forall f\in G\big\}\,.
\eeq
Hence, the kernel of $\mathcal{I}_{\gamma_F}$ must (at least) contain all those cocycles that do not lie in $\cH_G$. This subspace can also be described explicitly: we know that the intersection pairing on $\cH$ is invariant, and from standard group theory arguments we know that there is a decomposition
\beq
\cH = \cH_G \oplus \cH_G^{\perp}\,,
\eeq
where $\cH_G^{\perp}$ is the orthogonal complement of $\cH_G$ for the intersection pairing.
Hence, we conclude that 
\beq
\cH_G^{\perp}\subseteq \operatorname{Ker} \mathcal{I}_{\gamma_F}\,.
\eeq
In particular, this means that the restriction of $\mathcal{I}_{\gamma_F}$ to $\cH_G$ is surjective. In other words, every element from $\cV_{\gamma_F}$ can be represented by a cocycle from $\cH_G$. 
However, while the map $\cH_G \to \cV_{\gamma_F}$ is always surjective, it is in general not injective, and so we cannot in general identify $\cH_G$ with $\cV_{\gamma_F}$ (even if we assume that the kernel is only generated by symmetry transformations). The reason is that there may still be non-trivial symmetry transformations between subsectors, so that some subsectors that are distinct in $\cH_G$ become identified in $\cV_{\gamma_F}$. Moreover, there may also be non-linear changes of variables that relate different integrals from a same family, cf.~e.g.,~refs.~\cite{Buccioni:2023okz,Becchetti:2025rrz,Coro:2025vgn}. 

Let us discuss a scenario where we can identify the family of integrals with the subspace $\cH_G$ of a twisted cohomology group. We can focus on the top sector $\Theta_{\textrm{top}} = (1,\ldots,1)$ by focusing not on the Feynman integrals themselves, but on their maximal cuts in the top sector. This is equivalent to focusing on the associated graded space  $\operatorname{Gr}^S_{\Theta_{\textrm{top}}}\cV_{\gamma_F}$ (see eq.~\eqref{eq:grV}). Let us briefly discuss how we can associate a twisted cohomology theory to the maximal cuts in the top sector. We again work with the Lee-Pomeransky representation. In this representation, the sectors are encoded in the components of the relative boundaries $D_+$, and $\partial\gamma_F\subset D_+$. We can focus on the maximal cuts by choosing a cycle that is not related to the relative boundary $D_+$. A convenient such cycle is given by the chamber bounded by the Lee-Pomeransky polynomial in the top sector (for a discussion of how to choose cycles for cuts of Feynman integrals in the parametric representation, see refs.~\cite{Britto:2024mna,Britto:2025wzt}),
\beq
\gamma_{\cG} = \big\{\bx\in\mathbb{R}^P: \cG(\bx,\bs{s}) >0\big\}\,.
\eeq
We then have a bijection
\beq
\operatorname{Gr}^S_{\Theta_{\textrm{top}}}\cV_{\gamma_F} \simeq \cV_{\gamma_{\cG}}\,.
\eeq
Since $\partial\gamma_{\cG} \subset \Sigma$, it is not a relative cycle, and so we can describe $\cV_{\gamma_{\cG}}$ using a non-relative twisted cohomology theory,
\beq\label{eq:cH_mc}
\cH^{\textrm{m.c.}} := H^P(\mathbb{C}^P\setminus\Sigma,\nabla_{\omega})\,.
\eeq
As we will see, working with this cohomology group has several practical advantages. In particular, since $D_+=D_-=\emptyset$, the sector filtrations $\operatorname{Sec}_{\Theta^+}$ and $\operatorname{Sec}_{\Theta^-}$ are both trivial, and so there are no subsectors. Hence, there are no symmetry transformations relating different subsectors. The only relevant group is the group $G$ of symmetry transformations in the top sector $\Theta_{\textrm{top}}$. It is easy to see that we have
\beq
f_*\gamma_{\cG} = \det(f)\,\gamma_{\cG}\,,\quad\forall f\in G\,,
\eeq
because by definition elements of $G$ leave the twist in the top sector invariant.
As a consequence, the integration map $\mathcal{I}_{\gamma_{\cG}} : \cH^{\textrm{m.c.}}_G\to \cV_{\gamma_{\cG}}$ is injective, again under the assumption that the only elements in the kernel come from symmetry transformations. Combined with the surjectivity, we conclude that under the same assumptions we obtain an isomorphism, and thus
\beq\label{eq:dimV_dimHG}
\dim_{\cF} \cV_{\gamma_\cG} = \dim_{\cF} \cH_G^{\textrm{m.c.}} \,.
\eeq


\subsection{Kinematic dependence of the symmetry transformations}
\label{sec:kinematic_dependence}
So far we have only discussed symmetry transformations \emph{for a fixed value of $\bs{s}$}. In this section, we will discuss what we can say about the kinematic dependence of the symmetry transformations.

The integral family $\cV_{\gamma_F}$ and the twisted cohomology group $\cH$ depend parametrically on the external kinematic data $\bs{s}$. We should then rather think of the family of integrals and the twisted cohomology groups as vector bundles over the base defined by the external kinematic data $\bs{s}$, i.e., the Mandelstam invariants and the propagator masses. The fibers are $\cV_{\gamma_F}$ and $\cH$, respectively. Similarly, the symmetry group is a group bundle, i.e., a fiber bundle whose fiber over the point $\bs{s}$ is the symmetry group of the vector space $\cV_{\gamma_F}$ over that point. Since in this subsection we are interested in understanding how the effect of the symmetries changes with $\bs{s}$, we will indicate the dependence on the kinematic point explicitly, i.e., we will write $\cH_{\bs{s}}$, $\cV_{\gamma_F,\bs{s}}$, $G_{\bs{s}}$ and $\bD_{\bs{s}}$ for the respective fibers over the point $\bs{s}$.

It is easy to see that for most kinematic points, i.e., in the bulk of the kinematic space, the fiber will be the trivial group (because, e.g., the propagator masses must be compatible with the permutation of the propagators for a non-trivial symmetry to exist). If the group $G$ is trivial, we have $(\cH_{G})_{\bs{s}}=\cH_{\bs{s}}$ in the bulk. Note that even in that case the kernel of the integration map $\mathcal{I}_{\gamma_F}$ may be non-trivial, as the kernel may contain relations other than those generated by symmetry transformations. For simplicity, and for the sake of the discussion here, we assume for now that $\operatorname{Ker}\mathcal{I}_{\gamma_F}$ is only generated by symmetry transformations. In such a scenario, we have $\operatorname{Ker}\mathcal{I}_{\gamma_F} = 0$ in the bulk, and thus there is a bijection between the twisted cohomology group $\cH_{\bs{s}}$ and the family of integrals $\cV_{\gamma_F,\bs{s}}$.

At special points $\bs{s}=\bs{s}_0$ in the space of external kinematic data, the group fibers may be non-trivial. At such loci we may have $(\cH_G)_{\bs{s}_0}\neq \cH_{\bs{s}_0}$, and the number of master integrals will be reduced compared to the bulk. As we will discuss now, in a certain sense, the master integrals in the bulk  remember the symmetries at these special loci where the symmetry group is non-trivial. 

Consider another kinematic point $\bs{s}_1$ where the symmetry group $G_{\bs{s}_1}\subseteq G_{\bs{s}_0}$ is non-trivial. An important question is, what is the dimension of the invariant subspace $(\cH_{G})_{\bs{s}_1}$, because it encodes the number master integrals at this point (at least under our assumption that the kernel of the integration map is only generated by symmetry transformations). As we will argue now, this information is already encoded purely group-theoretically in the decomposition into irreducible representations at the point of enhanced symmetry $\bs{s}_0$ in eqs.~\eqref{eq:irrep_dec_1} and~\eqref{eq:irrep_dec_2}. 

Let us make the assumption (which is typically satisfied in applications) that the dimension of $\cH_{\bs{s}_1}$ is equal to the dimension of $\cH_{\bs{s}_0}$. We may then decompose each irreducible representation $\bD^{(R)}$ of $G_{\bs{s}_0}$ into irreducible representations of its subgroup $G_{\bs{s}_1}$,
\beq
\bD^{(R)} = \bigoplus_{R'}\mu_{RR'}\,\bs{\widetilde{D}}^{(R')}\,,
\eeq
where the sum runs over all irreducible representations $\bs{\widetilde{D}}^{(R')}$ of $G_{\bs{s}_1}$. The coefficients $\mu_{RR'}$ appearing in this decompositions are purely group-theoretical, and known as the \emph{branching coefficients} of $G_{\bs{s}_0}$ with respect to its subgroup $G_{\bs{s}_1}$. Since both $G_{\bs{s}_1}$ and $G_{\bs{s}_0}$ are finite groups, the branching coefficients are typically known in the literature, or they can easily be computed from the character tables of the groups using standard tools from group theory. It is now easy to see that we have
\beq
\dim(\cH_G)_{\bs{s}_1} =\sum_{R}m_{R,\bs{s}_0}\,\mu_{R\det}\,.
\eeq
We thus see that the number of master integrals at a point $\bs{s}_1$ of reduced symmetry is completely determined by  the decomposition into irreducible representations at a point $\bs{s}_0$ of larger symmetry, plus purely group-theoretical information coming from the relationship between the group $G_{\bs{s}_0}$ and its subgroup $G_{\bs{s}_1}$. In other words, if we understand the structure of the symmetry group at a point of enhanced symmetry $\bs{s}_0$ and how the group action on $\cH_{\bs{s}_0}$ decomposes into irreducible representations, then we can infer the dimension of $(\cH_G)_{\bs{s}_1}$, and thus in many cases also the number of master integrals, at all points with reduced symmetry.

\subsection{Symmetry transformations and canonical bases}
\label{eq:canonical_basis}

We now discuss another, and in some sense even stronger, incarnation of the fact that the master integrals in the bulk seem to remember the enlarged symmetry at a point $\bs{s}_0$. 
We know that there is (conjecturally) a distinguished basis of master integrals, namely a so-called canonical basis (see section~\ref{sec:summary_Feynman_integrals} for a review). We assume that a rotation to a canonical basis can be constructed for our family.
Throughout this section, we will follow the notations and conventions from section~\ref{sec:summary_Feynman_integrals}: $\bI(\bs{s},\eps)$ denotes a vector of master integrals for $
\cV_{\gamma_{\cG}}$ (not necessarily a canonical basis), and the rotation to a canonical basis $\bJ(\bs{s},\eps)$ is achieved via a matrix $\bU(\bs{s},\eps)$ as in eq.~\eqref{conventiontrafo}.

Assume that there is a point $\bs{s}_0$ with an enlarged symmetry group.
Our goal is to study how a canonical basis encodes the symmetries, both at the point $\bs{s}_0$ and away from it. As we have seen in the previous subsection, an important tool to achieve this is the decomposition of $\cH^{\textrm{m.c.}}$ into irreducible representations as in eq.~\eqref{eq:irrep_dec_1}. However, there is no reason to expect a priori that we may pick the canonical master integrals such that they all lie in an irreducible representation of $G_{\bs{s}_0}$.\footnote{We note that we had defined the scalars defining the vector space $\cH^{\textrm{m.c.}}$ to be the field $\cF$ of rational functions. It is well known that the rotation $\bU(\bs{s},\eps)$ to a canonical basis may involve transcendental functions. Since our conclusions do not rely on the specific choice of scalars, in the following we assume that the field of scalars has been appropriately extended to include these transcendental functions.}

\begin{lemma}\label{eq:canonical_symmetry}
    There is a basis of $\cH^{\mathrm{m.c.}}_{\bs{s}}$ that is at the same time canonical and such that at the point $\bs{s}_0$ all elements transform in irreducible representations of $G_{\bs{s}_0}$. Moreover, when working in that basis the representation matrices $\bD_{\bs{s}_0,f}$ are kinematics-independent.
\end{lemma}

\begin{proof}
Let $\bs{\varphi}_{\bs{s}} = \big(\varphi_{\bs{s},1},\ldots,\varphi_{\bs{s},N_{\mathrm{top}}}\big)^T$ be a basis of $\cH^{\mathrm{m.c.}}_{\bs{s}}$. We can rotate it to a canonical basis $\bs{\widetilde{\psi}}_{\bs{s}}$ via
\beq
\bs{\widetilde{\psi}}_{\bs{s}} = \bU(\bs{s},\eps)\bs{\varphi}_{\bs{s}}\,.
\eeq
Let $\bs{\widetilde{J}}(\bs{s},\eps) = \langle\bs{\widetilde{\psi}}_{\bs{s}}|\gamma_{\cG}]$ be the vector of canonical master integrals in the bulk. We keep the number of elements in this vector fixed as we vary $\bs{s}$, including at the point $\bs{s}_0$ of enlarged symmetry. It satisfies a differential equation of the form
\beq
\rd \bs{\widetilde{J}}(\bs{s},\eps)=\eps\,\bs{\widetilde{\Omega}}(\bs{s})\,\bs{\widetilde{J}}(\bs{s},\eps)\,.
\eeq
Let $\bs{\widetilde{P}}(\bs{s},\eps)$ be a fundamental solution matrix of this differential equation.

We now focus on the enhanced symmetry at the point $\bs{s}_0$.
Let $\bD_{\bs{s}_0,f}$ be the matrix representation of $f\in G_{\bs{s}_0}$ in the basis $\bs{\varphi}_{\bs{s}_0}$. In our canonical basis, the group acts via the matrix representation $\bs{\widetilde{D}}_{\bs{s}_0,f} := \bU(\bs{s}_0,\eps)\bD_{\bs{s}_0,f}\bU(\bs{s}_0,\eps)^{-1}$. Since $f$ is a twisted symmetry transformation for $\cH_{\bs{s}_0}^{\mathrm{m.c.}}$, its action on the fundamental solution matrix produces another fundamental solution matrix. This implies that there are  matrices $\bs{\widetilde{A}}_{f}(\eps)$ constant in $\bs{s}$ such that for all $f\in G_{\bs{s}_0}$,
\beq\label{eq:canonicalsym}
\bs{\widetilde{D}}_{\bs{s}_0,f}\bs{\widetilde{P}}(\bs{s}_0,\eps)=\bs{\widetilde{P}}(\bs{s}_0,\eps)\bs{\widetilde{A}}_{f}(\eps)\,.
\eeq
We now argue that this last equation implies that $\bs{\widetilde{D}}_{\bs{s}_0,f}$ is independent of the kinematic point $\bs{s}_0$. First, we note that in applications the point $\bs{s}_0$ where the symmetry is enlarged is not an isolated point, but there is typically a neighborhood $U$ of $\bs{s}_0$ such that $G_{\bs{s}}\simeq G_{\bs{s}_0}$, for all $\bs{s}\in U$ (because otherwise the integral would have no kinematic dependence at all). Hence, it makes sense to talk about the functional dependence of all quantities on $\bs{s}_0$, and we may consider this point as being a generic point with this enlarged symmetry. It then follows from an argument similar to the one presented in ref.~\cite{Duhr:2024xsy} that  $\widetilde{\bD}_{\bs{s}_0,f}=\sum_k \eps^k \widetilde{\bD}_{f}^{(k)}$ with $\widetilde{\bD}_{f}^{(k)}$ constant matrices with entries in $\mathbb{C}$. 
This is a consequence of the linear independence of the iterated integrals in a canonical basis.\footnote{In the notation of ref.~\cite{Duhr:2024xsy}, we expand both sides of eq.~\eqref{eq:canonicalsym} in $\eps$. Then each term on the right-hand side admits an expansion $\sum_{w\in \mathcal{B}_\mathbb{A}}c_w J(w)$ with $c_{w}\in \mathbb{C}$ and $J(w)\in V_\mathbb{A}$. Meanwhile, terms on the left-hand side are of the form $\sum_{w\in \mathcal{B}_\mathbb{A}}\eta_w J(w)$ with $\eta \in \mathcal{F}_\mathbb{C}$. Because of the linear independence of the iterated integrals $J(w)$ it follows from $\sum_{w\in \mathcal{B}_\mathbb{A}}(\eta_w-c_w) J(w)$ that $\eta_w=c_w$ for all $w$. Hence all entries of the matrix $\tilde{\bD}_f^{(k)}$ lie in $\mathbb{C}$.} 
Hence $\widetilde{\bD}_{\bs{s}_0,f}$ is independent of the parameter $\bs{s}$: 
\beq\label{eq:dD}
\rd\widetilde{\bD}_{\bs{s}_0,f}=0\,. 
\eeq
Similarly to ref.~\cite{Duhr:2024xsy}, there exists a value $\eps_0\in \mathbb{Q}$ such that $\bs{\widetilde{D}}_{\bs{s}_0,f}(\eps_0)$ is of full rank and a representation of the symmetry group $G_{\bs{s}_0}$.
Accordingly, there exists a constant and $\eps$-independent transformation $\bT$ such that the entries of $\bP_{\!c}(\bx,\eps_0)=\bT \bs{\widetilde{P}}(\bx,\eps_0)$ transform in irreducible representations. Then, the differential equation and the representation transform as $\bs{\Omega}_c(\bs{s})=\bT\bs{\widetilde{\Omega}}_c(\bs{s})\bT^{-1}$, $\bD_{f,c}=\bT\bs{\widetilde{D}}_f(\eps_0)\bT^{-1}$, respectively. Finally, from the constancy of $\bD_{f,c}$, it follows that
\begin{align}
    [\bD_{f,c},\bs{\Omega}_c(\bs{s})]=0\,.
\end{align}
Hence $\bs{\Omega}_c(\bs{s})$ admits a block-diagonal structure, where the blocks correspond to irreducible representations.
\end{proof}

Lemma~\ref{eq:canonical_symmetry} implies that we may always pick a canonical basis (if it exists) in the bulk (where the symmetry group is trivial) such that at a point $\bs{s}_0$ the elements of this canonical basis transform in irreducible representations of $G_{\bs{s}_0}$. In other words, the canonical basis in the bulk remembers in some sense that there is a point $\bs{s}_0$ of enlarged symmetry. 
This effect becomes even stronger when looking at the intersection matrix defined from the canonical basis.

Using the notations of the previous proof, let $\bC(\bs{s},\eps)$ denote the intersection matrix defined using the original basis $\bs{\varphi}_{\bs{s}}$ of $\cH_{\bs{s}}^{\textrm{m.c.}}$, and let $\bU(\bs{s},\eps)$ be the rotation to the canonical basis aligned with the irreducible representations of $G_{\bs{s}_0}$.
In ref.~\cite{Duhr:2024xsy} it was shown that the intersection matrix in a canonical basis takes a particularly simple form, 
\begin{align}\label{eq:canonical_Delta}
\bU(\bs{s},\eps)\bC(\bs{s},\eps)\bU(\bs{s},-\eps)^T=f(\eps)\bDelta\,,  
\end{align}
where $f(\eps)$ is a rational function of $\eps$ only and $\bDelta$ is a constant matrix with entries in $\mathbb{C}$, called the \emph{canonical intersection matrix}. In particular, we see that the right-hand side of eq.~\eqref{eq:canonical_Delta} is independent of the kinematics. This very special structure of the intersection matrix is a direct consequence of working with a canonical basis. If in addition the kinematic space has a point $\bs{s}_0$ where the symmetry is enlarged, then we can say even more about the structure of the canonical intersection matrix:
\begin{proposition}
    Assume that a family of integrals has a point $\bs{s}_0$ of enlarged symmetry, and that the hypersurface arrangements of the corresponding twist change continuously between a generic point $\bs{s}$ in the bulk and $\bs{s}_0$. If the canonical basis in the bulk is chosen such that it reflects the decomposition into irreducible representations of $G_{\bs{s}_0}$, then the canonical intersection matrix in the bulk is block-diagonal, with the blocks corresponding to the irreducible representations of $G_{\bs{s}_0}$. 
\end{proposition}
\begin{proof}
    We assume that we can connect a point $\bs{s}$ in the bulk to a point $\bs{s}_0$ of enlarged symmetry in a continuous fashion. Using the notations introduced earlier, the intersection matrix in the canonical basis $\bU(\bs{s}_0,\eps)\bC(\bs{s}_0,\eps)\bU(\bs{s}_0,-\eps)^T$ at the point $\bs{s}_0$ is block-diagonal, because the intersection pairing is zero unless the two cocycles transform in contragredient irreducible representations. However, if the basis is canonical, the intersection matrix is independent of the kinematics, and we have
    \begin{align}
 \bU(\bs{s},\eps)\bC(\bs{s},\eps)\bU(\bs{s},-\eps)^T=\bU(\bs{s}_0,\eps)\bC(\bs{s}_0,\eps)\bU(\bs{s}_0,-\eps)^T=f(\eps)\bDelta\,.   
\end{align}
Hence, $\bDelta$ is also block-diagonal.
\end{proof}

The results of this section indicate that there is a strong connection between canonical bases and the group of symmetry transformations of a sector. In particular, the structure of the irreducible representations at a point of enlarged symmetry is encoded into a canonical basis and the associated canonical intersection matrix in the bulk. We find it remarkable that the canonical bases in the bulk remember the enlarged symmetry at the point $\bs{s}_0$.


\subsection{Example: Appell \texorpdfstring{$F_2$}{F2} function}\label{ex2F}
Let us illustrate these ideas on the example of a two-variable hypergeometric function, namely the Appell $F_2$ function, which can be defined by the series:
\beq
F_2(a;b_1,b_2;c_1,c_2;y_1,y_2) = \sum_{m,n=0}^{\infty}\frac{(a)_{m+n}(b_1)_m(b_2)_n}{(c_1)_m(c_2)_n}\,\frac{y_1^m}{m!}\,\frac{y_2^n}{n!}\,.
\eeq
Hypergeometric functions are the prime examples of applications of twisted cohomology theory in pure mathematics (cf.~e.g.,~refs.~\cite{aomoto_theory_2011}), and they admit integral representations as twisted periods. The twisted cohomology group relevant to the Appell $F_2$ function was studied in detail in ref.~\cite{goto2013} (see also ref.~\cite{Abreu_2020} for a discussion in a physics context).
For simplicity and concreteness, and in order to have a context where we can easily identify a non-trivial groupoid of twisted symmetries, we focus on the special instances where ($i=1,2$)
\beq
a= \frac{1}{2}+\eps\,,\quad b_i = \frac{1}{2}-\eps\,,\qquad c_i = 1-2\eps\,. 
\eeq
We consider the non-relative twisted cohomology group defined by the twist
\begin{align}\label{twistF2}
\Psi_{F_2}&= x_1^{-\frac{1}{2}-\eps } x_2^{ -\frac{1}{2}-\eps}  (1-x_1)^{ -\frac{1}{2}-\eps}(1-x_2)^{ -\frac{1}{2}-\eps}(1-x_1 y_1-x_2 y_2)^{ -\frac{1}{2}-\eps}\,\\
&=:H_1^{ -\frac{1}{2}-\eps} H_2^{ -\frac{1}{2}-\eps} H_3^{ -\frac{1}{2}-\eps}H_4^{ -\frac{1}{2}-\eps}H_5^{ -\frac{1}{2}-\eps}\notag\,,
\end{align}
where $H_i$ are the (linear) polynomials defining individual hyperplanes in $\mathbb{C}^2$ with coordinates $\bx=(x_1,x_2)$, and $\bs{s}=(y_1,y_2)$ are external parameters. The twisted variety is the union of hyperplanes (see figure~\ref{boundchambf2})
\beq
\Sigma_{F_2} = \{(x_1,x_2)\in\mathbb{C}^2 : H_1\cdots H_5=0\}\,,
\eeq
and we consider the twisted cohomology group $\cH = H^2(\mathbb{C}^2\setminus \Sigma_{F_2}, \nabla_{F_2})$ with $\nabla_{F_2} = \rd + \rd\!\log\Psi_{F_2}$. A basis for $\cH$ is
\begin{align}\label{f2basis}
\bs{\varphi}_{F_2}&=\left(\frac{\rd x_1 \wedge \rd x_2 }{x_2},\frac{\rd x_1 \wedge \rd x_2}{x_1},-\frac{2\rd x_1 \wedge \rd x_2}{(1-x_1) (1-x_2)},\rd x_1 \wedge \rd x_2\right)\,.
\end{align}
A basis for the twisted homology group defined by the twist in eq.~\eqref{twistF2} is given by (regularized) chambers between regulated boundaries, (cf. refs.~\cite{aomoto_theory_2011,yoshida_hypergeometric_1997,kita_intersection_1994-2,goto2014intersectionnumberstwistedperiod,Duhr:2023bku,Bhardwaj:2023vvm} for details).
Here these read
\begin{align}\label{hombasisf2}
\bs{\gamma}_{F_2}=\big(\gamma_{125},\gamma_{235},\gamma_{345},\gamma_{145}\big)\,,
\end{align}
where $\gamma_{ijk}$ is the chamber bounded by the hyperplanes $H_i$, $H_j$ and $H_k$, as depicted in figure~\ref{boundchambf2}. Our branch choice is the same as in ref.~\cite{goto2013} for parameters $0<y_1,y_2<1$.

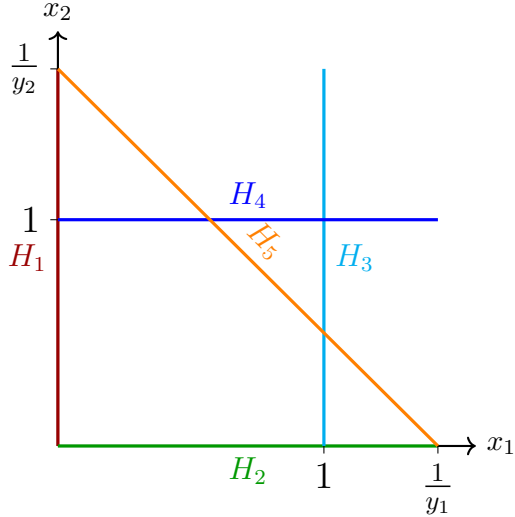
\begin{figure}[!th]
\centering
\begin{tikzpicture}
  \draw[->, thick] (0,0) -- (5.5,0) node[right] {$x_1$};
  \draw[->, thick] (0,0) -- (0,5.5) node[above] {$x_2$};

  \draw[very thick, red!60!black] (0,0) -- (0,5) node[midway, left, red!60!black] {\large $H_1$};
  \draw[very thick, green!60!black] (0,0) -- (5,0) node[midway, below, green!60!black] {\large $H_2$};
  \draw[very thick, cyan] (3.5,0) -- (3.5,5) node[midway, right, cyan] {\large $H_3$};
  \draw[very thick, blue] (0,3) -- (5,3) node[midway, above, blue] {\large $H_4$};
  \draw[very thick, orange] (0,5) -- (5,0) node[midway, above, rotate=-45, orange] {\large $H_5$};

  \draw (0,3) -- (-0.1,3) node[left] {\Large $1$};
  \draw (0,5) -- (-0.1,5) node[left] {\Large $\frac{1}{y_2}$};
  \draw (3.5,0) -- (3.5,-0.1) node[below] {\Large $1$};
  \draw (5,0) -- (5,-0.1) node[below] {\Large $\frac{1}{y_1}$};
\end{tikzpicture}

 \caption{Hyperplane arrangement defined by the twist in eq.~\eqref{twistF2}. 
}
\label{boundchambf2}
\end{figure}

If we pair elements from $\bs{\varphi}_{F_2}$ with elements from $\bs{\gamma}_{F_2}$ in all possible ways by integration, we obtain the period matrix $\bP_{F_2}$. All entries of the period matrix can be expressed in terms of Appell $F_2$ functions (see appendix~\ref{appF1} for explicit expressions). We can also fix bases for the dual homology and cohomology groups. In particular, if we pick the following basis of the dual cohomology group,
\begin{align}
\bs{\check{\varphi}}_{F_2}&=\frac{1}{(1-x_1) x_1 (1-x_2) x_2 (1-x_1 y_1-x_2 y_2)}\bs{\varphi}_{F_2}\,,
\end{align}
then we obtain a self-dual basis in the sense of ref.~\cite{Duhr:2024xsy} (see also eq.~\eqref{eq:self-duality}). From now on we work with this self-dual basis, and we assume that the dual period matrix $\bs{\check{P}}_{F_2}$ and $\bP_{F_2}$ are simply related by changing the sign of $\eps$.

Let us now discuss the twisted symmetry transformations. Since we are dealing with a non-relative twisted cohomology theory and all the exponents of the linear forms $H_i$ in the twist are non-integer, the sector filtration is trivial, and there is a single sector. Hence, the groupoid reduces to a group (because a groupoid with a single object is by definition a group). The precise structure of the group depends on the value of the external parameters $\bs{s}=(y_1,y_2)$:\footnote{There are other regions with non-trivial symmetry group where either $y_1$ or $y_2$ vanish. In those regions the twofold integrals decouple into two onefold integrals.}
\begin{itemize}
    \item In the bulk, where $(y_1,y_2)$ take generic values, the group of twisted symmetry transformations is trivial.
    \item On the locus $y_1=y_2$, the group of twisted symmetry transformations is isomorphic to $\mathbb{Z}_2=\{\id,\sigma\}$, where $\sigma$ exchanges $x_1$ and $x_2$, 
    \beq
    \sigma(x_1,x_2) = (x_2,x_1)\,.
    \eeq
        \item On the locus $y_1=-y_2$, the group of twisted symmetry transformations is isomorphic to $\mathbb{Z}_2=\{\id,\bar{\sigma}\}$, with
    \beq
    \bar{\sigma}(x_1,x_2) = (1-x_2,1-x_1)\,.
    \eeq
\end{itemize}
In the following we will only discuss the case $y_1=y_2=y$ in detail. The discussion of the case $y_1=-y_2$ is similar.

Let us start by discussing the representation of $G = \{\id,\sigma\}$ on $\cH_y := \cH_{(y,y)}$. We only need to give the representation matrix of the generator $\sigma$:
\beq\label{eq:D_F2}
\bD_{y,\sigma} = \left(
\begin{array}{cccc}
 0 & -1 & 0 & 0 \\
 -1 & 0 & 0 & 0 \\
 0 & 0 & -1 & 0 \\
 0 & 0 & 0 & -1 \\
\end{array}
\right)\,.
\eeq
It acts on the basis vector $\bs{\varphi}_{F_2}$ via
\beq
\sigma^*\bs{\varphi}_{F_2} = \bD_{y,\sigma}\bs{\varphi}_{F_2}\,.
\eeq

Let us now study the decomposition of the representation in eq.~\eqref{eq:D_F2} into irreducible representations. The group $\mathbb{Z}_2$ has two irreducible representations, the trivial representation (which sends $\sigma$ to 1) and the sign representation (with $\sign(\sigma)=-1$). The sign representation agrees with the determinant representation. An easy computation shows that 
\beq\label{eq:irreps_F2}
\cH_y =   (\cH_G)_y\oplus (\cH_{\bs{1}})_y\,,
\eeq
with 
\beq\bsp
(\cH_{G})_y &\,= \{\varphi\in\cH_y: \sigma^*\varphi = -\varphi\}\,,\\
(\cH_{\bs{1}})_y &\,= \{\varphi\in\cH_y: \sigma^*\varphi = \varphi\}\,,
\esp\eeq
and we have
\beq\label{eq:F2_dims}
\dim(\cH_{G})_y=3 \textrm{~~~and~~~} \dim (\cH_{\bs{1}})_y = 1 \,.
\eeq
We can rotate the basis such that it aligns with the decomposition into irreducible representations,
\begin{align}\label{symF2}
\bs{\varphi}_{F_2}^{\scriptscriptstyle \mathrm{irrep}}=\bT^{\scriptscriptstyle \mathrm{c}}_{\scriptscriptstyle F_2}\,\bs{\varphi}_{F_2}\,,\text{~~~~with~~~~} \bT^{\scriptscriptstyle \mathrm{c}}_{\scriptscriptstyle F_2}=\left(
\begin{array}{cccc}
 \frac{1}{2} & \frac{1}{2} & 0 & 0 \\
 0 & 0 & 1 & 0 \\
 0 & 0 & 0 & 1 \\
 \frac{1}{2} & -\frac{1}{2} & 0 & 0 \\
\end{array}
\right)\,.
\end{align}
The first three elements of $\bs{\varphi}_{\scriptscriptstyle F_2}^{\scriptscriptstyle \mathrm{irrep}}$ provide a basis for $(\cH_{G})_y$  and the last one generates $(\cH_{\bs{1}})_y$.

We may similarly rotate the basis in the homology group and the dual homology and cohomology groups to be aligned with the decomposition into irreducible representations. Let us assume that we have performed this rotation, and let us evaluate the pairings between the different groups in these new bases. For the cohomology intersection matrix, we find
\begin{equation}
\bC_{F_2}^{\scriptscriptstyle \mathrm{irrep}}=\left(\begin{array}{cc}
\bC_{F_2}^{({G})}&0 \\
0&\bC_{F_2}^{(\bs{1})}
\end{array}
\right)\,, 
\end{equation}
with 
\begin{align}
\bC_{F_2}^{(\bs{1})}&=-\frac{2}{(y-1) \left(4 \eps ^2-1\right)}\,,\\
\bC_{F_2}^{({G})}&=\left(
\begin{array}{ccc}
 \frac{2-4 y}{(y-1) \left(4 \eps ^2-1\right)} & \frac{8 (-6 y \eps +y+4 \eps )}{(y-1) (2 y-1) (1-2 \eps )^2 (2 \eps +1)} & 0 \notag\\
 -\frac{8 (6 y \eps +y-4 \eps )}{\left(2 y^2-3 y+1\right) (2 \eps -1) (2 \eps +1)^2} & \frac{16 (3 y-1) y^2+64 (y (y (37 y-63)+28)-4) \eps ^2}{(y-1) (2 y-1)^3 \left(1-4 \eps ^2\right)^2} & -\frac{8}{(2 y-1) (2 \eps +1)^2} \\
 0 & -\frac{8}{(2 y-1) (1-2 \eps )^2} & 0 \\
\end{array}
\right)\,.\notag
\end{align}
We see that the cohomology intersection matrix is block-diagonal, with the block representing the decomposition into irreducible representations in eq.~\eqref{eq:irreps_F2}. This structure is a direct consequence of the invariance of the cohomology intersection pairing under twisted symmetry transformations. Similarly, we find for the homology intersection matrix computed in bases aligned with the decomposition into irreducible representations,
\begin{equation}
\bH_{F_2}^{\scriptscriptstyle \mathrm{irrep}}=\left(\begin{array}{cc}
\bH_{F_2}^{(G)}&0 \\
0&\bH_{F_2}^{(\bs{1})}
\end{array}
\right)\, ,
\end{equation}
with
\begin{align}
\bH_{F_2}^{(\bs{1})}&=\frac{1}{8} \left(4-3 \,\mathrm{s}^2\right) \,,\\
\bH_{F_2}^{(G)}&=\frac{1}{16}\left(
\begin{array}{ccc}
 -\,\mathrm{c}^2-12 \,\mathrm{s}^2+25 & -\,\mathrm{c}^2-4 \,\mathrm{s}^2+15 & 9-\,\mathrm{c}^2 \\
 -\,\mathrm{c}^2-4 \,\mathrm{s}^2+15 & -\,\mathrm{c}^2-6 \,\mathrm{s}^2+17 & -\,\mathrm{c}^2-4 \,\mathrm{s}^2+15 \\
 9-\,\mathrm{c}^2 & -\,\mathrm{c}^2-4 \,\mathrm{s}^2+15 & -\,\mathrm{c}^2-12 \,\mathrm{s}^2+25 \\
\end{array}
\right)\,,\notag
\end{align}
where we introduced the abbreviations $\mathrm{s}:=\sec(\pi \eps)$ and  $\mathrm{c}:=\csc(\pi \eps)$. Also the period matrix becomes block-diagonal\footnote{The same holds true for the dual period matrix, because in a self-dual basis the period matrix and its dual only differ by the sign of $\eps$.}
\begin{align}
\bP^{\scriptscriptstyle \mathrm{irrep}}_{\scriptscriptstyle F_2}
=\left(
    \begin{array}{cccc}
    \cellcolor{ourlightblue}\,  \textcolor{ourblue}{\bullet} & 0   &  0 & 0 \\ 
        0& \cellcolor{ourlightblue}\,  \textcolor{ourblue}{\bullet}   & \cellcolor{ourlightblue}\,  \textcolor{ourblue}{\bullet} \, \, &     \cellcolor{ourlightblue}\,  \textcolor{ourblue}{\bullet} \\ 0& \cellcolor{ourlightblue}\,  \textcolor{ourblue}{\bullet}   & \cellcolor{ourlightblue}\,  \textcolor{ourblue}{\bullet} \, \, &     \cellcolor{ourlightblue}\,  \textcolor{ourblue}{\bullet} \\
0 &    \cellcolor{ourlightblue}\,  \textcolor{ourblue}{\bullet}&    \cellcolor{ourlightblue}\,  \textcolor{ourblue}{\bullet}& \cellcolor{ourlightblue}\,  \textcolor{ourblue}{\bullet} \, \,
    \end{array}\right)\,.
\end{align}
The explicit expressions are lengthy and not recorded here. Instead, we give the differential equation matrix, which also takes a block-diagonal form:
\begin{equation}
\rd \bP^{\scriptscriptstyle \mathrm{irrep}}_{\scriptscriptstyle F_2} = \bOmega_{F_2}^{\scriptscriptstyle \mathrm{irrep}}\bP^{\scriptscriptstyle \mathrm{irrep}}_{\scriptscriptstyle F_2}\,,\qquad \bOmega_{F_2}^{\scriptscriptstyle \mathrm{irrep}}=\left(\begin{array}{cc}
\bOmega_{F_2}^{(G)}&0 \\
0&\bOmega_{F_2}^{(\bs{1})}
\end{array}
\right)\,, 
\end{equation}
with
\begin{align}
\bOmega_{F_2}^{(\bs{1})}&=-\frac{\rd y\, (2 y \eps +y+4 \eps )}{2 (y-1) y}\,,\\
\bOmega_{F_2}^{(G)}&=\left(
\begin{array}{ccc}
 \frac{\rd y\, \left(8 y^2 \eps -y (22 \eps +3)+4 \eps \right)}{2 y \left(2 y^2-3 y+1\right)} & \frac{\rd y\, (2 y-1) (2 \eps +1)}{4 (y-1) y} & \frac{\rd y\, (10 \eps +1) (6 y \eps +y+2 \eps +1)}{2 (y-1) (2 y-1) (2 \eps +1)} \\
 -\frac{2 \rd y\, (y (4 \eps +2)+6 \eps +1)}{(1-2 y)^2 (y-1)} & -\frac{\rd y\, (y-2) (2 \eps +1)}{(y-1) (2 y-1)} & \frac{4 \rd y\, y \left(40 \eps ^2+14 \eps +1\right)}{(1-2 y)^2 (y-1) (2 \eps +1)} \\
 -\frac{2 \rd y\, \eps +\rd y\,}{y-2 y^2} & 0 & \frac{\rd y\, (2 y \eps +y+4 \eps )}{y-2 y^2} \\
\end{array}
\right)\,.\notag
\end{align}

So far we have only focused on the action of the the group of twisted symmetry transformations on the homology and cohomology groups and their duals. Let us briefly also discuss what happens if we consider a family of integrals. We consider a cycle $\gamma$ that is an (integer) linear combination of real bounded chambers. In particular, we can consider a linear combination of the basis elements from eq.~\eqref{hombasisf2}. The most general linear combination that transforms in the determinant representation of $G$ is
\beq
\gamma = m\,\gamma_{125} + n\,\gamma_{345} + l\,(\gamma_{145} + \gamma_{235})\,.
\eeq
Note that the space of such cycles is three-dimensional. This could have been anticipated, because on the cohomology side, the dimension of the irreducible subspace transforming in the determinant representation has dimension three, cf.~eq.~\eqref{eq:F2_dims}.
Let us now pick such a cycle $\gamma$, and let us consider the family of integrals $\cV_{\gamma,(y_1,y_2)}$. In the bulk, where $(y_1,y_2)$ take generic values, we have $\dim \cV_{\gamma,(y_1,y_2)} = \cH_{(y_1,y_2)}=4$. If $y_1=y_2$, however, since $\sigma^*\gamma=-\gamma$, the number of master integrals is reduced compared to the bulk, and we have (with $\cV_{\gamma,y} := \cV_{\gamma,(y,y)}$):
\beq
\dim \cV_{\gamma,y} = \dim (\cH_G)_y = 3\,.
\eeq


\subsection{Example: 3-loop banana integrals}
As a second example, we focus on the three-loop banana integral defined in eq.~\eqref{eq:ban_int_def}.
We have already seen in eq.~\eqref{eq:ban_em_Aut} that the symmetry group in the equal-mass case is 
$\Aut_G(\Theta_{\text{ban}})\simeq S_4$. Without loss of generality, we set $p^2=1$ throughout this section.

\paragraph{The representation in the equal-mass case.}
We first discuss how we can decompose the action of $G$  into irreducible 
representations for the equal-mass case. 
Let $\cH_{\mathrm{em}}$ denote the cohomology group associated with the maximal cut in the equal-mass case. It is eleven-dimensional. 
We pick a basis $\bs{\varphi}_{\mathrm{em}}$ of twisted cocycles of  $\cH_{\mathrm{em}}$ corresponding to the choice of master integrals in the unequal-mass case in eq.~\eqref{bananbasis}.
The group $G$ acts on this basis. Then the generators are represented by
\begin{align}  \tau^*\bs{\varphi}_{\mathrm{em}}=\bD^{\mathrm{em}}_{\tau} \,\bs{\varphi}_{\mathrm{em}} \text{~~~~and~~~~}\sigma^*\bs{\varphi}_{\mathrm{em}}=\bD^{\mathrm{em}}_{\sigma} \,\bs{\varphi}_{\mathrm{em}}\,,
\end{align}
where $\bD^{\mathrm{em}}_{\tau}$ and  $\bD^{\mathrm{em}}_{\sigma}$ are the representation matrices of the two generators of $\text{Aut}_G(\Theta_\text{ban})$ introduced in section~\ref{sec3loopbanana}:
\begin{align}
\bD^{\mathrm{em}}_{\tau}&=\left(
\begin{array}{ccccccccccc}
 -1 & 0 & 0 & 0 & 0 & 0 & 0 & 0 & 0 & 0 & 0 \\
 0 & 0 & 0 & 0 & -1 & 0 & 0 & 0 & 0 & 0 & 0 \\
 0 & 0 & -1 & 0 & 0 & 0 & 0 & 0 & 0 & 0 & 0 \\
 0 & 0 & 0 & -1 & 0 & 0 & 0 & 0 & 0 & 0 & 0 \\
 0 & -1 & 0 & 0 & 0 & 0 & 0 & 0 & 0 & 0 & 0 \\
 -2 m^2-1 & 0 & 0 & 0 & 0 & 0 & 0 & -2 & 2 & 2 & 0 \\
 1 & 0 & 0 & 0 & 0 & 0 & 0 & 1 & -1 & 0 & 0 \\
 1 & 0 & 0 & 0 & 0 & 0 & 1 & 0 & -1 & 0 & 0 \\
 0 & 0 & 0 & 0 & 0 & 0 & 0 & 0 & -1 & 0 & 0 \\
 \frac{1}{2}-m^2 & 0 & 0 & 0 & 0 & \frac{1}{2} & 1 & 0 & 0 & 0 & 0 \\
 0 & -\frac{\varepsilon +1}{2 m^2} & 0 & 0 & \frac{\varepsilon +1}{2 m^2} & 0 & 0 & 0 & 0 & 0 & -1 \\
\end{array}
\right)\,,
\end{align}
and
\begin{align}
\bD^{\mathrm{em}}_{\sigma}&=\left(
\begin{array}{ccccccccccc}
 -1 & 0 & 0 & 0 & 0 & 0 & 0 & 0 & 0 & 0 & 0 \\
 0 & 0 & 0 & 0 & -1 & 0 & 0 & 0 & 0 & 0 & 0 \\
 0 & -1 & 0 & 0 & 0 & 0 & 0 & 0 & 0 & 0 & 0 \\
 0 & 0 & -1 & 0 & 0 & 0 & 0 & 0 & 0 & 0 & 0 \\
 0 & 0 & 0 & -1 & 0 & 0 & 0 & 0 & 0 & 0 & 0 \\
 -1 & 0 & 0 & 0 & 0 & -1 & -2 & 0 & 0 & 0 & 0 \\
 1 & 0 & 0 & 0 & 0 & 0 & 1 & 0 & 0 & 0 & 0 \\
 1 & 0 & 0 & 0 & 0 & 0 & 1 & 0 & -1 & 0 & 0 \\
 0 & 0 & 0 & 0 & 0 & 0 & 0 & 1 & 0 & 0 & 0 \\
 0 & 0 & 0 & 0 & 0 & -\frac{1}{2} & 0 & -1 & 0 & 1 & 0 \\
 0 & -\frac{\eps +1}{2 m^2} & 0 & 0 & \frac{\eps +1}{2 m^2} & 0 & 0 & 0 & 0 & 0 & -1 \\
\end{array}
\right)\,.
\end{align}
They satisfy the relations (cf.~eq.~\eqref{eq:S4_presentation}):
\beq
(\bD_{\tau}^{\textrm{em}})^2=(\bD_{\sigma}^{\textrm{em}})^4=(\bD_{\tau}^{\textrm{em}}\bD_{\sigma}^{\textrm{em}})^3=\mathds{1}\,.
\eeq
The representation matrices depend on the kinematics. As explained in section \ref{subsec:symmgrouptwisted}, we can always find a basis in which the representation matrices are constant. In particular, we may find a decomposition into irreducible representations. Indeed, the representation $\bD^{\mathrm{em}}$ is reducible. To see this, recall that the inequivalent irreducible representations of the symmetric group $S_n$ are classified by Young tableaux with $n$ boxes (see appendix~\ref{recapgrouptheory} for a review). By computing its character, we find that $\bD^{\mathrm{em}}$ decomposes as 
\ytableausetup{aligntableaux=center,boxsize=0.5em}
\begin{equation}
    \bD^{\mathrm{em}}=2\,\,  \begin{ytableau} \ & \ \\ \ \\ \
\end{ytableau}\oplus \, \begin{ytableau}
\ & \  \\
\ & \
\end{ytableau} \oplus 3 \,\, 
\begin{ytableau}
\  \\ \\ \ \\\
\end{ytableau} \,.
\label{eq:banana_irreps}
\end{equation} 
In particular, the subspace $(\cH_{G})_{\mathrm{em}} = \{\varphi\in\cH_{\mathrm{em}}: \sigma^*\varphi = -\varphi\}$ transforming in the determinant representation is three-dimensional, and it is here simply the alternating representation of $S_4$, representated by the Young tableau $\begin{ytableau}
\  \\ \\ \ \\\
\end{ytableau}$.
By defining a new basis 
$\bs{\tilde{\varphi}}_{\mathrm{\mathrm{em}}}=\bT_{\mathrm{em}}\bs{\varphi}_{\mathrm{em}}$, with
\begin{align}\label{invsubspace}
\bT_{\mathrm{em}}=
\left(
\begin{array}{ccccccccccc}
 0 & 1 & -1 & 0 & 0 & 0 & 0 & 0 & 0 & 0 & 0 \\
 0 & 0 & 0 & 1 & -1 & 0 & 0 & 0 & 0 & 0 & 0 \\
 0 & -1 & 0 & 0 & 1 & 0 & 0 & 0 & 0 & 0 & 0 \\
 -1 & 0 & 0 & 0 & 0 & 0 & -2 & 0 & 0 & 0 & 0 \\
 0 & 0 & 0 & 0 & 0 & 0 & 1 & -1 & 1 & 0 & 0 \\
 0 & 0 & 0 & 0 & 0 & 0 & -\frac{1}{2} & 1 & 0 & 0 & 0 \\
 -2 m^2& 0 & 0 & 0 & 0 & 1 & 1 & -1 & 1 & 2 & 0 \\
 -1 & 0 & 0 & 0 & 0 & -1 & -1 & -2 & 2 & 4 & 0 \\
 1 & 0 & 0 & 0 & 0 & 0 & 0 & 0 & 0 & 0 & 0 \\
 0 & \frac{1}{4} & \frac{1}{4} & \frac{1}{4} & \frac{1}{4} & 0 & 0 & 0 & 0 & 0 & 0 \\
 0 & \frac{3 (\varepsilon +1)}{8 m^2} & -\frac{\varepsilon +1}{8 m^2} & -\frac{\varepsilon +1}{8 m^2} & -\frac{\varepsilon +1}{8 m^2} & 0 & 0 & 0 & 0 & 0 & 1 \\
\end{array}
\right)\,,
\end{align}
we find that the representation matrices split into blocks corresponding to the irreducible representations
\ytableausetup{boxsize=0.7em}
\begin{align}
\label{eq:banana_irreps_2}
\bs{T}_{\mathrm{em}}\bD^{\mathrm{em}}\bs{T}_{\mathrm{em}}^
{-1}=\scalebox{0.6}{$
\left(
\begin{array}{cccccc}
\scriptscriptstyle  \begin{ytableau} \ & \ \\ \ \\ \
\end{ytableau}&  & &&&  \\
  &
\scriptscriptstyle \begin{ytableau} \ & \ \\ \ \\ \
\end{ytableau}& && &  \\
  &  & \scriptscriptstyle\begin{ytableau}
\ & \ \\
\ & \ 
\end{ytableau}
 & && \\
  &  &  & \scriptscriptstyle\qquad \begin{ytableau}
\  \\ \\ \ \\\
\end{ytableau}&&\\
 \\
  & &  && \qquad \begin{ytableau}
\  \\ \\ \ \\\
\end{ytableau}& \\
 \\
  &  &  &&&\qquad \scriptscriptstyle \begin{ytableau}
\  \\ \\ \ \\\
\end{ytableau} \\
\end{array}
\right)$}\,.
\end{align}

\paragraph{The intersection matrix.}
Let us now discuss the structure of the intersection matrix. To this end, we first pick a basis $\bs{\check{\varphi}}_{\textrm{em}}$ for the dual cohomology group. Since we are working on the maximal cut, we may pick the dual basis such that we make the self-duality of the maximal cuts manifest, i.e., such that, if $\bI(\eps)$ and $\bs{\check{I}}(\eps)$ are the vectors of master integrals and their duals obtained from $\bs{\varphi}_{\textrm{em}}$ and $\bs{\check{\varphi}}_{\textrm{em}}$, respectively, then 
$\check{\bI}(\eps)=\bI(-\eps)$.
Let $\bC_{\textrm{em}}(m^2,\eps)$ denote the intersection matrix computed in these bases. We then see that
\beq
\bD_x^{\textrm{em}}(\eps)\bC_{\textrm{em}}(m^2,\eps)\bD_x^{\textrm{em}}(-\eps)^T = \bC_{\textrm{em}}(m^2,\eps)\,,\qquad x=\sigma,\tau\in S_4\,.
\eeq
This relation expresses the fact that the intersection pairing is invariant under twisted symmetry transformations. We can rotate the bases so that they transform in irreducible representations, and we see that the intersection matrix becomes block-diagonal,
\beq
\bs{\widetilde{C}}_{\textrm{em}}(m^2,\eps)=\bT_{\textrm{em}}(\eps)\bC_{\textrm{em}}(m^2,\eps)\bT_{\textrm{em}}(-\eps)^T=\left(
\begin{array}{cccc}
 \bs{\widetilde{C}}^{(1)}_{\mathrm{em}} & 0 & 0 & 0 \\
 0 &\bs{\widetilde{C}}^{(2)}_{\mathrm{em}}& 0 & 0 \\
 0 & 0 & \bs{\widetilde{C}}^{(3)}_{\mathrm{em}}& 0 \\
 0 & 0 & 0 & \bs{\widetilde{C}}^G_{\mathrm{em}} \\
\end{array}
\right)\,,
\eeq
where the blocks on the diagonal are given by
\beq\bsp
\bs{\widetilde{C}}^{(1)}_{\mathrm{em}} &=\frac{1}{4 m^4  (4  m^2-1)}\left(
\begin{array}{ccc}
 -2 & 0 & 1 \\
 0 & -2 & 1 \\
 1 & 1 & -2 \\
\end{array}
\right)\,,\\
 \bs{\widetilde{C}}^{(2)}_{\mathrm{em}} &=\frac{1}{6\eps^2}\left(
\begin{array}{ccc}
 1 & -\frac{1}{2} & 0 \\
 -\frac{1}{2} & \frac{1}{2} & -\frac{1}{8} \\
 0 & -\frac{1}{8} & \frac{1}{8} \\
\end{array}
\right)\,
\,,\\
\bs{\widetilde{C}}^{(3)}_{\mathrm{em}}&=\frac{1}{4\eps^2}\left(
\begin{array}{cc}
 1 & \frac{1}{2} \\
 \frac{1}{2} & 1 \\
\end{array}
\right)\,,\\
\label{eq:symmIntMatrixEM}
    \bs{\widetilde{C}}^G_{\mathrm{em}}&=-\frac{3}{16 \left(64 m^8-20 m^6+m^4\right)}\left(
\begin{array}{ccc}
 0 & 0 & -8 m^2 \\
 0 & 1 & (\widetilde{C}^G_{\mathrm{em}})_{2,3} \\
 -8 m^2 &(\widetilde{C}^G_{\mathrm{em}})_{3,2}& (\widetilde{C}^G_{\mathrm{em}})_{3,3}\\
\end{array}
\right)\,,
\esp\eeq
with 
\begin{align}
(\widetilde{C}^G_{\mathrm{em}})_{2,3}&=(\widetilde{C}^G_{\mathrm{em}})_{3,2}|_{\eps\rightarrow -\eps}=\frac{4 m^2 \left(64 m^4 (2 \eps +1)+m^2 (36 \eps -46)-9 \eps +7\right)+\eps -1}{2 \left(64 m^6-20 m^4+m^2\right)}\,,\notag\\
(\widetilde{C}^G_{\mathrm{em}})_{3,3}=& \frac{1}{4 m^4 \left(64 m^4-20 m^2+1\right)^2} \Big(8 m^2 \Big(2 m^2 \Big(4 m^2 \Big(-640 m^4+545 m^2-\\\notag
&4 \left(512 m^6-352 m^4+357 m^2-81\right) \eps ^2-151\Big)-121 \eps ^2+70\Big)+9 \eps ^2-7\Big)\notag\\
&-\eps ^2+1\Big)\,. \notag
\end{align}

\paragraph{The canonical basis.}
The rotation to a canonical form of this integral family with four unequal masses was derived in refs.~\cite{Duhr:2025kkq,Pogel:2025bca}, while a canonical basis in the equal-mass case was first obtained in ref.~\cite{Pogel:2022yat}. We start from the canonical basis in the unequal mass case, and then put all four masses equal to obtain $\bs{\tilde{\varphi}}_{\textrm{em},c}=\bU_{\mathrm{em}}\bs{\tilde{\varphi}}_{\textrm{em}}$. In the canonical basis, the group acts via the representation matrices $\bs{\widetilde{D}}^{\textrm{em}} = \bs{U}_{\textrm{em}}\bD^{\textrm{em}}\bs{U}_{\textrm{em}}^{-1}$. We find for the two generators $\tau$ and $\sigma$:
\begin{align}
\bs{\widetilde{D}}^{\mathrm{em}}_{\tau}=\left(
\begin{array}{ccccccccccc}
 -1 & 0 & 0 & 0 & 0 & 0 & 0 & 0 & 0 & 0 & 0 \\
 0 & 0 & 0 & 0 & -1 & 0 & 0 & 0 & 0 & 0 & 0 \\
 0 & 0 & -1 & 0 & 0 & 0 & 0 & 0 & 0 & 0 & 0 \\
 0 & 0 & 0 & -1 & 0 & 0 & 0 & 0 & 0 & 0 & 0 \\
 0 & -1 & 0 & 0 & 0 & 0 & 0 & 0 & 0 & 0 & 0 \\
 -1 & 0 & 0 & 0 & 0 & 0 & 0 & -2 & 2 & 2 & 0 \\
 1 & 0 & 0 & 0 & 0 & 0 & 0 & 1 & -1 & 0 & 0 \\
 1 & 0 & 0 & 0 & 0 & 0 & 1 & 0 & -1 & 0 & 0 \\
 0 & 0 & 0 & 0 & 0 & 0 & 0 & 0 & -1 & 0 & 0 \\
 \frac{1}{2} & 0 & 0 & 0 & 0 & \frac{1}{2} & 1 & 0 & 0 & 0 & 0 \\
 6 & 0 & 0 & 0 & 0 & 0 & 6 & 6 & -6 & 0 & -1 \\
\end{array}
\right)\,,
\end{align}
and
\begin{align}
  \bs{\widetilde{D}}^{\mathrm{em}}_{\sigma}=  \left(
\begin{array}{ccccccccccc}
 -1 & 0 & 0 & 0 & 0 & 0 & 0 & 0 & 0 & 0 & 0 \\
 0 & 0 & 0 & 0 & -1 & 0 & 0 & 0 & 0 & 0 & 0 \\
 0 & -1 & 0 & 0 & 0 & 0 & 0 & 0 & 0 & 0 & 0 \\
 0 & 0 & -1 & 0 & 0 & 0 & 0 & 0 & 0 & 0 & 0 \\
 0 & 0 & 0 & -1 & 0 & 0 & 0 & 0 & 0 & 0 & 0 \\
 -1 & 0 & 0 & 0 & 0 & -1 & -2 & 0 & 0 & 0 & 0 \\
 1 & 0 & 0 & 0 & 0 & 0 & 1 & 0 & 0 & 0 & 0 \\
 1 & 0 & 0 & 0 & 0 & 0 & 1 & 0 & -1 & 0 & 0 \\
 0 & 0 & 0 & 0 & 0 & 0 & 0 & 1 & 0 & 0 & 0 \\
 0 & 0 & 0 & 0 & 0 & -\frac{1}{2} & 0 & -1 & 0 & 1 & 0 \\
 8 & 0 & 0 & 0 & 0 & 2 & 14 & -2 & -10 & -8 & -1 \\
\end{array}
\right)\,,
\end{align}
with the relations (cf.~eq.~\eqref{eq:S4_presentation}):
\beq
(\bs{\widetilde{D}}_{\tau}^{\textrm{em}})^2=(\bs{\widetilde{D}}_{\sigma}^{\textrm{em}})^4=(\bs{\widetilde{D}}_{\tau}^{\textrm{em}}\widetilde{\bD}_{\sigma}^{\textrm{em}})^3=\mathds{1}\,.
\eeq
We see that in the canonical basis, the representation matrices are kinematics-independent. 

We know that if we use a self-dual basis, the intersection matrix after rotation to the canonical basis is constant. The $11\times 11$ canonical intersection matrix $\bs{\Delta}_{\textrm{ban}}$ for the three-loop banana integrals with four unequal masses was obtained in ref.~\cite{Duhr:2025kkq}. Since this matrix is independent of the masses, it is the same in the equal-mass limit. If we pass to a basis that is aligned with the decomposition of the $S_4$ symmetry into irreducible representations in the equal-mass limit, we see that the canonical intersection matrix $\bs{\Delta}_{\textrm{ban}}$ becomes block-diagonal, with the blocks representing the decomposition into irreducible representations in eqs.~\eqref{eq:banana_irreps} and~\eqref{eq:banana_irreps_2}:
\begin{align}\label{eq:ban_canon}
\bT_{\mathrm{em}}\bDelta_{\textrm{ban}}\bT_{\mathrm{em}}^T=\left(
\begin{array}{ccccccccccc}
 1 & 0 & -\frac{1}{2} & 0 & 0 & 0 & 0 & 0 & 0 & 0 & 0 \\
 0 & 1 & -\frac{1}{2} & 0 & 0 & 0 & 0 & 0 & 0 & 0 & 0 \\
 -\frac{1}{2} & -\frac{1}{2} & 1 & 0 & 0 & 0 & 0 & 0 & 0 & 0 & 0 \\
 0 & 0 & 0 & -\frac{1}{2} & -\frac{1}{4} & 0 & 0 & 0 & 0 & 0 & 0 \\
 0 & 0 & 0 & -\frac{1}{4} & -\frac{1}{2} & 0 & 0 & 0 & 0 & 0 & 0 \\
 0 & 0 & 0 & 0 & 0 & -\frac{1}{3} & \frac{1}{6} & 0 & 0 & 0 & 0 \\
 0 & 0 & 0 & 0 & 0 & \frac{1}{6} & -\frac{1}{6} & \frac{1}{24} & 0 & 0 & 0 \\
 0 & 0 & 0 & 0 & 0 & 0 & \frac{1}{24} & -\frac{1}{24} & 0 & 0 & 0 \\
 0 & 0 & 0 & 0 & 0 & 0 & 0 & 0 & 0 & 0 & -\frac{1}{2} \\
 0 & 0 & 0 & 0 & 0 & 0 & 0 & 0 & 0 & -6 & 0 \\
 0 & 0 & 0 & 0 & 0 & 0 & 0 & 0 & -\frac{1}{2} & 0 & \frac{11}{12} \\
\end{array}
\right)\,.
\end{align}
We stress that the matrix in eq.~\eqref{eq:ban_canon} is the canonical intersection matrix in the unequal mass case, even though it reflects that decomposition into irreducible representations from the equal-mass case (in particular the rotation matrix $\bs{T}_{\textrm{em}}$ was constructed from the decomposition into irreducible representations in the equal-mass limit).

\paragraph{Irreducible representations for unequal masses.}
Let us briefly comment on other mass configurations, and in particular how the representation of $S_4$ in the equal-mass limit encodes the number of master integrals for other configurations. This can be done using the branching rules for the symmetric group, which are reviewed in appendix~\ref{recapgrouptheory}.

For the two-mass configuration $m_1=m_3=m_4$, the symmetry group is $S_3$. We find that the $S_4$ representation $\bD^{\textrm{em}}$ from the equal-mass case decomposes into irreducible representations of $S_3$ as follows:
\ytableausetup{aligntableaux=center,boxsize=0.5em}
\begin{equation}
    \bD^{\mathrm{em}}\big|_{S_3}=3\, \begin{ytableau}
\ & \  \\
\ 
\end{ytableau} \oplus 5 \, 
\begin{ytableau}
\  \\ \\ \ 
\end{ytableau} \,.
\label{eq:banana_S3}
\end{equation} 
The coefficients appearing in this decomposition are determined from the decomposition of $\bD^{\textrm{em}}$ into irreducible representation of $S_4$ in eq.~\eqref{eq:banana_irreps} and from the branching rules for its subgroup $S_3$. The determinant representation is the alternating representation $\begin{ytableau}
\  \\ \\ \ 
\end{ytableau} $ of $S_3$. We then see from eq.~\eqref{eq:banana_S3} that this configuration has five master integrals, in agreement with refs.~\cite{Maggio:2025jel,Duhr:2025ouy} and table~\ref{tab:3-loop-banana}.
\ytableausetup{aligntableaux=center, boxsize=normal}

For the 
configuration $m_1=m_3,\,m_2=m_4$, the symmetry group is $\mathbb{Z}_2\times \mathbb{Z}_2$. In this case we find that the representation from the equal-mass case decomposes as 
\begin{equation}
     \bD^{\mathrm{em}}\big|_{\mathbb{Z}_2\times \mathbb{Z}_2}=(++)\oplus 2(+-)\oplus 2(-+)\oplus 6(--) \,,
     \label{eq:banana_Z2xZ2}
\end{equation}
where the irreducible representations are labeled by pairs $(s_1,s_2)$ with $s_i=\pm$ (recall that the irreducible representations of $\mathbb{Z}_2$ just correspond to decompositions into even and odd contributions). In this notation, the determinant representation is $(--)$, and so we see from eq.~\eqref{eq:banana_Z2xZ2} that this configuration has 6 master integrals, in agreement with ref.~\cite{Maggio:2025jel} and table~\ref{tab:3-loop-banana}.

Finally, for the configuration $m_1=m_2$, the symmetry group is $\mathbb{Z}_2$. The decomposition of the representation $\bD^{\textrm{em}}$ into irreducible representations for $\mathbb{Z}_2$ reads,
\begin{equation}
    \bD^{\mathrm{em}}\big|_{\mathbb{Z}_2}=3(+)\oplus 8(-) \,.
\end{equation}
The determinant representation is the odd representation of $\mathbb{Z}_2$, and so we see that this mass configuration has 8 master integrals, in agreement with table~\ref{tab:3-loop-banana}.

\section{From group theory to topology}
\label{sec:OrbitSpaceEuler}

In the previous section, we have argued that we may apply tools from the theory of finite groups to study the symmetry group of a sector. It is well known from the representation theory of finite groups that the decomposition into irreducible representations is governed by characters (see appendix~\ref{recapgrouptheory} for a review). 
In this section, we show that we can compute the character of the representation of the twisted symmetry transformations in a sector using purely topological information.

\subsection{The Euler characteristic}
The Euler characteristic of a (co-)homology theory is defined as the alternating sum of the dimensions of the cohomology groups. In our setting, we focus on the cohomology groups of a space $X = \mathbb{C}^n\setminus \Sigma$ (i.e., we assume $D_+=D_-=\emptyset$). For example, the Euler characteristic of the twisted cohomology theory $H^n(X,\nabla_{\omega})$ is
\begin{equation}\label{eq:twisted_Euler}
    \chi(X,\nabla_{\omega}):=\sum_{k=0}^{2n}(-1)^k\, \dim H^k(X,\nabla_{\omega}) \,.
\end{equation}
Similarly, we can consider the \emph{topological} Euler characteristic of $X$, computed from the ordinary (untwisted) cohomology groups
\begin{equation}\label{eq:topo_Euler}
    \chi(X):=\sum_{k=0}^{2n}(-1)^k\, \dim H^k(X,\mathbb{C}) \,.
\end{equation}
In general, the dimensions of the ordinary and twisted cohomology groups do not agree, $\dim H^k(X,\nabla_{\omega}) \neq \dim H^k(X,\mathbb{C})$. Hence, there is no reason to expect that the twisted and topological Euler characteristics of $X$ should be the same. Remarkably, it was shown that the two concepts always agree~\cite[sec. 2.2.13]{aomoto_theory_2011}:
\beq\label{eq:twistedchi_topochi}
\chi(X,\nabla_{\omega}) = \chi(X)\,.
\eeq
For this reason, we will not distinguish the two concepts anymore, and we will only talk about the Euler characteristic $\chi(X)$ in the rest of the paper.

In the context of twisted cohomology, there is a direct connection between the Euler characteristic and the dimension of the middle cohomology $H^n(X,\nabla_{\omega})$. More precisely, if the twisted variety $\Sigma$ defines a normal crossing divisor (meaning, all its irreducible components intersect transversely) and if there is no resonance (by which we mean that the monodromy of the twist around the twisted variety does not have finite order, including at infinity), then only the middle cohomology is non-zero~\cite{aomoto_theory_2011},\footnote{The vanishing theorem generalizes to certain partially-twisted cases where $D_-\neq \emptyset~$ \cite{Agostini:2022cgv}. }
\beq
\dim H^k(X,\nabla_{\omega}) = 0\,,\quad \textrm{if } k\neq n\,.
\eeq
We refer to this result as the \emph{vanishing theorem} for the twisted cohomology theory. It is easy to see that, if the vanishing theorem holds, the dimension of the middle cohomology is given by the Euler characteristic of $X$, up to a sign,
\beq\label{eq:dim_chi}
\dim H^n(X,\nabla_{\omega}) = (-1)^n\,\chi(X)\,,\quad \textrm{if the vanishing theorem holds.}
\eeq
Note that eq.~\eqref{eq:dim_chi} has an immediate consequence, which we will need later on. The Euler characteristic $\chi(X)$ is an integer number, either positive or negative. The left-hand side of eq.~\eqref{eq:dim_chi}, however, must be positive, and so
\beq\label{eq:sign_chi}
\operatorname{sign}(\chi(X)) =(-1)^n \quad\textrm{if the vanishing theorem holds.}
\eeq

The relationship between the Euler characteristic and the dimension of the middle cohomology group has an interesting implication for families of Feynman integrals. Let us use the notations and assumptions from section~\ref{sec:decomposition}, namely we assume that (1) the symmetry group is trivial, (2) there are no additional elements in the kernel of the integration map $\mathcal{I}_{\gamma_{\cG}}$, and (3) the vanishing theorem holds. We then see that the dimension of the cohomology group $\cH^{\textrm{m.c.}}$ associated to the maximal cuts in eq.~\eqref{eq:cH_mc} is given by the absolute value of the Euler characteristic:
\beq\label{eq:cH_mc_dim}
\dim \cV_{\gamma_{\cG}} = \dim \cH^{\textrm{m.c.}} = \big|\chi\big(\mathbb{C}^P\setminus \Sigma\big)\big|\,.
\eeq
This result was first pointed out in refs.~\cite{Bitoun:2017nre,bitoun2018numbermasterintegralseuler,Mastrolia:2018uzb}.\footnote{The result of refs.~\cite{Bitoun:2017nre,bitoun2018numbermasterintegralseuler} was not obtained in the context of twisted cohomology theories, but in the context of $D$-modules for families of Feynman integrals with generic kinematics (so no symmetries) and with arbitrary propagator powers. The last point essentially identifies the corresponding cohomology theory as a non-relative twisted cohomology theory, and the genericity assumptions exclude resonance effects, so that the vanishing theorem holds.} It relates the number of master integrals in a given sector to a purely topological property of the space $X=\mathbb{C}^P\setminus\Sigma$, namely the Euler characteristic $\chi\big(\mathbb{C}^P\setminus \Sigma\big)$. The absolute value of the Euler characteristic in turn can often be computed (under some mild assumptions) using Morse theory by counting the number of critical points~\cite{Lee:2013hzt}. 

The previous considerations, however, fail in the presence of a non-trivial symmetry group $G$ of symmetry transformations for the sector. Indeed, if $G$ is non-trivial, then the number of master integrals is given by the dimension of the irreducible subspace $\cH_{G}^{\textrm{m.c.}}$ transforming in the determinant representation, cf.~eq.~\eqref{eq:dimV_dimHG}. The dimension of $\cH_{G}^{\textrm{m.c.}}$ is typically strictly less than the dimension of $\cH^{\textrm{m.c.}}$. In the remainder of this section we discuss how we can compute the dimension of $\cH_{G}^{\textrm{m.c.}}$ using purely topological information. The proof of the result will heavily rely on tools from algebraic topology, which we introduce in the next subsections.

\subsection{Some background in algebraic topology}
In this section we review some background material in algebraic topology and we prove some lemmas that are needed to understand the relationship between the number of master integrals in a given sector and the Euler characteristic in the presence of symmetries. In order to have rigorous and complete proofs of our results, we include many details. The reader not interested in these details may skip this section (at first reading).

\subsubsection{Some technical lemmas}
We start by proving some technical lemmas from linear algebra which appear repeatedly in the subsequent proofs.

\begin{lemma}\label{lemma1fromnotes}
    Let $V_i$, $i=1,2,3$, be three complex vector spaces, and $A_i:V_i\to V_i$ linear maps. Assume that we have a short exact sequence,\footnote{This simply means that for every node, the image of the map on the incoming arrow equals the kernel of the map on the outgoing arrow.}
    \beq
    \label{eq:exactSeqLemma1}
    0\longrightarrow V_1\xrightarrow{\phantom{a}\alpha\phantom{a}}V_2\xrightarrow{\phantom{a}\beta\phantom{a}}V_3\longrightarrow 0\,.
    \eeq
    Assume that the linear maps $A_i$ commute with the linear maps $\alpha$ and $\beta$,
    \beq
    A_2\alpha = \alpha A_1\,,\qquad A_3\beta=\beta A_2\,.
    \eeq
    Then $\Tr(A_2) = \Tr(A_1) +\Tr(A_3)$.
\end{lemma}
\begin{proof}
    We will prove this by constructing convenient bases for the $V_i$ and computing the traces of the matrices representing the maps $A_i$ in these bases. To this end, we start by choosing some basis $v_1,\dots ,v_{d_1}$ of $V_1$, with $d_1=\dim V_1$. Then, since $\alpha$ is injective, the vectors $\alpha(v_1),\dots ,\alpha(v_{d_1})$ are linearly independent and can be completed into a basis of $V_2$, which we denote by
    \begin{equation}
        \{\alpha(v_1),\dots ,\alpha(v_{d_1}),w_1,\dots ,w_{d_3} \} \,,
    \end{equation}
    for some $d_3\geq 0$, such that $d_2=\dim V_2=d_1+d_3$. From the exactness of the sequence in eq.~\eqref{eq:exactSeqLemma1} it readily follows that $\dim V_3=d_3$:
    \begin{equation}
        \dim V_3=\dim V_2-\dim\mathrm{Ker}\,\beta=\dim V_2-\dim \mathrm{Im}\,\alpha=\dim V_2-\dim V_1=d_3 \,.
    \end{equation}  
    Now observe that $\beta(w_1),\dots ,\beta(w_{d_3})$ form a basis of $V_3$. To see this, first note that, since $\beta$ is surjective and $\beta(\alpha(v))=0$ for all $v\in V_1$, the vectors $\beta(w_1),\dots ,\beta(w_{d_3})$ certainly generate all elements of $V_3$. Since $d_3=\dim V_3$ this already implies that they form a basis.  

    To prove the claim let us now describe the form of the matrix $M_2$ representing $A_2$ in the chosen basis of $V_2$. Generally, we can write the matrix in a block form
    \begin{equation}
        M_2=\begin{pmatrix}
            M_2^{(11)} & M_2^{(12)} \\ M_2^{(21)} & M_2^{(22)} 
        \end{pmatrix} \,,
    \end{equation}
    where we split the $d_2$ rows and columns into a block of size $d_1$ and a block of size $d_3$. Then the trace of $A_2$ can be computed as
    \begin{equation}
    \label{eq:tracesLemmaProof}
        \Tr(A_2)=\Tr(M_2)=\Tr\big(M_2^{(11)}\big)+\Tr\big(M_2^{(22)}\big) \,.
    \end{equation}
    Now notice that
    \begin{equation}
    \label{eq:A2ActionLemmaProof}
        A_2(\alpha(v_j))=\alpha(A_1(v_j))=\sum_{k=1}^{d_1}\alpha(v_k(M_1)_{kj})=\sum_{k=1}^{d_1} \alpha(v_k)(M_1)_{kj} \,,
    \end{equation}
    for $j\in\{1,\dots ,d_1\}$, where $M_1$ is the matrix representing $A_1$ in the chosen basis for $V_1$. From eq.~\eqref{eq:A2ActionLemmaProof} it in particular follows that $M_2^{(11)}=M_1$ and hence $\Tr\big( M_2^{(11)}\big)=\Tr(A_1)$. Similarly we have
    \begin{equation}
        \beta(A_2(w_j))=A_3(\beta(w_j))=\sum_{k=1}^{d_3}\beta(w_k)(M_3)_{kj} \,,
    \end{equation}
    for $j\in\{1,\dots ,d_3\}$, where $M_3$ is the matrix representing $A_3$ in the chosen basis for $V_3$. Comparing this to
    \begin{equation}
        \beta(A_2(w_j))=\beta\left( \sum_{k=1}^{d_1}\alpha(v_k) \big(M_2^{(21)}\big)_{kj}+\sum_{k=1}^{d_3}w_k \big(M_2^{(22)}\big)_{kj}\right)=\sum_{k=1}^{d_3}\beta(w_k) \big(M_2^{(22)}\big)_{kj}\,,
    \end{equation}
    and using linear independence of $\beta(w_1),\dots ,\beta(w_{d_3})$, we see that $M_2^{(22)}=M_3$ and hence $\Tr\big(M_2^{(22)}\big)=\Tr(M_3)=\Tr(A_3)$. Inserting the results for $\Tr\big(M_2^{(11)}\big)$ and $\Tr\big(M_2^{(22)}\big)$ found above into eq.~\eqref{eq:tracesLemmaProof} proves the claim.
\end{proof}

\begin{lemma}\label{lemma2fromnotes}
    Consider a long exact sequence of vector spaces,
    \beq
    0\longrightarrow V_1 \xrightarrow{\phantom{a}d_1\phantom{a}}V_2 \xrightarrow{\phantom{a}d_2\phantom{a}} V_3 \xrightarrow{\phantom{a}d_3\phantom{a}} \ldots \xrightarrow{\phantom{a}d_{r-2}\phantom{a}}V_{r-1}\xrightarrow{\phantom{a}d_{r-1}\phantom{a}}V_r \longrightarrow0\,.
    \eeq
    Assume that there are linear maps $T_i: V_i\to V_i$ that commute with the $d_i$, $T_{i+1}d_i=d_iT_i$. Then the alternating sum of the traces is zero:
    \beq
    \label{eq:trTelescopeSumLemma2}
    \sum_{i=1}^r (-1)^i\,\Tr(T_i)=0\,.
    \eeq
\end{lemma}
\begin{proof}
    For $i=2,\dots ,r-1$ consider the short exact sequence
    \begin{equation}
        0\longrightarrow B_i \xrightarrow{\phantom{a}\iota\phantom{a}} V_i
        \xrightarrow{\phantom{a}d_i\phantom{a}} B_{i+1} \longrightarrow 0 \,,
    \end{equation}
    where $B_i=\im d_{i-1}=\ker d_i \subseteq V_i$, and $\iota :B_i\rightarrow V_i$ is the inclusion. Consider $v\in B_i$. Then
    \begin{equation}
        d_i(T_i(v))=T_{i+1}(d_i(v))=0 \,,
    \end{equation}
    and hence $T_i(v)\in B_i$, i.e.,  $T_i(B_i)\subseteq B_i$ and $T_i\vert_{B_i}$ provides a well-defined linear map $B_i\rightarrow B_i$ for all $i=2,\dots ,r-1$. Furthermore it is clear that
    \begin{equation}
        T_i\circ\iota=\iota \circ T_i\vert_{B_i} \,, \qquad T_{i+1}\vert_{B_{i+1}}\circ d_i=d_i\circ T_i \,.
    \end{equation}
   We are hence in the situation of Lemma~\ref{lemma1fromnotes}, which implies that
    \begin{equation}
    \label{eq:trRelProofLemma2}
        \Tr(T_i)=\Tr\big( T_i\vert_{B_i}\big)+\Tr\big( T_{i+1}\vert_{B_{i+1}}\big) \,,
    \end{equation}
    for $i=2,\dots ,r-1$. Since $d_1$ is injective and $d_{r-1}$ surjective, we furthermore have
    \begin{equation}
    \label{eq:additionalTrsProofLemma2}
        \Tr\big(T_1\vert_{B_1}) =0 \,, \qquad \Tr\big(T_r\vert_{B_r}) =\Tr\big(T_r) \,.
    \end{equation}
    Hence, by applying eq.~\eqref{eq:trRelProofLemma2} for $i=2,\dots ,r-1$, as well as eq.~\eqref{eq:additionalTrsProofLemma2} for the first and last term, we see that the alternating sum in eq.~\eqref{eq:trTelescopeSumLemma2} indeed telescopes to zero.
\end{proof}

\subsubsection{Homotopies and retractions}
In our setup, the relevant space $X=\mathbb{C}^n\setminus \Sigma$ is not compact. Many results in algebraic topology, however, apply to compact spaces. In the next section, we describe a way to replace our non-compact space $X$ by a compact space with isomorphic cohomology groups.

The main idea is that two spaces that can be continuously deformed into each other have isomorphic (co-)homology groups. This idea of continuous deformation is captured by the concept of a \emph{homotopy}. Consider two topological spaces $X$ and $Y$, and continuous maps $f,g:X\to Y$ between them. We say that $f$ and $g$ are \emph{homotopic}, $f\simeq g$, if there is a continuous function $h:[0,1]\times X\to Y$ such that $h(0,\cdot)=f$ and $h(1,\cdot)=g$. Homotopy defines an equivalence relation on continuous maps. 

The two spaces $X$ and $Y$ are said to be \emph{homotopy-equivalent}, $X\simeq Y$, if there are continuous maps $g_1:X\to Y$ and $g_2:Y\to X$ such that their compositions are homotopic to the identity maps, $g_1\circ g_2\simeq \id_Y$ and $g_2\circ g_1\simeq \id_X$. A fundamental result in algebraic topology states that if $X$ and $Y$ are homotopy-equivalent, then they have isomorphic homology and cohomology groups, i.e., $X\simeq Y$ implies $H^k(X)\simeq H^k(Y)$.\footnote{We state the result here for ordinary cohomology groups. The result also applies to twisted cohomology theories.}

Homotopy-invariance allows one, at least as far as the computation of cohomology groups is concerned, to replace $X$ by $Y$ without loosing any information. In our setup, we will encounter a special case of homotopy-equivalence where $Y$ is a subset of $X$. We say that $Y\subseteq X$ is a \emph{strong deformation retract} of $X$ if there is a continuous map $r:X\to Y$ and a continuous map $h:[0,1]\times X\to X$ such that 
\begin{itemize}
    \item  $h(0,\cdot) = \id_X$ and $h(1,\cdot)=r$,
    \item $h(s,y) = y$, for all $s\in[0,1]$ and $y\in Y$.
\end{itemize}
It is easy to see that, if $Y$ is a strong deformation retract of $X$, then $Y$ is also homotopy-equivalent to $X$ (it suffices to pick $g_1 =r $ and $g_2=\iota:Y\to X$ the inclusion). In particular, this means that if $Y$ is a strong deformation retract of $X$, then $X$ and $Y$ have isomorphic homology and cohomology groups.

\subsubsection{The compact model}\label{sec:compactmodel}
The main idea is now the following: we may replace our space $X=\mathbb{C}^n\setminus\Sigma$ by some other space that is compact and is a strong deformation retract of $X$. To this end, we will make use of Hironaka's theorem for the resolution of singularities~\cite{Hironaka1964}, as well as arguments from differential topology.  The arguments are standard, and analogous to those in ref.~\cite[section 2.2.13]{aomoto_theory_2011}. We will therefore not go into technical details.

Let us embed $X$ into projective space, i.e., we write $X=\mathbb{CP}^n\backslash \Sigma_p$ where $\Sigma_p$ includes the hypersurface at infinity. Then, there exists a finite sequence of blowups along smooth centers (see, e.g., ref.~\cite{GriffithsHarris1978} for introductory material about blowups):
\begin{align}
    \pi:\overline{W}\to \mathbb{CP}^n\,,
\end{align}
such that, with $\overline{X}=\pi^{-1}(X)$, the complement $D:=\overline{W}\backslash \overline{X}$ is a simple normal crossing divisor, meaning that there exist local coordinates $x_i$ such that locally $D$ can be described by $D=\{(x_1\dots x_k)\in X,x_1\dots x_k=0\}$. The ambient space $\overline{W}$ is a smooth variety. Note that $\pi|_{\overline{X}}:\overline{X}\rightarrow X$ is an isomorphism of algebraic varieties. 

Further, assume that we have an automorphism $f:\mathbb{CP}^n\to \mathbb{CP}^n$ such that $f(\Sigma)\subseteq \Sigma$. Then there exists a (functorially) induced smooth morphism $\widetilde{f}:\overline{W}\to\overline{W} $ (c.f. ref.~\cite[Theorem 1.0.2]{bierstone2012effectivehironakaresolutioncomplexity})
\begin{align}\label{eq:pi-f}
\pi \circ \widetilde{f}=f\circ \pi\,.
\end{align}
We can then see that $\widetilde{f}$ is bijective on $\overline{X}$ in the following way: There exists an $f^{-1}$ with lift  $\widetilde{f^{-1}}$ such that, since $\pi$ is functorial, we can compose $\pi \circ \widetilde{f^{-1}}=f^{-1}\circ \pi$. Then $\pi \circ \widetilde{f}\circ \widetilde{f^{-1}}=f\circ f^{-1}\circ \pi=\pi$. Hence on $\overline{X}$, where $\pi$ is an isomorphism, $\widetilde{f}\circ \widetilde{f^{-1}}=\text{id}$, and the lift of $f^{-1}$ is the inverse of $\widetilde{f}$ and is a smooth morphism. Hence $\widetilde{f}|_{\overline{X}}$ is an automorphism of $\overline{X}$ such that $\widetilde{f}(D)\subset D$.

The divisor $D$ might still be singular and we construct a neighborhood $T_\delta$ around $D$ such that $\partial T_\delta$ is smooth and $T_\delta$ deformation retracts to $D$.  
To achieve this it is important that, according to ref.~\cite[Theorem 1.0.2]{bierstone2012effectivehironakaresolutioncomplexity}, if $f\in G$ is a group action, then the group structure is preserved by the lift. In particular, if $f$ is of finite order $k$, $f^k=\text{id}$ then $\widetilde{f}^k=\text{id}$. This makes it possible to choose an $\widetilde{f}$-invariant function $\phi:\overline{W}\to \mathbb{R}$, with $\phi^{-1}(0)=D$, that measures the distance to $D$. Then there is an (open) tubular neighborhood for some small $\delta$
\begin{align}\label{eq:tube}
    T_\delta:=\{p\in \overline{W}, \phi(p)<\delta\}\,,
\end{align}
 such that $\widetilde{f}(T_\delta)=T_\delta$. Furthermore, $\partial T_\delta$ is a smooth hypersurface such that $X_\delta:=\overline{W}\backslash T_\delta$ is compact and there is a strong deformation retraction 
\begin{align}
    r:\overline{X}\to X_\delta\,.
\end{align}
 We do not provide an explicit discussion of the construction of $\phi$ and $T_\delta$  here. Details on tubular neighbourhoods of normal crossing divisors can be found in ref.~\cite{elzein1999topologyalgebraicvarieties} and references therein.

\subsection{Lefschetz numbers}
Consider a space $X$ of complex dimension $n$, and assume that we define a twisted cohomology theory on it with flat connection $\nabla_{\omega}$. Consider a map $f:X\to X$. We have already seen that $f$ induces via pullback a map $f^*$ between the twisted cohomology groups of $X$. In order to keep track of the precise cohomology groups, throughout this section it will be useful to denote the induced maps $f^*$ by 
\beq
H^k(f,X,\nabla_{\omega}): H^k(X,\nabla_{\omega})\to H^k(X,\nabla_{\omega})\,.
\eeq
The \emph{twisted Lefschetz number} of $f$ is defined as the alternating sum of the traces of the induced maps in cohomology:
\beq\label{eq:twisted_Lefschetz}
L(f,X,\nabla_{\omega}) = \sum_{k=0}^{2n}(-1)^k\,\Tr\big(H^k(f,X,\nabla_{\omega})\big)\,.
\eeq
Similarly, we can define the \emph{topological Lefschetz number} of $f$ as
\beq\label{eq:topo_Lefschetz}
L(f,X) = \sum_{k=0}^{2n}(-1)^k\,\Tr\big(H^k(f,X)\big)\,.
\eeq
The most famous application of the (topological) Lefschetz number is the celebrated \emph{Lefschetz fixed-point theorem}:
\begin{theorem}
Let $X$ be a compact space and $f:X\to X$ a continuous map. If $L(f,X)\neq 0$, then $f$ has at least one fixed-point in $X$.
\end{theorem}

There is a similarity between the definition of the twisted and topological Lefschetz numbers in eqs.~\eqref{eq:twisted_Lefschetz} and~\eqref{eq:topo_Lefschetz} and the corresponding Euler characteristics in eqs.~\eqref{eq:twisted_Euler} and~\eqref{eq:topo_Euler}. Indeed, it is easy to see from the definitions that the Euler characteristics are nothing but the Lefschetz numbers of the identity map,
\beq\label{eq:L_to_chi}
\chi(X,\nabla_{\omega}) = L(\id,X,\nabla_{\omega}) \textrm{~~~and~~~} \chi(X) = L(\id,X)\,.
\eeq
We have also seen in eq.~\eqref{eq:twistedchi_topochi} that the twisted and topological Euler characteristics agree. We now show that the same result extends to the Lefschetz numbers, assuming the same assumptions as in the proof of eq.~\eqref{eq:twistedchi_topochi} (see, e.g., section 2.2.13 of ref.~\cite{aomoto_theory_2011}), which in particular implies that we are in the setting of section~\ref{sec:compactmodel}. We start by proving some technical results that we will also need in the next subsection, and that in particular connect Lefschetz numbers and strong deformation retractions.

First, let us mention an obvious fact. Consider a finite-dimensional $\mathbb{C}$-vector space $V$ with basis $\bv$ and a map $f:V\to V$ that induces an action on the basis $f:\bv\to \bA \bv$ with $\bA\in \textrm{GL}(\text{dim}(V),\mathbb{C})$. Any two basis choices are related via an isomorphism and the trace is basis independent $\textrm{Tr}(f)=\textrm{Tr}(\bA)$. Similarly, for any isomorphic vector space $W\simeq V$, $f$ induces a map $f_W:W\to W$ on $W$ with $\textrm{Tr}(f_W)=\textrm{Tr}(f)$. In particular, the dual space $\check{V}$ is isomorphic to $V$. 
That is why in the following we will not differentiate between the traces of induced map on (co)homology groups that are related via isomorphisms/dualities. For example
\begin{align}
\label{eq:differentLefshetzSums}
    L(f,X)=&\sum_k (-1)^k \,\mathrm{Tr}(H_k(f,X))=\sum_k (-1)^k  \,\mathrm{Tr}(H_\mathrm{dR}^k(f,X))\notag\\
  = &\sum_k (-1)^k  \,\mathrm{Tr}(H_\mathrm{dR,c}^{2n-k}(f,X))=\sum_k (-1)^k  \,\mathrm{Tr}(H_\text{dR,c}^{k}(f,X))\,.
\end{align}
Details on the corresponding isomorphisms can be found in appendix~\ref{dualities}.

\begin{lemma}\label{eq:Lefschetz-retraction}
Let $X$ and $Y\subseteq X$ be topological spaces, and $f:X\to X$ a continuous map such that $f(Y)\subseteq Y$. If $Y$ is a strong deformation retract of $X$, then $L(f,X) = L(f,Y)$.
\end{lemma}
\begin{proof}
 Let $\iota:Y\rightarrow X$ be the inclusion and let $r:X\rightarrow Y$ be the strong deformation retraction. Then we have $r\circ \iota=\mathrm{id}_Y$, since $r\vert_Y=h(1,\cdot)\vert_Y=\mathrm{id}_Y$, where $h$ is the homotopy between $r$ and $\mathrm{id}_X$. Furthermore we have $\iota\circ r\simeq \mathrm{id}_X$, with the homotopy given by $h$. Since homotopic maps induce identical maps on (co)homology \cite{hatcherAlgTop} we have that
    \begin{equation}
        \iota^*\circ r^*=(r\circ \iota)^*=\mathrm{id}_{H^k(Y)}  \,, \qquad 
        r^*\circ \iota^*=(\iota\circ r)^*=\mathrm{id}_{H^k(X)}  \,,
    \end{equation}
    for all $k$. Recall that for two maps $\psi:M\rightarrow N$ and $\varphi:N\rightarrow M$, if $\psi\circ \varphi=\mathrm{id}_N$, then $\psi$ is surjective and $\varphi$ is injective. Applying this here, we see that the maps
    \begin{equation}
        \iota^*:H^k(X)\rightarrow H^k(Y) \,, \qquad r^*:H^k(Y)\rightarrow H^k(X) \,,
    \end{equation}
    are both bijections for all $k$ and are inverses to each other. In particular, we see that $H^k(X)\simeq H^k(Y)$ for all $k$. 
    Now let us define $g=f\vert_Y:Y\rightarrow Y$, which is well-defined by assumption. Then from $f\circ \iota=\iota\circ g$ it follows that
    \begin{equation}
        g=r\circ\iota\circ g=r\circ f\circ \iota \,,
    \end{equation}
    and hence
    \begin{equation}
        g^*=\iota^*\circ f^*\circ r^*=\iota^*\circ f^*\circ (\iota^*)^{-1} \,.
    \end{equation}
    We see that, since $g^*$ and $f^*$ are related by a similarity transformation, their traces are equal (on each cohomology group), i.e., we have
    \begin{equation}
        H^k(f,Y)=H^k(g,Y)=\Tr(g^*)=\Tr(f^*)=H^k(f,X) \,,
    \end{equation}
    for all $k$. The claim then follows by taking the alternating sum. 
\end{proof}

\begin{proposition}\label{prop:twistedL_TopL}
Let $f:X\to X$ be a twisted symmetry transformation in a given sector acting on a twisted cohomology group as defined in sec.~\ref{introtwisted}. Then, under the same assumptions as in section~\ref{sec:compactmodel}, the twisted and topological Lefschetz numbers of $f$ agree,
\beq
L(f,X) = L(f,X,\nabla_{\omega})\,.
\eeq
\end{proposition}
\begin{proof}
First note that, as described in subsection~\ref{sec:compactmodel}, by Hironaka's theorem, there exists a resolution $\pi:\overline{W}\to X$ such that $\overline{W}$ is smooth and for $\overline{X}:=\pi^{-1}(X)$, the complement $D:=\overline{W}\backslash \overline{X}$ is a simple normal crossing divisor.
Further, as described in subsection~\ref{sec:compactmodel}, the lift $\widetilde{f}$ of $f$ is an automorphism of $\overline{W}$ such that $\widetilde{f}(\overline{X})\subseteq \overline{X}$ and since $\pi$ is bijective on $X$ it follows that
\begin{align}
    L(f,X)=L(\tilde{f},\overline{X})\,.
\end{align}
We have seen that we can construct an open neighborhood $T_\delta(D)$ around $D$ such that $\partial T_\delta$ is a smooth hypersurface and there exists a strong deformation retraction $r:\overline{X}\to X_\delta$, where $X_\delta=\overline{W}\backslash T_\delta$  is smooth and compact. Then the two Lefschetz numbers agree 
\begin{align}
    L(f,X)=L(\tilde{f},X_\delta)\,. 
\end{align}

We will now make use of this construction to show the equivalence of the twisted and the topological Lefschetz number. To this end, we will first argue that the topological Lefschetz number can be written in terms of traces of permutation matrices acting on a simplicial complex associated with $X_{\delta}$. Then, we will show that the twist does not modify this result.

Let us describe the associated simplicial complex. To this end, we will use terminology from simplicial homology. Details can be found in appendix~\ref{simplicialhomology}. We know that smooth, compact manifolds admit finite triangulations, such that there exists a homeomorphism $ t: K_\delta\to X_\delta\,,$
where $K_\delta$ is the associated simplicial complex. Then $\widetilde{f}_t:=t^{-1}\circ \widetilde{f} \circ t:K_\delta\to K_\delta$ is a map between simplicial complexes. Let $\widetilde{K}_\delta$ be the $k^{\text{th}}$ barycentric subdivision, for some finite $k\in \mathbb{N}$ such that there exists a simplicial approximation $f_t$ to $\widetilde{f}_t$. In particular $f_t$ is a simplicial map (cf. appendix~\ref{simplicialhomology}).  As the two maps are homotopic $\tilde{f}_t \simeq f_t$ their Lefschetz numbers agree 
\begin{align}
    L(f,X)=L(\widetilde{f}_t,K_\delta)=L(f_t,\widetilde{K}_\delta)\,,
\end{align}
where 
\begin{align}
L(f_t,\widetilde{K}_\delta)=\sum_k (-1)^k\,\text{Tr}(   H_k(f_t,\widetilde{K}_\delta))\,,
\end{align}
and $H_k(f_t,\widetilde{K}_\delta)$ is the induced map on the simplicial homology group.
For simplicity, let us from now on denote by $f$ the simplicial approximation and drop the tilde on $\widetilde{K}_\delta$.
Since $f$ is simplicial, the induced chain map on the $k$-simplices $\sigma^{(j)}_k=[v^{(j)}_0,\dots, v^{(j)}_k]\in C_k(K_\delta)$ is a signed permutation
\begin{align}
 f_*( \sigma_k^{(j)})=  [f(v^{(j)}_0),\dots, f(v^{(j)}_k)]=({\bs{\kappa}_k})_j\, \sigma_k^{\tau_k(j)}\,,
\end{align}
 with $({\bs{\kappa}_k})_j\in\{\pm 1\}$, $\tau_k\in S_{d_{C_k}}$ and where $d_{C_k}$ is the number of $k$-simplices in the triangulation.  
Then, as argued in the appendix~\ref{app:inducedmaps}, eq.~\eqref{inducedC}, the Lefschetz number becomes
\begin{align}
   L(f,K_\delta)=\sum_k (-1)^k \text{Tr}(C_k(f,K_\delta))=\sum_k (-1)^k\text{Tr}(\bP^T_k)\,,
\end{align}
where $\bP_k$ is a $d_{C_k}\times d_{C_k}$ dimensional signed permutation matrix.

Now we want to establish the same result for the twisted case, where we call the induced map $H_k(f,K_\delta,\mathcal{L}^\vee)$. 
Details on (simplicial) homology with coefficients in a local system can be found in refs.~\cite{hatcherAlgTop, aomoto_theory_2011}. For a simplex with coefficients in a local system, the coefficient can be assigned at the barycenter
\begin{align}
    \sigma \otimes s(b_\sigma)\,,
\end{align}
where $s(b_\sigma)\in \mathcal{L}^\vee_\omega(b_\sigma)$ is a (flat) section in the fiber of the (flat) line bundle over the point $b_\sigma$. The barycenter transforms under $f$ as 
\begin{align}
b_{{[
   f(v_0),\dots, f(v_k)]}}&=\frac{1}{1+k} \sum_l f(v_l)=f(b_{{[
   v_0,\dots, v_k]}})\,.
\end{align}
The section in the fiber over a point $p$ in an open chart $U_i$  is given by the value of the locally single-valued twist. As the twist is by definition invariant under $f$, if $f$ is a twisted symmetry transformation, the associated transition functions are unchanged, and the line bundle with its  fibers and sections remains unchanged under $f:K_\delta\to K_\delta$.
Accordingly, a chain with coefficients in the local system transforms as
\begin{align}f_*:\sigma_k^{(j)}\otimes s(b_{\sigma_k^{(j)}})\mapsto  
&\sigma_k^{\tau_k(j)}\otimes s(f(b_{\,\sigma_k^{(j)}}))=\sigma_k^{\tau_k(j)}\otimes s(b_{\,\sigma_k^{\tau_k(j)}})=(\bP_k)_{ij}(\sigma_k^{i}\otimes s(b_{\,\sigma_k^{i}}))\,,
\end{align}
where $(\bP_k)_{ij}=(\bs{\kappa}_k)_j\, \delta_{i,\tau_k(j)}$. 
Then the Lefschetz number is the same as in the untwisted case 
\begin{align}
    L(f,X_\delta,\mathcal{L}^\vee)&=
    \sum_{k=0}^{2n}(-1)^k \text{Tr}(C_k(f,K_\delta,\mathcal{L}^\vee))=\sum_{k=0}^{2n}(-1)^k \text{Tr}(\bP^T_k)=\text{Tr}(C_k(f,K_\delta))
   \,.
\end{align}
\end{proof}

Since the two notions agree, we will from now on no longer distinguish between the twisted and topological Lefschetz numbers. We stress that for this conclusion to hold, $f$ needs to be a twisted symmetry, which means in particular that it leaves the local system invariant. In the remainder of this subsection, we prove several properties of Lefschetz numbers that we will need in the next section to prove our main result. 

\begin{lemma}\label{lemma4fromnotes}
Let $X$ be a triangulable space,\footnote{See appendix~\ref{simplicialhomology} for details.} $A\subseteq X$ a closed subset, and $f:X\to X$ a proper invertible map with $f(A)\subseteq A$. Then
\begin{align}\label{eq:relativelefschetz}
L(f,X, A) = L(f,X) - L(f,A)\,,
\end{align}
where 
\begin{align}
     L(f,X,A)=\sum_k (-1)^k\, \Tr (H_k(f,X,A))\,.
\end{align}
\end{lemma}
\begin{proof}
For such a pair $X$ and $A$, the following sequence between complexes is exact (cf.~e.g.,~ref.~\cite[p.116]{hatcherAlgTop})
\begin{align}
 0\to  C_k(A)\overset{\iota^k}{\to}C_k(X) \overset{\pi^k}{\to} C_k(X,A)\to 0\,,
\end{align}
with the inclusion $\iota^k:C_k(A) \hookrightarrow C_k(X)$ and the quotient map $\pi^k:C_k(X)\to C_k(X)/C_k(A)$.
For every exact sequence of complexes of the above type, there is a long exact sequence in homology (cf.~e.g.,~ ref.~\cite[p.117]{hatcherAlgTop})
\begin{align}
 0\to  H_{2n}(A)\overset{\iota_*^{2n}}{\to}  \dots &\to H_k(A)\overset{\iota_*^k}{\to}H_k(X) \overset{\pi_*^k}{\to} \\
 &H_k(X,A)\overset{\delta_*^k}{\to}H_{k-1}(A)\to\dots  \overset{\delta_*^1}{\to}H_{0}(A)\to 0\,,\notag
\end{align}
where $\delta_*^k$ is the (relative) boundary map, which assigns to a (relative) cycle the part of its boundary on $A$.

It is easy to show that the map $f$ commutes with the induced maps $\iota_*^k$, $\pi_*^k$ and $\delta_*^k$. 
It then follows from Lemma~\ref{lemma2fromnotes} that 
\begin{align}
  L(f,X,A)&=-L(f,A)+L(f,X)\,.
\end{align}
\end{proof}

\begin{lemma}\label{lemma:compact}
    If in addition to the assumptions from Lemma~\ref{lemma4fromnotes}, $X$ is compact, then
    \begin{align}\label{eqcompactlef}
L(f,X\backslash A)=L(f,X)-L(f,A)\,.
\end{align}
\end{lemma}
\begin{proof}
If $X$ is compact there is the isomorphism 
\begin{align}
    H_{\text{dR},c}^{k}(X\backslash A) \simeq H^{k}_{\text{dR}}(X,A)\,.
\end{align}
Together with Lemma~\ref{lemma4fromnotes}, this implies
\beq
L(f,X\setminus A) = L(f,X,A) = L(f,X)-L(f,A)\,.
\eeq
\end{proof}

\begin{lemma}\label{lemma5fromnotes}
Let $X$ be an oriented, real, even-dimensional manifold, $A,B\subseteq X$ closed subsets, and $f:X\to X$ a proper invertible map of finite order with $f(A)\subseteq A$ and $f(B)\subseteq B$. Then
\beq\bsp\label{compactversionadditivity}
L(f,X\setminus (A\cup B)) &\,= L(f,X\setminus A) + L(f,X\setminus B) - L(f,X\setminus (A\cap B))
 \,.
\esp\eeq
\end{lemma}
\begin{proof}
For two open sets $U$ and $V$, the following Mayer-Vietoris sequence is exact (cf.~e.g.,~ref.~\cite[p.23]{BottTu1982}) 
\begin{align}
   \dots  \to H^k(U\cup V)\overset{\rho^k}\to H^k(U)\oplus H^k(V)\overset{\sigma^k}\to H^k(U\cap V)\overset{\partial^k}\to H^{k+1}(U\cup V)\to \dots\,,
\end{align}
where $\rho^k(\omega)\mapsto (\iota^*_U\omega\,,\iota^*_V\omega)$ and $\sigma^k(\omega_1\,,\omega_2)\mapsto j^*_U\,\omega_1-j^*_V\,\omega_2$, with the inclusions $\iota_U:U\rightarrow U\cup V$, $\iota_V:V\rightarrow U\cup V$ and $j_U:U\cap V\to U$, $j_V:U\cap V\to V$, and the boundary map $\partial^k$, which is discussed below.

Now choose $U=X\backslash A$ and $V=X\backslash B$. Then the sequence reads 
\begin{align}
\label{eq:mayerVietorisAppl}
   \dots  \to H^k(X\backslash(A\cap B))\to H^k(X\backslash A)\oplus H^k(X\backslash B)\to H^k(X\backslash(A\cup B))\to \dots\,.
\end{align}
Since $f$ has finite order, it is easy to see that $f$ fixes the sets $A$ and $B$, $f(A)=A$ and $f(B)=B$. Hence, it also fixes $X\setminus A$, $X\setminus B$, as well as $A\cap B$, $A\cup B$, $X\setminus(A\cap B)$ and $X\setminus(A\cup B)$. As a consequence, the induced maps are well defined on all the cohomology groups that appear in the Mayer-Vietoris sequence.
We now show that $f^*$ commutes with the maps along the sequence \eqref{eq:mayerVietorisAppl}. The fact that $f^*$ commutes with the maps induced from the inclusions is clear, let us hence focus on the boundary map $\partial$.

Let us pick a cohomology class $[\omega]\in H^k(X\backslash(A\cup B))$, where $\omega$ is a closed differential form on $X\backslash(A\cup B)$.
Pick differential forms $\alpha$ and $\beta$ on $ X\setminus A$ and $X\setminus B$, respectively, such that 
\beq
j^*_{X\setminus A}\alpha-j^*_{X\setminus B}\beta=\omega\,.
\eeq
Then, $\rd \alpha-\rd \beta=0$ on $X\backslash(A\cup B)$.
Hence, there is a well-defined $\eta$ on $X\backslash(A\cap B)$, where
\begin{equation}
\label{eq:partial_omega}
    \partial^k(\omega)=\eta:=\left\{ \begin{array}{rl} 
    \rd\alpha & \quad\mathrm{on}\quad X\setminus A\,, \\
    \rd\beta & \quad\mathrm{on} \quad X\setminus B \,,
    \end{array}\right.
\end{equation}
and $\rd\eta=0$.
We can then write
\begin{equation}
\label{eq:partial_rhs}
    f^*\partial^k(\omega)=f^*\eta=\left\{ 
    \begin{array}{rl}
    &f^*_{|_{X\setminus A}}\rd\alpha=\rd f^*_{|_{X\setminus A}}\alpha\quad\mathrm{on}\quad X\setminus A\,,\\
    &f^*_{|_{X\setminus B}}\rd\beta=\rd f^*_{|_{X\setminus B}}\beta\quad\mathrm{on}\quad X\setminus B\,,
\end{array}\right.
\end{equation}
and considering that 
\beq
j^*_{X\setminus A}f^*_{|_{X\setminus A}}\alpha-j^*_{X\setminus B}f^*_{|_{X\setminus B}}\beta=f^*_{|_{X\backslash(A\cup B)}}(j^*_{|_{X\setminus A}}\alpha-j^*_{|_{X\setminus B}}\beta)=f^*_{|_{X\backslash(A\cup B)}}\omega\,,
\eeq
we can now apply the boundary operator $\partial^k$ as in eq.~\eqref{eq:partial_omega} 
\begin{equation}
\label{eq:partial_fomega}
    \partial^k(f^*\omega)=\left\{ \begin{array}{rl} 
    \rd f^*_{|_{X\setminus A}}\alpha & \quad\mathrm{on}\quad X\setminus A\,, \\
    \rd f^*_{|_{X\setminus B}}\beta & \quad\mathrm{on} \quad X\setminus B \,.
    \end{array}\right.
\end{equation}
Comparing eqs.~\eqref{eq:partial_rhs} and~\eqref{eq:partial_fomega}, we can conclude that $f^*\partial^k(\omega)=\partial^k(f^*\omega)$.

Then Lemma~\ref{lemma2fromnotes} implies that 
\begin{align}
0=-L(f,X\backslash(A\cup B)) +L(f,X\backslash A)+L(f,X\backslash B) -L(f,X\backslash(A\cap B))  \,.
\end{align}
\end{proof}

\subsection{Relating Lefschetz numbers and Euler characteristics}

In the previous section we have introduced Lefschetz numbers of a map $f$, and we studied some of their properties. We now discuss how we can actually evaluate the Lefschetz numbers relevant to our scenarios. 

We have already seen that, via the Lefschetz fixed-point theorem, the Lefschetz number of a map $f$ is tightly linked to its fixed points. In fact, more refined versions of the theorem relate the value of $L(f,X)$ to the properties of the fixed points, like the Lefschetz-Hopf theorem in the scenario where $f$ has isolated fixed points. There are also generalizations where the set of fixed points of $f$ in $X$,
\beq
\operatorname{Fix}(f,X) = \big\{x\in X:f(x)=x\big\} \,,
\eeq
is a continuous set. In particular, if $X$ is a smooth closed\footnote{A manifold is closed if it is compact and has no boundary.} manifold, and the fixed-point set is a union of submanifolds, the value of the Lefschetz number can be computed from the Euler characteristics of the submanifolds of fixed points~\cite{frumosu1998mathaiquillenformslefschetztheory}. However, the theorem of ref.~\cite{frumosu1998mathaiquillenformslefschetztheory} does not apply to our setup, because $X=\mathbb{C}^n\setminus \Sigma$ is a manifold, but it is not compact, and thus not closed. In the following we prove a variant of this result which applies to our setting. While our result is very close to existing theorems in the literature (cf.~e.g.,~ref.~\cite{frumosu1998mathaiquillenformslefschetztheory}), we are not aware of any existing theorem which precisely covers the scenario which we are interested in.

\begin{theorem}\label{thm:main}
    Let $\Sigma$ be a codimension-1 variety in $\mathbb{C}^n$, and $f:\mathbb{C}^n\to \mathbb{C}^n$ a proper invertible affine map of finite order such that $f(\Sigma)\subseteq \Sigma$. Let $V_f=\operatorname{Fix}(f,\mathbb{C}^n)$ be the set of fixed-points of $f$ in $\mathbb{C}^n$. Then
    \beq
    L(f,\mathbb{C}^n\setminus\Sigma) = \chi\big(V_f\setminus(V_f\cap \Sigma)\big)\,.
    \eeq
\end{theorem}

\begin{proof} 
The fixed-point set of $f$ in $\mathbb{C}^n$ is some affine subspace $\mathbb{C}^m\subseteq \mathbb{C}^n$ and hence not compact. 
We therefore pass to projective space, and we write $\mathbb{C}^n\backslash \Sigma=\mathbb{CP}^n\backslash \Sigma_p$, where $\Sigma_p$ includes the hypersurface at infinity. It is easy to see that $f$ extends to a projective linear map $\bar{f}:\mathbb{CP}^n\to \mathbb{CP}^n$ that keeps the hyperplane at infinity fixed and whose restriction to the affine chart $\mathbb{C}^n\setminus\Sigma$ is precisely $f$. It follows that $\bar{f}$ keeps $\Sigma_p$ fixed, $\bar{f}(\Sigma_p)\subseteq \Sigma_p$. We then have
\beq
L(f,\mathbb{C}^n\setminus\Sigma) = L(\bar{f},\mathbb{CP}^n\setminus\Sigma_p)\,.
\eeq
Let $\overline{V}_f$ denote the fixed-point set of $\bar{f}$ in $\mathbb{CP}^n$,
\beq
\overline{V}_f = \operatorname{Fix}(\bar{f},\mathbb{CP}^n)\,.
\eeq
Since $\bar{f}$ is a projective linear map, $\overline{V}_f$ is a union of projective subspaces, and thus compact.
Then, since $\Sigma_p$ and $\overline{V}_f$ are closed subsets of $\mathbb{CP}^n$ kept fixed by $\bar{f}$, we may apply Lemma~\ref{lemma5fromnotes},
\beq\bsp
L(\bar{f},\mathbb{CP}^n\setminus(\Sigma_p\cup \overline{V}_f)) &\,= L(\bar{f},\mathbb{CP}^n\setminus\Sigma_p)+L(\bar{f},\mathbb{CP}^n\setminus\overline{V}_f)-L(\bar{f},\mathbb{CP}^n\setminus(\Sigma_p\cap \overline{V}_f))\\
&\,= L({f},\mathbb{C}^n\setminus\Sigma)-L(\bar{f},\overline{V}_f)+L(\bar{f},\Sigma_p\cap \overline{V}_f)
\,,
\esp\eeq
where in the last step we used the fact that $\mathbb{CP}^n$ is compact, so that Lemma~\ref{lemma:compact} applies. Since furthermore $\overline{V}_f$ is compact, we have
\beq\bsp
L(\bar{f},\overline{V}_f)-L(\bar{f},\Sigma_p\cap \overline{V}_f)&\,=L(\bar{f},\overline{V}_f\setminus(\Sigma_p\cap \overline{V}_f))\\
&\,= \chi(\overline{V}_f\setminus(\Sigma_p\cap \overline{V}_f))\\
&\,= \chi({V}_f\setminus(\Sigma\cap {V}_f))\,,
\esp\eeq
where we used the fact that by definition $\bar{f}$ acts as the identity when restricted to its fixed-point set, so that the Lefschetz number reduces to the Euler characteristic (cf.~eq.~\eqref{eq:L_to_chi}).
Putting everything together, we arrive at
\beq\bsp
L(\bar{f},\mathbb{CP}^n\setminus(\Sigma_p\cup \overline{V}_f)) &\,=L({f},\mathbb{C}^n\setminus\Sigma)-\chi({V}_f\setminus(\Sigma\cap {V}_f))
\,.
\esp\eeq
Comparing this to the statement of the theorem, we see that it remains to prove that $L(\bar{f},\mathbb{CP}^n\setminus(\Sigma_p\cup \overline{V}_f))=0$.

To see that this is true, we apply the construction of the compact model from subsection~\ref{sec:compactmodel} to the space $X = \mathbb{CP}^n\setminus(\Sigma_p\cup \overline{V}_f)$. Using the notations from subsection~\ref{sec:compactmodel}, we see that we can embed $X$ into a compact space $\overline{W}$ such that $\pi^{-1}(X) = \overline{X} = \overline{W}\setminus D$, where $D$ is a normal crossing divisor. In subsection~\ref{sec:compactmodel}, we have also argued that we can construct an (open) tubular neighborhood $T_{\delta}$ (cf.~eq.~\eqref{eq:tube}) and a strong deformation retraction $r$ from $\overline{X}$ to $\overline{X}_{\delta} = \overline{W}\setminus T_{\delta}$. Then, since $\pi$ is a bijection on $X$, Lemma~\ref{eq:Lefschetz-retraction} implies that 
\beq
L(\bar{f},X) = L(\widetilde{f},\overline{X}) = L(\widetilde{f},\overline{X}_{\delta})\,.
\eeq
where $\widetilde{f}$ is the lift of $\bar{f}$ to $\overline{W}$ (cf. subsection~\ref{sec:compactmodel}, in particular eq.~\eqref{eq:pi-f}). We now observe that $\widetilde{f}$ has no fixed-points on $\overline{X}$ (and thus also not on $\overline{X}_{\delta}\subset \overline{X})$. Indeed, if there was a fixed-point $x\in \overline{X}$, $x = \widetilde{f}(x)$, then $\pi(x)\in X$ would be a fixed-point of $\bar{f}$, because eq.~\eqref{eq:pi-f} implies
\beq
\pi(x) = \pi\big(\widetilde{f}(x)\big) = \bar{f}\big(\pi(x)\big)\,.
\eeq
But this is impossible, because, by definition, $X$ was obtained by removing the fixed-point set $\overline{V}_{f}$. Finally, $\bar{f}$ is continuous, and so is its lift $\widetilde{f}$. Since $\overline{W}$ is compact and $T_{\delta}$ is open, $\overline{X}_{\delta}$ is compact, and so the Lefschetz Fixed-Point Theorem implies that $L(\widetilde{f},\overline{X}_{\delta})=0$, which finishes the proof.
\end{proof}

\subsection{Euler characteristics and  representations of twisted symmetry transformations}
We now apply the results from the previous section to study the decomposition into irreducible representations of the action of the group $G$ of twisted symmetry transformations on the twisted cohomology group $\cH:=H^n(\mathbb{C}^n\setminus\Sigma,\nabla_{\omega})$. 
Recall that $G$ acts on $\cH$ by pullbacks, and in a basis $\bs{\varphi}$ of $\cH$ the action is given by the matrix representation $\bD_{f}$ in eq.~\eqref{eq:Dsf}.
The character of this representation is as usual defined as the trace of the matrices:
\beq
\chi_{\bD}(f) := \Tr(\bD_{f})\,.
\eeq
For simplicity, let us define $X_f :=V_f\setminus(V_f\cap\Sigma)$ as the fixed-point set of $f$ in $\mathbb{C}^n\setminus\Sigma$. The representation admits a decomposition into irreducible representations as in eq.~\eqref{eq:irrep_dec_1}, and $\cH_R$ denotes the subspace of $\cH$ transforming in the irreducible representation $R$.
It is well known from the theory of finite groups that the characters are enough to entirely describe the decomposition into irreducible representations.

\begin{theorem}\label{thm:character_euler}
If the vanishing theorem holds, then the character is given up to a sign by the Euler characteristic of the fixed-point set:
\beq\label{eq:Euler_to_character}
\chi_{\bD}(f) = (-1)^n\,\chi\big(X_f\big)\,.
\eeq
Moreover, the multiplicity of the irreducible representation $R$ in $\bD$ is 
\beq
m_{R} = \frac{(-1)^n}{|G|}\sum_{f\in G}\chi_R(f)^*\,\chi\big(X_f\big)\,,
\eeq
and the dimension of $\cH_R$ is 
\beq\label{eq:thm_dim_chi}
\dim\cH_R = m_{R}\,d_R= \frac{(-1)^n\,d_R}{|G|}\sum_{f\in G}\chi_R(f)^*\,\chi\big(X_f\big)\,.
\eeq
\end{theorem}
\begin{proof}
If the vanishing theorem holds, we have
\beq\bsp
\chi_{\bD}(f) &\,= \Tr\big(H^n(f,\mathbb{C}^n\setminus\Sigma,\nabla_{\omega})\big) \\&\,= (-1)^n\,L(f,\mathbb{C}^n\setminus\Sigma,\nabla_{\omega}) \\&\,= (-1)^n\,L(f,\mathbb{C}^n\setminus\Sigma)\,,
\esp\eeq
where the last step follows from Proposition~\ref{prop:twistedL_TopL}. The claim then follows from Theorem~\ref{thm:main}.

The projector $\bP_{R}$ onto the irreducible representation $\cH_R$ is given by eq.~\eqref{eq:Reynolds}. We then have
\beq
\dim\cH_R = \Tr\big(\bP_{R}\big)= \frac{d_R}{|G|}\sum_{f\in G}\chi_R(f)^*\,\chi_{\bD}(f)\,,
\eeq
and the claim follows.
\end{proof}

Theorem~\ref{thm:character_euler} is one of the main mathematical results of our paper. Under the assumption that the vanishing theorem applies, it relates the properties of the action of the group of twisted symmetry transformations on the twisted cohomology group, in particular its character and the decomposition into irreducible representations, to topological properties of the fixed-point sets of the individual group elements. We find it remarkable that there is such a simple connection between topological and group-theoretical quantities.

Let us briefly discuss a direct consequence of Theorem~\ref{thm:character_euler} for the decomposition into irreducible representations. We have already seen that the representations $\bD$ and $\bA$ on the cohomology and homology group are equivalent, and the representations $\bs{\check{D}}$ and $\bs{\check{A}}$ on the dual groups are equivalent to the contragredient representation of $\bD$. Under the assumptions where Theorem~\ref{thm:character_euler} holds, we can say more. From eq.~\eqref{eq:Euler_to_character} we see that the character of $\bD$ is real. Since every finite-dimensional representation of a finite group is equivalent to a unitary representation, we get that the contragredient representation of $(\bD^{-1})^T$ is equivalent to the complex conjugate $\bD^*$. this implies
\beq\chi_{(\bD^{-1})^T} = \chi_{\bD^*} = \chi_{\bD}^* =\chi_{\bD}\,,
\eeq
where we used the fact that the trace is invariant under conjugation and the last step follows from the fact that the character of $\bD$ is real. In other words, we see that the contragredient representation $(\bD^{-1})^T$ has the same character as $\bD$. If two representations have the same character, they admit equivalent decompositions into irreducible representations. As a consequence, under the assumptions that Theorem~\ref{thm:character_euler} holds, all four representations  $\bD$, $\bA$, $\bs{\check{D}}$ and $\bs{\check{A}}$ are equivalent.

We have already seen that in the context of families of integrals the determinant representation plays a special role. Before we discuss the implications for families of Feynman integrals in more detail in the next section, we conclude this section by an application of Theorem~\ref{thm:character_euler} to the trivial representation.
Given an $n$-dimensional manifold $X$ and an action of a finite group $G$ on $X$, recall that the \emph{orbifold} $X/G$ is the space that locally looks like the quotient of an open set $U\subset \mathbb{C}^n$ by the group $G$. It is possible to define concepts from topology and geometry for orbifolds. In particular, the Euler characteristic of an orbifold $X/G$ is defined as~\cite{HirzebruchHöfer1990}
\beq\label{eq:orbifold_Euler}
\chi(X/G) = \frac{1}{|G|}\sum_{g\in G}\chi(X_g)\,,
\eeq
where $X_g$ is the set of fixed-points for $g$. It is also possible to define de Rham cohomology groups for $X/G$, and one can show that the (twisted) de Rham cohomology of $X/G$ can be identified with the subspace of the cohomology of $X$ left invariant by $G$, cf.~\cite{Satake1956}. In other words, we have
\beq
H^k(X/G,\nabla_\omega) = H^k(X,\nabla_\omega)^G\,,
\eeq
where $H^k(X,\nabla_\omega)^G$ is the space of $G$-invariants of $H^k(X,\nabla_\omega)$, i.e., the subspace of $H^k(X,\nabla_\omega)$ transforming in the trivial representation of $G$. If the vanishing theorem holds for $H^k(X,\nabla_\omega)$, then the only non-trivial cohomology group has $k=n$. Theorem~\ref{thm:character_euler} then allows us to relate the dimension of this subspace to the Euler characteristics of the fixed-points sets in $X$. Comparing with the definition of the orbifold Euler characteristic in eq.~\eqref{eq:orbifold_Euler}, we see that we have
\beq\label{eq:orbifold_Euler_2}
\dim H^n(X/G,\nabla_{\omega}) = (-1)^n\,\chi(X/G)\,,
\eeq
in agreement with eq.~\eqref{eq:dim_chi}.

\section{The number of master integrals and Euler characteristics}
\label{sec:number_MIs}

In this section we present a formula to compute the number of master integrals in a given sector in the presence of a non-trivial symmetry group. Our result extends the relation between the number of master integrals and the Euler characteristic to include non-trivial symmetry groups.

\subsection{The number of master integrals from group theory and topology}

In this section we apply Theorem~\ref{thm:character_euler} to obtain a formula for the number $N_{\Theta}$ of master integrals in a sector $\Theta$ of a family of Feynman integrals. The Lee-Pomeransky polynomial of the sector is denoted by $\cG$, and, following the notations from section~\ref{sec:single_sector}, we have (see eq.~\eqref{eq:dimV_dimHG})
\beq
N_{\Theta} = \dim_{\cF} \cV_{\gamma_{\cG}}= \dim_{\cF} \cH_G^{\textrm{m.c.}}\,,
\eeq
where $\cH_G^{\textrm{m.c.}}$ is the subspace of $\cH^{\textrm{m.c.}}$ transforming in the determinant representation of the symmetry group $G = \mathbb{G}(\cG)$ of the family. If we assume that the vanishing theorem holds for $\cH^{\textrm{m.c.}}$ (which is typically expected for Feynman integrals), we may apply Theorem~\ref{thm:character_euler} to compute the dimension of $\cH_G^{\textrm{m.c.}}$. Since the determinant representation is one-dimensional, it is irreducible, and equal to its character. Moreover, the elements of $G$ act by permuting the Feynman parameters of the sector, and for every permutation $\sigma\in G$ we have
\beq
\det (\sigma) = \sign(\sigma)\,.
\eeq
Hence, the formula for the dimension in eq.~\eqref{eq:thm_dim_chi} takes the form
\beq\label{eq:NMIs_raw}
N_{\Theta} = \frac{(-1)^P}{|G|}\sum_{\sigma\in G}\sign(\sigma)\,\chi(X_{\sigma})\,,
\eeq
where $X_{\sigma}$ is the set of points in $\mathbb{C}^P\setminus\Sigma$ left invariant by the permutation $\sigma$ of the Feynman parameters. Equation~\eqref{eq:NMIs_raw} is another main result of our paper, and expresses the number of master integrals $N_{\Theta}$ in terms of the Euler characteristics of the fixed-point sets $X_{\sigma}$. In the remainder of this section, we show that we can cast eq.~\eqref{eq:NMIs_raw} in an equivalent form that makes it more suitable for applications.

\subsection{The fixed-point sets} 

In this section we describe in detail the fixed-point sets $X_{\sigma}$ appearing in eq.~\eqref{eq:NMIs_raw}. Let $X = \mathbb{C}^P\setminus \Sigma$. An element $\sigma\in\mathbb{G}(\cG)\subseteq S_P$ acts on affine space $\mathbb{C}^P$ simply by permuting the coordinates, i.e., if $\bs{x}=(x_1,\ldots,x_P)^T$, then
\beq\sigma(\bx) = (x_{\sigma^{-1}(1)},\ldots,x_{\sigma^{-1}(P)})^T\,.
\eeq
This action of $\mathbb{G}(\cG)$ on $\mathbb{C}^P$ is linear, and we have
\beq
\sigma(\bs{x}) = \bP_{\!\sigma}\,\bs{x}  \,,
\eeq
where $\bP_{\!\sigma}$ is a permutation matrix, $(\bP_{\!\sigma})_{ij} = \delta_{i\sigma(j)}$. Since by definition every $\sigma\in\mathbb{G}(\cG)$ leaves the Lee-Pomeransky polynomial invariant, and thus also the twisted variety $\Sigma$, we obtain an action of $\mathbb{G}(\cG)$ on $X$, i.e., we obtain a map $\sigma:X\to X$. The fixed-point sets are then explicitly given by
\beq\label{fixedpointset}
X_{\sigma} = \big\{\bs{x}\in X: \sigma(\bs{x})=\bP_{\!\sigma}\,\bs{x}=\bs{x}\big\}\,.
\eeq
We now discuss various properties of these fixed-point sets that allow us to simplify eq.~\eqref{eq:NMIs_raw}.

First, we show that we do not need to evaluate the Euler characteristic for each $X_{\sigma}$ separately. Recall that two elements $\sigma,\sigma'\in G$ are said to be \emph{conjugate in $G$}  if there is a $\tau\in G$ such that $\sigma'=\tau \sigma\tau^{-1}$. We then write $\sigma \sim\sigma'$. Conjugacy defines an equivalence relation on the elements of a group, and we denote the set of conjugacy classes of $G$ by
\beq
\mathcal{C}_G := \big\{[\sigma]: \sigma\in G\big\}\,,\textrm{~~~with~~~} [\sigma] = \big\{\sigma'\in G: \sigma'\sim \sigma\big\}\,.
\eeq
\begin{lemma}
    If $\sigma,\sigma^\prime\in G$ are conjugate in $G$, then the Euler characteristics of $X_\sigma$ and $X_{\sigma^\prime}$ agree:
    \begin{align}\label{eulerconjugacy}
    \chi(X_\sigma)=\chi(X_{\sigma^\prime})\,,\qquad \textrm{if~~~} \sigma\sim \sigma'.
    \end{align}
\end{lemma}
\begin{proof}
If $\sigma\sim \sigma^\prime$, then there exists a $\tau\in G$ such that $\sigma^\prime=\tau\sigma \tau^{-1}$.
If $\bx^\prime\in X_{\sigma^\prime}$ we can locally write
\begin{align}
\bx^\prime=\bP_{\!\sigma^\prime}\,\bx^\prime=\bP_{\!\tau}\bP_{\!\sigma}\bP_{\!\tau^{-1}}\,\bx^\prime\,.
\end{align}
It follows that
\begin{align}\label{eqinv}
\bP_{\!\sigma}\left(\bP_{\!\tau}^{-1}\,\bx^\prime\right)=\bP_{\!\tau}^{-1}\,\bx^\prime\,. \end{align}
Since the space $X$ is closed under the group action, we know that $\bP_{\!\tau}^{-1}\bx^\prime\in X$. Then, according to eq.~\eqref{eqinv}, we conclude that $\bP_{\!\tau}^{-1}\,\bx^\prime\in X_{\sigma}$. Conversely, if $\bx\in X_\sigma$ then analogously $\bP_{\!\tau}\,\bx\in X_{\sigma^\prime}$. 

The permutation matrix $\bP_{\!\tau}$ defines a homeomorphism from $\mathbb{C}^P$ to itself which is preserved under restriction to $X_{\sigma^\prime}$.
Consequently, the restriction 
\begin{align}
\bP_{\!\tau}^{-1}|_{X_{\sigma^\prime}}:X_{\sigma^\prime}\rightarrow   X_{\sigma}\,,
\end{align}
is a global homeomorphism between $X_{\sigma}$ and $X_{\sigma^\prime}$.
Since the fixed-point sets $X_\sigma$ and $X_{\sigma^\prime}$ are homeomorphic, their Euler characteristics are equal.
\end{proof}

Using the previous lemma, and the fact that the signature of a permutation is constant on conjugacy classes, we see that we can reduce the sum in eq.~\eqref{eq:NMIs_raw} to a sum over the conjugacy classes of $G$,
\beq\label{eq:NMIs_2}
N_{\Theta} = \frac{(-1)^P}{|G|}\sum_{\sigma\in \mathcal{C}_G}\,n_{[\sigma]}\,\sign(\sigma)\,\chi(X_{\sigma})\,,
\eeq
where $n_{[\sigma]}=\#[\sigma]$ is the cardinality of the conjugacy class $[\sigma]$, i.e., the number of distinct elements conjugate to $\sigma$. Since the number of conjugacy classes is typically much lower than the number of elements of the group, this often considerably reduces the number of Euler characteristics that need to be computed.

Let us now discuss what exactly the fixed-point sets are. For $\sigma=\id$, we obviously have $X_{\id}=X$. For $\sigma\neq \mathrm{id}$, we proceed by parameterizing the fixed-point set $V_{\sigma}$ of $\sigma$ in  $\mathbb{C}^P$ and we will then intersect it with $X$ to obtain 
the fixed-point set  $X_\sigma = V_\sigma\cap X$.

The spaces $V_{\sigma}$ can be described explicitly. Recall that every permutation $\sigma\in S_P$ admits a decomposition into a product of disjoint cycles, and this decomposition is unique up to the order of terms in the product. Here a \emph{cycle of length $k$} is a permutation $c$ such that
\beq
c : i_{1}\to i_{2}\to\ldots \to i_{k} \to i_{1}\,,
\eeq
while all other elements remain fixed. Every permutation $\sigma$ admits a representation as a product of cycles, $\sigma=c_1\cdots c_r$. We denote the number $r$ of cycles in $\sigma$ by $c(\sigma)$. It is easy to see that $V_{\sigma}$ is the subspace of $\mathbb{C}^P$ where we identify the coordinates that lie on the same cycle $c_k$. In particular, it follows that the dimension of $V_{\sigma}$ is given by the number of cycles in $\sigma$, 
\beq
\dim X_{\sigma} = \dim V_{\sigma} = c(\sigma)\,.
\eeq
We then see from eq.~\eqref{eq:sign_chi} that, if the vanishing theorem holds for the twist restricted to the locus $X_{\sigma}$, then
\beq
\sign(\chi(X_{\sigma}))=(-1)^{c(\sigma)}\,,
\eeq
so that 
\beq\label{eq:final_1}
(-1)^{c(\sigma)}\,\chi(X_{\sigma})=|\chi(X_{\sigma})|\,.
\eeq
Note that we have the relation
\beq\label{eq:final_2}
\sign(\sigma) = (-1)^{P-c(\sigma)}.
\eeq
We see that in the context of the Lee-Pomeransky representation, we have a very concrete description of the fixed-point sets $X_{\sigma}$.

\subsection{The final formula for the number of master integrals}
We are now in a position to state the main formula. Let us stress again our two assumptions:
\begin{enumerate}
    \item the vanishing theorem applies to the twisted cohomology group for $X$ and all the fixed-points sets $X_{\sigma}$.
    \item There are no additional relations that reduce the number of master integrals, like relations coming from supersectors or non-linear transformations like those considered in refs.~\cite{Buccioni:2023okz,Becchetti:2025rrz,Coro:2025vgn}. 
\end{enumerate}
Then,
inserting eqs.~\eqref{eq:final_1} and~\eqref{eq:final_2} into~eq.~\eqref{eq:NMIs_2}, we obtain our final formula for the number of master integrals in the sector $\Theta$,\footnote{We recall eq.~\eqref{eq:N_Theta} for our definition of the number of master integrals in a sector.}
\beq\label{eq:NMIs_3}
N_{\Theta} =  \frac{1}{|G|}\sum_{\sigma\in G}\,|\chi(X_{\sigma})|=\frac{1}{|G|}\sum_{\sigma\in \mathcal{C}_G}\,n_{[\sigma]}\,|\chi(X_{\sigma})|\,.
\eeq
In other words, the number of master integrals is obtained by taking the average over the whole group $G$ of the absolute values of the Euler characteristics of the fixed-points sets. 

Let us make some comments. First, 
if our two assumptions fail, then there is no guarantee that eq.~\eqref{eq:NMIs_3} computes the number of master integrals in the sector $\Theta$. We stress that these assumptions are similar to the assumptions that underlie the original result of refs.~\cite{Lee:2013hzt,Bitoun:2017nre,bitoun2018numbermasterintegralseuler}
valid in the absence of symmetries. In fact, eq.~\eqref{eq:NMIs_3} reduces to the result of 
refs.~\cite{Lee:2013hzt,Bitoun:2017nre,bitoun2018numbermasterintegralseuler} 
if $G=\{\id\}$ is trivial: in that case the whole space is invariant, 
$X_{\id}=X$, and we obtain $N_{\Theta} = |\chi(X)|$. 
Second, note that eq.~\eqref{eq:NMIs_3} is strikingly 
similar to the formula for the orbifold Euler characteristic in eqs.~\eqref{eq:orbifold_Euler} and~\eqref{eq:orbifold_Euler_2}, and they only differ by the overall sign factor and the absolute value in the sum. This distinction, however, is crucial: without the absolute value, we obtain the dimension of the invariant subspace of the twisted cohomology group. As argued in section~\ref{sec:twistedsymparam}, the number of master integrals, instead, is related to the subspace transforming in the determinant representation.

Finally, let us point out that the fact that only the absolute values of the Euler characteristics enter eq.~\eqref{eq:NMIs_3} has an important practical implication.
To see this, recall that for the identity element, the relevant Euler characteristic $|\chi(X_{\mathrm{id}})|$ in eq.~\eqref{eq:NMIs_2} is just the absolute value of the Euler characteristic of $X$. Under some assumptions there is a clear relation between the Euler characteristic and critical point counting~\cite{Huh_2013}. However, even if these assumptions are not satisfied it is observed that the Euler characteristic $\chi(X)$ can often be computed by the number of critical points of $X$, i.e. solutions to the equations (cf. e.g., ref.~\cite{Lee:2013hzt})
\beq\label{critpointcount}
\partial_{x_i}\log\cG(\bs{x},\bs{s}) = 0\,,\qquad i=1,\ldots,P\,.
\eeq
We emphasize that there are limitations to this approach. For example the critical points have to be isolated. Further, there are cases known in the literature where eq.~\eqref{critpointcount} under- or overcounts the number of master integrals, cf.~e.g., refs.~\cite{Boels:2015yna,Bitoun:2017nre} (for a summary of known failures of this formula cf.~e.g., ref.~\cite{Chestnov:2026mpo}).
We then have (cf.~e.g., ref.~\cite{Lee:2013hzt} for an application in the context of Feynman integrals)
\beq
|\chi(X)| = \#(\textrm{critical points of $X$})\,.
\eeq
We now argue that also for the remaining elements $\sigma\neq\id$ the computation of the Euler characteristics can be reduced to counting numbers of critical points (assuming that the critical points form a finite set).

We start by providing an explicit parametrization of the fixed-point set of $V_{\sigma}\subseteq\mathbb{C}^P$.
Note that $V_{\sigma}$ is spanned by the eigenvectors of $\bP_{\!\sigma}$ with eigenvalue $1$:
\begin{align}\label{invsubspaceC}
V_\sigma=\mathrm{Ker}(\bP_{\!\sigma}-\mathds{1})=\mathrm{Span}\{\bs{v}_{\sigma}^1,\dots,\bs{v}_{\sigma}^m\}\,.
\end{align}
A possible parametrization of $V_\sigma$ is a homeomorphism $f_\sigma$ with 
\begin{align}\label{invsubspaceparam}
f_\sigma:\mathbb{C}^m\rightarrow V_\sigma,\quad \bs{y} \mapsto  \sum_{i=1}^{m}y_i\,\bs{v}_{\sigma}^i  \,,
\end{align}
such that $V_\sigma=\mathrm{Im}(f_\sigma)$. 
Let us define the pullback to the fixed-point locus  
\begin{align}\label{twistpullback}
\mathcal{G}_\sigma:\mathbb{C}^m\rightarrow \mathbb{C}\,, \text{~~~with~~~} \mathcal{G}_\sigma(\bs{y}):=\mathcal{G}(f_\sigma(\bs{y}))\,.
\end{align}
We may then describe the fixed-point set $X_{\sigma}$ as
\begin{align}
X_\sigma=\big\{f(\bs{y}): \bs{y}\in \mathbb{C}^m,\, \mathcal{G}_\sigma(\bs{y})\neq 0\big\}\,.
\end{align}
Since $f_\sigma$ is a homeomorphism it follows that 
\begin{align}
X_\sigma\simeq \big\{\bs{y}\in \mathbb{C}^m:\mathcal{G}_\sigma(\bs{y})\neq0\big\}=\mathbb{C}^m\backslash \{\mathcal{G}_\sigma=0\}\,.
\end{align}
We know how to compute the Euler characteristic of such a space by counting critical points of 
\begin{align}
\partial_{y_i}\log(\mathcal{G}_\sigma(\bs{y}))=0\,, \quad i\in \{1,\dots,m\}\,.
\end{align}
If the fixed-point set for a given group element is a single isolated point,  the Euler characteristic is $\chi(\mathrm{point})=1$.  

The previous considerations show that, just like in the case of a trivial symmetry group, we can reduce the problem of computing the number of master integrals to counting the number of critical points. This problem can be addressed using modern computer algebra packages.
A \textsc{Mathematica} package based on the SP$\mathbb{Q}$R package~\cite{ChestnovCrisanti2025} to compute the number of master integrals taking into account possible symmetries is currently being developed~\cite{packagemasters}.

\section{Conclusions}
\label{sec:conclusion}
In this paper we performed a comprehensive study of certain classes of discrete symmetries that naturally arise in the context of multiloop Feynman integrals in dimensional regularization. From a mathematical perspective, these symmetries are most appropriately described in the context of twisted cohomology theories as invertible affine maps that leave the twist invariant and map different irreducible components of $D_+$ and $D_-$ into each other. These irreducible components provide a natural notion of sectors, which generalizes the one known from Feynman integrals. The symmetry transformations naturally organize into a groupoid whose objects are the sectors and the morphisms are the affine maps. The affine maps then define linear maps on the twisted homology and cohomology groups and their duals via pullbacks and pushforwards, and these linear maps naturally preserve the pairings between these groups. If we focus on the subgroupoid of transformations that leave a specific twisted cycle $\gamma$ setwise invariant (but they may change the orientation of the cycle), then we obtain a groupoid of symmetry transformations defined on the vector space generated by all twisted periods obtained by integrating twisted cocycles over this fixed cycle $\gamma$.

If we further focus on families of Feynman integrals defined in loop-momentum space, then this notion of symmetry transformation agrees with the corresponding notion introduced for Feynman integrals in the literature, cf., e.g.,~refs.~\cite{Lee:2012cn,Lee:2013mka,wu2024neatibp10packagegenerating,Wu:2024paw}. However, there is no unique twisted cohomology theory that one may associate with a family of Feynman integrals, and so the question naturally arises if and how the set of symmetry transformations of a family of Feynman integrals depends on the choice of twisted cohomology theory used to define the family. Remarkably, we find that for the loop-momentum, (democratic) Baikov, Lee-Pomeransky and Feynman parameter representations, the resulting groupoid of symmetry transformations has a factorized form, where one factor depends on the integral representation and acts trivially on the integrals (at least for an appropriately chosen basis of cocycles), while the other is universal and can be identified with the set of permutations of the Feynman parameters that map the Lee-Pomeransky polynomials of the two sectors into each other. We explicitly showed how we can lift such a permutation to a symmetry transformation in the loop-momentum or Baikov representations. In particular, the resulting affine maps in loop-momentum space admit very natural interpretations in terms of operations on the underlying Feynman graphs as certain isomorphisms between the graphical matroids. Moreover, our construction naturally explains the appearance of kinematics-dependent symmetry transformations between subsectors recently discussed in the literature~\cite{Wu:2024paw}. 

If we focus on the symmetry transformations from a sector to itself, we obtain a finite group acting on the twisted cohomology group. A substantial part of our paper consisted in applying tools from the representation theory of finite groups to study the decomposition into irreducible representations. One of our main mathematical results is the fact that, under certain assumptions, the character of the representation is given by the Euler characteristic of the corresponding fixed-point set, up to a sign. The proof of this result heavily relies on tools from algebraic topology. We have also studied the implications of our results for families of Feynman integrals. We find that there is a connection between canonical bases for families of Feynman integrals and the group of symmetry transformations in a sector. In particular, the canonical intersection matrix of a sector can be chosen block-diagonal, even in the region of the parameter space where the symmetry is not realized. Finally we have derived a formula which relates the number of master integrals in a sector to the average over the group of the absolute values of the Euler characteristics of the fixed-point sets. The absolute values of the Euler characteristics can efficiently be evaluated by counting numbers of critical points. Our result generalizes a previously known result that relates the number of master integrals in a sector to the Euler characteristic or the number of critical points in the absence of symmetries~\cite{Lee:2013hzt,bitoun2018numbermasterintegralseuler,Mastrolia:2018uzb}.

We have identified various directions for future research. First, the connection between the Euler characteristic and the number of master integrals relies on certain assumptions, like the fact that the vanishing theorem from twisted cohomology theory applies. It would be interesting to better understand when this theorem applies, and how the formula for the master integrals needs to be extended if it does not. Second, by relating the characters of the representation to the Euler characteristics of the fixed-point sets, our analysis has revealed a connection between the topological properties of the underlying space and the representation theory of finite groups. This analysis, however, was restricted to the case of finite groups, which was sufficient to study our discrete symmetries. It would be interesting to investigate if the connection between topology and representation theory can be extended to other classes of symmetries, like Lie groups of Yangian symmetries. Finally, while in this work we have focused on affine changes of variables, it would be interesting to see if it is possible incorporate also non-linear changes of variables, like those recently considered in ref.~\cite{Buccioni:2023okz,Becchetti:2025rrz,Coro:2025vgn}.

\acknowledgments

We are grateful to Giulio Crisanti, Franziska Porkert and Yoann Sohnle for collaborations on related topics, and Giulio Crisanti and Einan Gardi for interesting and useful discussions. We are particularly thankful to Erik Panzer for suggesting to compute the number of master integrals in the presence of a symmetry in terms of the Euler characteristics of the corresponding fixed point sets.
Furthermore, SM, CS and SFS would like to thank the Erwin Schrödinger International Institute for Mathematics and Physics (ESI), University of Vienna (Austria), for the opportunity to participate in the Thematic Programme “Amplitudes and Algebraic Geometry" in 2026 where a part of this work has been carried out and for the support given. Finally CS would like to thank the organizers of the workshop "MathemAmplitudes" (2025) and the Mainz Institute for Theoretical Physics (MITP), where parts of this work have been carried out, for the invitation to the workshop and the possibility to present parts of this work there.
This work is supported in part by the the Deutsche Forschungsgemeinschaft (DFG, German Research Foundation) under Germany’s Excellence Strategy – EXC 3107 – Project-ID 53376636 and the CRC
1639 “NuMeriQS” (CS), and by the European Union (ERC Consolidator Grant LoCoMotive
101043686) (CD, SM, SFS). Views and opinions expressed are however those of the
author(s) only and do not necessarily reflect those of the European Union or the European
Research Council. Neither the European Union nor the granting authority can be held
responsible for them.

\begin{appendix}

\section{A brief review of matroid theory}
\label{app:matroidIntro}

In this appendix we will review some basic facts from matroid theory, with the goal of stating Whitney's classification of matroid isomorphisms in a precise way. 

A \emph{matroid} is a tuple $(E,I)$, where $E$ is some set, the \emph{ground set}, and $I$ is a set of subsets of $E$ such that
\begin{enumerate}
\item $\emptyset \in I$,
\item For all $A,B\subset E$ with $A\subset B$: if $B\in I$ then $A\in I$,
\item If $A,B\in I$ and $|A|>|B|$, then there is an $x\in A\setminus B$ such that $B\cup\{x\}\in I$,
\end{enumerate}
These properties abstract the notion of linear independence, and the elements of $I$ are referred to as the \emph{independent sets} in $E$. This connection can be made explicit through the notion of a \emph{representation} of a matroid. A representation of a matroid is a family of vectors in a vector space, typically given as the columns of a matrix, such that the linear independence relations of these vectors are precisely those of the matroid. Not every matroid admits a representation, and those that do are referred to as \emph{linear} matroids. Another important concept that we will need is that of an \emph{isomorphism} of matroids. Two matroids $M_1=(E_1,I_1)$ and $M_2=(E_2,I_2)$ are isomorphic if there is a bijection $\phi:E_1\rightarrow E_2$ such that $A\in I_1$ if and only if $\phi(A)\in I_2$. 

An important fact about matroid theory is that matroids can be characterized in many different ways. Here we will review three equivalent characterizations which will be of use to us. To this end, we need to introduce some further terminology. Firstly, a \emph{basis} of a matroid is an independent set $B\in I$ that becomes dependent if any element of $E$ is added. We will refer to the set of bases of the matroid by $\mathcal{B}$. Similarly we define a \emph{circuit} of a matroid as a dependent set that becomes independent after removing any element. We will refer to the set of circuits by $\mathcal{C}$. Finally, a \emph{cocircuit} of a matroid is a set $C^*$ which has nontrivial intersection with every basis, but fails to do so if any element of $C^*$ is removed. We will denote the set of cocircuits by $\mathcal{C}^*$.

The set of bases, the set of circuits as well as the set of cocircuits individually already fully determine the matroid, i.e., we can recover the entire set $I$ from either of these. In particular, we can formulate an isomorphism of two matroids as a bijection $\phi$ on the ground sets such that $A$ is a basis if and only if $\phi(A)$ is a basis, and similarly for circuits or cocircuits. Said otherwise, matroid isomorphisms necessarily preserve bases and (co)circuits.

We will actually not be interested in arbitrary matroids, but we will consider \emph{cycle matroids} or \emph{graphical matroids}. These are associated with a graph $G$ and are defined by the tuple of the set of edges $E_G$ and the set of all forests (subsets of edges without cycles) of $G$ as the set of independent sets. We will refer to this matroid by $M(G)$. In this case, the bases of $M(G)$ correspond to the spanning trees, the circuits correspond to the cycles and the cocircuits correspond to minimal cuts, i.e., minimal sets of edges disconnecting the graph. In particular, one equivalent way of defining the independent sets of $M(G)$ is the union of the spanning trees with all of their subsets. Further note that every cycle matroid is linear (indeed it is \emph{regular}, i.e., representable over any field). A possible representation of it is for example given by the incidence matrix of the graph.

There is an important relation between the isomorphisms of cycle matroids and isomorphisms between the associated graphs. This is captured by \emph{Whitney's $2$-Isomorphism Theorem}, which states that two cycle matroids $M(G_1),M(G_2)$ are isomorphic if and only if the graphs $G_1,G_2$ are $2$-isomorphic \cite{whitneyTheorem}. Two graphs are $2$-isomorphic if one can be obtained from the other by a graph isomorphism  of by a sequence of the following moves
\begin{enumerate}
    \item \emph{Vertex identification:} Identifying two vertices of two disconnected components.
    \item \emph{Vertex cleaving:} Disconnecting the graph at a vertex, i.e., performing the inverse operation to vertex identification.
    \item \emph{(Whitney) twisting:} Disconnecting the graph at two vertices and reidentifying in the opposite way.
\end{enumerate}
Note that if we focus on connected graphs, then only the Whitney twist is a valid operation and we can hence in this case understand all isomorphisms between cycle matroids as isomorphisms between the underlying graphs or Whitney twists.

\section{Proof of Lemma \ref{lemmaLPToSymm}}
\label{sec:proofSymLemma}
Let us now prove Lemma \ref{lemmaLPToSymm}, which was needed to establish eq.~\eqref{eq:S_T1_T2}. Most of the graph-theoretic ingredients used in this section are standard. We nevertheless include complete proofs for the convenience of the reader and because we could not find them stated in the precise form required for our application to Feynman integrals. First, let us associate a \emph{signed permutation matrix}
\begin{equation}\label{signedpermutation}
    (\bP_{\!\sigma})_{ij}=\kappa_j\delta_{i\alpha(j)} \,,
\end{equation}
to the signed bijection $\sigma=(\bkappa,\alpha)\in \mathbb{Z}_2^P\times \mathbb{S}(\cG_1,\cG_2)$.
We will use this matrix to explicitly construct a symmetry transformation that implements this (signed) bijection in loop-momentum space. Note that the result was already given in eq.~\eqref{eq:symmMatricesFromAut_sum}. 
There we allowed for general parametrizations of the edge momenta, cf.~eq.~\eqref{eq:q12ParamGeneral}. For the proof it will be convenient to change variables to the orthogonal parametrization described in section \ref{sec:edgeMomParams}. Explicitly, let us shift the loop momenta as
\begin{equation}
    \bk= \tilde{\bk}-\bGamma_i\bp \,,
\end{equation}
with
\begin{equation}
    \bGamma_i=(\bC_i\bC_i^T)^{-1}\bC_i\bE_i^T \,,
\end{equation}
such that
\begin{equation}
    \bq_i(\bk(\tilde{\bk},\bp),\bp)=\bC_i^T\tilde{\bk}+\bD_i^T\bp \,,
\end{equation}
with
\begin{equation}
\label{eq:DiDef}
    \bD_i=\bE_i\bPi_{C_i}^{\perp} \,.
\end{equation}
for $i=1,2$, with $\bPi_{C_i}^{\perp}$ defined as in eq.~\eqref{eq:projectorDef}. In this parametrization we now have $\bD_i\bC_i^T=0$. We define the matrices 
\begin{align}
\label{eq:LMNfromIsomorph}
    \widetilde{\bL}_{\sigma}&=(\bC_1\bC_1^T)^{-1}\bC_1\bP_{\!\sigma}\bC_2^T \,, \nonumber \\
    \widetilde{\bM}_{\sigma}&=(\bC_1 \bC_1^T)^{-1}\bC_1\bP_{\!\sigma}\bD_2^T \,, \\
    \widetilde{\bN}_{\sigma}&=(\bD_1 \bD_1^T)^{-1}\bD_1\bP_{\!\sigma}\bD_2^T \,. \nonumber
\end{align}
It will be useful to characterize the images of the matrices $\bC_i^T$ and $\bD_i^T$.\footnote{By the image of a matrix we mean the image of the linear map that is defined by the matrix.} To this end, let us define the $\mathbb{Q}$-vector space
\begin{equation}
\label{eq:QDef}
    Q_{G_i}=\mathrm{Im}\,\bC_i^T + \mathrm{Im}\, \bE_i^T \,,
\end{equation}
in an arbitrary parametrization as in eq.~\eqref{eq:q12ParamGeneral}. 
This is precisely the subspace of $E_{G_i,\mathbb{Q}}^{\rint}=\mathbb{Q}\otimes E_{G_i}^{\rint}$ spanned by the edge-momenta in the graph $G_i$. Let us further characterize $Q_{G_i}$. 
\begin{lemma} Let $G$ be an arbitrary Feynman graph associated with a momentum-grouped sector $\Theta$. Let $\bC,\,\bE$ be the matrices in the edge-momentum functions of an arbitrary parametrization in this sector and define $Q_G$ as in eq.~\eqref{eq:QDef}.
Then, the following statements hold.
    \begin{enumerate}[label=(\alph*)]
        \item The basis elements of $\mathrm{Im}\,\bC^T$ and $\mathrm{Im}\,\bE^T$ are linearly independent over $\mathbb{Q}$. In particular the sum in eq.~\eqref{eq:QDef} is direct.
        \item The space $Q_G$ is independent of the chosen parametrization.
        \item We can characterize $Q_G\subset E_{G,\mathbb{Q}}^{\rint}$ as the subspace that satisfies momentum conservation at all internal vertices not connected to an external momentum. More precisely, $Q_G$ is the subspace spanned by vectors $\bq\in E_{G,\mathbb{Q}}^{\rint}$ satisfying (cf.~eq.~\eqref{eq:incidenceMatrix})
        \begin{equation}
        \label{eq:QCharCond}
            (\bB\bq)_v=0 \,,
        \end{equation}
        for $v\in V_{G}^{\rint}$ not connected to an external vertex.
    \end{enumerate}
\label{lemma:C_E}    
\end{lemma}
\begin{proof}
Let us address the various points in turn.
\begin{enumerate}[label=(\alph*)]
    \item Note that a general linear combination of the basis elements of $\mathrm{Im}\,\bC^T$ and $\mathrm{Im}\,\bE^T$ precisely takes the form of the edge-momentum function $\bq(\bk,\bp)$ corresponding to the chosen parametrization. Let us now assume that these basis elements are not linearly independent. Then there exists a non-trivial choice of $\bk,\bp$ such that all edge-momenta vanish. However, momentum conservation at the internal vertices connected to external ones immediately implies $\bp=0$, because we are considering a momentum-grouped sector. Then the requirement reduces to $\bC^T\bk=0$, which implies $\bk=0$, because $\bC$ is of rank $L$. Hence the basis elements are linearly independent.
    \item Let us consider the different possible changes in the parametrization. First we could change our basis of loop momenta or change our choice of independent external momenta. This simply corresponds to a multiplication of $\bC$ or $\bE$ from the left by an element of $\mathrm{GL}(L,\mathbb{Q})$ or $\mathrm{GL}(E,\mathbb{Q})$, respectively. This, however, does not change the image of $\bC^T$ or $\bE^T$, as it simply corresponds to a change of basis in the domain. Finally we could shift the loop momenta by a linear combination of the external momenta, which leads to a shift 
    \begin{equation}
        \bE^T\rightarrow\bE^T+\bC^T\bS \,,
    \end{equation}
    for some $\bS\in\mathbb{Q}^{L\times E}$. This clearly leaves $Q_G$ invariant, because the image of $\bC^T\bS$ is a subspace of $\mathrm{Im}\,\bC^T$.
    \item It is clear that our momentum parametrization has to satisfy momentum conservation at all vertices. Hence, we only need to show that the conditions in eq.~\eqref{eq:QCharCond} fully characterize the subspace $Q_G\subset E_{G,\mathbb{Q}}^{\rint}$. We do this by comparing dimensions. Imposing the conditions in eq.~\eqref{eq:QCharCond} leads to a subspace of dimension 
    \begin{equation}\label{dimQ}
        P-(V-V_{\rext})=L+V_{\rext}-1 \,,
    \end{equation}
    using the relation
    \begin{equation}
        L=P-V+1 \,,
    \end{equation}
    for connected graphs. This precisely matches the dimension of $Q_G$, since 
    \begin{equation}
        \dim Q_G=\mathrm{rank}\,\bC+\mathrm{rank}\,\hat{\bD} =L+V_{\rext}-1 \,,
    \end{equation}
    making use of (b) to evaluate the dimension in a pure loop-momentum representation.
\end{enumerate}
\end{proof}
Using these basic properties we can now understand how the bijection $\bP_{\!\sigma}$ acts on $Q_{G_1}$ and $Q_{G_2}$.
\begin{lemma}
\label{lemmaIsomorphOnQ}
    The bijection defined by the map $\bP_{\!\sigma}$ (cf.~eq.~\eqref{signedpermutation}) satisfies the following properties:
    \begin{enumerate}[label=(\alph*)]
    \item It maps $H_1(G_2,\mathbb{Z})$ to $H_1(G_1,\mathbb{Z})$.
    \item It maps $Q_{G_2}$ to $Q_{G_1}$.
    \end{enumerate}
\end{lemma}
\begin{proof}
     First note that one can characterize the rows of the incidence matrix $\bB$ as the (signed) edges making up cocircuits of the graph, which are minimal sets of edges that need to be cut to disconnect the graph \cite{tutteNotes}. As reviewed in appendix \ref{app:matroidIntro}, matroid isomorphisms map cocircuits to cocircuits, and hence preserve the row space of the incidence matrix. We hence know that there is some matrix $\bM\in\mathrm{GL}(V,\mathbb{Z})$ implementing this map (and preserving orientations) such that
        \begin{equation}
        \label{eq:incMatrixTransf}
            \bB_1\bP_{\!\sigma}=\bM\bB_2 \,.
        \end{equation}
    \begin{enumerate}[label=(\alph*)]
        \item From eq.~\eqref{eq:incMatrixTransf} it follows that
        \begin{equation}
            \bB_1\bP_{\!\sigma}\bC_2^T=\bM\bB_2\bC_2^T=0 \,,
        \end{equation}
        i.e., the basis elements of $H_1(G_2,\mathbb{Z})$ (the rows of $\bC_2$) are mapped to the kernel of $\bB_1$, which is $H_1(G_1,\mathbb{Z})$. The other direction follows similarly. This just reproduces the elementary fact that matroid isomorphisms map cycles to cycles.
        \item The edge momenta in $G_1$ satisfy
        \begin{equation}
            \bB_1\bq_1=\bSigma_1 \,,
        \end{equation}
        where $\bSigma_1$ is the vector of inflowing external momenta on all internal vertices of $G_1$. Restricting this condition to rows corresponding to vertices not connected to an external momentum, this condition reduces to eq.~\eqref{eq:QCharCond}. From eq.~\eqref{eq:incMatrixTransf} it follows that
        \begin{equation}
            \bB_1\bP_{\!\sigma}\bq_2=\bM\bSigma_2 \,,
        \end{equation}
        since $\bq_2\in Q_{G_2}$. Thus, $\bP_{\!\sigma}\bq_2 \in Q_{G_1}$ if and only if 
        \begin{equation}
        \label{eq:QtoQProofCond}
            (\bM\bSigma_2)_v=0 \,,
        \end{equation}
        for all $v\in V_{G_1}^{\rint}$ not connected to an external vertex. Recall that matroid isomorphisms are either graph isomorphisms or Whitney twists. For graph isomorphisms, the matrix $\bM$ is a permutation matrix corresponding to the vertex bijection. Hence eq.~\eqref{eq:QtoQProofCond} is trivially satisfied since vertices not connected to an external vertex are mapped to such vertices. For Whitney twists this furthermore follows from Lemma \ref{lemmaWhitneyKin}, since there can be no total momentum flowing into the component of the graph affected by the twist.
    \end{enumerate}
\end{proof}
Using the previous two lemmas it is now simple to derive some important properties of the matrices defined in eq.~\eqref{eq:LMNfromIsomorph}.
\begin{lemma} 
The matrices $\widetilde{\bL}_{\sigma},\widetilde{\bM}_{\sigma},\widetilde{\bN}_{\sigma}$ in eq.~\eqref{eq:LMNfromIsomorph} satisfy the following properties.
    \begin{enumerate}[label=(\alph*)]
        \item $\bC_1^T\widetilde{\bL}_{\sigma}=\bP_{\!\sigma}\bC_2^T$. 
        \item $\bC_1^T\widetilde{\bM}_{\sigma}+\bD_1^T\widetilde{\bN}_{\sigma}=\bP_{\!\sigma}\bD_2^T$.
    \end{enumerate} 
\end{lemma}
\begin{proof}
We will prove the two statements in turn.
\begin{enumerate}[label=(\alph*)]
    \item The statement we want to prove is equivalent to
    \begin{equation}
        \bPi_{C_1}\bP_{\!\sigma}\bC_2^T=\bP_{\!\sigma} \bC_2^T \,.
    \end{equation}
    We have seen above that $\bP_{\!\sigma}$ maps $H_1(G_2,\mathbb{Z})$ to $H_1(G_1,\mathbb{Z})$, and hence the projector $\bPi_{C_1}$ simply acts as the identity, proving the claim.
    \item The statement we want to prove is equivalent to
    \begin{equation}
        (\bPi_{C_1}+\bPi_{D_1})\bP_{\!\sigma}\bD_2^T=\bP_{\!\sigma}\bD_2^T \,,
    \end{equation}
    where $\bPi_{D_1}$ is the projector onto  $\mathrm{Im}\,\bD_1^T$, defined analogously to eq.~\eqref{eq:projectorDef}. Hence, the combination $(\bPi_{C_1}+\bPi_{D_1})$ simply projects onto $Q_{G_1}$, and since $\bP_{\!\sigma}$ maps $\mathrm{Im}\,\bD_2^T\subset Q_{G_2}$  into $Q_{G_1}$, the claim follows.
\end{enumerate}
\end{proof}
\noindent We can now use the matrices in eq.~\eqref{eq:LMNfromIsomorph} to define a transformation in the first sector 
\begin{equation}
\label{eq:symmTransfProof}
    \tilde{\bk}=\widetilde{\bL}_{\sigma}\tilde{\bk}'+\widetilde{\bM}_{\sigma}\bp', \qquad \bp=\widetilde{\bN}_{\sigma} \bp' \,.
\end{equation}
It is then easy to see that
\begin{equation}
    \tilde{\bq}_1(\tilde{\bk},\bp)=\bP_{\!\sigma} \tilde{\bq}_2(\tilde{\bk'},\bp') \,,
\end{equation}
or written in components
\begin{equation}
    \tilde{q}_{1,\alpha(i)}(\tilde{\bk},\bp)=\kappa_i\tilde{q}_{2,i}(\tilde{\bk}',\bp') \,,
\end{equation}
where 
\begin{equation}
    \tilde{\bq}_i(\tilde{\bk},\bp)=\bq_i(\bk(\tilde{\bk},\bp),\bp) \,, \quad i=1,2 \,.
\end{equation}
We can now go back to the original parametrizations in eq.~\eqref{eq:q12ParamGeneral} by changing variables in the two sectors as
\begin{equation}
    \tilde{\bk}'= \bk'+\bGamma_i\bp' \,,
\end{equation}
which satisfies
\begin{equation}
    \tilde{\bq}_i(\tilde{\bk}',\bp')=\bq_i(\bk',\bp') \,.
\end{equation}
Putting the transformations together, we end up with the transformation in sector $\Theta_1$
\begin{equation}
\label{eq:fullTransfLemmaProof}
    \bk=\widetilde{\bL}_{\sigma}\bk'+(\widetilde{\bM}_{\sigma}+\widetilde{\bL}_{\sigma}\bGamma_2-\bGamma_1\widetilde{\bN}_{\sigma})\bp' , \qquad \bp=\widetilde{\bN}_{\sigma}\bp' \,.
\end{equation}
By construction, it satisfies
\begin{equation}
    \bq_1(\bk,\bp)=\bP_{\!\sigma}\bq_2(\bk',\bp') \,.
\end{equation}
We see that the transformation in eq.~\eqref{eq:fullTransfLemmaProof} indeed takes the form in eq.~\eqref{eq:symmTrans} with the matrices
\begin{equation}
\label{eq:fromOrthToGenParam}
    \bL_{\sigma}=\widetilde{\bL}_{\sigma}, \qquad \bM_{\sigma}=\widetilde{\bM}_{\sigma}+\widetilde{\bL}_\sigma\bGamma_2-\bGamma_1\widetilde{\bN}_{\sigma} ,\qquad \bN_{\sigma}=\widetilde{\bN}_{\sigma} \,,
\end{equation}
which yields precisely eq.~\eqref{eq:symmMatricesFromAut_sum} after inserting the expressions in eq.~\eqref{eq:LMNfromIsomorph}. To show that this transformation is indeed a symmetry transformation, we need some further properties of the matrices $\bL_{\sigma}$ and $\bN_{\sigma}$. To this end let us first establish a different interpretation of the matrices $\widetilde{\bL}_{\sigma},\widetilde{\bM}_{\sigma},\widetilde{\bN}_{\sigma}$.
To achieve this, observe that from Lemma \ref{lemmaIsomorphOnQ} it follows that there are matrices $\bZ\in\mathrm{GL}(L,\mathbb{Z}), \, \bS\in {\mathrm{GL}(E,\mathbb{Q})}$ and $\bW\in\mathbb{Q}^{L\times E}$ such that
\begin{equation}
    \bP_{\!\sigma}\bC_2^T=\bC_1^T\bZ, \qquad \bP_{\!\sigma}\bD_2^T=\bD_1^T\bS+\bC_1^T\bW \,.
\end{equation}
Inserting into eq.~\eqref{eq:LMNfromIsomorph}, we see that
\begin{equation}
    \widetilde{\bL}_{\sigma}=\bZ, \qquad \widetilde{\bM}_{\sigma}=\bW, \qquad \widetilde{\bN}_{\sigma}=\bS \,.
\end{equation}
In particular it immediately follows that 
\begin{equation}
    \det\bL_{\sigma}=\det\widetilde{\bL}_{\sigma}=\pm 1\,,
\end{equation}
since $\widetilde{\bL}_{\sigma}=\bZ\in\mathrm{GL}(L,\mathbb{Z})$. Furthermore we can show that $\bN_{\sigma}$ implements the `obvious' map on the external momenta.
\begin{lemma} 
The matrix $\bN_{\sigma}$ implements the map on the external momenta that is dictated by the operation of $\sigma$ on the graph.
\end{lemma}
\begin{proof}
    The idea to show the claim is to first construct a convenient edge-momentum parametrization where this is easy to see, and then show that the statement carries over to the orthogonal parametrization that we are considering here.
    
    Let us start with a pure loop momentum parametrization and use the statement of Lemma \ref{lemmaIsomorphOnQ} to write
    \begin{equation}
    \label{eq:lemma5,b}
    \bP_{\!\sigma}\hat{\bD}_2^T=\hat{\bD}{}_1^T\hat{\bS}+\bC_1^T\hat{\bW} \,,
    \end{equation}
    for some $\hat{\bS}\in {\mathrm{GL}(E,\mathbb{Q})},\,\hat{\bW}\in\mathbb{Q}^{L\times E}$, as above and $\hat{\bD}_i$ as in eq.~\eqref{edgeflow}. In the chosen parametrization, the basis elements of $\mathrm{Im}\,\hat{\bD}{}_i^T$  for $i=1,2$ have a clear interpretation as the paths $\gamma^{(i)}_v$ in the chosen spanning tree $T_i$ from an external vertex $v$ to the fixed vertex $v_0^{(i)}$. Let us now distinguish between $\sigma$ corresponding to a graph isomorphism or a Whitney twist. 

    If $\sigma$ is a Whitney twist, we choose the pure loop momentum parametrizations in such a way that the paths $\gamma_v^{(i)}$ are all in the untwisted components of the graph, which is possible due to Lemma \ref{lemmaWhitneyKin}. They are hence unaffected by the twist, and we necessarily have $\hat{\bS}=\mathds{1}$.

    Let us now assume that $\sigma$ is a graph isomorphism.
    Then we choose the pure loop-momentum parametrizations in such a way that the internal vertices connected to the fixed vertices $v_0^{(1)},v_0^{(2)}$ are mapped to each other via the vertex bijection of the graph isomorphism. 
    Then the edge bijection maps a path $\gamma_v$ in $T_2$ to a path $\gamma_{v'}$ in $T_1$ (up to loops), where $v,v'$ are external vertices such that the internal vertices connected to them are mapped to each other via the vertex bijection. Thus, we can identify $\hat{\bS}=\bP_{\!\beta}$ as a permutation matrix associated with a permutation $\beta\in S_E$ acting on the external vertices, i.e., the external momenta.\footnote{A priori, it acts on the internal vertices connected to external ones, but these are in one-to-one correspondence to the external vertices.}
    
    We have now established that in a particular parametrization we have $\hat{\bS}=\bP_{\!\beta}$  for some $\beta\in S_E$. Let us now translate this into the orthogonal parametrization that we need. We can achieve this by shifting the loop momenta by an appropriate combination of the external momenta and possibly rotating the loop and external momenta amongst each other. We can hence relate the matrices $\bD_i$ as defined in eq.~\eqref{eq:DiDef} to the matrices $\hat{\bD}_i$ appearing in the auxiliary pure loop momentum parametrization chosen above, via
    \begin{equation}
        \bD_i^T=(\hat{\bD}{}_i^T-\bPi_{C_i}\hat{\bD}_i^T)\bs{Y}_i^T \,.
    \end{equation}
    Here, the $\bs{Y}_i\in \mathrm{GL}(E,\mathbb{Z})$ are the matrices that implement the change of basis in the external momenta. We then find
    \begin{equation}
        \bP_{\!\sigma}\bD_2^T=\bD^T_1\bs{Y}_1^{-1T}\bP_{\!\beta}\bs{Y}_2^T+\bPi_{C_1}\hat{\bD}{}_1^T\bP_{\!\beta}\bs{Y}_2^T+\bC_1^T\hat{\bW}\bs{Y}_2^T-\bP_{\!\sigma}\bPi_{C_2}\hat{\bD}{}_2^T\bs{Y}_2^T \,.
    \end{equation}
    We can now see that every term on the right hand side but the first maps into $H_1(G_1,\mathbb{Z})$. Hence we can identify $\bS=\bs{Y}_1^{-1T}\bP_{\!\beta}\bs{Y}_2^T$. Thus 
    \begin{equation}
        \bN_{\sigma}=\widetilde{\bN}_{\sigma}=\bs{Y}_1^{-1T}\bP_{\!\beta}\bs{Y}_2^T \,,
    \end{equation}
    is precisely the matrix representing the action of the graph isomorphism (for a Whitney twist this is trivial) on the external edges expressed in the chosen parametrizations. Hence, $\bN_{\sigma}$ implements the `obvious' action on the external momenta, as claimed.
\end{proof}

Note that we can further conclude that the scalar products between the external momenta are invariant under the map implemented by the graph operation, i.e., by $\bN_{\sigma}$, which can be seen as follows. For a Whitney twist this is trivial since by Lemma \ref{lemmaWhitneyKin} it does not act on the external momenta. For a graph isomorphism we can argue as follows. If a 2-forest $(T_i,T_j)$ is mapped to some other 2-forest $(T_k,T_l)$ by the graph isomorphism, the squared momenta $s_{T_i,T_j}$ and $s_{T_k,T_l}$ associated with these 2-forests have to agree. This readily follows from the assumption that the $\cF$ polynomials of the two sectors are mapped into each other by (the ordinary permutation corresponding to) $\sigma$. However at the same time we have seen that the map on the external momenta, i.e., $\bN_\sigma$, is consistent with the mapping of the graph, so the equality $s_{T_i,T_j}=s_{T_k,T_l}$ implies the equality of the corresponding Mandelstam variable before and after applying the map on the external momenta. In other words, this combination of scalar products is invariant. Since the only kinematical dependence of the integrals comes from the $\cF$ polynomials, this shows that all scalar products that the integrals depend on are invariant.

Putting everything together it is now clear that we have indeed defined a symmetry transformation, and hence (the signed bijection associated to) every element of $\mathbb{Z}_2^c\times \mathbb{S}(\cG_1,\cG_2)$ gives rise to a symmetry transformation between the associated sectors, as claimed.

\section{A brief review of group and representation theory}\label{recapgrouptheory}
In this section we briefly review some basic facts about group and representation theory. We focus on the representation theory of symmetric groups. The content of this section is standard material and can be found in many textbooks, e.g., refs.~\cite{fultonHarris,fuchsSchweigert,saganBook}.

\subsection{Basic terminology}
Consider a group $G$ and a finite-dimensional linear representation $\bR$, i.e., a map $\bR: G\rightarrow \mathrm{GL}(V)$, where $V$ is some finite-dimensional vector space. In general this representation will be (completely) reducible, i.e., there exists a matrix $\bM$ such that
\begin{equation}
    \bM \bR(g) \bM^{-1}=\begin{pmatrix}
        \bR_1 & & \\ & \ddots & \\ & & \bR_k
    \end{pmatrix} \,,
\end{equation}
takes a block-diagonal form. We also write $\bR=\bR_1\oplus\dots \oplus \bR_k$. If such a rotation does not exist, the representation is irreducible. Irreducible representations are of high importance, because they form the building blocks of any representation. An important result about irreducible representations is \emph{Schur's lemma}. In the formulation that we need it, it states the following: Let $\bR_1:G\rightarrow \mathrm{GL}(V), \, \bR_2:G\rightarrow \mathrm{GL}(W)$ be two irreducible representations of a group $G$ acting on non-isomorphic vector spaces $V,W$. Then a $\dim V\times \dim W$ matrix $\bA$ satisfying
\begin{equation}
    \bR_1(g)\bA=\bA\bR_2(g) \,,
\end{equation}
for all $g\in G$ has to be trivial, i.e., $\bA=0$.

An important tool to check if a representation is irreducible or not are \emph{characters}. The character of a representation is the function 
\begin{equation}
    \chi_{\bR}:G\rightarrow \mathbb{C}, \quad g\mapsto \chi_{\bR}(g)=\mathrm{Tr}(\bR(g)) \,.
\end{equation}
The characters of the irreducible representations of a group are the \emph{irreducible characters}. An important property of characters is that they are class functions, i.e., they are constant on conjugacy classes
\begin{equation}
    \chi_{\bR}(hgh^{-1})=\chi_{\bR}(g) \,,\quad\forall g,h\in G\,.
\end{equation}
Thus, a character is fully determined by its values on one representative of every conjugacy class. 

Let us now focus on finite groups. Then there are finitely many conjugacy classes and equally-many irreducible characters. For finite groups, one can write down the \emph{character table}, which is a list of all irreducible characters and their values on the conjugacy classes.
We can define an inner product on characters by
\begin{equation}
    \langle \chi_1, \chi_2 \rangle= \frac{1}{|G|}\sum_{g\in G} \chi_1(g){\chi}_2(g)^* =\sum_{c\in\mathcal{C}_G} \frac{n_c}{|G|}\chi_1(g_c){\chi}_2(g_c)^* \,,
\end{equation}
where $\mathcal{C}_G$ is the set of conjugacy classes of $G$, and $n_c$ and $g_c$ are the cardinality and a representative of $c$, respectively.
Importantly, irreducible characters are an orthonormal basis for the vector space of class functions with respect to this inner product. In particular, a character $\chi_{\bR}$ is irreducible if and only if it satisfies
\begin{equation}
    \langle \chi_{\bR},\chi_{\bR} \rangle=1 \,.
\end{equation}
This constitutes a very simple criterion to check if a character is irreducible.

If a character $\chi_{\bR}$ is not irreducible, its representation $\bR$ can be decomposed into irreducible representations $\bR=\bR_1\oplus \dots \oplus \bR_k$, from which it follows that $\chi_{\bR}=\chi_{\bR_1}+\dots +\chi_{\bR_k}$. In general, some of the $\bR_i$ could be identical, and hence we generally have
\begin{equation}
    \chi_{\bR}=\sum_i m_i\chi_{\bR_i} \,,
\end{equation}
where $\chi_{\bR_i}$ are the irreducible characters. The coefficients can simply by projected out using the inner product
\begin{equation}
    m_i=\langle \chi_{\bR_i},\chi_{\bR}\rangle \,.
\end{equation}
This provides a simple way of finding the decomposition of a reducible representation into irreducible representations.

An important application of characters is the computation of the \emph{restriction} $\bR\big|_H$ of a representation $\bR$ of $G$ to a subgroup $H$, i.e., the map
\begin{equation}
    \bR\big|_H:H\rightarrow \mathrm{GL}(V),\quad  h\mapsto \bR(h) \,,
\end{equation}
where $V$ is the representation space of $\bR$. Even if $\bR$ is an irreducible representation of $G$, it might be reducible as a representation of $H$. Explicitly, let $\bR_i$ be an irreducible representation of $G$ and denote the irreducible representations of $H$ by $\bS_1,\dots ,\bS_r$. We then have
\begin{equation}
    \bR_i=\bigoplus_{j=1}^r \mu_{ij}\bS_j \,,
\end{equation}
for some appropriate coefficients $\mu_{ij}\in\mathbb{Z}_{\geq 0}$. Such decompositions are referred to as \emph{branching coefficients}. Using characters we can also write down an explicit formula for the coefficients $\mu_{ij}$, namely
\begin{equation}
\label{eq:muij}
    \mu_{ij}=\frac{1}{|H|}\sum_{h\in H} \chi_i^G(h){\chi}_j^H(h)^* \,,
\end{equation}
where $\chi_i^G$ and $\chi_j^H$ are the characters of $\bR_i$ and $\bS_j$, respectively.
\subsection{The representation theory of the symmetric groups}
In the main text we mostly deal with the symmetric groups $S_n$ and their representations. Here we briefly review their representation theory.

Generally, the conjugacy classes and the irreducible representations of $S_n$ are labeled by partitions of $n$,
\begin{equation}
    \lambda=(\lambda_1,\dots ,\lambda_k) \quad \text{with}\quad \lambda_1\geq \dots \geq\lambda_k \quad \text{and} \quad \lambda_1+\dots \lambda_k=n \,.
\end{equation}
All irreducible representations of symmetric groups are real and can be realized by integer matrices. The combinatorics of the partitions of $n$ is conveniently captured by Young diagrams, i.e., left-justified arrays of square boxes with $k$ rows of lengths $\lambda_1,\dots ,\lambda_k$. The dimension of a representation of $S_n$ corresponding to the Young diagram $\lambda$ is given by the hook length formula
\begin{equation}
    \dim \bR_{\lambda}= \frac{n!}{\prod_h \ell_h} \,.
\end{equation}
Here the product is taken over all hooks (entering the diagram from the bottom and exiting to the right) admitted by the diagram and $\ell_h$ denotes the respective hook lengths, i.e., the number of boxes the hook passes through.

For representations of the symmetric group $S_n$, it is also particularly simple to compute the branching rules for the restriction to $S_{n-1}$, which are captured by the \emph{Branching Theorem}. Explicitly, for $[\lambda]$ the representation associated with the partition $\lambda$ of $n$ we have
\begin{equation}
    \left. [\lambda] \right|_{S_{n-1}}=\sum_{A}[\lambda\setminus A] \,,
\end{equation}
where the sum is over all \emph{removable boxes} of $\lambda$, i.e., all hooks of length $1$.

As a example consider the group $S_3$. There are $3$ distinct partitions of $3$ labeled by the Young diagrams
\ytableausetup{boxsize=0.7em}
\begin{equation}
    \begin{ytableau} \ & \ & \  \end{ytableau}\,, \quad
    \begin{ytableau} \ & \ \\ \  \end{ytableau} \,, \quad
    \begin{ytableau} \ \\ \ \\ \  \end{ytableau}\,,
\end{equation}
and hence $3$ conjugacy classes and irreducible representations. The conjugacy classes correspond to the different cycle types and have representatives
\begin{equation}
    \mathds{1},\quad (12), \quad (123) \,.
\end{equation}
These have the cardinalities
\begin{equation}
    n_{[\mathds{1}]}=1, \quad n_{[(12)]}=3, \quad n_{[(123)]}=2 \,.
\end{equation}
The character table takes the form 
\ytableausetup{boxsize=0.35em}
\begin{center}
\begin{tabular}{ |c||c|c|c| } 
 \hline
 $S_3$ & $\mathds{1}$ & $(12)$ & $(123)$ \\ 
 \hline \hline
 \begin{ytableau} \ & \ & \  \end{ytableau} & 1 & 1 & 1   \\
 \hline
 \begin{ytableau} \ \\ \ \\ \  \end{ytableau} & 1 & -1 & 1   \\
 \hline
 \begin{ytableau} \ & \ \\ \  \end{ytableau} & 2 & 0 & -1  \\
 \hline
\end{tabular}
\end{center}
\ytableausetup{boxsize=normal}
Note that the value of the characters on the identity given in the first column, yield precisely the dimensions of the respective representations.

In the context of the equal-mass banana integral we further encounter the group $S_4$. Here there are $5$ distinct partitions and hence $5$ conjugacy classes and irreducible representations. They are labeled by the Young diagrams
\ytableausetup{boxsize=0.7em}
\begin{equation}
    \begin{ytableau} \ & \ & \ & \ \end{ytableau}\,, \quad
    \begin{ytableau} \ & \ \\ \ & \ \end{ytableau} \,, \quad
    \begin{ytableau} \ & \ & \ \\ \  \end{ytableau} \,, \quad
    \begin{ytableau} \ & \ \\ \ \\ \  \end{ytableau} \,, \quad
    \begin{ytableau} \ \\ \ \\ \ \\ \ \end{ytableau}\,.
\end{equation}
\ytableausetup{boxsize=normal}
The conjugacy classes correspond to the different cycle types with representatives
\begin{equation}
    \mathds{1},\quad [(12)(34)],\quad [(12)], \quad [(123)], \quad [(1234)] \,,
\end{equation}
and cardinalities
\begin{equation}
    n_{\mathds{1}}=1, \quad n_{[(12)(34)]}=3, \quad n_{[(12)]}=6, \quad n_{[(123)]}=8, \quad n_{[(1234)]}=6 \,.
\end{equation}
The character table then takes the form 
\ytableausetup{boxsize=0.35em}
\begin{center}
\begin{tabular}{ |c||c|c|c|c|c| } 
 \hline
 $S_4$ & $\mathds{1}$ & $(12)(34)$ & $(12)$ & $(123)$ & $(1234)$ \\ 
 \hline \hline
 \begin{ytableau} \ & \ & \ & \ \end{ytableau} & 1 & 1 & 1 & 1 & 1  \\
 \hline
 \begin{ytableau} \ \\ \ \\ \ \\ \ \end{ytableau} & 1 & 1 & -1 & 1 & -1  \\
 \hline
 \begin{ytableau} \ & \ \\ \ & \ \end{ytableau} & 2 & 2 & 0 & -1 & 0  \\
 \hline
 \begin{ytableau} \ & \ & \ \\ \  \end{ytableau} & 3 & -1 & 1 & 0 & -1  \\
 \hline
 \begin{ytableau} \ & \ \\ \ \\ \  \end{ytableau} & 3 & -1 & -1 & 0 & 1  \\
 \hline
\end{tabular}
\end{center}
\ytableausetup{boxsize=normal}

\section{Period matrix for \texorpdfstring{$F_2$}{F2}}
\label{appF1}
In this appendix, we present an explicit expression for the period matrix $\bP_{F_2}$ in example~\ref{ex2F}. The matrices  $\bs{\mathcal{T}}_c$ and $\bT_h$  transform from the bases in refs.~\cite{goto2013,Abreu_2020} to the ones defined in eqs.~\eqref{f2basis} and~\eqref{hombasisf2}. The period matrix is given by 
\begin{align}
\bP_{F_2}=\bs{\mathcal{T}}_c\,\widetilde{\bP}^{(F_2)}\bs{\mathcal{T}}^T_h\,,
\end{align} where $\widetilde{\bP}_{F_2}$ is the period matrix in terms of $F_2$ functions derived in ref.~\cite{goto2013}:
\begin{align} 
\label{eq:periodMatF2}
&\widetilde{\bP}^{(F_2)}_{ij}=(-1)^{i+1}\frac{1}{4} (2 \eps +1)^2\,\tilde{\mathcal{F}}^{(j)}_2\left(\eps +\frac{1}{2},\bx_i,-2 \eps ,-2 \eps\right)\,.\end{align}
In each row the argument changes according to 
\begin{align}
\bx_1=\left(-\eps -\frac{1}{2},-\eps -\frac{1}{2} \right)\,,&\quad\bx_2=\left(\frac{1}{2}-\eps ,-\eps -\frac{1}{2}\right)\notag\,,\\
\bx_3=\left(\frac{1}{2}-\eps ,\frac{1}{2}-\eps\right)\,,&\quad\bx_4=\left(\frac{1}{2}-\eps ,\frac{1}{2}-\eps\right)\,,
\end{align}
and the column changes the function 
\begin{align}
\widetilde{\mathcal{F}}^{(1)}_2(\alpha,\beta,\beta',\gamma,\gamma')&=\frac{\Gamma (\beta ) \Gamma (\beta') \Gamma (\gamma -\beta ) \Gamma (\gamma'-\beta') }{\Gamma (\gamma ) \Gamma (\gamma')}F_2(\alpha ;\beta ,\beta';\gamma ,\gamma';x,y)\,,\notag\\
\widetilde{\mathcal{F}}^{(2)}_2(\alpha,\beta,\beta',\gamma,\gamma')&=-\frac{\Gamma (1-\alpha ) \Gamma (\beta') \Gamma (\gamma -1) x^{1-\gamma } \Gamma (\gamma'-\beta') \exp (i \pi  (\beta -\gamma )) }{\Gamma (\gamma') \Gamma (\gamma -\alpha )}\notag\\
&\times F_2(\alpha -\gamma +1;\beta -\gamma +1,\beta';2-\gamma ,\gamma';x,y)\,,\notag\\
\widetilde{\mathcal{F}}^{(3)}_2(\alpha,\beta,\beta',\gamma,\gamma')&=-\frac{\Gamma (1-\alpha ) \Gamma (\beta ) \Gamma (\gamma'-1) y^{1-\gamma'} \Gamma (\gamma -\beta ) \exp (i \pi  (\beta'-\gamma')) }{\Gamma (\gamma ) \Gamma (\gamma'-\alpha )}\notag\\
&\times F_2(\alpha -\gamma'+1;\beta ,\beta'-\gamma'+1;\gamma ,2-\gamma';x,y)\,, \notag\\
\widetilde{\mathcal{F}}^{(4)}_2(\alpha,\beta,\beta',\gamma,\gamma')&=\frac{\Gamma (1-\alpha ) \Gamma (\gamma -1) \Gamma (\gamma'-1) x^{1-\gamma } y^{1-\gamma'} \exp (i \pi  (\beta +\beta'-\gamma -\gamma')) }{\Gamma (-\alpha +\gamma +\gamma'-1)}\notag\\
&\times F_2(\alpha -\gamma -\gamma'+2;\beta -\gamma +1,\beta'-\gamma'+1;2-\gamma ,2-\gamma';x,y)\,,\notag
\end{align}
with the transformations to $\bP_{F_2}$ 
\beq
\bs{\mathcal{T}}_c=\frac{x y (6 \eps +1) (10 \eps +3)}{4 (x+y-1)}\widetilde{\bs{\mathcal{T}}}_c\,, 
\eeq
with
\begin{align}
\notag\widetilde{\bs{\mathcal{T}}}_c&=\left(
\begin{array}{cccc}
 \frac{2 y (4 \eps +1)}{(2 \eps +1) (x+y-1)} & -\frac{2 (x-1) y (4 \eps +1)}{(2 \eps +1) (x+y-1)} & y & \frac{(y-1) y}{x+y-1} \\
 \frac{2 x (4 \eps +1)}{(2 \eps +1) (x+y-1)} & \frac{(x-1) x}{x+y-1} & x & -\frac{2 x (y-1) (4 \eps +1)}{(2 \eps +1) (x+y-1)} \\
 0 & 0 & -\frac{2 x y (6 \eps +1) (10 \eps +3)}{(2 \eps +1)^2 (x+y-1)} & 0 \\
 \frac{2 \eps  (x+y+3)+x+y+1}{(10 \eps +1) (x+y-1)} & \frac{(x-1) (x (8 \eps +2)-2 (y+3) \eps -y-1)}{(10 \eps +1) (x+y-1)} & \frac{8 \eps  (x+y)+2 x+2 y-6 \eps -1}{10 \eps +1} & -\frac{(y-1) (2 \eps  (x-4 y+3)+x-2 y+1)}{(10 \eps +1) (x+y-1)} \\
\end{array}
\right)\,,\\
\bs{\mathcal{T}}_h&=\left(
\begin{array}{cccc}
 \frac{1}{4} & \frac{1}{2} & \frac{1}{2} & 1 \\
 -\frac{1}{4} & \frac{1}{2} & -\frac{1}{2} & 1 \\
 \frac{1}{4} & -\frac{1}{2} & -\frac{1}{2} & 1 \\
 -\frac{1}{4} & -\frac{1}{2} & \frac{1}{2} & 1 \\
\end{array}
\right)\,.
\end{align}
\section{Some algebraic topology}
\subsection{Simplicial homology}\label{simplicialhomology}
In the following we review some basics about simplicial (co)homology that are relevant for the proof of Proposition~\ref{prop:twistedL_TopL}. Introductory material can be found, e.g., in ref.~\cite{hatcherAlgTop}.

The simplical $k$-chain group $C_k(X)$ of a triangulable space\footnote{A topological space $X$ is triangulable if it admits a homeomorphism to a simplicial complex.} $X$ is the free $\mathbb{Z}$-module 
\begin{align}
C_k(X):=\bigoplus_{\sigma_k^{(i)} \in C_X}\mathbb{Z}\,\sigma_k^{(i)}=\left\{   \sum_{i} n_i\,\sigma_k^{(i)}\,,\quad n_i\in \mathbb{Z}\right\}\,,
\end{align}
where $\sigma_k$ is a $k$-simplex of the associated simplicial complex $C_X$.
The \textit{chain complex} is then 
\begin{align}
    \dots \overset{\delta_{k+1}}{\to} C_k(X)\overset{\delta_{k}}{\to}C_{k-1}(X) \overset{\delta_{k-1}}\to \dots\,,
\end{align}
where the boundary maps $\delta_k:C_k(X)\rightarrow C_{k-1}(X)$, satisfy $\delta_{k-1}\circ \delta_k=0$. 
The $k^{\textrm{th}}$ simplicial homology group is defined as
    \begin{align}
H_k(X):=\frac{\ker(\delta_k:C_k(X)\to C_{k-1}(X))}{\Ima(\delta_{k+1}:C_{k+1}(X)\to C_{k}(X))}=\frac{Z_k(X)}{B_{k}(X)}\,.
    \end{align}
With coefficients in $\mathbb{C}$ there is an isomorphism
\begin{align}\label{decompCX}
   C_k(X)\simeq B_k(X)\oplus H_k(X)\oplus B_{k-1}(X)\,.
\end{align}
Now denote by $b_k^{(i)}$ a basis of $B_k(X)$: 
\begin{align}\label{basisBk}
  B_k(X)=\left\langle  b_k^{(i)},\dots , b_k^{(d_{B_k})}\right\rangle_\mathbb{C}\,,
\end{align}
where $d_{B_k}$ is the dimension of the $\mathbb{C}$-vector space $B_k(X)$. And similarly 
\begin{align}
  B_{k-1}(X)=\left\langle  b_{k-1}^{(i)},\dots , b_{k-1}^{(d_{B_{k-1}})}\right\rangle_\mathbb{C}\,,\quad  H_k(X)=\left\langle  h_k^{(i)},\dots , h_k^{(d_{H_k})}\right\rangle_\mathbb{C}\,,
\end{align}
with $h_k^{(i)}\in Z_k(X)$, where $d_{B_{k-1}}$ and $d_{H_k}$ are the dimensions of $  B_{k-1}(X)$ and $ H_k(X)$ respectively. 
Now we want to lift the basis of $ B_{k-1}(X)$ to a basis of some isomorphic vector space $\widetilde{B}_{k}(X) \simeq B_{k-1}(X)$ of $k$-simplices. Therefore, we use that $\delta_{k}:C_k(X)\to B_{k-1}(X)$ is surjective, such that there exist $\tilde{b}_{k}^{(i)} \in C_{k}(X)$ with $\delta_{k}(\tilde{b}_{k}^{(i)})=b_{k-1}^{(i)} $, for all $i$. Then define
\begin{align}\label{isombasisB}
\tilde{B}_k(X):=\left\langle \tilde{b}_k^{(1)},\dots ,\tilde{b}_{k}^{(d_{B_{k-1}})}\right \rangle\simeq B_{k-1}(X)\,.
\end{align}
Then note that the map $\delta_{k}|_{\tilde{B}_k(X)}:\tilde{B}_k(X)\to B_{k-1}(X)$ is injective, as $\ker(\delta_{k}|_{\tilde{B}_k(X)})=\{0\}$.
Then
\begin{align}\label{basisCX}
C_k(X)=   \left\langle b_k^{(1)},\dots,b_{k}^{(d_{B_{k}})},h_k^{(1)},\dots,h_{k}^{(d_{H_k})},\tilde{b}_k^{(1)},\dots,\tilde{b}_{k}^{(d_{B_{k-1}})} \right\rangle\,.
\end{align}
\subsubsection{Induced maps and the simplicial Lefschetz number}\label{app:inducedmaps}
\label{inducedmap}
A map $f:X\to Y$ between triangulable spaces $X$ and $Y$ induces maps $f_*:C_n(X)\to C_n(Y)$ by simply acting on respective $n$-simplices. Note that in particular any such map commutes with the boundary operator, $f_*\circ \partial=\partial \circ f_*$ (cf.~e.g.,~ref.~\cite[p.111. Proposition 2.9, Proposition 2.10]{Hatcher2002}).
Further, a chain map between chain complexes induces homomorphisms between the homology groups of the two complexes and they induce the same homomorphisms if and only if two maps are homotopic.

For a simplicial complex, the Lefschetz number of a map $f:X\to X$ is determined by the induced maps on $C_k(X)$. This can be seen in the following way.
We use the isomorphism in eq.~\eqref{decompCX},
to write
any $c_k\in C_k(X)$ as 
\begin{align}
c_k=  \sum^{d_{B_{k}}}_{i=1} \alpha_i b_k^{(i)}+ \sum_{j=1}^{d_{H_k}}\beta_j h_{k}^{(j)}+\sum_{l=1}^{d_{B_{k-1}}}\gamma_l \tilde{b}_k^{(l)}\,.
\end{align}
Note that $f$ acts one these subspaces as 
\begin{align}
 f( B_k(X))&\subseteq B_k(X)\,,\quad    f(\tilde{H}_k(X))\subseteq \tilde{H}_k(X)\oplus B_k(X)\,,\\
  f(\tilde{B}_k(X))&\subseteq \tilde{B}_k(X)\oplus Z_k(X)\simeq  \tilde{B}_k(X)\oplus B_k(X)\oplus  \tilde{H}_k(X)\notag\,.
\end{align}
Then the induced map on $ C_k(K_\delta)$ is lower block triangular. Its trace is determined by the diagonal blocks
\begin{align}
B_k(f,X)&:B_{k}(X)\to B_k(X)\,, \quad \tilde{B}_k(f,X):\tilde{B}_{k}(X)\to \tilde{B}_{k}(X)\,, \\
\tilde{H}_{k}(f,X)&:\tilde{H}_{k}(X)\to \tilde{H}_{k}(X)\,.\notag
\end{align}
Then we notice that 
\begin{align}
    \sum_{k=0}^{2n}&(-1)^k \text{Tr}(C_k(f,X))
  \\
\notag    &  =\text{Tr}(C_0(f,X))+ \sum_{k=1}^{2n}(-1)^k \Big[ \text{Tr}(B_{k-1}(f,X))+ \text{Tr}\left(\tilde{H}_{k}(f,X)\right)+\text{Tr}\left( B_k(f,X)\right)   \Big]\,.
\end{align}
Additionally, notice that for $k=0$ we have $   C_0(X)=Z_0(X)\simeq B_0(X)\oplus H_0(X) $.
Then the trace of the induced chain map decomposes as 
\begin{align}
    \sum_{k=0}^{2n}&(-1)^k \text{Tr}(C_k(f,X))
\notag  \\
    &  =L(f,X)+ \sum_{k=1}^{2n}(-1)^k \Big( \text{Tr}\left(B_{k-1}(f,X)\right)+\text{Tr}\left( B_k(f,X)\right)   \Big)+\text{Tr}\left( B_0(f,X)\right)
      \\
    &  =L(f,X)-\sum_{k^\prime=0}^{2n}(-1)^{k^\prime} \text{Tr}\left(B_{k^\prime}(f,X)\right)+\sum_{k=0}^{2n}(-1)^k\text{Tr}\left( B_k(f,X)\right)\,.  \notag 
\end{align}
Finally, 
\begin{align}\label{inducedC}
      \sum_{k=0}^{2n}(-1)^k \text{Tr}(C_k(f,X))=L(f,X) \,.
\end{align}

For the proof of Lemma~\ref{prop:twistedL_TopL} a particular type of induced map is relevant, namely a
\textit{simplicial map}. If $K$ and $L$ are simplicial complexes, then a simplicial map is a map $f:K\to L$ such that each vertex $v_i\in K$ is sent to a vertex $f(v_i)\in L$ and each simplex $[v_0,\dots, v_k]\in K$ is mapped to a simplex $f([v_0,\dots, v_k])=[f(v_0),\dots, f(v_k)]\in L$. 
For a bijective map $f:K\to K$ this in particular implies that the induced map on the simplices is a signed permutation.
For any simplicial complex, there exists a triangulation such that any map $f$ becomes homotopic to a simplicial map. This is due to the simplicial approximation theorem. 
\begin{theorem}[Simplicial Approximation Theorem]  (cf.~e.g., ref.~\cite[Theorem 2C.1]{hatcherAlgTop})\label{simplicialapproxtheorem} Let $K$ be a finite simplicial complex, and consider some map  $f:K\to K$. For some $r\in \mathbb{N}$, there exists a simplicial map $\phi:K^{r}\to K$ such that $\phi$ is homotopic to $f$.  Here $K^r$  is the $r^{\textrm{th}}$ barycentric subdivision.
\end{theorem} 
The \textit{barycenter} of a $k$-simplex $\Delta_k=\sum_{i=0}^n \lambda_i v_i$ is 
\begin{align}
b_{[v_0,\dots , v_k]}=\frac{1}{1+k}\sum_{l=0}^kv_l\,.
\end{align} 
and a \textit{barycentric subdivision} is a triangulation such that the vertices are the barycenters of the simplices in $K$. Note that in particular a vertex is a $0$-simplex with itself as its barycenter, such that the vertices of $K$ are still vertices of $K^r$.

\subsection{Various dualities and isomorphisms}\label{dualities}
In the following, we summarize some isomorphisms between different types of cohomology groups. 
First, as a consequence of the universal coefficients theorem, the singular cohomology group $H^k(X,\mathbb{C})$ with coefficients in $\mathbb{C}$ is dual to the singular homology group (cf.~e.g., ref.~\cite[p.198]{hatcherAlgTop}) 
\begin{align}\label{singularduality}
  H^k(X,\mathbb{C}) \simeq H_k(X,\mathbb{C})^\vee\,,
\end{align}
for $X$ a topological space.
Further, because of the non-degenerate period pairing, there is an isomorphism between de Rham and singular cohomology (see, e.g., ref.~\cite[appendix D, Theorem D.3.2]{Conlon2001}) 
\begin{align}\label{deRhamduality}
    H^k_\text{dR}(X,\mathbb{C})\simeq H^k(X,\mathbb{C})\,.
\end{align}
For $X$ be an orientable, real manifold of dimension $2 n$, Poincaré duality relates de Rham cohomology to its compactly supported version
(see, e.g., ref.~\cite[p.46]{BottTu1982}) 
\begin{align}\label{poicaredualitycomplex}
    H_{\text{dR}}^k(X,\mathbb{C})\simeq H_{\text{dR},c}^{2n-k}(X,\mathbb{C})^\vee \,.
\end{align}
Finally, if $X$ is a compact, oriented, real $2n$-dimensional closed manifold without boundary and $A\subset X$ a closed subset, then there are isomorphisms (cf.~e.g., ref.~\cite{brown2025positivegeometriescanonicalforms}) 
\begin{align}
   H_{\mathrm{dR},c}^{k}(X\backslash A,\mathbb{C})\simeq H^{k}_{\mathrm{dR}}(X,A,\mathbb{C})\,,\quad   H_{k}(X\backslash A)\simeq H_{k}^\text{lf}(X,A)\,.
\end{align}
\end{appendix}

\bibliographystyle{JHEP}
\bibliography{biblio.bib}

\end{document}